\documentclass[11pt]{article}
\pdfoutput=1

\usepackage[utf8]{inputenc}
\usepackage{multirow}
\usepackage{amsmath, amsfonts, amssymb}
\usepackage{comment}
\usepackage{graphicx}
\usepackage{float}
\usepackage{psfrag}
\usepackage{amsthm}
\usepackage{transparent}
\usepackage{caption}
\usepackage{mathtools}
\usepackage{diagbox}
\usepackage{subcaption}
\usepackage{empheq}
\usepackage{lscape}
\usepackage{faktor}

\usepackage{import}
\usepackage{mathrsfs}
\usepackage{bm}
\usepackage[dvipsnames,svgnames,table]{xcolor}
\usepackage{enumerate}
\usepackage{arydshln}
\usepackage{soul}
 \usepackage{slashed}
 \usepackage{adjustbox}
\usepackage{mathrsfs}
\usepackage{mathalfa}
\usepackage{makecell}
\usepackage{longtable}
\allowdisplaybreaks
 \usepackage{a4wide}
  \usepackage{tikz}
  \usepackage{tikz-cd}
  \usetikzlibrary{shapes.geometric,arrows}
  \usepackage[most]{tcolorbox}
\tikzstyle{startstop} = [rectangle, rounded corners, minimum width=3cm, minimum height=1cm,text centered, draw=black, fill=red!30]
\tikzstyle{io} = [trapezium, trapezium left angle=70, trapezium right angle=110, minimum width=3cm, minimum height=1cm, text centered, draw=black, fill=blue!30]
\tikzstyle{process} = [rectangle, minimum width=3cm, minimum height=1cm, text centered, draw=black, fill=orange!30]
\tikzstyle{decision} = [diamond, minimum width=3cm, minimum height=1cm, text centered, draw=black, fill=green!30]
\tikzstyle{arrow} = [thick,->,>=stealth]
\tikzcdset{row sep/normal=0.75cm}

\usepackage[labelformat=simple]{subcaption}    
  \usepackage{color}
  \definecolor{dark-gray}{gray}{0.20}
  \definecolor{gray}{gray}{0.30}
  \definecolor{light-gray}{gray}{0.80}
  \definecolor{dark-red}{rgb}{0.7,0,0}
  \definecolor{dark-green}{rgb}{0,0.5,0}
  \definecolor{dark-blue}{rgb}{0.3,0.3,0.7}
  \definecolor{light-blue}{rgb}{0.8,0.8,1}
  \definecolor{swamp}{RGB}{240, 199, 197}
  \definecolor{myred}{RGB}{242,0,0}
  \definecolor{myyellow}{RGB}{255,171,0}
  \definecolor{myblue}{RGB}{13,117,255}
    \definecolor{mygreen}{RGB}{0,128,102}
  \definecolor{mypurple}{RGB}{156,82,242}

  \usepackage{pifont}

\usepackage{newunicodechar} 

\usepackage{setspace}

\newcommand{\be}{\begin{equation}}
\newcommand{\ee}{\end{equation}}

\def\be{\begin{equation}}
\def\ee{\end{equation}}
\def\bea{\begin{eqnarray}}
\def\eea{\end{eqnarray}}

\newcommand{\beq}{\begin{equation}}  \newcommand{\eeq}{\end{equation}}
\newcommand{\bal}{\begin{aligned}}   \newcommand{\eal}{\end{aligned}}
\def\beqa{\begin{eqnarray}}
\def\eeqa{\end{eqnarray}}

\newcommand{\dd}{\mathrm{d}}

\newcommand{\I}{\text{Im}\,}

\DeclareMathOperator{\rank}{rank}

\usepackage{stackengine}
\usepackage{calc}
\newlength\shlength
\newcommand\vv[2][0]{\setlength\shlength{#1pt}%
  \stackengine{-5.6pt}{$#2$}{\smash{$\kern\shlength%
    \stackengine{7.55pt}{$\mathchar"017E$}%
      {\rule{\widthof{$#2$}}{.57pt}\kern.4pt}{O}{r}{F}{F}{L}\kern-\shlength$}}%
      {O}{c}{F}{T}{S}}

\newcommand{\C}{\mathbb{C}}

\newcommand{\abs}[1]{\lvert #1 \rvert}

\def\simleq{\; \raise0.3ex\hbox{$<$\kern-0.75em
      \raise-1.1ex\hbox{$\sim$}}\; }
   \def\simgeq{\; \raise0.3ex\hbox{$>$\kern-0.75em
      \raise-1.1ex\hbox{$\sim$}}\; }

\numberwithin{equation}{section}

\usepackage{jheppub}
\usepackage{hyperref}
\usepackage{cleveref}

\hypersetup{
	colorlinks=true,
	linkcolor=dark-blue,
	citecolor=dark-red,
	urlcolor=dark-green,
	linktoc=page
}
\newtheorem{thm}[equation]{Theorem}

\newtheorem{cor}[equation]{Corollary}
\theoremstyle{remark}

\crefname{appendix}{Appendix}{Appendices}

\title{\centering Bordisms between 9d type IIB supergravities and commutator widths of duality groups}

\author[a]{Camilo las Heras}

\affiliation[a]{Instituto de F\'{i}sica y Astronom\'{i}a, Universidad de Valpara\'{i}so, A. Gran Bretaña 1111, Valpara\'{i}so, Chile}
\emailAdd{camilo.lasheras@uv.cl}

\author[b]{and Ignacio Ruiz}

\affiliation[b]{Theoretical Physics Department, CERN, 1211 Geneva 23, Switzerland}
\emailAdd{ignacio.ruiz.garcia@cern.ch}

\preprint{CERN-TH-2026-107}

\abstract{
We study the topological properties of bordisms interpolating between different 9d gauged supergravities obtained from compactification of type IIB string theory on $\mathbb{S}^1$ with a non-trivial $\mathsf{SL}(2,\mathbb{Z})$ bundle. We describe how such bordisms implement the needed monodromies through stacks of $[p,q]$ 7-branes or gravitational solitons of non-trivial topology. For the later mechanism, we see that the topology of the bordism becomes increasingly complicated for large monodromies, which results in the associated bordisms being arbitrarily suppressed, against expectations on the breaking of global symmetries in Quantum Gravity. Motivated by this, we propose a refinement of the Swampland Cobordism Conjecture for the first bordism group $\Omega_1({\rm B}G)$ with a $G$ duality bundle. We argue that even if gravitational solitons can realize the monodromies associated with elements of the commutator subgroup of $G$, if the number of needed commutators is unbounded (in other words, the \textit{commutator width} of $G$ diverges) then an infinite number of duality defects realizing elements in $G$ need to be included. We test this proposal for different duality groups $G$, and see that our expectations are realized, often in non-trivial ways.
}

\begin{document}
\hypersetup{pageanchor=false}
\makeatletter
\let\old@fpheader\@fpheader

\makeatother
\maketitle

\newcommand{\remove}[1]{\textcolor{red}{\sout{#1}}}

\pagenumbering{roman}
\newpage
\pagenumbering{arabic}
\setcounter{page}{1} 
\section{Introduction and motivation}
\label{sec:intro}

Consistency of the bosonic worldsheet quantization result in a single critical bosonic string theory in $D=26$ dimensions. Analogous considerations for the supersymetric worldsheet in $D=10$ result in two $\mathcal{N}=2$ theories (type IIA and type IIB), three $\mathcal{N}=1$ (type I and heterotic $\mathsf{SO}(32)$ and $\mathsf{E}_8\times\mathsf{E}_8$) as well as five non-tachyonic $\mathcal{N}=0$ theories ($\mathsf{USp}(32)$ or Sagnotti string, $\mathsf{U}(32)$ or type O$'$B and $\mathsf{SO}(16)^2$ heterotic), and nine additional tachyonic theories (type 0B, type 0A, an orientifold of the later and six additional heterotic theories). Additionally, there is a unique 11d $\mathcal{N}=1$ Supergravity, which is expected to describe the low-energy dynamics of M-theory. This seems to point to a quite reduced set of 10d theories, and while naively one would expect them to form disconnected components in the QG landscape (as opposed to the unique bosonic string theory), there exists a rich web of dualities connecting the different theories \cite{Font:1990gx,Sen:1994fa,Sen:1994wr,Hull:1994ys,Witten:1995ex,Townsend:1995kk,Horava:1995qa,Horava:1996ma,Sen:1998kd}.\footnote{See \cite{Bergman:1997rf,Blum:1997cs,Blum:1997gw,Kachru:1998yy,Blumenhagen:1999ad,Bergman:1999km,Fabinger:2000jd,Dudas:2001wd,Bossard:2024mls,Fraiman:2025yrx,Baykara:2026gem,Altavista:2026evd,Baykara:2026vdc} for additional efforts trying to connect some of the tachyonic non-supersymetric theories to this duality web.}

However, this small number of theories in $D=10$ grows spectacularly once we consider the effective field theories resulting from dimensionally reducing them along compact manifolds $X_n$ down to $d=10-n$. Not only does the number of available manifolds dramatically increase with $n$, but also the possibility of adding brane arrangements and fluxes along these compact dimensions results in an astoundingly large \emph{Landscape} of potential theories \cite{Susskind:2003kw,Douglas:2003um,Taylor:2015xtz,Baykara:2025nnc}. As we will review in this work, already in $d=9$ (down to which we can only compactify the 10d theory on $\mathbb S^1$ or $\mathbb S^1/\mathbb{Z}_2$) there is a phenomenologically rich set of theories that are built exploiting the $\mathsf{SL}(2,\mathbb{Z})$ self-duality group of 10d Type IIB string theory.\\

This exceptionally large number of EFTs derived from higher-dimensional String/M-theory led to \emph{Landscape problem}, used by critics of String Theory to argue that any apparently consistent low-energy theory could probably be obtained from some string compactification, what would result in string theory not having any predictive power on the real world. However, as pointed out in \cite{Vafa:2005ui}, this is not the case at all! As it turns out, the set of EFTs weakly coupled to Einstein gravity with a QG completion (i.e., the \emph{Landscape}) is much smaller than those not admitting such completion (and said to be in the \emph{Swampland}). Determining the precise set of constrains a QG completion (be it String/M-theory or some other hypothetical theory) is the goal of the \emph{Swampland Program} \cite{Vafa:2005ui,Brennan:2017rbf,Palti:2019pca,vanBeest:2021lhn,Grana:2021zvf,Harlow:2022ich,Agmon:2022thq,VanRiet:2023pnx}. Said restrictions that are expected to be satisfied by a consistent theory of gravity are known as \emph{Swampland constrains}, forming a wide net of interconnected conjectures with various levels of concreteness and phenomenological implications.

One of the most established Swampland conjectures is the \emph{No Global Symmetries conjecture} \cite{Banks:2010zn} (see \cite{Hawking:1975vcx,Zeldovich:1976vq,Zeldovich:1977be,Banks:1988yz,Kallosh:1995hi} for earlier historical developments), which postulates that a consistent theory of Quantum Gravity cannot have global symmetries, as these must be either broken or gauged. This conjecture has been extensively tested (and in some particular settings, even proven \cite{Polchinski:1998rr,Harlow:2018fse,McNamara:2019rup}), and is probably the Swampland conjecture that enjoys the largest support and evidence. One of said symmetries is known as \emph{bordism class}. One expected feature in Quantum Gravity is spacetime topology change \cite{Giddings:1987cg,Greene:1990ud,Aspinwall:1993yb,Witten:1993yc,Giveon:1994fu,Strominger:1995cz,Greene:1995hu,Witten:1998zw,Atiyah:2000zz,Greene:2000yb,Acharya:2004qe,Garcia-Etxebarria:2018ajm,Demulder:2023vlo}, such as that of the compact $n$-manifold $\mathcal{X}_n$ on which we compactify. The mathematical objects that tell us whether two $n$-manifolds $\mathcal{X}_n$ and $\mathcal{Y}_n$ can be connected to one another through a series of dynamically allowed topology changes are known as \emph{cobordism classes} $\Omega_n^{\mathsf{g}}({\rm B}G)$,\footnote{Here $\mathsf{g}$ refers to a given structure on the tangent bundle, and ${\rm B}G$ denotes the \emph{classifying space} \cite[Section 2.3]{HatcherTOPOLOGY} of the group $G$, which for discrete groups simply corresponds to a path-connected topological space with $\pi_1({\rm B}G)\simeq G$ and trivial higher homotopy groups. If such group is trivial then ${\rm B}G={\rm pt}$.} which consists in the set of equivalence classes defined by those $n$-dimensional $\mathsf{g}$-manifolds $\mathcal{X}_n$ and $\mathcal{Y}_n$ with a (possibly non-trivial) $G$-bundle for which there exists a $n+1$-dimensional $\mathsf{g}$-manifold $\mathcal{B}_{n+1}$ equipped with a $G$-bundle with $\partial\mathcal{B}_{n+1}=\mathcal{X}_n\sqcup\overline{\mathcal{Y}}_n$. The $\mathsf{g}$-structure and $G$-bundle of $\mathcal{B}_{n+1}$ must be such that, restricted to its boundary, recovers those of $\mathcal{X}_n$ and $\mathcal{Y}_n$. As depicted in Figure \ref{fig.cobord}, from the lower $d=D-n$ dimensional EFT point of view, moving across the bordism is equivalent to crossing a codimension-one domain wall connecting two different $d$-dimensional theories (or nothing, if $\mathcal{Y}_n$ is the empty set). Since we can understand the different classes of $\Omega_n^{\mathsf{g}}({\rm B}G)$ as a global charge (as this could not change under allowed topology changes), which would be in tension with the \emph{No Global Symmetries conjecture} mentioned above, the \emph{Swampland Cobordism Conjecture} proposes

\begin{tcolorbox}[enhanced jigsaw,breakable,pad at break*=1mm,colback=green!5!white,colframe=Green,title={Swampland Cobordism Conjecture \cite{McNamara:2019rup}}]
For any $D$-dimensional consistent theory of Quantum Gravity compactified on a $n$-dimensional internal manifold, all cobordism classes must vanish,
\begin{equation}\label{eq.cobord}
    \Omega^{\rm QG}_n=0\quad\forall\,n\,,
\end{equation}
as otherwise this can be associated with a $(D-n-1)$-global symmetry given by the bordism class in $\Omega^{\rm QG}_n$. This means that every compactification of any consistent theory of gravity should be able to be dynamically connected to \emph{nothing}, i.e., the trivial theory without any fields or spacetime.\\
In general, since it is not known what the full non-perturbative structure of QG is, the procedure is to approximate such structure by some $\mathsf{g}$, in such a a way that if $\Omega_n^{\mathsf{g}}\neq 0$, we either need to refine the spacetime structure, so that $\Omega_n^{\mathsf{g}'}=0$, or include the needed defect that kill said cobordism classes, $\Omega_n^{\mathsf{g}+\rm defects}=0$.
\end{tcolorbox}
There has been extensive work in trying to compute the cobordism groups for different known theories derived from String theory and M-theory, which in turn has resulted in deeper insights about the properties of these theories, as well as in the prediction of new (often non-supersymmetric) defects \cite{McNamara:2019rup,Montero:2020icj,Dierigl:2022reg,Kaidi:2023tqo,Kaidi:2024cbx,Fukuda:2024pvu,Heckman:2025wqd}.

    \begin{figure}[ht]
\begin{center}
\begin{subfigure}[b]{0.55\textwidth}
\center
\includegraphics[width=0.90\textwidth]{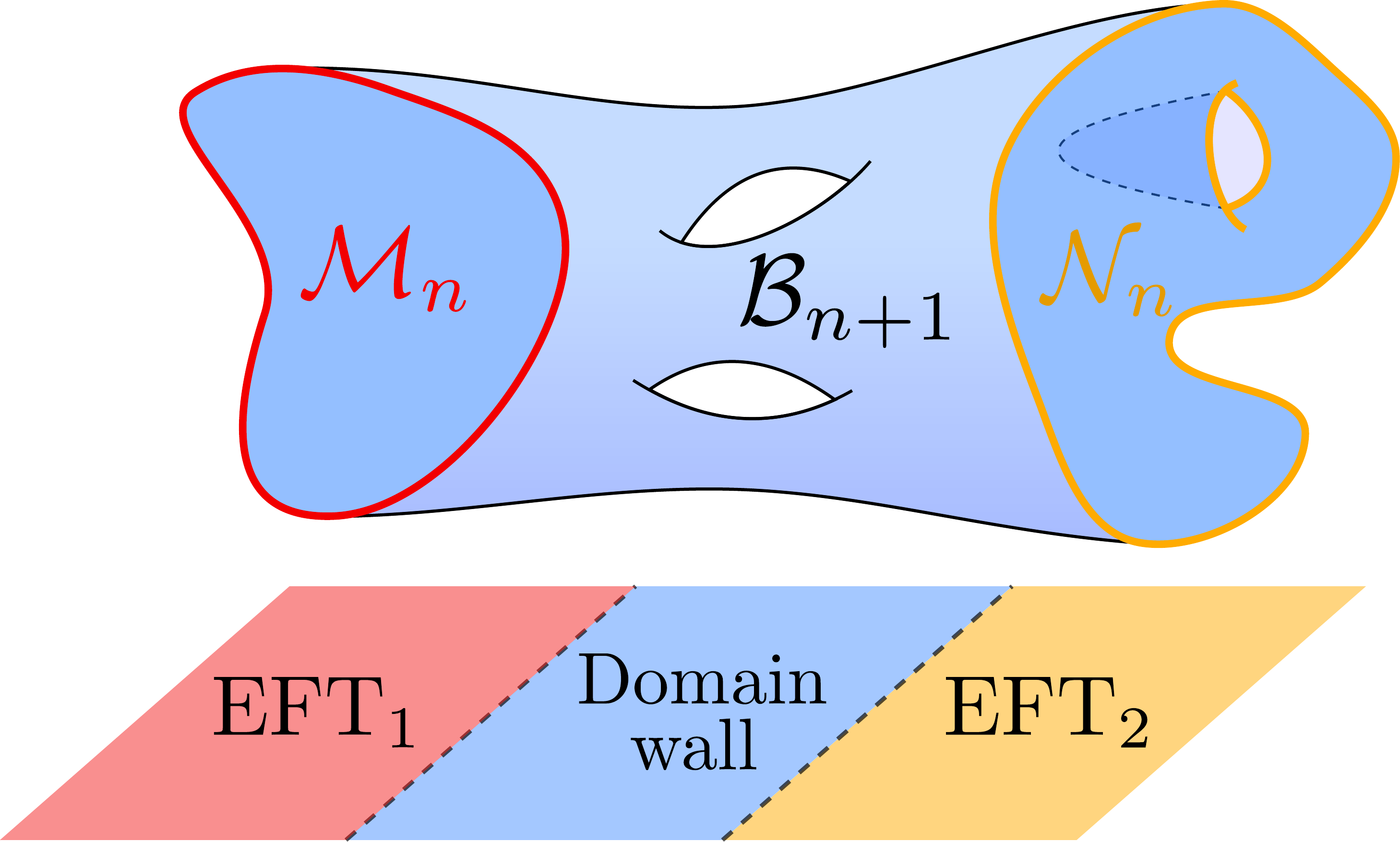}
\caption{\hspace{-0.3em}) Cobordism between two manifolds.} \label{ff.cob1}
\end{subfigure}
\begin{subfigure}[b]{0.44\textwidth}
\center
\includegraphics[width=0.87\textwidth]{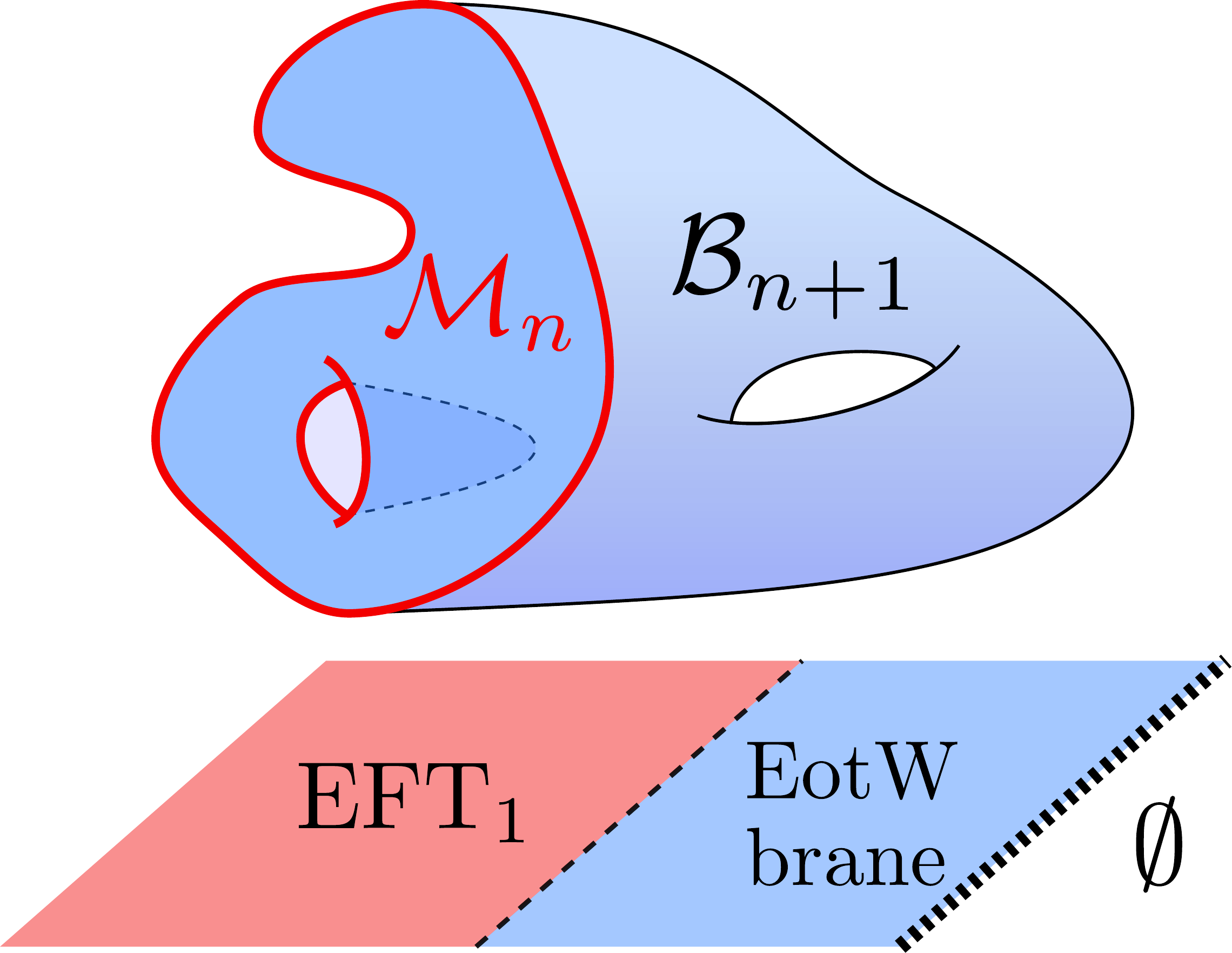}
\caption{\hspace{-0.3em}) Cobordism from a manifold to nothing.} \label{ff.cob2}
\end{subfigure}
\caption{Illustration of bordisms between two $n$-manifolds. In subfigure \ref{ff.cob1} the boundary of the $\mathcal{B}_{n+1}$ corresponds to the manifold $\mathcal{M}_n$ and $\mathcal{N}_n$ (with opposite orientations), and traversing $\mathcal{B}_{n+1}$ be seen as crossing a domain wall between the two EFTs, each obtained from the compactification of the higher-dimensional theory respectively on  $\mathcal{M}_n$ or $\mathcal{N}_n$. In subfigure \ref{ff.cob2} the bordism has a single boundary component, $\mathcal{M}_n=\partial\mathcal{B}_{n+1}$, so that it can be understood as connecting the given EFT to nothing. From the lower-dimensional point of view this is seen as a boundary of spacetime, in what it is known as \emph{End of the World brane}, see e.g., \cite{Angius:2022aeq}.
\label{fig.cobord}}
\end{center}
\end{figure}

However, one crucial aspect of the Swampland Cobordism Conjecture is that is \emph{not constructive}, this is, it only tells us that there should be some configuration interpolating between any pair of two $d$-dimensional theories and/or nothing. While we know that the appropriate $\mathsf{g}$-structure or $G$ bundles must be extended to the $\mathcal{B}_{n+1}$, bulk, their precise configuration, or the topological properties of the bordism are generically not known. These can be of great importance for phenomenological purposes, if for example one wants to estimate the tension of the domain wall connecting two given lower-dimensional theories, or the decay rate of the theory into nothing. Several top-down constructions in the literature \cite{Witten:1981gj,Fabinger:2000jd,Dine:2004uw,Horowitz:2007pr,Yang:2009wz,Blanco-Pillado:2010xww,Blanco-Pillado:2010vdp,Brown:2010mf,Blanco-Pillado:2011fcm,Brown:2011gt,deAlwis:2013gka,Brown:2014rka,Blanco-Pillado:2016xvf,Ooguri:2017njy,Dibitetto:2020csn,GarciaEtxebarria:2020xsr,Bomans:2021ara,Draper:2021ujg,Draper:2021qtc,Petri:2022yhy,Bandos:2023yyo,Ookouchi:2024tfz,Heckman:2024zdo} show that generically $\mathcal{B}_{n+1}$ can be relatively involved, and in \cite{Ruiz:2024jiz} a program aiming to understand the topological properties of $\mathcal{B}_{n+1}$ knowing those of the boundary $\partial\mathcal{B}_{n+1}$ was started.

In this paper, we will use the 9d type IIB supergravities as a non-trivial arena were we obtain information about the topological properties of the bordism $\mathcal{B}_2$ connecting them and/or nothing. As we will see, a non-trivial $\mathsf{SL}(2,\mathbb{Z})$-bundle on the compact $\mathbb{S}^1$ translates to properties of the arrangement of bordism (duality) defects and the topology of $\mathcal{B}_2$ (in the case the genus of said 2-manifold). If said topology becomes arbitrarily complicated, this would have implications that are in tension with Swampland/QG expectations regarding the scale at which global symmetries are broken, which motivates us to propose a sharpening of the Swampland Cobordism conjecture for $\Omega_1^{\mathsf g}({\rm B}G)$ in terms of an algebraic property of $G$ known as \emph{commutator width}, where $G$ is the duality group of the effective theory. We will further test this proposal for other theories where both $G$ and the spectrum of defects is known.\\

The structure of the paper is as follows. In Section \ref{sec:9dIIBgSUGRA} we will give a review of type IIB 9d supergravities and their origin in 10d Type IIB string theory, together with the appropriate F-theory picture. In Section \ref{sec:bordisms} we explain how the topology of the bordism can be obtained from the $\mathbb{S}^1$ boundary and its specific $\mathsf{SL}(2,\mathbb{Z})$ bundle. In turn, In Section \ref{sec: com len} we discuss how some algebraic properties of the duality group (here $\mathsf{SL}(2,\mathbb{Z})$) can result in problematic properties for our bordism, and propose a conjecture affecting $\Omega_1^{\mathsf g}({\rm B}G)$, which we later test in Section \ref{sec: other groups} for other known duality groups. Finally, in Section \ref{s. conc} we conclude and point towards possible future directions, followed by a series of appendices containing different technical results and computations.

\section{A review of type IIB 9d gauged supergravities}
\label{sec:9dIIBgSUGRA}

Type II supergravities in $D=9$ have been completely classified, see \cite{Bergshoeff5,Hull4,Hull8,Melgarejo}, corresponding to a unique maximal supergravity, and eight additional massive (gauged) supergravities, as shown in Figure \ref{fig1}.\footnote{See \cite{AbouZeid} to take into account subtleties in the massive sector of this statement} Kaluza-Klein (KK) or (generalized) Scherk-Schwarz (GSS) reductions of higher-dimensional theories can yield these eight gauge supergravities in nine dimensions. Alternatively, they can also be produced as a deformation of the 9d maximal supergravity, using the embedding tensor formalism.

Type IIA supergravities in 10 dimensions (maximal \cite{PhysRevD.30.325,CAMPBELL1984112,Huq:1983im} or massive \cite{Romans:1985tz,Howe:1997qt}), can yield four out of these eight gauge supergravities in nine dimensions. Therefore, they can be said to be on the type IIA sector. The four remaining gauged supergravities are on the type IIB sector, being associated with the gauging of the global symmetry $\mathsf{GL}(2,\mathbb{R})=\mathbb{R}\times \mathsf{SL}(2,\mathbb{R})$. The gauging of $\mathbb{R}$ corresponds to the gauging of a trombone symmetry of the equations of motion, which will not be further examined in this work. The breaking of $\mathsf{SL}(2,\mathbb{R})$ to any of its one-parameter subgroups ($\mathbb{R}$, $\mathsf{SO}(2)$, or $\mathsf{SO}(1,1)^+$ ) is connected to the other 3 type IIB gauge supergravities in 9d. These cases were analyzed in \cite{Bergshoeff8,Pope,Kaloper2} for particular elements of $\mathsf{SL}(2,\mathbb{R})$ and in \cite{Ortin, Meessen,Cowdall} for a general description.

In addition to the eight massive deformations of type II supergravities in 9d depicted in Figure \ref{fig1}, it is possible to combine mass parameters, in a way that is consistent with supersymmetry. Such cases were studied in \cite{Bergshoeff9,Bergshoeff5} and will not be considered in this work. Certain domain wall solutions have been found and presented in terms of type IIB string theory branes \cite{Bergshoeff9,Cowdall}.

In this section, we will briefly review some relevant features of type IIB supergravity in $D=10$, as well as the description of type IIB supergravities in 9d, both maximal and gauged, also discussing the interpretation in terms of F-theory \cite{Hull4,Hull8} on twisted 3-tori. In these $\mathcal{N}=2$ supergravities,  the breaking of the global symmetry group introduces mass parameters in the reduced theory modifying the Lagrangian, inducing the appearance of covariant derivatives, and changing the field strengths of the background fields, as well as flux potentials that spontaneously break supersymmetry.

\subsection{A lightning review of Type IIB 10d supergravity\label{ss.typeIIB}}
Before describing the properties of (gauged) Type IIB 9d supergravities, let us quickly review the 10d parent theory. Type IIB 10d Supergravity, which arises as the low-energy limit of Type IIB string theory, is a $\mathcal{N}=2$, chiral theory with the following bosonic spectrum:
\begin{equation}
    \big\{\underbrace{e_M^a,B_{MN},\Phi}_{\rm NSNS},\underbrace{C,C_{MN},C_{MNPQ}}_{\rm RR}\big\}\,,
\end{equation}
where both the NSNS and RR sectors feature scalars and 2-forms. We can partner each pair in the so-called \emph{axio-dilaton} $\tau=C_0+ie^{-\Phi}$ and doublet $\vec C_2=(B_2,C_2)^\intercal$, and write an effective \emph{pseudo-action}\footnote{The 5-form field strength $G_5=\dd C_4-\frac{1}{2}\epsilon_{ij}C_2^{i}\wedge \dd C_2^j$ must satisfy the self-duality condition $G_5=\star G_5$, and thus prevents us from writing a kinetic term $-\frac{1}{2}G_5\wedge\star G_5$. One then considers a pseudo-action where such self-duality is later imposed in the equations of motion.}
 of the form
\begin{equation}\label{eq.10dAction}
    S_{\rm IIB}=\frac{1}{2\kappa^{2}_{10}}\int\left\{R\star1-\frac{\dd \tau\wedge\star\dd\tau}{2({\rm Im}\,\tau)^2}-\frac{1}{2}\Lambda_{ij}H_3^{i}\wedge\star H_3^{j}-\frac{1}{4}G_5\wedge\star G_5\right\}+S_{\rm CS}\,,
\end{equation}
where $\kappa_{10}^2=M_{\rm Pl,10}^{-8}$ is the gravitational strength and
\begin{equation}\label{eq.LAMBDAij}
    \Lambda_{ij}=\frac{1}{\rm Im\,\tau}\begin{pmatrix}
        1&\rm Re\,\tau\\\rm Re\,\tau&|\tau|^2
    \end{pmatrix}=e^{\Phi}\begin{pmatrix}
        1&C_0\\C_0&C_0^2+e^{-2\Phi}
    \end{pmatrix}
\end{equation}
is a $\mathsf{SL}(2,\mathbb{R})$-symmetric matrix parameterizing the coset space $\mathsf{SL}(2,\mathbb{R})/\mathsf{SO}(2)$. The action is invariant under a $\mathsf{SL}(2,\mathbb{R})$ transformation (broken to $\mathsf{SL}(2,\mathbb{Z})$ upon instanton charge quantization \cite{Hull})\footnote{\label{fn.presentationSL2Z}For completeness and future reference, the presentation of $\mathsf{SL}(2,\mathbb{Z})$ is given by
\begin{equation}\label{eq.presSL2Z}
    \mathsf{SL}(2,\mathbb{Z})=\big\langle U,S|S^4=1,\,S^2=U^3\big\rangle\simeq\mathbb{Z}_6\ast_{\mathbb{Z}_2}\mathbb{Z}_4\,.
\end{equation}
Of great importance is also the element $T=SU^{-1}$, which generates a subgroup of infinite order. In terms of $2\times 2$ matrices, we can write the above elements as
\begin{equation}\label{eq.matrix UST}
   U=\begin{pmatrix}
        1&1\\-1&0
    \end{pmatrix}\,,\quad S=\begin{pmatrix}
        0&1\\-1&0
    \end{pmatrix}\,,\quad   T=\begin{pmatrix}
        1&1\\0&1
    \end{pmatrix}
\end{equation}\,.
}
\begin{equation}
    \tau\to\frac{a\tau+b}{c\tau+d}\,,\quad\vec C_2\to \mathcal{M}\vec{C}_2\,,\quad C_4\to C_4\,,\quad g_{MN}\to g_{MN}\,,
\end{equation}
with $\mathcal{M}=\begin{psmallmatrix}
        a&b\\c&d
    \end{psmallmatrix}\in \mathsf{SL}(2,\mathbb{Z})$. Magnetically charged with respect to the axion $C_0$ we find the D7-brane, on which fundamental type IIB strings end, with unit charge $\int_{\mathbb{S}^1}\star F_9=\int\dd C_0=1$. Preservation of (half) supersymmetry requires the equation of motion $\bar\partial\tau=0$ away from the D7-brane location $z_0$, with solution $\tau(z)\approx\frac{1}{2\pi i}\log(z-z_0)$ \cite{Greene:1989ya}, and a monodromy
    \begin{equation}
        \tau\to \tau+1\,,\quad (B_2,C_2)^{\intercal}\to
            (B_2+C_2,C_2)^{\intercal}\,,
    \end{equation}
    precisely corresponding to a $\mathsf{SL}(2,\mathbb{Z})$ transformation $T=\begin{psmallmatrix}
        1&1\\0&1
    \end{psmallmatrix}$, identified with $T$ in the presentation \eqref{eq.matrix UST} from Footnote \ref{fn.presentationSL2Z}. Apart from the fundamental F1-string (charged under the NSNS $B_2$ field), type IIB string theory also contains D1-strings, charged under the RR $C_2$ form. We can thus consider general $(p,q)$-strings (with $p$ and $q$ coprime \cite{Witten:1995im}), charged under $pB_2+qC_2$, consistent in a BPS bound state of $p$ F1 (0,1)-strings, and $q$ D1 (1,0)-strings. These precisely end on $[p,q]$ 7-branes. Being codimension-2 objects, a closed loop around them result in a $\mathsf{SL}(2,\mathbb{Z})$ monodromy given by
    \begin{equation}\label{eq.monpq}
        \mathcal{M}_{[p,q]}=\begin{pmatrix}
            1+pq&p^2\\-q^2&1-pq
        \end{pmatrix}\,,
    \end{equation}
    see \cite[Section 2.1]{Weigand:2018rez} for more details.

\subsection{Type IIB (maximal and gauged) supergravities in 9d\label{ss.9dIIB}}
KK reduction on a circle of any of the type II (IIB or massless IIA) supergravities in 10 dimensions leads to type II maximal supergravity in $d=9$. Alternatively, it can also be obtained from a toroidal compactification of eleven-dimensional supergravity. Its (bosonic) field content is given by
\begin{equation}
    \left\lbrace e_{\mu}^{a}, \rho,\tau=C_0+ie^{-\Phi}, A_\mu, A_\mu^{(1)}, A_\mu^{(2)},   B_{\mu\nu}^{(1)},B_{\mu\nu}^{(2)},C_{\mu\nu\rho} \right\rbrace\,,
\end{equation}
where from the Type IIB reduction perspective we have the (canonically normalized) radion $\rho=\sqrt{\frac{8}{7}}\log\frac{R}{\ell_{10}}$ (with $R$ the radius of the circle compactification), $\tau=C_0+ie^{-\Phi}$ is the inherited axio-dilaton, $A_\mu$ is the graviphoton, the doublets $\vec A_\mu$ and $\vec B_{\mu\nu}$ come from the dimensional reduction of the 10d 2-form doublet in 10d Type IIB (both with and without a leg on $\mathbb{S}^1$), and finally the 3-form $C_{\mu\nu\rho}$ results from the reduction of the self-dual 10d 4-form with one leg along $\mathbb{S}^1$. Analogous interpretation from the reduction of the Type IIA massless spectrum is also straightforward, see e.g., \cite[Section 5.1.1]{Ibanez:2012zz}.
The reduction of the \eqref{eq.10dAction} action results in
\begin{align}\label{eq.9dungauged}
    S_{\rm II}^{(\rm 9d)}&=\frac{1}{2\kappa_9^2}\int\bigg\{\mathcal{R}\star1-\dd\rho\wedge\star\dd\rho-\frac{\dd\tau\wedge\star\dd\tau}{2(\rm Im\,\tau)^2}-\frac{e^{4\sqrt{\frac{2}{7}}\rho}}{2}F_2\wedge\star F_2-\frac{e^{-3\sqrt{\frac{2}{7}}\rho}}{2}\Lambda_{ij}F_2^{i}\wedge\star F_2^j\notag\\&~~~~~~~~~~~~~~~~~-\frac{e^{\sqrt{\frac{2}{7}}\rho}}{2}\Lambda_{ij}H_3^i\wedge\star H_3^j-\frac{e^{-\frac{5}{\sqrt{14}}\rho}}{2}G_4\wedge\star G_4\bigg\}+S_{\rm CS}^{(\rm 9d)}\,,
\end{align}
where $\kappa_9^2=(2\pi R)^{-1}\kappa_{10}^2$ is the 9d gravitational coupling.

\begin{figure}[ht]
    \centering
    \resizebox{0.9\textwidth}{!}{%
    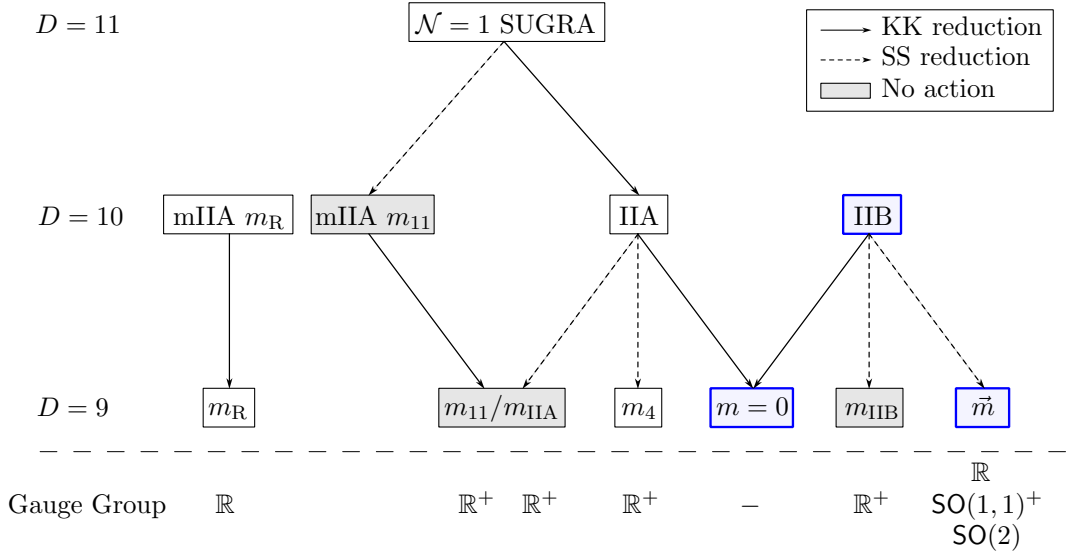
    }
    \caption{Type II (gauged) supergravities in $D=9$, together with their construction and massive deformations. We highlight in blue the 10d type IIB theory and the associated 9d type II SUGRAs we will study in this paper. Figure adapted from \cite{Bergshoeff5}. Note that in the highlighted case the gauge group is broken once the non-perturbative effects break $\mathsf{SL}(2,\mathbb{R})\to\mathsf{SL}(2,\mathbb{Z})$ in the 10d theory.}
    \label{fig1}
\end{figure}

Symmetries from the ten-dimensional theories are inherited by the massless theory in nine dimensions, corresponding to $\mathbb{R}^+\times \mathsf{SL}(2,\mathbb{R})$ from type IIB supergravity.
The nontrivial action of the $\mathsf{SL}(2,\mathbb{R})$ global symmetry in nine dimensions (again, broken to $\mathsf{SL}(2,\mathbb{Z})$ with the inclusion of non-perturbative effects) on the bosonic sector is given by
\begin{eqnarray}
     \tau &\rightarrow& \frac{a\tau + b}{c\tau+d} \,, \qquad \vec{B}\rightarrow \mathcal{M}\vec{B} \,, \qquad \vec{A}\rightarrow \mathcal{M}\vec{A} \,,
 \end{eqnarray}
 where $\mathcal{M}\in \mathsf{SL}(2,\mathbb{Z})$. Additionally the expected gauge symmetries associated to the dynamical $p$-forms $\{A_1,\vec A_1,\vec B_2,C_3\}$ are present.

\vspace{0.5cm}

The 9d type II SUGRA theory we have just described can be directly obtained from compactifying on a circle $\mathbb{S}^1$, in such a way that the 10d fields have trivial profiles along the internal dimension. In contrast, type IIB gauged supergravities in nine dimensions can be obtained by means of Generalized Scherk--Schwarz (GSS) reduction of the 10d type IIB supergravity in ten dimensions. The main idea of the Scherk--Schwarz (SS) mechanism \cite{Schwarz5} is to introduce mass parameters by considering a reduction of the higher-dimensional fields such that they explicitly depend on the internal coordinate $y\in[0,1]$ parameterizing $\mathbb{S}^1$. Following \cite{Hull4}, we will consider the generalized Scherk-Schwarz reduction (GSS), where a twist with an element of the global symmetry group $G$, one takes
\begin{eqnarray}
    \phi (x^\mu,y)=g_y(\phi^i(x^\mu)),
\end{eqnarray}
where $g_y=g(y)$ is a symmetry transformation contained in a subgroup of $G$. The above map is not periodic, but rather has a monodromy $\mathcal{M}(g)=\exp M$, where $\mathcal{M}\in G$ and $M=g^{-1}\partial_y g$ where $M \in \mathfrak{sl}(2,\mathbb{R})$ is an element of the Lie algebra, with explicit expression $g_y=\exp(My)$. The element $M$ generating the transformation can be written as linear combination of the generators of $\mathfrak{sl}(2,\mathbb{R})$ in terms of the Pauli matrices $\{\sigma_1,i\sigma_2,\sigma_3\}$, 
\begin{equation}\label{massmatrix}
    M=m_1\sigma_3+m_2\sigma_1+m_3(i\sigma_2)= \begin{pmatrix}
      m_1   & m_2+m_3 \\
        m_2 - m_3 &  -m_1
    \end{pmatrix}, 
\end{equation}

and consequently, the monodromy is given by $\mathcal{M}=g(y=1)=e^M\in\mathsf{SL}(2,\mathbb{R})$.\footnote{\label{fn.SL2Zsimplyconnected}
Note that, since $\mathsf{SL}(2,\mathbb{R})$ is not simply connected, $\exp\mathfrak{sl}(2,\mathbb{R})$ does not cover all $\mathsf{SL}(2,\mathbb{R})$. Since elements $X\in\mathfrak{sl}(2,\mathbb{R})$ are $2\times 2$ matrices with zero trace, the sum of their eigenvalues $\lambda_{\pm}=\frac{1}{2}\big[{\rm Tr}\,X\pm\sqrt{({\rm Tr}\,X)^2-4\det(X)}\big]=\pm i\sqrt{\det(X)}$
is 0, i.e., of the form $(\lambda,-\lambda)$, with $\lambda\in\mathbb{R}\cup i\mathbb{R}$ (since $\det(X)\in\mathbb{R}$). This means then that ${\rm Tr}(X)=e^{\lambda}+e^{-\lambda}=2\cosh(\lambda)\geq-2$ for $\lambda\in\mathbb{R}\cup i\mathbb{R}$. Thus the elements of $\mathsf{SL}(2,\mathbb{R})$ with ${\rm Tr}(\mathcal{M})<-2$ lie in a patch outside $\exp\mathfrak{sl}(2,\mathbb{R})$ and cannot be accessed through the exponential map around the identity matrix.} The values $\{m_i\}_{i=1}^3$ correspond to the \emph{mass parameters} of the reduced theory. The interpretation as ``masses'' of the lower-dimensional scalars will become more clear once we see how such parameters enter in the scalar potential.

The elements of $\mathsf{SL}(2,\mathbb{R})$ can be classified according to their trace in parabolic, elliptic, or hyperbolic elements for $\abs{{\rm Tr}(\mathcal{M})}=2$, $\abs{{\rm Tr}(\mathcal{M})}<2$, or $\abs{{\rm Tr}(\mathcal{M})}>2$, respectively. This classification, which precisely corresponds to the three different factors in the Iwasawa/KAN decomposition of $\mathsf{SL}(2,\mathbb{R})$, is based on subsets and does not generally correspond to subgroups. Nevertheless, it is known that each matrix of any of these subsets is conjugate to a member of one of the above three one-parameter subgroups \cite{Hull4}, which can be written via the exponential map in terms of the mass matrix \eqref{massmatrix}, as they appear in Table \ref{tab:param subgroups}. Such reductions were considered in \cite{Bergshoeff8,Pope,Kaloper2} for particular elements and in \cite{Ortin, Meessen,Cowdall} for general elements of $\mathsf{SL}(2,\mathbb{R})$. We note that, even though the set of hyperbolic matrices $\abs{{\rm Tr}(\mathcal{M})}>2$ has two component, only the one with positive trace is in $\exp\mathfrak{sl}(2,\mathbb{R})\in\mathsf{SL}(2,\mathbb{R})$, with the negative trace one not appearing in GSS reductions. Nonetheless, since $\exp\mathfrak{sl}(2,\mathbb{R})$ is not a subgroup, hyperbolic monodromies with negative traces will appear later along the text when considering bordisms of our 9d theories.

Therefore, the GSS reduction of the type IIB supergravity in ten dimensions corresponds to considering a monodromy matrix $\mathcal{M}$ as those from Table \ref{tab:param subgroups}, with their corresponding mass matrix $M$ given by \eqref{massmatrix} with appropriate values of $m_1$, $m_2$, and $m_3$.

Nevertheless, these type IIB gauged supergravities in nine dimensions, can also be obtained using the embedding tensor formalism, which is not a higher-dimensional approach but a method to gauge the global symmetry group of the massless theory in the same dimension. Let us consider som Lie algebra $\mathfrak{g}_0$ \cite{Samtleben}, with generators $\{t_i\}_{i=1}^{\dim (\mathfrak{g}_0)}$ and structure constants defined through $\left[t_i,t_j\right] = f_{ij}{}^k t_k$. The embedding tensor formalism couples the gauge vector fields $A_\mu^M$ to the $\mathfrak{g}_0$ symmetry generators $t_\alpha$ in the form of a covariant derivative
\begin{eqnarray}\label{eq.cov}
    {\rm D}_\mu = \partial_\mu - g A_\mu^M\Theta_M^it_i ,
\end{eqnarray}
where $\mu$ is space-time index, $g$ is the gauge coupling constant, $M$ labeled the representation of the global symmetry group and $\Theta_M^i$ is the embedding tensor. The latter is a rectangular matrix coupling the vector fields with the symmetry generators, and can be understood as a map from the representation of the global symmetry group to the gauge group of local symmetry. Indeed, $X_M=\Theta_M^it_i$ are interpreted as the generators of the gauge group that will appear on the gauge transformations of the fields. By considering a particular gauge, the embedding tensor breaks down the invariance from the global symmetry group to the local gauge group.

\begin{table}[htb]
    \centering
    \resizebox{\textwidth}{!}{
    \begin{tabular}{|c||c|c|c|}\hline
        Subset & Parabolic $\mathbb{R}$ & Elliptic $\mathsf{SO}(2)$ & Hyperbolic $\mathsf{SO}(1,1)^+$ \\\hline\hline
        ${\rm Tr}(\mathcal{M})$ &$2$ & $\in(-2,2)$ & $>2$ \\\hline
        Representative & $\begin{pmatrix}
            1&k\\0&1
        \end{pmatrix}$, $k\in \mathbb{R}$ & $\begin{pmatrix}
      \cos(\theta)   &  \pm\sin(\theta) \\
        \mp\sin(\theta) &  \cos(\theta)
    \end{pmatrix}$, $\theta\in[0,2\pi)$ & $\begin{pmatrix}
       e^a  & 0 \\
        0 & e^{-a}
    \end{pmatrix}$, $a\in\mathbb{R}_{(\geq 0)}$ \\\hline
        $\vec m=(m_1,m_2,m_3)$ & $(0,\frac{k}{2},\frac{k}{2})$  & $(0,0,\pm\theta)$ &  $(a,0,0)$\\\hline
        ${\rm D}\tau$ & $\dd\tau-kA$ & $\dd\tau\mp\theta(1+\tau^2)A$ & $\dd\tau-2a\tau A$ \\\hline
        $V(\rho,\tau)$  &$\frac{k^2}{2}e^{-4\sqrt{\frac{2}{7}}\rho+2\Phi}$ & $\frac{\theta^2}{2}e^{-4\sqrt{\frac{2}{7}}\rho-2\Phi}\left[1+2e^{2\Phi}(C_0^2-1)+e^{4\Phi}(C_0^2+1)^2\right]$& $2a^2e^{-4\sqrt{\frac{2}{7}}\rho}\left(1+e^{2\Phi}C_0^2\right)$\\\hline
    \end{tabular}}
    \caption{Different subsets of $\exp\mathfrak{sl}(2,\mathbb{R})\subset\mathsf{SL}(2,\mathbb{R})$ in terms of their traces, together with their representatives and mass parameters $\vec m$ from which they can be generated. We also depict the form of the covariant derivative for the axio-dilaton and the scalar potential. It is immediate that setting $\vec{m}=\vec{0}$ recovers the ungauged case. Note that in hyperbolic matrices we choose the representative with $a\geq 0$, since $a\to-a$ leaves the trace and eigenvalues (and thus the conjugacy class) invariant. As discussed in Footnote \ref{fn.SL2Zsimplyconnected}, since $\mathsf{SL}(2,\mathbb{R})$ is not simply connected, only the elements with ${\rm Tr}\,\mathcal{M}\geq -2$ are in $\exp\mathfrak{sl}(2,\mathbb{R})$, which means that the above description does not cover the hyperbolic matrices with ${\rm Tr}\,\mathcal{M}< -2$}
    \label{tab:param subgroups}
\end{table}

In the particular case of 9d type IIB gauged supergravities, we have that the nonzero components of the embedding tensor are given by $\Theta_0^n=m_n$ with $m_n$ the mass parameters of the matrix \eqref{massmatrix} and the only 1-form entering the covariant derivative \eqref{eq.cov} being the graviphoton. 

These gauged supergravities can be produced by both methods: GSS reduction and the Embedding Tensor formalism. The corresponding action, will be modified from \eqref{eq.9dungauged} through the introduction of mass parameters, gauge couplings and scalar potentials, with the general form

\begin{equation}\label{eq.9dgauged}
    S_{\rm II}^{(\rm 9d)}=\frac{1}{2\kappa_9^2}\int\bigg\{\mathcal{R}\star1-{\rm D}\rho\wedge\star{\rm D}\rho-\frac{{\rm D}\tau\wedge\star{\rm D}\tau}{2(\rm Im\,\tau)^2}-2\kappa_9^2V(\rho,\tau)\star1\bigg\}+S_{\rm gauge}^{(\rm 9d)}+S_{\rm CS}^{(\rm 9d)}\,.
\end{equation}
We note that the above action is modified from the ungauged one \eqref{eq.9dungauged} (to which it reduces when all the mass parameters are zero, $\vec m=(0,0,0)$) in two ways. First of all the usual exterior derivatives are replaced by the covariant ${\rm D}=d-g \Theta_M^it_i A^M\wedge$, which for scalar fields translates in \cite{Cowdall} $D_\mu\varphi=\partial_\mu\varphi-\delta_{\mathcal{M}}\varphi A_\mu$, in particular 
\begin{equation}
    {\rm D}_\mu\rho=\partial_\mu\rho\,,\quad {\rm D}_\mu\tau=\partial_\mu\tau-\big[(m_2+m_3)+2m_1\tau+(m_3-m_2)\tau^2\big]A_\mu\,,
\end{equation}
since the 10d spacetime metric (and in particular the radion $\rho$) is invariant under the $\mathsf{SL}(2,\mathbb{R})$ that we are gauging. In Table \ref{tab:param subgroups} we depict the covariant derivative for the axio-dilaton for the different classes of $\mathsf{SL}(2,\mathbb{R})$ subsets.

As for the $p$-form gauge fields, this introduces a St\"uckelberg/Higgs-like coupling between scalars, graviphoton 1-form and other $p$-forms, through which some of them can become massive. The gauge couplings of the form $e^{\alpha\rho}$ to the radion are the same as in the ungauged case in \eqref{eq.9dungauged}. The modified field strengths and the dynamics of the $p$-form fields will not be the focus of this work, and will not be discussed further.

Secondly, a scalar potential $V(\rho,\tau)$ acting on the radion and axio-dilaton is introduced, with general form 
\begin{equation}
    V(\rho,\tau)=\frac{1}{2}e^{-4\sqrt{\frac{2}{7}}\rho}{\rm Tr}\big(M^2+M\Lambda^{-1}M^{\intercal}\Lambda\big)\,,
\end{equation}
with $M$ given by \eqref{massmatrix} and $\Lambda$ the $\mathsf{SL}(2,\mathbb{R})$-symmetric matrix from \eqref{eq.LAMBDAij}. We depict the potential for the different cases in Table \ref{tab:param subgroups}.

We notice that for the parabolic and hyperbolic monodromies, $V(\rho,\tau)>0$, so that supersymmetry is spontaneously broken. The potential has runaway direction towards large radius and small coupling, with the axion $C_0$ having a flat direction in the parabolic case, and being stabilized at $C_0=0$, with a light mass $m_{C_0}^2\sim 4M_{\rm Pl,9}^2a^2 e^{-4\sqrt{\frac{2}{7}}\rho}\to 0$ for hyperbolic monodromies. On the other hand, for the elliptic case there exists a (non-supersymmetric) Minkowski minimum at tree level, located at $\tau=i$ (i.e., $C_0=\Phi=0$) \cite{Cowdall}, with $m_{C_0}^2\sim m_{\Phi}^2\sim 8M_{\rm Pl,8}^2\theta^2e^{-4\sqrt{\frac{2}{7}}\rho}$, with the radion $\rho$ having a flat direction.

\subsection{The F-theory picture}
\label{ss:Fth}

In the previous section, when studying the type IIB gauged supergravities in 9d, we have treated the 10d type IIB supergravity as having a $\mathsf{SL}(2,\mathbb{R})$ continuous global symmetry, which allowed us to consider any monodromy $\mathcal{M}\in \mathsf{SL}(2,\mathbb{R})$ in the GSS reduction.  However, we know that the symmetry group $\mathsf{SL}(2,\mathbb{R})$ breaks into the discrete subgroup $\mathsf{SL}(2,\mathbb{Z})$ in the quantum theory, such that it guarantees the quantization of the charges \cite{Hull}. Consequently, the quantum-consistent generalized Scherk-Schwarz reductions of type IIB superstring are those for which the monodromy resides in the conjugacy classes of $\mathsf{SL}(2,\mathbb{Z})$ \cite{Zwiebach,Zwiebach2}. This can be seen as a restriction on the choice of $M\in\mathfrak{sl}(2,\mathbb{R})$ such that $\mathcal{M}=\exp M$ has integer entries. Therefore, the compact or non-compact parameters of the subgroups of $\mathsf{SL}(2,\mathbb{R})$ in Table \ref{tab:param subgroups} become restricted in the quantum theory:
\begin{enumerate}
    \item \textbf{Parabolic monodromies}: There is an infinite number of parabolic $\mathsf{SL}(2,\mathbb{Z})$ conjugacy classes, with representatives
    \begin{eqnarray}
        \mathcal{M}_k = \begin{pmatrix}
           1  & k \\
            0 & 1
        \end{pmatrix} ,\label{parabolicmonodromyUV}\qquad\text{with }\, k\in\mathbb{Z}\,.
    \end{eqnarray}
The real parameter of the parabolic supergravity is quantized in the higher-dimensional theory, breaking $\mathbb{R}\to\mathbb{Z}$.
\item \textbf{Elliptic monodromies}: There are four elliptic $\mathsf{SL}(2,\mathbb{Z})$ conjugacy classes, with traces ${\rm Tr}\mathcal{M}\in\{-2,-1,0,1\}$,\footnote{The elliptic integer monodromy with $\rm Tr\,\mathcal{M}=2$ corresponds to the identity matrix, which results in the ungauged, maximally supersymmetric 9d SUGRA. The same happens to the parabolic monodromy with $k=0$. While the $ \mathcal{M}_{\mathbb{Z}_2} = \begin{psmallmatrix}
           -1  & 0 \\
            0 & -1
\end{psmallmatrix}$ has trace -2, it does not correspond to a parabolic monodromy, since it is conjugate to a $\pi$ $\mathsf{SO}(2)$ rotation.} represented by
\begin{equation}
     \mathcal{M}_{\mathbb{Z}_2} = \begin{pmatrix}
           -1  & 0 \\
            0 & -1
\end{pmatrix}, \;        \mathcal{M}_{\mathbb{Z}_3} = \begin{pmatrix}
           -1  & 1 \\
            -1 & 0
        \end{pmatrix} 
, \;      \mathcal{M}_{\mathbb{Z}_4} = \begin{pmatrix}
           0  & 1 \\
            -1 & 0
\end{pmatrix}, \;      \mathcal{M}_{\mathbb{Z}_6} = \begin{pmatrix}
           1  & 1 \\
            -1 & 0
\end{pmatrix}, \label{ellipticmonodromyUV} 
\end{equation}
labeled by the group the generate. These matrices are conjugate to particular rotations of the $\mathsf{SO}(2)$ group. Indeed, these monodromies correspond to $\mathsf{SO}(2)$ gaugings at specific angles. Indeed, from the elliptic classes in $\mathsf{SL}(2,\mathbb{R})$ from Table \ref{tab:param subgroups} we see that they are respectively conjugate to $\pm\theta\mod 2\pi\in\{\pi,\frac{2\pi}{3},\frac{\pi}{2},\frac{\pi}{3}\}$, to which the compact parameter $\theta$ of the $\mathsf{SO}(2)$ must be restricted. 
\item \textbf{Hyperbolic monodromies}: There are an infinite number of hyperbolic conjugacy classes with ${\rm Tr}(\mathcal{M})>2$ for $n\geq 3$ (notice again that only the component of hyperbolic matrices with positive trace appears), generated by
\begin{eqnarray}
\mathcal{M}_{h} = \begin{pmatrix}
           n  & 1 \\
            -1 & 0
\end{pmatrix},
\label{hiperbolicmonodromyUV}
\end{eqnarray}
and it can be seen they are conjugate to the hyperbolic matrices from Table \ref{tab:param subgroups}. It is not difficult to see that it requires the parameter $a$ to be expressed in terms of $n$ as
\begin{eqnarray}
    a=\ln{\left( \frac{n}{2}+\sqrt{\frac{n^2}{4}-4}\right)}>0, 
\end{eqnarray} 
such that $n$ is discrete. Once again, we can observe that the parameter $a$ becomes restricted to be particular integers.
\end{enumerate}

In \cite{Hull4,Hull8} the authors proposed that the GSS reduction of type IIB string theory on a circle should correspond to F-theory on a three-dimensional twisted torus. As explained in the appendix \ref{Ape1}, twisted tori in three dimensions are orientable solvmanifolds that admit a global description in terms of torus bundles over a circle with monodromy in $\mathsf{SL}(2,\mathbb{Z})$. These are in one-to-one correspondence with the three-dimensional solvable Lie algebras. It can be checked that, for parabolic monodromies, the twisted tori corresponds to a nilmanifold, known as the Heisenberg nilmanifold in three dimensions, while for the other monodromies, they are just solvmanifolds, see \cite{Aschieri,Hull5}.

In \cite{mpgm2,mpgm7,mpgm23,mpgm24} it was also shown that type IIB gauged supergravities are in correspondence with the inequivalent classes of M2-branes with non-vanishing winding on a torus and nontrivial $\mathsf{SL}(2,\mathbb{Z})$ monodromy. These are M2-branes, whose corresponding Hamiltonian has an $\mathsf{SU}(N)$ regularized model with a discrete supersymmetric spectrum. The global description is given in terms of twisted torus bundles with monodromy contained in $\mathsf{SL}(2,\mathbb{Z})$ \cite {mpgm10}. The fiber is a Heisenberg nilmanifold, and the base corresponds to a Riemann surface of genus 1 associated with the spatial directions of the M2-brane worldvolume. The group transformation between equivalent twisted torus bundles for a given flux and monodromy corresponds to the restricted gauge symmetries at low energies. This allows us to give a $D=11$ interpretation to the discrete values that the mass parameters take in the quantum description for each case. The global symmetry will correspond to the quantization of $\mathsf{SL}(2,\mathbb{R})$, and it will be, in general, a discrete subgroup contained in $\mathsf{SL}(2,\mathbb{Z})$.

\section{Bordisms between type IIB 9d gauged SUGRAs \label{sec:bordisms}}

The different 9d type IIB gauged supergravities described in the previous section have a clear geometrical origin through Scherk-Schwarz compactifications of 10d type IIB supergravity, with different theories being associated with distinct conjugacy classes within $\mathsf{SL}(2,\mathbb{Z})$. More interestingly, from an F-theory perspective, this corresponds to a twisted 3-torus $\mathbb{E}\hookrightarrow T^3_\mathcal{M}\to\mathbb{S}^1$ with some appropriate monodromy $\mathcal{M}\in\mathsf{SL}(2,\mathbb{Z})$ along the $\mathbb{S}^1$ base. While these theories make sense on their own, from the Swampland Cobordism Conjecture \ref{eq.cobord}, we expect the existence of topology changing processes that dynamically connect them. Assuming that an F-theory description remains valid during this process, bordisms from a given type IIB 9d gauged SUGRA to nothing or to analogous theories can be understood as moving along an elliptically fibered manifold with boundary
\begin{equation}\label{eq. B4 fibration}  
\mathbb{E}\hookrightarrow\mathcal{B}_4\to\tilde{\mathcal{B}}_2\,,
\end{equation}
where $\partial\mathcal{B}_4=T^3_\mathcal{M}$ and $\partial\tilde{\mathcal{B}}_2=\mathbb{S}^1_\mathcal{M}$ for bordisms to nothing, and $\partial\mathcal{B}_4=T^3_\mathcal{M}\sqcup T^3_{\mathcal{N}^{-1}}$ and $\partial\tilde{\mathcal{B}}_2=\mathbb{S}^1\sqcup\overline{\mathbb{S}^{1}}$ for bordisms between the 9d gauged SUGRAs associated to monodromies $\mathcal{M}$ and $\mathcal{N}$.\footnote{We write the second boundary, $T^3_{\mathcal{N}^{-1}}$, with the inverse monodromy, $\mathcal{N}^{-1}$, rather than $\mathcal{N}$ due to the inverse orientation this boundary has with respect to the first one, $T^3_{\mathcal{M}}$. }\\

The Cobordism Conjecture is not a \emph{constructive} statement, since it only tells us that such $\mathcal{B}_4$ \emph{must exist}, rather than telling us \emph{how it looks}. However, as shown in \cite{Ruiz:2024jiz}, quite some information about its topology (for our purposes, integral homology) can be said. In this section we will show how one can indeed obtain information about the F-theoretical bordism $\mathcal{B}_4$ simply by knowing its boundary $\partial\mathcal{B}_4$, and indeed give the topological properties (in the form of the homology groups and possible Kodaira singularities) of our elliptic fibrations for the different nullbordisms and interpolation between different theories.

\subsection{Recovering the homology of the bordism\label{sec.top}}

Before focusing in the elliptic fibrations of the form \eqref{eq. B4 fibration}, we can review the general argument from  \cite[Section 3.1]{Ruiz:2024jiz} to show how one can \emph{bound} the homology\footnote{Given a $n$-manifold $X_n$, we respectively denote by $b_k(X_n)$ and $T_k(X_n)$ the Betti numbers and (integer) torsion part of the $k$-th homology group 
\begin{equation}\label{eq. prime exp tors}
    H_k(X_n;\mathbb{Z})\simeq \mathbb{Z}^{b_k(X_n)}\oplus T_k(X_n)\simeq\mathbb{Z}^{b_k(X_n)}\bigoplus_pT_{k,(p)}(X_n)\simeq\mathbb{Z}^{b_k(X_n)}\bigoplus_p\Big(\bigoplus_{i}\mathbb{Z}_{p^{a_{p,i}}}\Big)\,,
\end{equation}
with uniquely defined $a_{p,i}\in\mathbb{Z}_{\geq 0}$,
where we have further decomposed the torsion part in terms of its prime expansion.
} of a manifold $\mathcal{B}_{n+1}$ knowing that of its boundary $\mathcal{C}_n=\partial\mathcal{B}_{n+1}$ (with one or more connected components). One starts with the long exact sequence involving the relative integral homology:
\begin{equation}\label{eq. relative homo}    \dots\xrightarrow{\partial_*}H_k(\mathcal{C}_n;\mathbb{Z})\xrightarrow{i_*}H_k(\mathcal{B}_{n+1};\mathbb{Z})\xrightarrow{j_*}H_k(\mathcal{B}_{n+1},\mathcal{C}_{n};\mathbb{Z})\xrightarrow{\partial_*}H_{k-1}(\mathcal{C}_n;\mathbb{Z})\rightarrow\dots\,,
\end{equation}
which using Lefschetz duality, $H_k(\mathcal{B}_{n+1},\partial\mathcal{B}_{n+1};\mathbb{Z})\simeq H^{n+1-k}(\mathcal{B}_{n+1};\mathbb{Z})$, and the Universal Coefficient Theorem, $H_k(\mathcal{B}_{n+1};\mathbb{Z})\simeq\mathbb{Z}^{\oplus b_k(\mathcal{B}_{n+1})}\oplus T_k(\mathcal{B}_{n+1})$ and $H^k(\mathcal{B}_{n+1};\mathbb{Z})\simeq\mathbb{Z}^{\oplus b_k(\mathcal{B}_{n+1})}\oplus T_{k-1}(\mathcal{B}_{n+1})$, can be rewritten as
{\begin{equation}\label{eq.les}
    \dots\to\mathbb{Z}^{\oplus b_{n-k}(\mathcal{B}_{n+1})}\oplus T_{n-k-1}(\mathcal{B}_{n+1})\to\mathbb{Z}^{\oplus b_{k}(\mathcal{C}_{n})}\oplus T_{k}(\mathcal{C}_{n})\to\mathbb{Z}^{\oplus b_{k}(\mathcal{B}_{n+1})}\oplus T_{k}(\mathcal{B}_{n+1})\to\dots\,.
\end{equation}}
As explained in \cite{Ruiz:2024jiz}, after tensoring by $\mathbb{Q}$ one can see that the boundary of torsion-less cycles is also torsion-free, what allows us to separate the above long exact sequence in terms of free and torsion parts. Regarding the former, the exactness of the sequence allows us to obtain the following bounds on the Betti numbers of $\mathcal{B}_{n+1}$ and $\mathcal{C}_n$,
\begin{equation}\label{eq.betti ineq}
    \tcbhighmath{b_k(\mathcal{B}_{n+1})+b_{n-k}(\mathcal{B}_{n+1})\geq b_k(\mathcal{C}_n)=b_{n-k}(\mathcal{C}_n)\quad\forall\;k=0,\dots, n\,.}
\end{equation}
This implies that in general the free part of the bordism $\mathcal{B}_{n+1}$ can be arbitrarily ``complicated'', since the above inequalities only impose lower bounds on the Betti numbers of $\mathcal{B}_{n+1}$. Considering $\mathcal{B}_{n+1}$ to have a single connected component, this further imposes $b_0(\mathcal{B}_{n+1})=1$ and $b_{n+1}(\mathcal{B}_{n+1})=0$ (since the unique component has a boundary).\\

As for the torsion part, things are not as straightforward as for the Betti numbers. Since for our purposes $n=3$, the relative low dimensionality of the manifolds involved will be of some help,\footnote{Similar relations could be obtained for higher dimensional boundaries $\mathcal{C}_n$ and bordisms $\mathcal{B}_{n+1}$ (with $n>3$). However, in general more torsion groups will be non-vanishing and the subsequent bounds might not be as sharp.} and we can take torsion parts of \eqref{eq. relative homo},
\begin{subequations}
    \begin{align}
    T_3(\mathcal{C}_3)\simeq 0\to &T_3(\mathcal{B}_4)\to T_0(\mathcal{B}_4)\simeq 0 \label{3.4a}\\
        \dots\to T_1(\mathcal{C}_3)\to &T_1(\mathcal{B}_4)\to T_2(\mathcal{B}_4)\to T_0(\mathcal{C}_3)\simeq 0 \label{3.4b} \\
        T_2(\mathcal{C}_3)\simeq0\to &T_2(\mathcal{B}_4)\to  T_1(\mathcal{B}_4)\to T_1(\mathcal{C}_3)\to\dots\ \label{3.4c}\,,
    \end{align}
\end{subequations}
where some of the torsion groups are automatically trivial, such as the top and bottom homology, as well $T_2(\mathcal{C}_3)\simeq T_0(\mathcal{C}_3)=0$ through Poincar\'e duality, since $\partial T^3=\emptyset$. This in turn implies that $T_3(\mathcal{B}_4)=0$. Considering the prime expansion \eqref{eq. prime exp tors} for the above groups, one recovers 
\begin{equation}
    T_{2,(p)}(\mathcal{B}_4)\simeq\frac{T_{1,(p)}(\mathcal{B}_4)}{{\rm im}[T_{1,(p)}(\mathcal{C}_3)\to T_{1,(p)}(\mathcal{B}_4)]}\,,\quad {\rm im}[T_{1,(p)}(\mathcal{B}_4)\to T_{1,(p)}(\mathcal{C}_3)]\simeq\frac{T_{1,(p)}(\mathcal{B}_4)}{T_{2,(p)}(\mathcal{B}_4)}\,,
\end{equation}
for all $p$ prime.
Since both ${\rm im}[T_{1,(p)}(\mathcal{C}_3)\to T_{1,(p)}(\mathcal{B}_4)]$ and ${\rm im}[T_{1,(p)}(\mathcal{B}_4)\to T_{1,(p)}(\mathcal{C}_3)]$ are isomorphic to subgroups of $T_{1,(p)}(\mathcal{B}_4)$, this implies that, given the expansion
\begin{equation}
    T_1(\mathcal{C}_3)=\bigoplus_p\Big(\bigoplus_i\mathbb{Z}_{p^{a_{p,i}}}\Big)\,,\quad
        T_1(\mathcal{B}_4)=\bigoplus_p\Big(\bigoplus_i\mathbb{Z}_{p^{b_{p,i}}}\Big)\,,\quad
        T_2(\mathcal{B}_4)=\bigoplus_p\Big(\bigoplus_i\mathbb{Z}_{p^{c_{p,i}}}\Big)\,,
\end{equation}
with some unique $a_{p,i}$, $b_{p,i}$ and $c_{p,i}$, one has
\begin{equation}\label{eq.prime dec}
    \tcbhighmath{\sum_ib_{p,i}=x_p+\sum_i c_{p,i}\quad \text{for some }x_p\leq\sum_ia_{p,i} \,,\quad \forall\,p\;\text{ prime}\,.}
\end{equation}
Note that, similarly to \eqref{eq.betti ineq}, this again only sets ``lower bounds'' on the torsion of the bordism $\mathcal{B}_4$, but the obtained expression does not treat differently $T_{1}(\mathcal{B}_4)$ and $T_{2}(\mathcal{B}_4)$. The above expression is therefore not very useful in order to obtain information regarding $H_\bullet(\mathcal{B}_4;\mathbb{Z})$.\\

However, as previously mentioned, our bordism $\mathcal{B}_4$ and boundary $\mathcal{C}_3=\partial\mathcal{B}_4$ are not just any arbitrary 4- and 3-manifolds. Since we expect a F-theory description, we ask both manifolds to be elliptically fibered. As shown in Appendix \ref{app.leraySerre} (to where we refer for more details of the appropriate computations through Leray-Serre spectral sequence), this fibration structure allows us to obtain sharper expressions for $H_\bullet(\mathcal{B}_4;\mathbb{Z})$, which also satisfy \eqref{eq.betti ineq} and \eqref{eq.prime dec}. 

As described in Section \ref{ss:Fth}, our boundary to be the disjoint sum of $n_b$ twisted 3-torus, each with $\mathsf{SL}(2,\mathbb{Z})$ monodromy $\mathcal{M}_i$:
\begin{equation}\label{Boundaries}
\mathcal{C}_3=\bigsqcup_{i=1}^{n_b}\big(\mathbb{T}^2\hookrightarrow T^3_{\mathcal{M}_i}\to \mathbb{S}^1\big)\;.
\end{equation}
The bordism to nothing will correspond to $n_b=1$, while the bordism between two different  type IIB 9d gauged Supergravities will have $n_b=2$. The homology of $\mathcal{C}_3$ will be the direct sum of the homologies of each $T^3_{\mathcal{M}_i}$, computed in \eqref{eq. hom T3} from Appendix \ref{app.leraySerre}, so that
\begin{equation}\label{eq. hom C3}
\tcbhighmath{\begin{array}{ll}
     H_0(\mathcal{C}_3;\mathbb{Z})\simeq\mathbb{Z}^{\oplus n_b}\,,& H_1(\mathcal{C}_3;\mathbb{Z})\simeq\mathbb{Z}^{\oplus n_b}\bigoplus_{i=1}^{n_b}\frac{\mathbb{Z}^{\oplus2}}{{\rm im}\,(\mathcal{M}_i-{\rm Id}_2)} \,,\\
     H_2(\mathcal{C}_3;\mathbb{Z})\simeq\mathbb{Z}^{\oplus n_b}\bigoplus_{i=1}^{n_b} \ker(\mathcal{M}_i-{\rm Id}_2)\,,\quad& H_3(\mathcal{C}_3;\mathbb{Z})\simeq\mathbb{Z}^{\oplus n_b}
\end{array}\;.   }
\end{equation}
The expressions for $H_0(\mathcal{C}_3;\mathbb{Z})$ and $H_3(\mathcal{C}_3;\mathbb{Z})$ are straightforward from the counting of connected components. However, the twisting of the fiber due to the $\mathcal{M}_i$ monodromy along the $\mathbb{S}^1$ basis can reduce the homology with respect to that of the untwisted 3-torus with trivial monodromy, $H_1(\mathbb{T}^3;\mathbb{Z})\simeq H_2(\mathbb{T}^3;\mathbb{Z})\simeq \mathbb{Z}^{\oplus 3}$. This is due to the fact that 1-cycles of the fiber $\mathbb{E}$ might not return to themselves after the monodromy, which can induce torsion to $H_1(T^3_{\mathcal{M}_i};\mathbb{Z})$, while the $\mathcal{M}_i$-twisted boundary operator $\partial_1:C_1\otimes\mathbb{Z}^{\oplus2}\to C_0\otimes\mathbb{Z}^{\oplus2}$ can further lower the rank of $H_2(T^3_{\mathcal{M}_i};\mathbb{Z})\simeq\mathbb{Z}\oplus H_1(\mathbb{S}^1;\mathbb{Z}^2_{\mathcal{M}_i})$ (this time without introducing any torsion). Again, we refer to Appendix \ref{app.leraySerre} for more details.\\
In Table \ref{tab. hom T3} we present the homology groups for $T^3_{\mathcal{M}}$ with the $\mathsf{SL}(2,\mathbb{Z})$ monodromies discussed in Section \ref{ss:Fth}. Note again that, since the homology of the disjoint union is the direct product of the individual homologies, this allows us to compute $H_\bullet(\mathcal{C}_3;\mathbb{Z})$ for any number of connected components.

\begin{table}[htb!]
    \centering
    \resizebox{\textwidth}{!}{
    \begin{tabular}{|c||c|c|c|c|c|c|c|}\hline
        $\mathcal{M}$ & {\scriptsize$\begin{pmatrix}1&0\\0&1\end{pmatrix}$} &  {\scriptsize$\begin{pmatrix}1&n\\0&1\end{pmatrix}$}&{\scriptsize$\begin{pmatrix}-1&0\\0&-1\end{pmatrix}$}  &{\scriptsize$\begin{pmatrix}-1&1\\-1&0\end{pmatrix}$}  &{\scriptsize$\begin{pmatrix}0&1\\-1&0\end{pmatrix}$}  &{\scriptsize$\begin{pmatrix}1&1\\-1&0\end{pmatrix}$}  &{\scriptsize$\begin{pmatrix}n&1\\-1&0\end{pmatrix}$},  $|n|\geq 3$\\\hline\hline
         $H_1(T_{\mathcal{M}};\mathbb{Z})$ & $\mathbb{Z}^{\oplus 3}$ &$\mathbb{Z}^{\oplus 2}\oplus\mathbb{Z}_{|n|}$  & $\mathbb{Z}\oplus\mathbb{Z}_2^{\oplus 2}$  & $\mathbb{Z}\oplus\mathbb{Z}_3$  & $\mathbb{Z}\oplus\mathbb{Z}_2$  &$\mathbb{Z}$& $\mathbb{Z}\oplus\mathbb{Z}_{|n-2|}$ \\\hline
         $H_2(T_{\mathcal{M}};\mathbb{Z})$& $\mathbb{Z}^{\oplus 3}$ & $\mathbb{Z}^{\oplus 2}$  & $\mathbb{Z}$ & $\mathbb{Z}$ & $\mathbb{Z}$ & $\mathbb{Z}$ & $\mathbb{Z}$ \\\hline
          \textbf{Solvable?}& Abelian & Nil. & Solv. & Solv. & Solv. & Solv. & Solv.\\\hline
    \end{tabular}
        }
    \caption{Middle (integer) homology groups for the twisted 3-tori $\mathbb{E}\hookrightarrow T^3_\mathcal{M}\to \mathbb{S}^1$ in terms of the monodromies $\mathcal{M}\in\mathsf{SL}(2,\mathbb{Z})$ along the base, described in Section \ref{ss:Fth}. Additional information regarding the Abelian/nilmanifold/solvmanifold property of $T^3_\mathcal{M}$ is also included, see Appendix \ref{Ape1}.}
    \label{tab. hom T3}
\end{table}

\begin{figure}[hbt]
    \centering
     \begin{subfigure}{\textwidth}
     \centering
    \resizebox{0.66\textwidth}{!}{%
\begingroup%
  \makeatletter%
  \providecommand\color[2][]{%
    \errmessage{(Inkscape) Color is used for the text in Inkscape, but the package 'color.sty' is not loaded}%
    \renewcommand\color[2][]{}%
  }%
  \providecommand\transparent[1]{%
    \errmessage{(Inkscape) Transparency is used (non-zero) for the text in Inkscape, but the package 'transparent.sty' is not loaded}%
    \renewcommand\transparent[1]{}%
  }%
  \providecommand\rotatebox[2]{#2}%
  \newcommand*\fsize{\dimexpr\f@size pt\relax}%
  \newcommand*\lineheight[1]{\fontsize{\fsize}{#1\fsize}\selectfont}%
  \ifx\svgwidth\undefined%
    \setlength{\unitlength}{267.03948734bp}%
    \ifx\svgscale\undefined%
      \relax%
    \else%
      \setlength{\unitlength}{\unitlength * \real{\svgscale}}%
    \fi%
  \else%
    \setlength{\unitlength}{\svgwidth}%
  \fi%
  \global\let\svgwidth\undefined%
  \global\let\svgscale\undefined%
  \makeatother%
  \begin{picture}(1,0.29498378)%
    \lineheight{1}%
    \setlength\tabcolsep{0pt}%
    \put(0,0){\includegraphics[width=\unitlength,page=1]{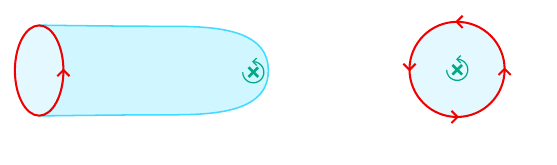}}%
    \put(0.56698518,0.14499831){\color[rgb]{0,0,0}\makebox(0,0)[lt]{\lineheight{1.25}\smash{\begin{tabular}[t]{l}$\Longleftrightarrow$\end{tabular}}}}%
    \put(0.4557016,0.04929168){\color[rgb]{0,0,0}\makebox(0,0)[lt]{\lineheight{1.25}\smash{\begin{tabular}[t]{l}${\color{myred}\mathcal{M}_1}={\color{mygreen}\mathcal{M}_{a}}$\end{tabular}}}}%
    \put(0.12700809,0.1519684){\color[rgb]{0,0,0}\makebox(0,0)[lt]{\lineheight{1.25}\smash{\begin{tabular}[t]{l}${\color{myred}\mathcal{M}_1}$\end{tabular}}}}%
    \put(0.36376233,0.15148564){\color[rgb]{0,0,0}\makebox(0,0)[lt]{\lineheight{1.25}\smash{\begin{tabular}[t]{l}${\color{mygreen}\mathcal{M}_a}$\end{tabular}}}}%
  \end{picture}%
\endgroup%

    }
    \caption{\hspace{-0.3em}) Nullbordism from $\mathbb{S}^1_{\mathcal{M}_1}$ to nothing, with a codimension-two brane of monodromy $\mathcal{M}_a=\mathcal{M}_1$ as cobordism defect.}
    \label{fig.ex1}
    \end{subfigure}
    \begin{subfigure}{\textwidth}
    \centering
    \resizebox{0.66\textwidth}{!}{%
\begingroup%
  \makeatletter%
  \providecommand\color[2][]{%
    \errmessage{(Inkscape) Color is used for the text in Inkscape, but the package 'color.sty' is not loaded}%
    \renewcommand\color[2][]{}%
  }%
  \providecommand\transparent[1]{%
    \errmessage{(Inkscape) Transparency is used (non-zero) for the text in Inkscape, but the package 'transparent.sty' is not loaded}%
    \renewcommand\transparent[1]{}%
  }%
  \providecommand\rotatebox[2]{#2}%
  \newcommand*\fsize{\dimexpr\f@size pt\relax}%
  \newcommand*\lineheight[1]{\fontsize{\fsize}{#1\fsize}\selectfont}%
  \ifx\svgwidth\undefined%
    \setlength{\unitlength}{267.03948734bp}%
    \ifx\svgscale\undefined%
      \relax%
    \else%
      \setlength{\unitlength}{\unitlength * \real{\svgscale}}%
    \fi%
  \else%
    \setlength{\unitlength}{\svgwidth}%
  \fi%
  \global\let\svgwidth\undefined%
  \global\let\svgscale\undefined%
  \makeatother%
  \begin{picture}(1,0.29498378)%
    \lineheight{1}%
    \setlength\tabcolsep{0pt}%
    \put(0,0){\includegraphics[width=\unitlength,page=1]{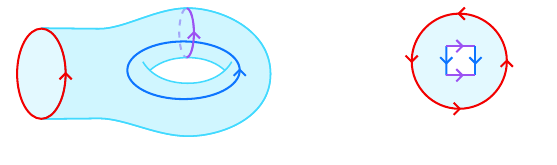}}%
    \put(0.53761991,0.16683531){\color[rgb]{0,0,0}\makebox(0,0)[lt]{\lineheight{1.25}\smash{\begin{tabular}[t]{l}$\Longleftrightarrow$\end{tabular}}}}%
    \put(0.40756299,0.00887336){\color[rgb]{0,0,0}\makebox(0,0)[lt]{\lineheight{1.25}\smash{\begin{tabular}[t]{l}${\color{myred}\mathcal{M}_1}=[{\color{myblue}\mathcal{M}_{a}},{\color{mypurple}\mathcal{M}_b}]$\end{tabular}}}}%
    \put(0.13883531,0.14900373){\color[rgb]{0,0,0}\makebox(0,0)[lt]{\lineheight{1.25}\smash{\begin{tabular}[t]{l}${\color{myred}\mathcal{M}_1}$\end{tabular}}}}%
    \put(0.29949442,0.08137843){\color[rgb]{0,0,0}\makebox(0,0)[lt]{\lineheight{1.25}\smash{\begin{tabular}[t]{l}${\color{myblue}\mathcal{M}_{a}}$\end{tabular}}}}%
    \put(0.36492547,0.23185663){\color[rgb]{0,0,0}\makebox(0,0)[lt]{\lineheight{1.25}\smash{\begin{tabular}[t]{l}${\color{mypurple}\mathcal{M}_{b}}$\end{tabular}}}}%
  \end{picture}%
\endgroup%

    }
    \caption{\hspace{-0.3em}) Nullbordism from $\mathbb{S}^1_{\mathcal{M}_1}$ to nothing, consisting in a gravitational soliton with appropriate monodromies along non-trivial 1-cycles.}
    \label{fig.ex2}
    \end{subfigure}
    \begin{subfigure}{\textwidth}
    \centering
    \resizebox{0.66\textwidth}{!}{%
\begingroup%
  \makeatletter%
  \providecommand\color[2][]{%
    \errmessage{(Inkscape) Color is used for the text in Inkscape, but the package 'color.sty' is not loaded}%
    \renewcommand\color[2][]{}%
  }%
  \providecommand\transparent[1]{%
    \errmessage{(Inkscape) Transparency is used (non-zero) for the text in Inkscape, but the package 'transparent.sty' is not loaded}%
    \renewcommand\transparent[1]{}%
  }%
  \providecommand\rotatebox[2]{#2}%
  \newcommand*\fsize{\dimexpr\f@size pt\relax}%
  \newcommand*\lineheight[1]{\fontsize{\fsize}{#1\fsize}\selectfont}%
  \ifx\svgwidth\undefined%
    \setlength{\unitlength}{267.03948734bp}%
    \ifx\svgscale\undefined%
      \relax%
    \else%
      \setlength{\unitlength}{\unitlength * \real{\svgscale}}%
    \fi%
  \else%
    \setlength{\unitlength}{\svgwidth}%
  \fi%
  \global\let\svgwidth\undefined%
  \global\let\svgscale\undefined%
  \makeatother%
  \begin{picture}(1,0.29498378)%
    \lineheight{1}%
    \setlength\tabcolsep{0pt}%
    \put(0,0){\includegraphics[width=\unitlength,page=1]{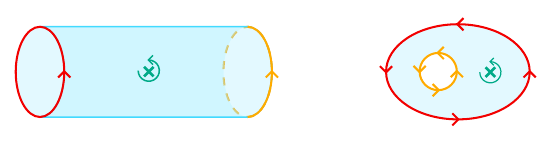}}%
    \put(0.54991185,0.16391363){\color[rgb]{0,0,0}\makebox(0,0)[lt]{\lineheight{1.25}\smash{\begin{tabular}[t]{l}$\Longleftrightarrow$\end{tabular}}}}%
    \put(0.412903,0.01546136){\color[rgb]{0,0,0}\makebox(0,0)[lt]{\lineheight{1.25}\smash{\begin{tabular}[t]{l}${\color{myred}\mathcal{M}_1}={\color{mygreen}\mathcal{M}_{a}}{\color{myyellow}\mathcal{M}_2}$\end{tabular}}}}%
    \put(0.13422677,0.15063486){\color[rgb]{0,0,0}\makebox(0,0)[lt]{\lineheight{1.25}\smash{\begin{tabular}[t]{l}${\color{myred}\mathcal{M}_1}$\end{tabular}}}}%
    \put(0.24095282,0.20399777){\color[rgb]{0,0,0}\makebox(0,0)[lt]{\lineheight{1.25}\smash{\begin{tabular}[t]{l}${\color{mygreen}\mathcal{M}_a}$\end{tabular}}}}%
    \put(0.40595662,0.14923088){\color[rgb]{0,0,0}\makebox(0,0)[lt]{\lineheight{1.25}\smash{\begin{tabular}[t]{l}${\color{myyellow}\mathcal{M}_{2}}$\end{tabular}}}}%
  \end{picture}%
\endgroup%

    }
    \caption{\hspace{-0.3em}) Bordism between $\mathbb{S}^1_{\mathcal{M}_1}$ and $\mathbb{S}^1_{\mathcal{M}_2}$, with a codimension-two brane of monodromy $\mathcal{M}_a=\mathcal{M}_1\mathcal{M}_2^{-1}$ as cobordism defect.}
    \label{fig.ex3}
    \end{subfigure}
    \begin{subfigure}{\textwidth}
    \centering
    \resizebox{0.66\textwidth}{!}{%
    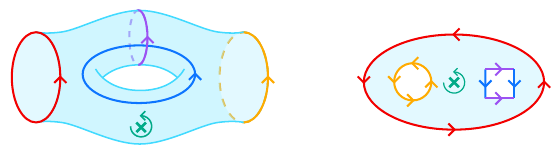
    }
    \caption{\hspace{-0.3em}) Bordism between $\mathbb{S}^1_{\mathcal{M}_1}$ and $\mathbb{S}^1_{\mathcal{M}_2}$, with both non-trivial topology (in terms of a gravitational soliton with appropriate monodromies) and a codimension-two brane defect.}
    \label{fig.ex4}
    \end{subfigure}
    \caption{Sketch of various bordisms between different theories resulting from compactifying on $\mathbb{S}^1$ with non-trivial monodromy and/or nothing. We show different possibilities, where the trivialization of the monodromies is achieved either through gravitational solitons (i.e, base $\tilde{\mathcal{B}}_2$ with genus $g>0$) or codimension-two defects. We show the closing relation \eqref{eq.closing} along the different monodromies, as well as the fundamental polygons through which is is made more explicit.}
    \label{fig:examples}
\end{figure}

We now turn to the elliptically fibered bordism, $\mathbb{E}\hookrightarrow\mathcal{B}_4\xrightarrow{\pi}\tilde{\mathcal{B}}_2$. By the classification of Riemann surfaces, we have that $\tilde{\mathcal{B}}_2$ is characterized by its genus $g\in\mathbb{Z}_{\geq 0}$ and the number $n_b$ boundaries (i.e., disks/points removed, $n_b\in\{1,2\}$ for our purposes).\footnote{Since we are considering the F-theory picture of $\Omega_1^{\mathsf{Spin}}({\rm B}\mathsf{SL}(2,\mathbb{Z}))$, in order for the bordism to keep the $\mathsf{Spin}$ structure $\tilde{\mathcal{B}}_2$ must be orientable.} Requiring the overall $\mathcal{B}_4$ to be smooth, the elliptic fiber will be smooth, $\pi^{-1}(x)\simeq\mathbb{T}^2$ in all points of $\tilde{\mathcal{B}}_2$ save for a finite number of them, $\{x_i\}_{i=1}^{n_s}\subset \tilde{\mathcal{B}}_2$ where $\pi^{-1}(x_i)$ is singular. We can then consider a $\tilde{\mathcal{B}}_2^*\simeq \tilde{\mathcal{B}}_2-\{x_i\}_{i=1}^{n_s}$ genus-$g$ Riemann surface with $n=n_b+n_s$ punctures, around which we have monodromies $\{\mathcal{M}_j\}_{j=1}^{n=n_b+n_s}$. Similarly, for $g>1$, around the non-trivial $2g$ 1-cycles spanned by the $g$ handles of our surface we can have non-trivial monodromies $\{\mathcal{M}_{a_i},\mathcal{M}_{b_i}\}_{i=1}^g$.\footnote{For our purposes, we will study the $\mathsf{SL}(2,\mathbb{Z})$ action as monodromies the 1-cycles of the base generating $H_1(\tilde{\mathcal{B}}_2^*;\mathbb{Z})$. One could also take $\mathsf{SL}(2,\mathbb{Z})$ transition function along the Poincar\'e dual 1-cycles of $\tilde{\mathcal{B}}_2^*$. While along this paper directly working with the monodromies is more straightforward, for higher dimensionalities of the base transition functions might be more tractable, see e.g., \cite{McNamara:2022lrw,Debray:2023yrs,Braeger:2025kra,Chakrabhavi:2025bfi}.} This generates a $\mathsf{SL}(2,\mathbb{Z})$ action on $H_1(\mathbb{T}^2;\mathbb{Z})$ of the non-singular fiber,
\begin{equation}\label{eq. action}
\begin{array}{rccc}
\rho:&\pi_1\big(\tilde{\mathcal{B}}_2^*\big)\simeq\mathbb{Z}^{\ast2g+n-1}&\longrightarrow
    &\mathsf{SL}(2,\mathbb{Z})  \\
     & \gamma&\longmapsto&\mathcal{M}_\gamma
\end{array}
\end{equation}
generated by the monodromy along the 1-cycles spanning $\pi_1\big(\tilde{\mathcal{B}}_2^*\big)$. Now, in order for the above action to be well-defined, the following closing relation must hold:

\begin{equation}\label{eq.closing}
\prod_{j=1}^n\mathcal{M}_{c_j}\prod_{i=1}^g[\mathcal{M}_{a_i},\mathcal{M}_{b_i}]=\mathrm{Id}_2\,.
\end{equation}
For our purposes in order to describe nullbordisms and cobordisms interpolating between different theories with monodromies $\mathcal{M}_1$ and $\mathcal{M}_2$ along the boundaries, we write
\begin{equation}
    \left\{
    \begin{array}{rll}
       \mathcal{M}_1=\prod_{i=1}^g[\mathcal{M}_{a_i},\mathcal{M}_{b_i}]  & \text{for nullbordisms} \\
       \mathcal{M}_1\mathcal{M}_2^{-1}=\prod_{i=1}^g[\mathcal{M}_{a_i},\mathcal{M}_{b_i}]  & \text{for interpolating bordisms}
    \end{array}
    \right.\quad,
\end{equation}
where the inverse $\mathcal{M}_1$ comes from the fact that we are entering the base $\tilde{\mathcal{B}}_2$ from ``outside'', and thus has an opposite orientation to the rest, see Figure \ref{fig:examples}.

As it was the case for the twisted tori $\mathbb{E}\hookrightarrow T^3_{\mathcal{M}}\to\mathbb{S}^1$ in \eqref{eq. hom C3}, a non-trivial action $\rho:\pi_1(\tilde{\mathcal{B}}_2^*)\to\mathsf{SL}(2,\mathbb{Z})$ can reduce the homology of the generic fiber $\mathbb{T}^2$, so that our bordism is not a Cartesian product and K\"unneth formula does not apply, $H_k(\mathcal{B}_4;\mathbb{Z})\not\simeq \oplus_{p+l=k}H_p(\mathbb{T}^2;\mathbb{Z})\oplus H_k(\tilde{B}_2;\mathbb{Z})$. Instead, through the use of Leray-Serre spectral sequence (see Appendix \ref{app.leraySerre} for more details), one obtains

\begin{equation}\label{eq. hom bord4}
    \tcbhighmath{\begin{array}{l}
         H_0(\mathcal{B}_4;\mathbb{Z})\simeq\mathbb{Z}\,,\qquad
    H_1(\mathcal{B}_4;\mathbb{Z})\simeq\mathbb{Z}^{\oplus2g+n_b+n_s-1}\oplus(\mathbb{Z}^{\oplus 2})_{\rho}\,,\\
    H_2(\mathcal{B}_4;\mathbb{Z})\simeq\mathbb{Z}^{\oplus4g+2(n_b+n_s)-3}\oplus(\mathbb{Z}^{\oplus 2})_{\rho}\,,\\
    H_3(\mathcal{B}_4;\mathbb{Z})\simeq\mathbb{Z}^{\oplus2g+n_b+n_s-1}\,,\qquad
    H_4(\mathcal{B}_4;\mathbb{Z})\simeq0\,,
    \end{array}   }
\end{equation}
where we define the coinvariant of the $\rho:\pi_1(\tilde{\mathcal{B}}_2^*)\to\mathsf{SL}(2,\mathbb{Z})$ action as\footnote{
As explained in Appendix \ref{app.leraySerre}, the closing relation \eqref{eq.closing} makes one of monodromy matrices redundant. Furthermore, in order to compute the coinvariant $(\mathbb{Z}^{\oplus 2})_\rho$, we employ the Smith normal form of the $2\times(2g+n)$-matrix
\begin{equation}
    \hat{\mathcal{M}}=\begin{pmatrix}
        \mathcal{M}_{a_1}-{\rm Id}_2|\cdots|\mathcal{M}_{a_g}-{\rm Id}_2|\mathcal{M}_{b_1}-{\rm Id}_2|\cdots|\mathcal{M}_{b_g}-{\rm Id}_2|\mathcal{M}_{c_1}-{\rm Id}_2|\cdots|\mathcal{M}_{c_n}-{\rm Id}_2
    \end{pmatrix}\,,
\end{equation}
so that for some matrices $\hat{\mathcal{S}}$ and $\hat{\mathcal{T}}$, $\hat{\mathcal{S}}\hat{\mathcal{M}}\hat{\mathcal{T}}=\begin{psmallmatrix}d_1&0&\cdots &0\\
        0&d_2&\dots &0\end{psmallmatrix}\,$. This way $${\langle(\rho(\gamma)-{\rm Id}_2)v\,:\,\gamma\in\pi_1(\Sigma_{g,n}),\,v\in\mathbb{Z}^{\oplus 2}\rangle}=d_1\mathbb{Z}\oplus d_2\mathbb{Z}\,,$$ and as a result$(\mathbb{Z}^{\oplus 2})_\rho\simeq\mathbb{Z}_{d_1}\oplus\mathbb{Z}_{d_2}$, where we denote $\mathbb{Z}_0\equiv\mathbb{Z}$ and $\mathbb{Z}_1\equiv\{0\}$.
}
\begin{equation}\label{eq. coinv}
   (\mathbb{Z}^{\oplus 2})_\rho=\frac{\mathbb{Z}^{\oplus 2}}{\bigoplus_{\gamma\in\pi_1(\tilde{\mathcal{B}}_2^*)}{\rm im}(\mathcal{M}_\gamma-{\rm Id}_2)} \,.
\end{equation}
Note that then, knowing the genus $g$ of the base and the number of points over which the fiber is singular, the homology of $\mathcal{B}_4$ its set by $(\mathbb{Z}^{\oplus 2})_\rho$.\\

Before continuing with the explicit bordisms between Type IIB 9d gauged SUGRAs, we can explicitly check how \eqref{eq. hom C3} and \eqref{eq. hom bord4} fulfill the general bounds \eqref{eq.betti ineq} and \eqref{eq.prime dec} for the Betti numbers and torsion of $\mathcal{B}_4$. Starting with the free part, we find
\begin{subequations}
    \begin{align}
        b_0(\mathcal{B}_4)+{b_3(\mathcal{B}_4)}&=2g+n_b+n_s\geq n_b=b_0(\mathcal{C}_3)=b_3(\mathcal{C}_3)\\
        b_1(\mathcal{B}_4)+ {b_2(\mathcal{B}_4)}&=3(2g+n_b+n_s)+2(\rank (\mathbb{Z^{\oplus}})_\rho-2)\notag\\&\geq n_b+\sum_{j=1}^{n_b}\rank \ker(\mathcal{M}_i-{\rm Id}_2)\geq b_1(\mathcal{C}_3)=b_2(\mathcal{C}_3)\,,\label{eq. second free ineq}
    \end{align}
\end{subequations}
where the second inequality follows from the closing relation \eqref{eq.closing}.\footnote{To see this, notice that, since $\rank\ker(\mathcal{M}_i-{\rm Id}_2)\leq 2$, we have that $n_b\leq b_1(\mathcal{C}_3)=b_2(\mathcal{C}_3)\leq 3n_b$, which is saturated when all the monodromies of the boundary components are trivial. It is clear that, regardless of the value of $g$, $n_s$ and $\rank (\mathbb{Z^{\oplus}})_\rho$, the second inequality in \eqref{eq. second free ineq} holds if there are two or more $\mathcal{M}_i$ matrices different from the identity.

Suppose then that all the such matrices are $\rm Id_2$. Then we must have that $\rank (\mathbb{Z^{\oplus}})_\rho=2$ unless there is some monodromy along some internal cycle or some singularity of the fiber. The former case means $g\geq 1$, which automatically implies inequality \eqref{eq. second free ineq}. As for the later, we have that the closing relation \eqref{eq.closing} means that either we have a single singularity ($n_s=1$) with trivial monodromy, or at least two singularities, which again implies \eqref{eq. second free ineq}.

Consider finally the case where there is a single monodromy along the boundary which is non-trivial. This means that $b_1(\mathcal{C}_3)=b_2(\mathcal{C}_3)\leq 3n_b-1$ or $3n_b-2$. Now, this implies $\rank (\mathbb{Z^{\oplus}})_\rho\leq 1$, but further, due to the closing relation \eqref{eq.closing}, there must be at least one singularity or pair of internal 1-cycles with monodromy, which means $n_s$ or $g\geq 1$, which in turn results in \eqref{eq. second free ineq}. 

} It is easy to see that the only way to saturate the above inequalities is if $g=n_s=0$ and $\rank (\mathbb{Z^{\oplus}})_\rho=\rank \ker (\mathcal{M}_i-{\rm Id}_2)=2$ for all $i$. In other words, the different boundaries have trivial monodromy and $\mathcal{B}_4\simeq\mathbb{T}^2\times(\mathbb{S}^2-\{x_i\}_{i=1}^{n_b})$, so we are interpolating between ungauged 9d type IIB supergravities (and/or nothing).

As for the torsion, we have that the torsion parts entering in \eqref{eq.prime dec} are
\begin{subequations}
    \begin{align}
        T_1(\mathcal{C}_3)&\simeq \bigoplus_{i=1}^{n_b}{\rm Tor}\left(\frac{\mathbb{Z}^{\oplus 2}}{{\rm im}(\mathcal{M}_i-{\rm Id}_2)}\right)\\
        T_1(\mathcal{B}_4)\simeq T_2(\mathcal{B}_4)&\simeq {\rm Tor}(\mathbb{Z}^{\oplus 2})_\rho= {\rm Tor} \frac{\mathbb{Z}^{\oplus 2}}{\bigoplus_{\gamma\in\pi_1(\tilde{\mathcal{B}}_2)}{\rm im}(\mathcal{M}_\gamma-{\rm Id}_2)}\,.
    \end{align}
\end{subequations}
Since $T_1(\mathcal{B}_4)\simeq T_2(\mathcal{B}_4)$, it is clear that \eqref{eq.prime dec} is trivially fulfilled by choosing $x_p=0$ for all primes $p$. Note how in general there is no clear relation between the torsion of the boundary and the torsion of the bordism, since the sum of the torsions of the individual ${\mathbb{Z}^{\oplus 2}}/{{\rm im}(\mathcal{M}_i-{\rm Id}_2)}$ of the boundary will generically be larger than that of the whole coinvariant $(\mathbb{Z}^{\oplus 2})_\rho$ \eqref{eq. coinv} once one considers the additional contributions from the singularities in the bulk of $\tilde{\mathcal{B}}_2$ and monodromies along the internal 1-cycles.

\subsection{Gravitation solitons and the commutator subgroup}

After having computed the general expression of the homology of our $\mathbb{E}\hookrightarrow\mathcal{B}_4\to\tilde{\mathcal{B}}_2$ bordisms and before studying the particular interpolating $\mathcal{B}_4$'s between different type IIB 9d gauged SUGRAs and/or nothing, we take a quick detour to comment some nuances on the perks of allowing for a non-trivial base $\tilde{\mathcal{B}}_2$.

As discussed on the previous section, having a non-trivial (discrete) duality $G$-bundle (with $G\equiv\mathsf{SL}(2,\mathbb{Z})$ for our purposes) induces a $G$-action \eqref{eq. action} $\rho:\pi_1(\tilde{\mathcal{B}}_2^*)\to G$ (where we allow duality defects on $\tilde{\mathcal{B}}_2$ inducing additional monodromies) acting on the fields of our theory (from the F-theory perspective of course this implies a $\mathsf{SL}(2,\mathbb{Z})$ transformation on the elliptic fiber), together with a closing relation \eqref{eq.closing}. Furthermore, if $\tilde{\mathcal{B}}_2$ has non-contractible (pairs of) 1-cycles $\{\gamma_{a_i},\gamma_{b_i}\}_{i=1}^g$, $G$-monodromies $\{\mathcal{M}_{a_i},\mathcal{M}_{b_i}\}_{i=1}^g$ around them can realize $\{[\mathcal{M}_{a_i},\mathcal{M}_{b_i}]\}_{i=1}^g$ as one winds around each of the $g$ handles of $\tilde{\mathcal{B}}_2$, which would be equivalent to a duality defect with the appropriate $\mathcal{M}_i=[\mathcal{M}_{a_i},\mathcal{M}_{b_i}]$ monodromy.\\

This means that there exists a subset of monodromies of $G$ that can be directly realized by non-trivial topology (and appropriate monodromies) on the base $\tilde{\mathcal{B}}_2$ without the need of codimension-2 duality defects. This precisely corresponds to the \emph{commutator} or \emph{derived} group
\begin{equation}\label{eq. comm group}
    G':=[G,G]=\langle\{[a,b]=aba^{-1}b^{-1}\,|\,a,b\in G\}\rangle\,,
\end{equation}
whose elements consist in commutators of elements in $G$.
This way, we expect that Quantum Gravity allows for topology changes that break the symmetry from $G$ by quotienting out the commutator subgroup \cite{McNamara:2021cuo}, down to the \textbf{abelianization} of $G$, defined as
\begin{equation}\label{eq. abel def}
    G^{\rm ab}=\faktor{G}{G'}\,.
\end{equation}
By definition, $G^{\rm ab}$ is Abelian and thus admits a decomposition $G^{\rm ab}\simeq\mathbb{Z}^r\oplus\bigoplus_{p}\Big(\bigoplus_i\mathbb{Z}_{p^{s_{p,i}}}\Big)$. We will see in Section \ref{sec: other groups} that for theories with a lot of supersymmetry $G^{\rm ab}$ is pure torsion, while for lower number of supercharges it can have free part.

The abelianization $G^{\rm ab}$ corresponds with the first integral homology of $G$, $H_1(G;\mathbb{Z})=G^{\rm ab}$, see e.g., \cite[Chapter II]{Brown1982} and, since the monodromies of $G^{\rm ab}$ \emph{cannot} be realized through pure gravitational solitons, $\mathbb{S}^1$ with a monodromy on $G^{\rm ab}$ represent non-trivial classes of the first cobordism group. More precisely, for $\mathsf{Spin}$-bordisms with $G$-duality bundle, one can use that \cite{anderson1967structure,Yonekura:2022reu}
\begin{equation}\label{eq.reduced bord}
    \Omega_k^{\mathsf{Spin}}({\rm B}G)\simeq \Omega_k^{\mathsf{Spin}}({\rm pt})\oplus\widetilde{\Omega}_k^{\mathsf{Spin}}({\rm B}G)\;,\quad \text{with}\qquad
        \begin{array}{|c|cccccccc|}
        \hline
           k  & 0&1&2&3&4&5&6&7 \\\hline
            \Omega_k^{\mathsf{Spin}}({\rm pt}) &\mathbb{Z} &\mathbb{Z}_2&\mathbb{Z}_2&0&\mathbb{Z}&0&0&0\\\hline 
        \end{array}\;,
\end{equation}
where $\widetilde{\Omega}_k^{\mathsf{Spin}}({\rm B}G)$ is the \emph{reduced bordism} group, which can be understood as the remainder of $\Omega_k^{\mathsf{Spin}}({\rm B}G)$ after quotienting out the $\mathsf{Spin}$ defects. In the case $k=1$ one has \cite{Yonekura:2022reu}
\begin{equation}\label{eq.1spinBord}
    \Omega_1^{\mathsf{Spin}}({\rm B}G)\simeq \underbrace{H_0(G,\mathbb{Z}_2)}_{ \Omega_1^{\mathsf{Spin}}({\rm pt})\simeq \mathbb{Z}_2}\oplus \underbrace{H_1(G,\mathbb{Z})}_{ \widetilde{\Omega}_1^{\mathsf{Spin}}({\rm B}G)}\simeq \mathbb{Z}_2\oplus G^{\rm ab}\;,
\end{equation}
with precisely the needed bordism defects being the $\mathsf{Spin}$ one killing $\mathbb{S}^1$ with periodic boundary conditions (see \cite{McNamara:2019rup,Hamada:2025duq}) and the codimension-2 duality defects generating the $G^{\rm ab}$ monodromies.\\

For the concrete $G=\mathsf{SL}(2,\mathbb{Z})$ case of interest, it is not difficult to show (see \cite{ConradSL2Z} for a series of relevant results on $\mathsf{SL}(2,\mathbb{Z})$) that $\mathsf{SL}(2,\mathbb{Z})'$ is isomorphic to the rank-2 free group $F_2$, generated by, e.g.,
\begin{align}\label{eq.sl2z com}
    \mathsf{SL}(2,\mathbb{Z})'&=\left\langle\text{\small$\begin{pmatrix}
        1&1\\1&2
    \end{pmatrix}=\left[\begin{pmatrix}
        3&2\\1&1
    \end{pmatrix},\begin{pmatrix}
        2&1\\-1&0
    \end{pmatrix}\right],\begin{pmatrix}
        1&-1\\-1&2
    \end{pmatrix}=\left[\begin{pmatrix}
        3&-1\\-2&   1
    \end{pmatrix},\begin{pmatrix}
        3&-1\\1&0
    \end{pmatrix}\right]$}\right\rangle\;,
\end{align}
Note that $\mathsf{SL}(2,\mathbb{Z})$ is not a \textit{perfect group}, this is, $\mathsf{SL}(2,\mathbb{Z})\neq  \mathsf{SL}(2,\mathbb{Z})'$,\footnote{Given that $\mathsf{SL}(2,\mathbb{Z})'$ is a free group, it does not have torsion elements such as minus the identity matrix, which is in $ \mathsf{SL}(2,\mathbb{Z})$.} and thus $\mathsf{SL}(2,\mathbb{Z})^{\rm ab}\neq0$ and there are elements that cannot be expressed as a the product of commutators. Following the $\mathsf{SL}(2,\mathbb{Z})$ presentation introduced in Footnote \ref{fn.presentationSL2Z}, we find
\begin{equation}\label{eq. ab SL2Z}
    \mathsf{SL}(2,\mathbb{Z})^{\rm ab}=\faktor{\mathsf{SL}(2,\mathbb{Z})}{\mathsf{SL}(2,\mathbb{Z})'}=\Big\langle\left.\overline{T}=\overline{\begin{psmallmatrix}
        1&1\\0&1    \end{psmallmatrix}}\right|\overline{T}^{12}=e\Big\rangle\simeq\mathbb{Z}_{12}\simeq\mathbb{Z}_4\oplus\mathbb{Z}_3\,,
\end{equation}
precisely generated by the lift of the D7-brane monodromy.\footnote{\label{fn.all same class}Note that, since the monodromy \eqref{eq.monpq} of any $[p,q]$ 7-brane is conjugated to that of a D7-brane,
\begin{equation}
        \mathcal{M}_{[p,q]}=\begin{pmatrix}
            1+pq&p^2\\-q^2&1-pq
        \end{pmatrix}=\begin{pmatrix}
            p&r\\-q&s
        \end{pmatrix}\begin{pmatrix}
            1&1\\0&1
        \end{pmatrix}\begin{pmatrix}
            p&r\\-q&s
        \end{pmatrix}^{-1}\,,
    \end{equation}
    all $[p,q]$ 7-branes (including the D7-brane) have the same image under ${\rm ab}:\mathsf{SL}(2,\mathbb{Z})\to \mathsf{SL}(2,\mathbb{Z})^{\rm ab}$. This means that the monodromy around any $[p,q]$ 7-brane can be obtained by a single D7-brane and some set of gravitational solitons.
}

From \eqref{eq.1spinBord}, this means  
\begin{equation}
    \tilde{\Omega}_1^{\mathsf{Spin}}(B\mathsf{SL}(2,\mathbb{Z}))\simeq H_1(\mathsf{SL}(2,\mathbb{Z}),\mathbb{Z})\simeq\mathsf{SL}(2,\mathbb{Z})^{\rm ab}\simeq \mathbb{Z}_{12}
\end{equation}
generated by $\mathbb{S}^1$ with $T={\begin{psmallmatrix}
    1&1\\0&1
\end{psmallmatrix}}$ monodromy, see also \cite{Debray:2023yrs}, which cannot be bounded by a pure gravitational soliton, regardless of the monodromy along the internal 1-cycles.

In Appendix \ref{App. real comm} we show the expression of the monodromies that need to be realized by the bordism connecting the boundaries from Table \ref{tab. hom T3} and/or nothing, in terms of $\mathsf{SL}(2,\mathbb{Z})'$ elements and ``remainders'' given by the monodromies associated to at most 11 (anti-)~D7 branes. We also show the realization of said elements when larger monodromy groups, such as $\mathsf{GL}(2,\mathbb{Z})$ are considered, something that will be discussed in more detail in Section \ref{ss. mp gl}.
\begin{tcolorbox}[enhanced jigsaw,breakable,pad at break*=1mm,colback=cyan!5!white,colframe=cyan!90!black,title={Showing $T^{12}\in\mathsf{SL}(2,\mathbb{Z})^{\rm ab}$}]
We can show this explicitly by noting that
\begin{equation}\label{eq. T12 mono}
    T^{12}=\begin{pmatrix}
        1&12\\0&1
    \end{pmatrix}=\Bigg[\!\begin{pmatrix}
        1&1\\1&2
    \end{pmatrix},\begin{pmatrix}
        3&1\\-1&0
    \end{pmatrix}\!\Bigg]^2\,,
\end{equation}
so that $T^{12}$ indeed can be written as the product of (the same two) commutators, which from now on we denote collectively by $\mathbf{T^{12}}$.

Now, one could wonder whether two is the minimal number of commutators needed, or if a single one is possible. To show that indeed this is not the case, by \emph{reductio ad absurdum} consider $X,Y\in\mathsf{SL}(2,\mathbb{Z})$, with $[X,Y]=T^{12}$, so that $XY=T^{12}YX$. Consider now $YX=\begin{psmallmatrix}
    a&b\\c&d
\end{psmallmatrix}\in\mathsf{SL}(2,\mathbb{Z})$, such that
\begin{equation}
    {\rm Tr}(XY)={\rm Tr}(YX)=a+d\,,\qquad{\rm Tr}\big(T^{12}YX\big)={\rm Tr}\begin{psmallmatrix}
        a&12a+b\\c&12c+d
    \end{psmallmatrix}=a+12c+d\,,
\end{equation}
which forces $d=0$, in such a way that $YX=\begin{psmallmatrix}
    \pm1&b\\0&\pm1
\end{psmallmatrix}$ is upper-triangular, where we have used that $YX\in\mathsf{SL}(2,\mathbb{Z})$. On the other hand, since $[X,Y]YX=X(YX)X^{-1}=T^{12}YX$ must be also upper-triangular, if $X=\begin{psmallmatrix}
    \alpha&\beta\\\gamma&\delta
\end{psmallmatrix}\in\mathsf{SL}(2,\mathbb{Z})$, then 
\begin{equation}
    X(YX)X^{-1}=\begin{pmatrix}
    \pm\alpha  \delta -\gamma  (\alpha  b\pm\beta ) & \alpha ^2 b \\
 -b \gamma ^2 & \pm\alpha  \delta +\gamma  (\alpha  b\mp\beta )
\end{pmatrix}\,,
\end{equation}
which forces $\gamma=0$, making $X$ also upper-diagonal. Since upper-diagonal matrices are a subgroup of $\mathsf{SL}(2,\mathbb{Z})$, then $Y=X^{-1}T^{12}YX$ would also be upper triangular. But any pair of upper-diagonal matrices commute, which would result in $[X,Y]=T^{12}={\rm Id}_2$, which is certainly not the case! We thus conclude that $T^{12}$ cannot be written as a single commutator (but it is indeed the square of one).\\

We have just shown that $T^{12}$ can be expressed as a minimum of two commutators (actually the same commutator squared). One might then feel tempted to conclude that $T^{12k}$, for $k\in\mathbb{Z}_{>0}$, needs $2k$ commutators to be expressed, as $T^{12k}=\begin{psmallmatrix}
        1&12k\\0&1
    \end{psmallmatrix}=\big[\!\begin{psmallmatrix}
        1&1\\1&2
    \end{psmallmatrix},\begin{psmallmatrix}
        3&1\\-1&0
    \end{psmallmatrix}\!\big]^{2k}$. However, as shown in \cite[Example 2.6]{CULLER1981133}, for any group $G$ and $u,\,v\in G$, $[u,v]^k$ can be expressed as the product of $\big\lfloor\frac{k}{2}\big\rfloor+1$ commutators (generally different ones). This way it is clear that generally, 
    \begin{equation}\label{eq.kplus1comm}
        \textit{$T^{12k}$ can be expressed as the product of $k+1$ commutators.}
    \end{equation}   
    We will denote such arrangement (of which we have no closed expression) as $(\mathbf{T^{12}})^{k}$.
\end{tcolorbox}

\subsection{Elliptically fibered bordisms, branes, and gravitational solitons\label{eq. gen results}} 

As described in previous subsections, interpolating from a 9d type IIB gauged SUGRA associated to a twisted $T^3_{\mathcal{M}}$ to another with $T^3_{\mathcal{M}'}$ (equivalently nothing, if $\mathcal{M}'={\rm Id}_2$, since the base circle has no obstruction to close) need the bulk of the bordism $\tilde{\mathcal{B}}_2$ to implement a non-trivial $\mathsf{SL}(2,\mathbb{Z})$ monodromy $\mathcal{M}\mathcal{M}'{}^{-1}$ in order for the $\mathsf{SL}(2,\mathbb{Z})$ action \eqref{eq. action} (and thus the bordism) to be well defined. In Table \ref{tab.int} we show how such bordism monodromies can be classified in four different sets, which include parabolic, elliptic and hyperbolic matrices from Section \ref{ss:Fth}, as well as their products and inverses. Note that in general the monodromy that needs to be implemented by the bordism can be hyperbolic, i.e., $|\rm Tr(\mathcal{M}\mathcal{M}'{}^{-1})|>2$, and thus is not realized in the Kodaira classification of Table \ref{tab:kodaira pq}

\begin{table}[htb]
    \centering
    \begin{tabular}{|c||c|c|c|}\hline
       \diagbox{$\mathcal{M}'$}{$\mathcal{M}$}  & \cellcolor{yellow!15}\scriptsize$\begin{pmatrix}1&n\\0&1\end{pmatrix}$& \cellcolor{green!15}\scriptsize$\begin{pmatrix}-1&0\\0&-1\end{pmatrix}$&\cellcolor{blue!15} \scriptsize$\begin{pmatrix}n&1\\-1&0\end{pmatrix}$ \\\hline\hline
       \cellcolor{yellow!15} \scriptsize$\begin{pmatrix}1&m\\0&1\end{pmatrix}$&\cellcolor{yellow!15}\scriptsize$\begin{pmatrix}1&n-m\\0&1\end{pmatrix}$ &\cellcolor{green!15}\scriptsize$\begin{pmatrix}-1&m\\0&-1\end{pmatrix}$ &\cellcolor{blue!15}\scriptsize$\begin{pmatrix}n&1-nm\\-1&m\end{pmatrix}$\\\hline
       \cellcolor{green!15}\scriptsize$\begin{pmatrix}-1&0\\0&-1\end{pmatrix}$&\cellcolor{green!15}\scriptsize$\begin{pmatrix}-1&-n\\0&-1\end{pmatrix}$ &\cellcolor{yellow!15}\scriptsize$\begin{pmatrix}1&0\\0&1\end{pmatrix}$ &\cellcolor{red!15}\scriptsize$\begin{pmatrix}-n&-1\\1&0\end{pmatrix}$\\\hline
       \cellcolor{blue!15}\scriptsize$\begin{pmatrix}m&1\\-1&0\end{pmatrix}$&\cellcolor{red!15}\scriptsize$\begin{pmatrix}n&nm-1\\1&m\end{pmatrix}$ & \cellcolor{blue!15}\scriptsize$\begin{pmatrix}0&1\\-1&-m\end{pmatrix}$ &\cellcolor{yellow!15}\scriptsize$\begin{pmatrix}1&m-n\\0&1\end{pmatrix}$\\\hline       
    \end{tabular}
    \caption{Monodromies $\mathcal{M}\mathcal{M}'{}^{-1}$ that need to be implemented by the bordisms $\tilde{\mathcal{B}}_2^*$ in order to interpolate from $\mathbb{S}^1_\mathcal{M}$ to $\mathbb{S}^1_{\mathcal{M}'}$, with the nullbordism case is captured by $\mathcal{M}'=\begin{psmallmatrix}
        1&0\\0&1\end{psmallmatrix}$. We arrange the different matrices by their structure, identifying it by different colors. For example, it is clear that for the representatives from Section \ref{ss:Fth}, parabolic monodromies appear in \colorbox{yellow!15}{yellow}, while parabolic elliptic and hyperbolic in \colorbox{blue!15}{blue} for $n\in\{-1,0,1\}$ and $|n|>3$, respectively.
    \label{tab.int}}
\end{table}

In order to realize such monodromies, two approaches can be considered. First of all, we can take a trivial base with zero genus, so that $\tilde{\mathcal{B}}_2\simeq\mathbb{S}^1\times[0,1]$ for bordisms between theories and $\mathbb{D}^2$ for nullbordisms. Since $\tilde{\mathcal{B}}_2$ lacks non-trivial 1-cycles, we must relay on codimension-2 duality defects \cite{Delgado:2024skw} that perform $\mathsf{SL}(2,\mathbb{Z})$ monodromies around them. This is precisely the role of (stacks of) $[p,q]$-7-branes, with
\begin{equation}
    \mathcal{M}_{[p,q]}=\begin{pmatrix}
        1+pq&p^2\\-q^2&1-pq
    \end{pmatrix}\,,
\end{equation}
with $[1,0]$-7-branes the usual D7-branes. Since we can realize $U=\begin{psmallmatrix}
    1&1\\-1&0
\end{psmallmatrix}=\mathcal{M}_{[1,0]}\mathcal{M}_{[1,1]}$ and $T=\begin{psmallmatrix}
    1&1\\0&1
\end{psmallmatrix}=\mathcal{M}_{[1,0]}$ in the presentation from Footnote \ref{fn.presentationSL2Z}, any element of $\mathsf{SL}(2,\mathbb{Z})$ can thus be realized by winding around an arrangement of (anti)$[p,q]$-7-branes. From an F-theory point of view, their location corresponds to the point of the base over which the elliptic fiber becomes singular (even if $\mathcal{B}_4$ as a whole remains smooth). Since $\mathbb{E}\hookrightarrow\mathcal{B}_4\to\tilde{\mathcal{B}}_2$ can be understood as a smooth elliptic surface (with boundaries), such singular fibers fall under the classification of Kodaira and N\'eron \cite{Kodaira+2015+1142+1156,Neron}, see also \cite[Chapter 4]{Weigand:2018rez} in the minimal case where $\mathcal{B}_4$ does not have -1-curves. See Table \ref{tab:kodaira pq} for the set of Kodaira singularities, together with the associated monodromy $\mathcal{M}$, ADE gauge group realized and realization in terms of a basis of $[p,q]$-7-branes.

\begin{table}[hbt!]
    \centering
    \resizebox{\textwidth}{!}{
    \begin{tabular}{|c||c|c|c|c|c|c|c|c|}\hline
        Kodaira & $I_n$ &$II$  &$III$  &$IV$  &$I_n^*$  &$IV^*$  &$III^*$  & $II^*$  \\\hline
        $\mathcal{M}$ & $\begin{psmallmatrix}1&n\\0&1\end{psmallmatrix}$ & $\begin{psmallmatrix}1&1\\-1&0\end{psmallmatrix}$ & $\begin{psmallmatrix}0&1\\-1&0\end{psmallmatrix}$ & $\begin{psmallmatrix}-1&1\\-1&0\end{psmallmatrix}$& $\begin{psmallmatrix}-1&-n\\0&-1\end{psmallmatrix}$ & $\begin{psmallmatrix}-1&-1\\1&0\end{psmallmatrix}$ & $\begin{psmallmatrix}0&-1\\1&0\end{psmallmatrix}$ & $\begin{psmallmatrix}1&-1\\1&0
        \end{psmallmatrix}$   \\\hline        
        ADE group & $\mathsf{SU}(n)$  & 1  & $\mathsf{SU}(2)$  &$\mathsf{SU}(3)$  & $\mathsf{SO}(8+2n)$  & $\mathsf{E}_6$ & $\mathsf{E}_7$  &   $\mathsf{E}_8$ \\\hline
        $[p,q]$-stack & $A^n$  & $AC$  & $A^2C$ &$A^3C$  & $A^{n+4}BC$  &$A^5BC^2$  &$A^6BC^2$  & $A^7BC^2$  \\\hline
    \end{tabular}}
    \caption{Set of Kodaira singularities (with $n\geq 0$), together with the associated ADE group, monodromy and realization in term of a stack of $[p,q]$-7-branes, where $A\equiv[1,0]$-7-brane, $B\equiv[3,1]$-7-brane and $C\equiv[1,1]$-7-brane, respectively with monodromies $\mathcal{M}_A=\begin{psmallmatrix}1&1\\0&1\end{psmallmatrix}$, $\mathcal{M}_B=\begin{psmallmatrix}4&9\\-1&-2\end{psmallmatrix}$ and $\mathcal{M}_C=\begin{psmallmatrix}2&1\\-1&0\end{psmallmatrix}$.}
    \label{tab:kodaira pq}
\end{table}

Given some monodromy $\mathcal{M}=\begin{psmallmatrix}
    a&b\\c&d
\end{psmallmatrix}\in\mathsf{SL}(2,\mathbb{Z})$, we can associate a measure of its ``complexity'',\footnote{Note that the trace of the matrix is not a good indicator of this, since all parabolic matrices $\begin{psmallmatrix}
    1&k\\0&1
\end{psmallmatrix}$ have trace 2, but it is clear that for large $k$ they implement larger monodromies, and thus they do not have the same ``complexity'' or charge.} given by the spectral norm $\|\mathcal{M}\|_2$ (this is, the square root of the largest eigenvalue of $\mathcal{M}^\dagger\mathcal{M}$, which for large entries can be approximated by the Frobenius norm $\|\mathcal{M}\|_2\approx\sqrt{a^2+b^2+c^2+d^2}$), in such a way that for the monodromies of Table \ref{tab.int}
\begin{subequations}
  \begin{align}
      \left\|\begin{psmallmatrix}
        1&n\\0&1
    \end{psmallmatrix}\right\|_2&=\left[1+\tfrac{n^2}{2}\Big(1+\sqrt{1+\tfrac{4}{n^2}}\Big)\right]^{\frac{\text{sign}(n)}{2}}=\left\{\begin{array}{ll}
       1+\mathcal{O}(|n|)  & \text{for } |n|\sim 1
         \\ 
       |n|+\mathcal{O}(|n|^{-1})  & \text{for } n\to\pm\infty
    \end{array}\right.\;,\\\left\|\begin{psmallmatrix}
        n&mn-1\\1&m
    \end{psmallmatrix}\right\|_2&=\tfrac{\sqrt{m^2 n^2+(m-n)^2+2+\sqrt{\left(m^2 n^2+(m-n)^2+2\right)^2-4}}}{\sqrt{2}}\notag\\
    &=\left\{\begin{array}{ll}
       1+\mathcal{O}(|m|,|n|)  & \text{for } |n|,|m|\sim 1
         \\ 
       |m n|+\mathcal{O}(|m|^{-1},|n|^{-1})  & \text{for } |n|\text{ or }|m|\to\pm\infty
    \end{array}\right.\;,
  \end{align}  
\end{subequations}
and as expected monodromies $\mathcal{M}$ with large entries will have a large spectral norm $\|\mathcal{M}\|_2$. We can relate this with the ``charge'' of the $[p,q]$-branes needed to realize such monodromy. Barring elliptic monodromies, with $|\rm Tr(\mathcal{M})|<2$, which from Table \ref{tab:kodaira pq} are immediately realized by a single brane stack, let us consider parabolic $|\rm Tr(\mathcal{M})|=2$ and hyperbolic, $|\rm Tr(\mathcal{M})|>2$ monodromies. The former can be realized by (a) stack(s) of identical or mutually local (i.e., with commuting monodromy matrices) 7-branes, themselves parabolic. Now, for a fixed $\mathcal{M}_{[p,q]}$ and large $k>0$, $\|\mathcal{M}^k_{[p,q]}\|_2\approx k \|\mathcal{M}_{[p,q]}\|_2=\mathcal{O}(k)$. On the other hand, for hyperbolic matrices we cannot have mutually local 7-brane stacks, but precisely this allows us to obtain the desired trace/spectral norm with less brane stacks. Take for example  $\mathcal{M}_{[1,0]}=\begin{psmallmatrix}1&1\\0&1\end{psmallmatrix}$ and  $\mathcal{M}_{[0,1]}=\begin{psmallmatrix}1&0\\-1&1\end{psmallmatrix}$. We find that $\mathcal{M}_*=\mathcal{M}_{[1,0]}^3\mathcal{M}_{[0,1]}^2=\begin{psmallmatrix}-5&3\\-2&1\end{psmallmatrix}$, which is clearly hyperbolic. Now, since
\begin{equation}
    {\rm Tr}(\mathcal{M}_*^k)=\big(-\sqrt{3}-2\big)^k+\big(\sqrt{3}-2\big)^k\sim (-1)^k(2+\sqrt{3})^{k}\,,\quad \|\mathcal{M}_*^k\|_2\sim\sqrt{\tfrac{37}{12}}(2+\sqrt{3})^k\,,
\end{equation}
for large $k$. It is then clear that, unlike in the parabolic case, it is possible to have an exponential growth of the trace/spectral norm of the monodromy in terms of the number of 7-branes.\footnote{While on Table \ref{tab:trivbase} we only consider realizations of monodromies through Kodaira singularities, but as just discussed one can consider general words written by $[p,q]$ 7-branes, which allow for large hyperbolic norm much faster.} The above approach is in general not possible for parabolic monodromies, since the inclusion of non-local brane stacks usually kills the ${\rm Tr}(\mathcal{M})=2$ condition.\\
We thus conclude that, for a given monodromy $\mathcal{M}$, the number $N_{\mathcal{M}}$ of $[p,q]$ 7-branes needed to realize such monodromy grows like
\begin{equation}\label{eq. bound no branes}
    \mathcal{O}(\log\|\mathcal{M}\|_2)\lesssim N_{\mathcal{M}}\lesssim\mathcal{O}(\|\mathcal{M}\|_2)\;,
\end{equation}
which, as we will see in the following section, will have important physical consequences.\\

So far we have only discussed a trivial base $\tilde{B}_2$ for the bordisms, where the realization of the monodromies along the boundary is fully performed by $[p,q]$ 7-branes. However, as discussed in the previous section, we can also consider the effect of \emph{gravitational solitons} by taking $\tilde{B}_2$ to have non-trivial topology, i.e., genus $g>0$, in such a way that non-trivial $\mathsf{SL}(2,\mathbb{Z})$ monodromies along the $g$ pairs dual 1-cycles of the base realizes elements of $\mathsf{SL}(2,\mathbb{Z})'$. With this in mind, only defects associated to $\mathsf{SL}(2,\mathbb{Z})^{\rm ab}\simeq\mathbb{Z}_{12}$ would be needed, associated, e.g., to the monodromy of D7-branes, with $\mathcal{M}_{[1,0]}=T=\begin{psmallmatrix}
    1&1\\0&1
\end{psmallmatrix}$. One then has that under this consideration the complete set of $[p,q]$ 7-branes would not be needed, and at most 11 D7- (or $\overline{\text{D7}}$-)branes suffice, since $T^{12}$ can be expressed as the product of two commutators (i.e., two handles in $\tilde{B}_2$), see \eqref{eq. T12 mono}.

Now, regarding the genus needed to realize a given monodromy $\mathcal{M}$ (modulo the at most 11 $[p,q]$ 7-branes, see Footnote \ref{fn.all same class}), write $\mathcal{M}=\overline{\mathcal{M}}\,T^{a}$, with $T^a$ a representative in $\{T^a\}_{a=0}^{11}$ of ${\rm ab}:\mathsf{SL}(2,\mathbb{Z})\to \mathsf{SL}(2,\mathbb{Z})^{\rm ab}$,\footnote{As explained in Footnote \ref{fn.all same class}, we could take as representative any $[p,q]$ 7-brane. We will take D7-branes (equivalently [1,0]-branes), since their monodromies are the one with smallest spectral norm $\|\mathcal{M}_{[1,0]}\|=\|\mathcal{M}_{[0,1]}\|=\frac{1+\sqrt{5}}{2}$, and we want the ``heavy lifting'' of the monodromy realization to be done by the gravitational solitons.} and $\overline{\mathcal{M}}\in \mathsf{SL}(2\,,\mathbb{Z})'$. Take then
\begin{equation}\label{eq. min alpha0}
    X=\overline{\mathcal{M}}T^{-12 \alpha_0}\,,\quad \text{with}\quad \alpha_0\,\text{ solution to }\,\text{minimize }\|X\|_2=\|\overline{\mathcal{M}}T^{-12 \alpha}\|_2\;\text{   for }\; \alpha\in\mathbb{Z}\,,
\end{equation}
where the above minimization problem is well defined since $T^{12}$ is not elliptic and $\|\overline{\mathcal{M}}T^{-12 \alpha}\|_2$ is a convex function.\footnote{As a matter of fact, solving \eqref{eq. min alpha0} is quite straightforward. Taking $\overline{\mathcal{M}}=\begin{psmallmatrix}
    a&b\\c&d
\end{psmallmatrix}$ and considering for simplicity $\alpha\in\mathbb{R}$, we have that
\begin{equation}
    \|\overline{\mathcal{M}}T^{-12\alpha}\|_2^2\approx a^2+(b-12\alpha a)^2+c^2+(d-12\alpha c)^2\Longrightarrow\partial_\alpha\|\overline{\mathcal{M}}T^{12\alpha}\|_2^2\approx 24[(ab+cd)-12(a^2+c^2)\alpha]\,,
\end{equation}
so that $\alpha_{0}=\lfloor-\frac{ab+cd}{12(a^2+c^2)}\rfloor$ or $\lceil-\frac{ab+cd}{12(a^2+c^2)}\rceil$, now an integer. Plugging back shows that $ \|\overline{\mathcal{M}}T^{-12\alpha_0}\|_2^2\geq 2$.
} Now, using $\|AB\|_2\leq\|A\|_2\|B\|_2$ and that $\|X\|_2^2\gtrsim2$, we find that the genus needed to realize $\overline{\mathcal{M}}$, $g_\mathcal{M}\geq\frac{|\alpha_0|}{6}$ (since at least we know how it would grow in terms of the two handles needed to realize $T^{12}$),  grows like
\begin{equation}\label{eq. growth of genus}
    g_{\mathcal{M}}\gtrsim \mathcal{O}(\|\overline{\mathcal{M}}\|_2)=\mathcal{O}(\|{\mathcal{M}}\|_2)\,.
\end{equation}
As a sanity check, for $\mathcal{M}=T^n$ we have $g_{T^n}=\lfloor\frac{n}{12}\rfloor+1=\mathcal{O}(n)=\mathcal{O}(\|T^n\|_2)$ for any $n$, see . We thus conclude that the genus $g$ of the bordism base $\tilde{\mathcal{B}}_2$ grows linearly with the spectral norm of monodromy $\mathcal{M}$ we want to realize, and thus is not bounded. We will see in Section \ref{sec: com len} why this becomes problematic for $\|{\mathcal{M}}\|_2\to\infty$.\\

In Tables \ref{tab:trivbase} and \ref{tab:nontrivbase} we show the realization of the different monodromies allowing the interpolation between 9d gauged supergravities associated to monodromies $\mathcal{M}_1$ and $\mathcal{M}_2$ (with nullbordisms to nothing, $\emptyset$ associated to $\mathcal{M}_2={\rm Id}_2$). We do this for trivial base $\tilde{\mathcal{B}}_2\simeq\mathbb{S}^1\times[0,1]$ or $\mathbb{D}^2$ in terms of different Kodaira singularities, while for non-trivial $\tilde{\mathcal{B}}_2$ we specified the needed genus $g>0$ and number of D7-branes. We additionally show the contribution of the coinvariants $(\mathbb{Z}^{\oplus 2})_\rho$ \eqref{eq. coinv} to the homology \eqref{eq. hom B4} of the bordism $\mathcal{B}_4$, since from \eqref{eq. hom bord4} $H_\bullet(\mathcal{B}_4;\mathbb{Z})$ is set by the number of punctures (be them boundaries or fiber singularities) and the action $\rho:\pi_1\big(\tilde{\mathcal{B}}_2^*\big)\to \mathsf{SL}(2,\mathbb{Z})$. It can be seen that said torsion contribution is trivial in the majority of cases once enough 7-brane defects or gravitation solitons are added.

\subsubsection*{A couple of comments on flat tori \texorpdfstring{$\mathbb{T}^3/\Gamma$}{T3/Gamma} quotients}

In \cite{Montero}, the authors considered nullbordisms of freely acting quotients of a flat torus $T^3/\Gamma$ as elliptic fibrations over a disk with punctures. There are six possible geometries $G1,G2,\dots,G6$, and 28 different possibilities considering spin structures, see \cite{Pfaffle2000,Acharya:2019mcu}. They explicitly discussed the nullbordisms and the Weierstrass models of the $G1$ and $G3$ geometries, see Table \ref{tab:geom}. The former is given by a rational elliptic surface $dP_9$, also known as half-K3,\footnote{Unlike in our case, where the elliptic fibration appears in the context of the F-theory interpretation of out type IIB theories, in reference \cite{Montero} the three compact dimensions are part of spacetime. Furthermore, \cite{Montero} explicitly deals with the spin structure of the geometries, and considers $G_1$ (i.e. $\mathbb{T}^3$) with periodic spin structure. Since $\Omega_3^{\mathsf{Spin}}(\rm pt)=0$ \cite{anderson1967structure}, such nullbordism exists without the need of considering spin defects, but such spin structure prevents $\mathcal{B}_4$ from simply consisting in one of the 1-cycles shrinking to a point, rather consisting in the above geometry with 12 $I_1$ singularities.} with a disk removed, while the latter is given by an elliptic fibration over a disk with a $IV^*$ singularity. 

  The majority of the flat torus quotients are given by torus bundles over a circle with quotient action in $\mathsf{PSL}(2,\mathbb{Z})$. However, these are not the three-dimensional solvmanifolds that we are considering in this work. Indeed, while the flat torus quotients admit a Riemann-flat metric, twisted tori in three dimensions are, in general, not Ricci-flat, and the corresponding algebra is solvable, not abelian. Additionally, while in the geometries in Table \ref{tab:geom} the quotient action can fix a constant complex structure $\tau$ for the fiber, this is not the case for our twisted tori, where the monodromy acts on the type IIB axio-dilaton.

\begin{landscape}
\begin{table}
    \centering
    \resizebox{22.5cm}{!}{
    \begin{tabular}{|c||c|c|c|c|c|c|}\hline
         \diagbox{$\mathcal{M}_2$}{$\mathcal{M}_1$}& {\scriptsize$\begin{pmatrix}1&n\\0&1\end{pmatrix}\,,\quad n\geq 0$} &{\scriptsize$\begin{pmatrix}-1&0\\0&-1\end{pmatrix}$}  &{\scriptsize$\begin{pmatrix}-1&1\\-1&0\end{pmatrix}$}  &{\scriptsize$\begin{pmatrix}0&1\\-1&0\end{pmatrix}$}  &{\scriptsize$\begin{pmatrix}1&1\\-1&0\end{pmatrix}$}  &{\scriptsize$\begin{pmatrix}n&1\\-1&0\end{pmatrix}\,,\quad |n|\geq 3$} \\\hline\hline
         {\scriptsize$\begin{pmatrix}1&m\\0&1\end{pmatrix}$}& $\scriptstyle\begin{array}{c}\left\{\begin{array}{ll}
             \mathbb{Z}\oplus\mathbb{Z}_{\gcd(n,m)} & \text{if }n\geq m \\
             0 & \text{if }n<m
         \end{array}\right.\\\left\{\begin{array}{ll}
             I_{n-m} & \text{if }n\geq m \\
             (II+III)^{m-n}+(I_0^*)^{m-n\mod 2} & \text{if }n<m
         \end{array}\right.\end{array}$ &\cellcolor{gray!15}  &\cellcolor{gray!15}  &\cellcolor{gray!15}  &\cellcolor{gray!15} &\cellcolor{gray!15} \\\hline
         {\scriptsize$\begin{pmatrix}-1&0\\0&-1\end{pmatrix}$}&$\scriptstyle\begin{array}{c}\mathbb{Z}_2\oplus\mathbb{Z}_{\gcd(n,2)}
             \\I_n+I_0^*\end{array}$  & $\scriptstyle\begin{array}{c}\mathbb{Z}_2\oplus\mathbb{Z}_2\\-\end{array}$ &\cellcolor{gray!15}  & \cellcolor{gray!15} &\cellcolor{gray!15}  &\cellcolor{gray!15} \\\hline
         {\scriptsize$\begin{pmatrix}-1&1\\-1&0\end{pmatrix}$}& $\scriptstyle\begin{array}{c}\mathbb{Z}_{\gcd(n,3)}\\I_n+IV+IV\end{array}$  &  $\scriptstyle\begin{array}{c}0\\III+II+III^*\end{array}$& $\scriptstyle\begin{array}{c}\mathbb{Z}_3\\-\end{array}$ &  \cellcolor{gray!15}&  \cellcolor{gray!15}&\cellcolor{gray!15} \\\hline
         {\scriptsize$\begin{pmatrix}0&1\\-1&0\end{pmatrix}$}&$\scriptstyle\begin{array}{c}\mathbb{Z}_{\gcd(n,2)}\\I_n+III^*\end{array}$  & $\scriptstyle\begin{array}{c}\mathbb{Z}_2\\III\end{array}$ &$\scriptstyle\begin{array}{c}0\\I\end{array}$  &$\scriptstyle\begin{array}{c}\mathbb{Z}_2\\-\end{array}$  &\cellcolor{gray!15}  &\cellcolor{gray!15} \\\hline
         
         {\scriptsize$\begin{pmatrix}1&1\\-1&0\end{pmatrix}$}& $\scriptstyle\begin{array}{c}0\\I_n+II+\stackrel{5}{\cdots}+II\end{array}$ &$\scriptstyle\begin{array}{c}0\\III+I_1\end{array}$  & $\scriptstyle\begin{array}{c}0\\I_2\end{array}$ &$\scriptstyle\begin{array}{c}0\\I\end{array}$  & $\scriptstyle\begin{array}{c}0\\-\end{array}$ & \cellcolor{gray!15}\\\hline
         
         {\scriptsize$\begin{pmatrix}m&1\\-1&0\end{pmatrix}$}& 
         $\scriptstyle\begin{array}{c}\left\{\begin{array}{ll}
             \mathbb{Z}_{\gcd(2,n,m)} & \text{if }m\geq 3 \\
             0 & \text{if }m\leq-3
         \end{array}\right. \\\left\{\begin{array}{cc}
             I_n+III^*+I_m &\text{if }m\geq 3  \\
             I_n+(III+II)^{-m}+(I_{0}^*)^{m-1\mod 2} &\text{if }m\leq-3 
         \end{array}\right.\end{array}$

         &

         $\scriptstyle\begin{array}{c}\left\{
         \begin{array}{ll}
             \mathbb{Z}_{\gcd(2,m)} & \text{if }m\geq 3 \\
             0 & \text{if }m\leq-3
         \end{array}\right.\\\left\{\begin{array}{ll}
            III+I_m  &\text{if }m\geq3  \\
            (III+II)^{-m}+III+(I_0^*)^{m\mod 2}  &\text{if }m\leq -3 
         \end{array}\right.\end{array}$

         &  $\scriptstyle\begin{array}{c}\left\{\begin{array}{ll}
              \mathbb{Z}_{\gcd(3,m+1)}& \text{if }m\geq 3 \\
              0& \text{if }m\leq-3 
         \end{array}\right.\\\left\{\begin{array}{ll}
            I_{m+1}  &\text{if }m\geq 3  \\
            (II+III)^{-m-1}+(I_0^*)^{m+1\mod 2}  &\text{if }m\leq-3 
         \end{array}\right.\end{array}$&$\scriptstyle\begin{array}{c}\left\{
         \begin{array}{ll}
             \mathbb{Z}_{\gcd(2,m)} & \text{if }m\geq 3 \\
             0 & \text{if }m\leq-3
         \end{array}\right. \\\left\{\begin{array}{cc}
            I_m  &\text{if }m\geq 3  \\
           (II+III)^{-m}+(I_0^*)^{m\mod 2}  &\text{if }m\leq -3
         \end{array}\right.\end{array}$  & $\scriptstyle\begin{array}{c}0\\\left\{\begin{array}{cc}
            I_{m-1}  &\text{if }m\geq 3  \\
           (II+III)^{1-m}+(I_0^*)^{m-1\mod 2}  &\text{if }m\leq -3 
         \end{array}\right.\end{array}$ & $\scriptstyle\begin{array}{c}\left\{\begin{array}{ll}
            \mathbb{Z}_{\gcd(n,m,2)}  & \text{if }n>m \\
            \mathbb{Z}_{n-2}  &\text{if } n=m\\
            0&\text{if }n<m
         \end{array}\right.\\\left\{\begin{array}{cc}
            I_{m-n}  &\text{if }m\geq n  \\
            (II+III)^{n-m}+(I_0^*)^{m-m\mod 2}  &\text{if }m<n 
         \end{array}\right.\end{array}$\\\hline
         
        $\emptyset$ &$\scriptstyle\begin{array}{c}\mathbb{Z}\oplus\mathbb{Z}_n\\I_n\end{array}$  & $\scriptstyle\begin{array}{c}\mathbb{Z}_2\oplus\mathbb{Z}_2\\I_0^*\end{array}$ &$\scriptstyle\begin{array}{c}\mathbb{Z}_3\\IV\end{array}$  & $\scriptstyle\begin{array}{c}\mathbb{Z}_2\\III\end{array}$ &$\scriptstyle\begin{array}{c}0\\II\end{array}$  &
        $\scriptstyle\begin{array}{c}0\;\text{ if }n>0\,,\qquad \mathbb{Z}_{\gcd(n,2)}\;\text{ if }n<0\\ \left\{\begin{array}{cc}
II+(III+II)^{n-1}+(I_0^*)^{n-1\mod 2}  &\text{if }n\geq 3  \\
            I_{-n}+III  &\text{if }n\leq -3
    \end{array}\right.\end{array}$

    \\\hline
    \end{tabular}
    }
    \caption{
    coinvariants $(\mathbb{Z}^{\oplus2})_\rho$ and required Kodaira singularities considering a trivial basis, i.e., $\tilde{\mathcal{B}}_2\simeq\mathbb{D}^2$ for nullbordisms and $\tilde{\mathcal{B}}_2\simeq\mathbb{S}^1\times[0,1]$ for bordisms between gauged 9d supergravities. We denote $(II+III)^n=II+III+\stackrel{n}{\cdots}+II+III$, similar for $(III+II)^n$. We recall that $\mathbb{Z}_1\simeq 0$ and $\mathbb{Z}_0\simeq \mathbb{Z}$.
    } 
    \label{tab:trivbase}
\end{table}

\begin{table}
    \centering
    \resizebox{22.5cm}{!}{
    \begin{tabular}{|c||c|c|c|c|c|c|}\hline
         \diagbox{$\mathcal{M}_2$}{$\mathcal{M}_1$}& {\scriptsize$\begin{pmatrix}1&n\\0&1\end{pmatrix}\,,\quad n\geq 0$} &{\scriptsize$\begin{pmatrix}-1&0\\0&-1\end{pmatrix}$}  &{\scriptsize$\begin{pmatrix}-1&1\\-1&0\end{pmatrix}$}  &{\scriptsize$\begin{pmatrix}0&1\\-1&0\end{pmatrix}$}  &{\scriptsize$\begin{pmatrix}1&1\\-1&0\end{pmatrix}$}  &{\scriptsize$\begin{pmatrix}n&1\\-1&0\end{pmatrix}\,,\quad |n|\geq 3$} \\\hline\hline
         {\scriptsize$\begin{pmatrix}1&m\\0&1\end{pmatrix}$}& $\scriptstyle\begin{array}{c}\left\{\begin{array}{ll}
           \mathbb{Z}\oplus\mathbb{Z}_{\gcd(n,m)}   &\text{if }\,12+m>n\geq m  \\
           \mathbb{Z}_2\oplus\mathbb{Z}_{\gcd(2,n,m)}&\text{if }\,12+n>m>n  \\
             0 &\text{otherwise} 
         \end{array}\right.\\g=\left|2\left\lfloor\frac{n-m}{12}\right\rfloor\right|\,,\quad I_{n-m\mod 12}\end{array}$ &\cellcolor{gray!15}  &\cellcolor{gray!15}  &\cellcolor{gray!15}  &\cellcolor{gray!15} &\cellcolor{gray!15} \\\hline
         {\scriptsize$\begin{pmatrix}-1&0\\0&-1\end{pmatrix}$}&$\scriptstyle\begin{array}{c}0\,,\quad g=2+2\left|\left\lfloor\frac{n-6}{12}\right\rfloor\right|\\\quad I_{n-6\mod 12}\end{array}$  & $\scriptstyle\begin{array}{c}\mathbb{Z}_2\oplus\mathbb{Z}_2\\g=0\,,\quad-\end{array}$ &\cellcolor{gray!15}  & \cellcolor{gray!15} &\cellcolor{gray!15}  &\cellcolor{gray!15} \\\hline
         {\scriptsize$\begin{pmatrix}-1&1\\-1&0\end{pmatrix}$}& $\scriptstyle\begin{array}{c}0\,,\quad g=1+2\left|\left\lfloor\frac{n-4}{12}\right\rfloor\right|\\\quad I_{n-4\mod 12}\end{array}$  &  $\scriptstyle\begin{array}{c}0\,,\quad g=1\\I_2\end{array}$& $\scriptstyle\begin{array}{c}\mathbb{Z}_3\\g=0\,,\quad-\end{array}$ &  \cellcolor{gray!15}&  \cellcolor{gray!15}&\cellcolor{gray!15} \\\hline
         {\scriptsize$\begin{pmatrix}0&1\\-1&0\end{pmatrix}$}&$\scriptstyle\begin{array}{c}0\,,\quad g=2+2\abs{\left\lfloor\frac{n-3}{12}\right\rfloor}\\\quad I_{n-3\mod 12}\end{array}$  & $\scriptstyle\begin{array}{c}0\,,\quad g=1\\I_2\end{array}$ &$\scriptstyle\begin{array}{c}0\,,\quad g=0\\I_1\end{array}$  &$\scriptstyle\begin{array}{c}\mathbb{Z}_2\\g=0\,,\quad-\end{array}$  &\cellcolor{gray!15}  &\cellcolor{gray!15} \\\hline
         {\scriptsize$\begin{pmatrix}1&1\\-1&0\end{pmatrix}$}& $\scriptstyle\begin{array}{c}0\,,\quad g=1+2\abs{\left\lfloor\frac{n-4}{12}\right\rfloor}\\\quad I_{n-4\mod 12}\end{array}$ &$\scriptstyle\begin{array}{c}0\,,\quad g=1\\I_4\end{array}$  & $\scriptstyle\begin{array}{c}0\,,\quad g=0\\I_2\end{array}$ &$\scriptstyle\begin{array}{c}0\,,\quad g=1\\I\end{array}$  & $\scriptstyle\begin{array}{c}0\\g=0\,,\quad-\end{array}$ & \cellcolor{gray!15}\\\hline
         {\scriptsize$\begin{pmatrix}m&1\\-1&0\end{pmatrix}$}& $\scriptstyle\begin{array}{c}0\,,\quad g=1+2\abs{\left\lfloor\frac{n+m-3}{12}\right\rfloor}\\\quad I_{n+m-3\mod 12}\end{array}$ &$\scriptstyle\begin{array}{c}0\,,\quad g=1+2\abs{\left\lfloor\frac{9-m}{12}\right\rfloor}\\\quad I_{9-m\mod 12}\end{array}$ &  $\scriptstyle\begin{array}{c}0\,,\quad g=2\abs{\left\lfloor\frac{m+1}{12}\right\rfloor}\\\quad I_{m+1\mod 12}\end{array}$ &$\scriptstyle\begin{array}{c}\left\{\begin{array}{ll}
           \mathbb{Z}_2   &\text{if }\,|m|<12 \text{ and }m \text{ even} \\
             0 &\text{otherwise} 
         \end{array}\right.\\g=2\abs{\left\lfloor\frac{m}{12}\right\rfloor}\,,\quad I_{m \mod 12}\end{array}$  & $\scriptstyle\begin{array}{c}\left\{\begin{array}{ll}
           \mathbb{Z}_2   &\text{if }\,|m|<12 \text{ and }m \text{ odd} \\
             0 &\text{otherwise} 
         \end{array}\right.\\g=2\abs{\left\lfloor\frac{m-1}{12}\right\rfloor}\,,\quad I_{m-1\mod 12}\end{array}$ & $\scriptstyle\begin{array}{c}\left\{\begin{array}{ll}
           \mathbb{Z}_{\gcd(n-2,m-2)}   &\text{if }\,|n-m|<12 \\
             0 &\text{otherwise} 
         \end{array}\right.\\g=2\abs{\left\lfloor\frac{n-m}{12}\right\rfloor}\,,\quad I_{n-m\mod 12}\end{array}$\\\hline
        $\emptyset$ &$\scriptstyle\begin{array}{c}\left\{\begin{array}{ll}
           \mathbb{Z}\oplus\mathbb{Z}_n   &\text{if }\,n<12 \\
             0 &\text{otherwise} \end{array}\right.\\g=2\left\lfloor\frac{n}{12}\right\rfloor\,,\quad I_{n\mod 12}\end{array}$  & $\scriptstyle\begin{array}{c}0\,,\quad g=2\\I_6\end{array}$ &$\scriptstyle\begin{array}{c}0\,,\quad g=1\\I_4\end{array}$ & $\scriptstyle\begin{array}{c}0\,,\quad g=1\\I_3\end{array}$  &$\scriptstyle\begin{array}{c}0\,,\quad g=1\\I_2\end{array}$  &$\scriptstyle\begin{array}{c}0\,,\quad g=1+2\abs{\left\lfloor\frac{n-3}{12}\right\rfloor}\\I_{3-n\mod 12}\end{array}$\\\hline
    \end{tabular}}
    \caption{coinvariants $(\mathbb{Z}^{\oplus2})_\rho$ and required Kodaira singularities considering a basis of genus $g$, i.e., $\tilde{\mathcal{B}}_2\simeq\Sigma_g-\{\rm pt\}$ for nullbordisms and $\tilde{\mathcal{B}}_2\simeq\Sigma_g-\{{\rm pt}_i\}_{i=1}^2$ for bordisms between gauged 9d supergravities. Note that for all cases with $g>1$ the coinvariant is trivial. We recall that $\mathbb{Z}_1\simeq 0$ and $\mathbb{Z}_0\simeq \mathbb{Z}$.}
    \label{tab:nontrivbase}
\end{table}
\end{landscape}

\begin{table}[hbt!]
        \centering
        \begin{tabular}{|c|c|c|c|}\hline
             Class & $\mathsf{SL}(2,\mathbb{Z})$ action & Kodaira sing. & Fixed $\tau$ \\ \hline \hline 
             $G1$  & $\begin{psmallmatrix}
                 1 & 0 \\
                 0 & 1
             \end{psmallmatrix}$ & - / $(I_1)^{12}$ & - \\ 
             \hline                          $G2$   &\small $\begin{psmallmatrix}
                 -1 & 0 \\
                 0 & -1
             \end{psmallmatrix}$ & $I_0^*$ & - \\ 
             \hline
             $G3$   &$\begin{psmallmatrix}
                 -1 & 1 \\
                 -1 & 0
             \end{psmallmatrix}$ , $\begin{psmallmatrix}
                 -1 & -1 \\
                 1 & 0
             \end{psmallmatrix}$ & $IV,IV^*$ &$e^{\pm\frac{\pi}{3}i }$,  $e^{\pm\frac{2\pi}{3} i}$\\ 
             \hline
             $G4$   &$\begin{psmallmatrix}
                 0 & -1 \\
                 1 & 0
             \end{psmallmatrix}$ , $\begin{psmallmatrix}
                 0 & 1 \\
                 -1 & 0
             \end{psmallmatrix}$ & $III,III^*$ & $\pm i$\\ 
             \hline
             $G5$   & $\begin{psmallmatrix}
                 1 & -1 \\
                 1 & 0
             \end{psmallmatrix}$ , $\begin{psmallmatrix}
                 1 & 1 \\
                 -1 & 0
             \end{psmallmatrix}$ & $IV,IV^*$& $e^{\pm\frac{\pi}{3}i }$, $e^{\pm\frac{2\pi}{3} i}$\\ 
             \hline
            $G6$& $\begin{psmallmatrix}
                 -1 & 0 \\
                 0 & -1
             \end{psmallmatrix}$ and $\begin{array}{rccl}
                  \omega:&\mathbb{S}^1\times\mathbb{T}^2 &\longrightarrow&\mathbb{S}^1\times\mathbb{T}^2 \\
                 &(\theta,z) &\longmapsto& (-\theta,\tfrac{1+\tau}{2}+\overline{z}) 
             \end{array}$& $I_1+I_1+I_1+I_1$&-
             \\\hline
        \end{tabular}
        \caption{Flat tori quotients geometries from \cite{Acharya:2019mcu}, with their interpretation of torus fibration with a $\mathsf{SL}(2,\mathbb{Z})$ action. The table also shows the Kodaira singularity of the corresponding nullbordisms, as well as the fixed complex structure of the fiber. For $G_1$ the $\mathbb{T}^3$ with at least one 1-cycle with anti-periodic boundary conditions needs no singularities in the bordism, while if all circles are periodic then 12 singularities are needed. Notice that, due to the additional $\omega$ $\mathbb{Z}_2$-action, $G6$ is not topologically a twisted torus but rather a Klein bottle fibrated over an interval. Adapted from Table 1 of \cite{Montero}.}
        \label{tab:geom}
    \end{table}

\subsection{Weierstrass models and breaking of supersymmetry\label{ss. Weierstrass}}
We finish this section by noting that, while we have been able to qualitatively describe the topology of the bordisms $\mathcal{B}_4$ between the type IIB 9d gauged supergravities, we have not given an explicit metric to $\mathcal{B}_4$ or solved the equations of motion (see \cite{Montero} for detailed computation for a setting similar to ours, but where the base $\tilde{\mathcal{B}}_2$ has genus 0). It is easy to imagine that for a generic choice of the bordisms above described such task is quite complicated and might not have analytic expressions. However, as we will argue now, in general our bordisms will not correspond to a SUSY-preserving geometry.

In order to see this, we take into account that, if one wants to preserve SUSY and solve Einstein's equations, see e.g. \cite{Weigand:2018rez}, the (elliptically fibered) bordism $\mathbb{E}\hookrightarrow\mathcal{B}_4\to\tilde{\mathcal{B}}_2$ must be Calabi-Yau. Considering the holomorphic line bundle $\mathcal{L}$ defined by the (1,1)-form $\dd A$, with $A=\frac{i}{2}(\partial\Phi-\bar\partial\Phi)$ (here $\Phi$ is the 10d type IIB dilaton), then its first Chern class must match that of the base, $ c_1(\mathcal{L})=c_1(\tilde{\mathcal{B}}_2)$, so that $\mathcal{L}=-K_{\tilde{\mathcal{B}}_2}$, with
\begin{equation}
   \int_{\tilde{\mathcal{B}}_2}c_1(\tilde{\mathcal{B}}_2)=\chi(\tilde{\mathcal{B}}_2)=2-2g-n\,,
\end{equation}
through Chern-Gauss-Bonnet theorem, where $g$ is the genus of the base and $n\in\{0,1\}$ the number of points removed. On the other hand, given a Weierstrass model
\begin{equation}
    y^2=x^3+f(z)x+g(z)\quad\text{with}\quad f\in H^0(\tilde{\mathcal{B}}_2,-4K_{\tilde{\mathcal{B}}_2})\,,\quad g\in H^0(\tilde{\mathcal{B}}_2,-6K_{\tilde{\mathcal{B}}_2})\,,
\end{equation}
the discriminant $\Delta=4f^3+27g^2\in H^0(\tilde{\mathcal{B}}_2,-12K_{\tilde{\mathcal{B}}_2})$ is a meromorphic function inducing a divisor
\begin{equation}
    (\Delta)=\sum_{p\in \tilde{\mathcal{B}}_2} {\rm ord}_p(\Delta) p=\sum_{\rm zeros} {\rm ord}_p(\Delta) p- \sum_{\rm poles} |{\rm ord}_q(\Delta)|q\,,
\end{equation}
where we allow $\Delta$ to have poles (with negative order) at the points removed     . Now, $(\Delta)$ is a representative of the first Chern class $c_1(-12K_{\tilde{\mathcal{B}}_2})$ by taking, e.g., $c_1(-12K_{\tilde{\mathcal{B}}_2})=\frac{i}{2\pi}\partial\bar\partial\log|\Delta|^2$. Through the Calabi-Yau condition, we would have $c_1(-12K_{\tilde{\mathcal{B}}_2})=12c_1(\tilde{\mathcal{B}}_2)$. Integrating now on both sides, we have
\begin{equation}
    \sum_{\rm zeros} {\rm ord}_p(\Delta) - \sum_{\rm poles} |{\rm ord}_q(\Delta)|=12(2-2g-n)\,,
\end{equation}
in order for $\mathbb{E}\hookrightarrow\mathcal{B}_4\to\tilde{\mathcal{B}}_2$ to admit a Calabi-Yau metric.\footnote{As a sanity check, note that taking $n=g=0$, so that $\tilde{\mathcal{B}}_2\simeq\mathbb{S}^2\simeq\mathbb{P}^1$ (since there are no removed points we cannot have poles), results in $\sum_{\rm zeros} {\rm ord}_p(\Delta)=24$, this is, a $\mathbb{E}\hookrightarrow{\rm K3}\to\mathbb{P}^1$ geometry.} Back to our geometries of interest, the above translates in the following inequality for the number and order of the $\Delta$ zeroes (since the removed points in principle allow for the axio-dilaton flux to slip away to infinity)
\begin{equation}
    0\leq\sum_{\rm zeros} {\rm ord}_p(\Delta)\leq12(2-2g-n)=\left\{
    \begin{array}{rl}
        12-24g &  \text{if }n=1\\
      -24g   & \text{if }n=2
    \end{array}
    \right.\quad.
\end{equation}
It is clear that the above inequality cannot be fulfilled for $n>0$ if $g\neq 0$, meaning that all gravitational solitons break supersymmetry.\footnote{If we allow for a compact $\tilde{\mathcal{B}}_2$ without boundary (this is, $n=0$), then we have a solution with $g=1$ if $\sum_{\rm zeros} {\rm ord}_p(\Delta)=0$, what results in $\mathcal{B}_4$ being the trivial fibration $\mathcal{B}_4\simeq \mathbb{T}^2\times \mathbb{E}\simeq \mathbb{T}^4$, trivially Calabi-Yau.}

For nullbordisms (this is, $n=1$) we find that no Calabi-Yau geometry is possible for $g\geq 1$. On the other hand, for $g=0$ (i.e., $\tilde{\mathcal{B}}_2\simeq\mathbb{D}^2\simeq \mathbb{C}$) we can have at most a total order of 12 for $\Delta$.\footnote{For example, this would amount to a maximum of 12 D7-branes. Since each D7-brane has an effective deficit angle $\frac{\pi}{6}$, see e.g., \cite{Blumenhagen:2013fgp}, this means that fewer than 12 D7-branes would result in the deficit angle being less than $2\pi$, and thus the radius of the bordisms would grow linearly with the radial distance $\rho$ to the branes, $R\sim \frac{\pi}{6}(12-n)\rho$, which can be be directly linked with the existence of a pole at the point removed. This effectively prevents supersymmetric compactifications unless one saturates $\sum_{\rm zeros} {\rm ord}_p(\Delta)=12$. This precisely corresponds to the explicit construction of \cite{Montero}. Only the $\mathbb{S}^1$ compactification with $T^{12}$ monodromy saturates such inequality, and would in principle be allowed without explicitly breaking SUSY.} As for bordisms between two 9d supergravities (where $n=2$), we see that the only way to respect the Calabi-Yau condition is to have $g=0$ and $\sum_{\rm zeros} {\rm ord}_p(\Delta)=0$. Since this results in the trivial bordism, one concludes that it is \emph{not} possible to connect two \emph{different} 9d gauged SUGRAs in a \emph{supersymmetric way}. See \cite{Kleban:2007kk} for a related study on how type IIB settings with too many branes (thus breaking SUSY) result in time-dependent expanding solutions.

\section{Commutator lengths in gravitational solitons and possible problems} 
\label{sec: com len}

As discussed in the previous section, the fact that \eqref{eq.1spinBord} (again, see \cite{Yonekura:2022reu})
\begin{equation}
    \Omega_1^{\mathsf{Spin}}({\rm B}G)\simeq\mathbb{Z}_2\oplus G^{\rm ab}\,,
\end{equation}
where the first $\mathbb{Z}_2$ factor is generated by $\mathbb{S}^1$ with periodic boundary conditions, results from the possibility of realizing any element of $x=[a_1,b_1]\dots[a_g,b_g]\in G':=[G,G]$ by gravitational solitons of genus (at least) $g$, with each pair of monodromies from the $g$ commutators implemented along pairs of dual 1-cycles. This implies that only those monodromies that are not realized as products of commutators, (i.e., those of $G^{\rm ab}=G/G'$) cannot be realized by pure geometry of the bordism, and UV defects need to be included in order to trivialize the bordism.

For the case $G=\mathsf{SL}(2,\mathbb{Z})$, we then showed in \eqref{eq. growth of genus} that the genus of the gravitational soliton needed to realize a given monodromy $\mathcal{M}$ (trying to minimize the contribution from $\mathsf{SL}(2,\mathbb{Z})^{\rm ab}$ representatives) grows linearly with the spectral norm $\|\mathcal M\|_2$. Since for parabolic and hyperbolic monodromies in $\mathsf{SL}(2,\mathbb{Z})$ the norm is not bounded from above, this results in the need to have arbitrarily convoluted bordisms if we want to only use the defects required to trivialize $\mathsf{SL}(2,\mathbb{Z})^{\rm ab}\simeq \mathbb{Z}_{12}$, e.g., at most 11 D7-branes. In this section we will try to give more intuition to when an arbitrarily large genus are needed, as well as the problems it might entail.

\vspace{0.5cm}

Consider for this some group $G$ and its commutator subgroup $G'=[G,G]$. The \emph{commutator length} of a given element $g\in G'$ is defined as 
\begin{equation}\label{eq. com len}
    {\rm cl}_G(g)=\min\left\{n\,|\,g=[a_1,b_1]\dots[a_n,b_n]\,,\;\text{with}\; a_i,\,b_i\in G\right\}\;.
\end{equation}
For a given set $R\subset G$ of representatives for $G^{\rm ab}=G/G'$ (which from a cobordism conjecture point of view we can identify with the specific defects trivializing $\tilde{\Omega}_1^{\mathsf{Spin}}({\rm B}G)\simeq G^{\rm ab}$), we can extend this to any element $g\in G$ modding out the elements in $R$,\footnote{For example, for $G=\mathsf{SL}(2,\mathbb{Z})$, with $\mathsf{SL}(2,\mathbb{Z})^{\rm ab}\simeq \mathbb{Z}_{12}\simeq\langle \overline{T}\rangle$, we have as possible generators $\{T^a\}_{a\in\{1,5,7,11\}}$. We first write
\begin{equation}
    T^n=(T^{a})^{a n\ {\rm mod}\,12}(\mathbf{T^{12}})^{\tfrac{n-a(a n\ {\rm mod}\,12)}{12}}\,,
\end{equation}  
for $T^a$ ($a\in\{1,5,7,11\}$, where here $a^{-1}=a$ for all these $a$) as generator. Now, from \eqref{eq.kplus1comm} we know that $(\mathbf{T^{12}})^{\tfrac{n-a(a n\ {\rm mod}\,12)}{12}}$ can be written as the product of $\frac{n-a(a n\ {\rm mod}\,12)}{12}+1$ commutators, so ${\rm cl}_{\mathsf{SL}(2,\mathbb{Z})}(T^n)\lesssim \frac{n-a(a n\ {\rm mod}\,12)}{12}+1$ , which for large $n$ grows like ${\rm cl}_{\mathsf{SL}(2,\mathbb{Z})}(T^n)\sim \frac{n}{12}$, independently of the choice of $a$. See Theorem \ref{thm. large n} from Appendix \ref{app. results comm} to see how this is always the case for commutator subgroups of finite index.
} in such a way that we can define the \emph{commutator width} of $G$ as
\begin{equation}\label{eq. com w}
    {\rm cl}^{(R)}(G)=\sup_{g\in G'}\, {\rm cl}_G^{(R)}(g)\,.
\end{equation}
By definition, for our group of interest, we will have that ${\rm cl}(\mathsf{SL}(2,\mathbb{Z}))=\infty$ (independent of the choice of $R$ representing $\mathsf{SL}(2,\mathbb{Z})^{\rm ab}\simeq\mathbb{Z}_{12}$), which means that we will have elements in $\mathsf{SL}(2,\mathbb{Z})$ expressed as the product of an arbitrarily large number of commutators (and possibly up to 11 $\mathsf{SL}(2,\mathbb{Z})^{\rm ab}$ defects).\footnote{As we have seen, this is actually true for the smaller (and for us more physically relevant) $\exp\mathfrak{sl}(2,\mathbb{R})\cap \mathsf{SL}(2,\mathbb{Z})$, since $T^n$ belongs to said set.} This means that the gravitational solitons acting as cobordism for $\mathbb{S}^1$ with $\mathsf{SL}(2,\mathbb{Z})$ monodromies will have arbitrarily large genus if we only allow for $\mathsf{SL}(2,\mathbb{Z})^{\rm ab}\simeq \mathbb{Z}_{12}$ defects, as we already showed in \eqref{eq. growth of genus}.

\subsection{Bubbles of nothing with large \texorpdfstring{$\mathsf{SL}(2,\mathbb{Z})$}{SL(2,Z)} monodromy and their decay rate\label{ss. decay rate}}

We will now give an qualitative argument for why requiring gravitational instantons of arbitrarily high genus might pose important problems. Consider our tree-level action (including the boundary Gibbons-Hawking-York term),
\begin{equation}
    S_{\rm E}=\underbrace{-\frac{1}{2\kappa_{d+1}^2}\int_\mathcal{X}\dd^{d+1}x\sqrt{G}\left(\mathcal{R}_G-\frac{\partial_M\tau\partial^M\bar\tau}{2({\rm Im}\,\tau)^2}\right)}_{S_{\rm E}^{\rm (bulk)}}-\underbrace{\frac{1}{\kappa_{d+1}^2}\int_{\partial \mathcal{X}}\dd ^{d}x\sqrt{H}(K-K_0)}_{_{S_{\rm E}^{\rm (GHY)}}}\,,
\end{equation}
with $\tau$ the axio-dilaton and $K$ and $K_0$ the extrinsic curvature of the spacetime boundary $\mathcal{X}$, $\partial\mathcal{X}=\mathbb{S}^{d-1}\times\partial\mathcal{B}_2\simeq\mathbb{S}^{d-1}\times\mathbb{S}^1$. We will study the vacuum decay from a given 9d gauged SUGRA (i.e., compactification of 10d type IIB on $\mathbb{S}^1$ with a given monodromy $\mathcal{M}\in\mathsf{SL}(2,\mathbb{Z})$) to nothing, i.e., the nullbordism dynamically described by a generalization of Witten's Bubble of Nothing \cite{Witten:1981gj}. As we saw in Section \ref{ss.9dIIB}, the generalized Scherk–Schwarz reduction induces a non-trivial positive potential on the radion $\rho$ and axio-dilaton $\tau=C_0+ie^{-\Phi}$ that spontaneously breaks supersymmetry and drives our fields to the large volume and weak coupling limit for parabolic and hyperbolic monodromies, as well as a tree-level Minkowski vacuum with $\tau=i$ and flat $\sigma$ for parabolic monodromies, see Table \ref{tab:param subgroups}. For our purposes, we will study the bordism solution in the 10d theory, with the 9d theories corresponding to the low energy, lower-dimensional asymptotic description far from the bordism, which in 9d is interpreted as a domain wall. In line with what is expected from the interpretation as \emph{local dynamical cobordism} and \emph{End of the World Branes} \cite{Angius:2022aeq} (see also \cite{Dudas:2000ff,Dudas:2002dg,Dudas:2004nd,Hellerman:2010dv,Basile:2018irz,Mourad:2021roa,Buratti:2021yia,Buratti:2021fiv,Angius:2022mgh,Blumenhagen:2022mqw,Blumenhagen:2023abk,Huertas:2023syg,Hassfeld:2023kpu,Angius:2023uqk,Angius:2024pqk,Makridou:2026jzy}), in the limit where the $\mathbb{S}^1$ shrinks to a point, i.e., $\rho=\sqrt{\frac{8}{7}}\log\frac{R}{\ell_{10}}\to-\infty$, the scalar potential in the 9d theory diverges as $V(\rho)\sim e^{-4\sqrt{\frac{2}{7}}\rho}\to\infty$. However, from the 10d perspective such potential is not present, and it is enough to study the dynamics of the metric $G_{MN}$ and the axiodilaton $\tau=C_0+ie^{-\Phi}$.\\

For simplicity we will consider this setting with a single boundary component, but we expect that our results extrapolate to bordisms between two different theories. We will provide the general expression in $d$-dimensions, with $d=9$ for reduction to our 9d gauged supergravities.  In the same spirit as \cite{Montero,Ruiz:2024jiz}, we take a $\mathsf{SO}(d-1)$-symmetric ansatz to parameterize our geometry $\mathcal{X}$
\begin{equation}\label{eq. bord param}
    \dd s^2_\mathcal{X}=G_{MN}\dd x^M\dd x^N=W(y)^2R^2\dd\Omega^2_{d-1}+h_{\alpha\beta}(y)\dd y^\alpha\dd y^\beta\,,
\end{equation}
where we parameterize $\mathcal{B}_2$ by $y^1$ and $y^2$, with metric $h_{\alpha\beta}$, while $W(y)$ acts as a warping function depending on $y$, and $R$ finally is the nucleation scale, which is expected to be of the order of the KK scale. Taking the radial direction $r=RW(y)\in[0,\infty)$, close to the boundary of $\partial\mathcal{X}\simeq\mathbb{S}^{d-1}\times \partial\mathcal{B}_2\simeq \mathbb{S}^{d-1}\times \mathbb{S}^1$, this is $r\to\infty$, we can express the above metric through a Schwarzschild-like expression
\begin{equation}
    \dd s^2_\mathcal{X}\approx r^2\dd\Omega_{d-1}^2+\underbrace{f(r)^{-1}\dd r^2+f(r)\dd \theta^2}_{h_{\alpha\beta}(y)\dd y^\alpha\dd y^\beta}\;,
\end{equation}
with $\theta\in[0,2\pi R)$ an angular coordinate along the boundary $\partial\mathcal{B}_2\simeq\mathbb{S}^1$. The explicit expression for $f(r)$ is given in Appendix \ref{ap. onshell}.

Evaluating on-shell, the bulk part of the Euclidean action vanishes, $S_{\rm E}^{\rm (bulk)}=0$, so that $S_E=S_{\rm E}^{\rm (GHY)}$. Unlike in the original work by Witten \cite{Witten:1981gj}, the explicit expression of the $\mathcal{B}_2$ metric or the axio-dilaton profile is not known, what prevents us from obtaining a closed, analytic expression for $S_E$. However, as argued in Appendix \ref{ap. onshell}, we can infer the scaling with the genus $g$ and the $\{\mathcal{M}_{a_i},\mathcal{M}_{b_i}\}_{i=1}^{g}$ monodromies along the $2g$ 1-cycles $\{a_i,b_i\}_{i=1}^g$
\begin{equation}\label{eq. onshell action}
    S_{\rm E}=(M_{{\rm Pl},d+1}R)^{d-1}\mathcal{O}\bigg(1-2g+\underbrace{\sum_{i=1}^g\log\big|{\rm Tr}([\mathcal{M}_{a_i},\mathcal{M}_{b_i}])\big|}_{{\mathsf{a}(\mathcal{B}_2)}}\bigg)\,.
\end{equation}
Here $\mathsf{a}(\mathcal{B}_2)\gtrsim g$ gives an approximation of the area of image of $\mathcal{B}_2$ through the map $\tau:\mathcal{B}_2\to\mathbb{H}$. As a sanity check we recover $S_{\rm E}\sim(M_{{\rm Pl},d+1}R)^{d-1}\sim(M_{{\rm Pl},d}R)^{d-2}$ for Witten's BoN \cite{Witten:1981gj}, where $g=0$ and there are no monodromies. For the conversion between Planck scales we have used $M_{{\rm Pl},d}^{d-2}=M_{{\rm Pl},d+1}^{d-1}R$.

Now, we see that for genus $g>0$, the contribution $1-2g$ from the curvature (or topology, through Gauss-Bonnet theorem) and the dynamics of the axio-dilaton (through the monodromies) have opposite signs. We would then have, rewriting \eqref{eq. onshell action} as
\begin{align}\label{eq. onshell action 2}
    S_{\rm E}&\sim(M_{{\rm Pl},d+1}R)^{d-1}\big[C_{\rm grav}(1- 2g)+C_{ \tau}\mathsf{a}(\mathcal{B}_2)\big]\notag\\
    &\sim(M_{{\rm Pl},d}R)^{d-2}\big[C_{\rm grav}(1- 2g)+C_{ \tau}\mathsf{a}(\mathcal{B}_2)\big]\quad\text{for }\; g>0\,,
\end{align}
with both $C_{\rm grav}$ and $C_{ \tau}$ positive. In order for the BoN to be stable, we would need $S_{\rm E}>0$, i.e., $C_{\rm grav}(1- 2g)+C_{ \tau}\mathsf{a}(\mathcal{B}_2)>0$. For large $g$, this translates into a condition for the average trace of the monodromies,
\begin{equation}\label{eq. mean trace}
   \overline{\log\big|{\rm Tr}([\mathcal{M}_{a},\mathcal{M}_{b}])\big|}=\frac{\sum_{i=1}^g\log\big|{\rm Tr}([\mathcal{M}_{a_i},\mathcal{M}_{b_i}])\big|}{g}>\frac{2 C_{\rm grav}}{C_\tau} \quad\text{for}\;g\gg1\,.
\end{equation}
In general it is not straightforward whether this is indeed the case. As an example, for a compactification to 9d on $\mathbb{S}^1$ with monodromy $T^{12k}=\begin{psmallmatrix}
    1&12k\\0&1
\end{psmallmatrix}$, which as we have seen from \eqref{eq. T12 mono} and \eqref{eq.kplus1comm} can be bounded by a genus $g=2k$ or $k+1$ $\mathcal{B}_2$, we have that $\overline{\log\big|{\rm Tr}([\mathcal{M}_{a},\mathcal{M}_{b}])\big|}\leq\log 2\approx 0.69315$, but explicitly computing the ratio $2 C_{\rm grav}C_\tau^{-1}$ is out of the reach of this paper, so we will assume that the inequality \eqref{eq. mean trace} is satisfied, and $ S_{\rm E}\approx C\,(M_{{\rm Pl},d+1}R)^{d-1}g>0$ for large $g$.

As mentioned before, based on different explicit constructions of bordisms to nothing (such as bubbles of nothing \cite{Witten:1981gj,Blanco-Pillado:2016xvf,Montero,Draper:2021qtc,Hassfeld:2023kpu,Ookouchi:2024tfz}), we expect the characteristic length of the gravitational soliton to be that of the KK scale. Now, in a similar approach as before, using the equations of motion, see Appendix \ref{ap. onshell}, we see that the mean curvature of the $\mathcal{B}_2$ bordism scales as 
\begin{equation}
   |\overline{\mathcal{R}}_{\mathcal{B}_2}|=\left|\frac{\int_{\mathcal{B}_2}\sqrt{h_2} \dd^2x\mathcal{R}_{\mathcal{B}_2}}{\int_{\mathcal{B}_2}\sqrt{h_2} \dd^2x}\right|=\mathcal{O}\bigg(m_{\rm KK}^2\times \underbrace{\sum_{i=1}^g\log\big|{\rm Tr}([\mathcal{M}_{a_i},\mathcal{M}_{b_i}])\big|}_{\mathsf{a}(\mathcal{B}_2)}\bigg)\,.
\end{equation}
Considering the generic form of the higher-derivative expansion of the gravitational EFT describing the low energy regime of a QG theory \cite{vandeHeisteeg:2023ubh,vandeHeisteeg:2023dlw,Castellano:2023aum,Bedroya:2024ubj,Castellano:2024bna,Calderon-Infante:2025ldq},
\begin{align}
    S_{{\rm EFT,} D}&\supset\frac{1}{2\kappa_D^2}\int\dd^{D}x\sqrt{-G_D}\sum_{n>2}\frac{\mathcal{O}_n(\mathcal{R})}{\Lambda_{{\rm QG},D}^{n-2}}\notag\\&=\frac{1}{\kappa_D^2}\int\dd^D x\sqrt{-G_D}\,\Lambda_{{\rm QG},D}^2\sum_{n>2}\mathcal{O}\left(\sqrt{\mathsf{a}(\mathcal{B}_2)}\frac{m_{\rm KK}}{\Lambda_{{\rm QG},D}}\right)^{n}\;,
\end{align}
where $\Lambda_{{\rm QG},D}$ is the scale at which Quantum Gravity effects become important (also known as \emph{species scale} in the higher-dimensional $D=d+1$ theory, see e.g., \cite{Han:2004wt,Dvali:2007wp, Dvali:2007hz,Anber:2011ut,vandeHeisteeg:2022btw,Calderon-Infante:2023ler,Caron-Huot:2024lbf,ValeixoBento:2025bmv}), and no local EFT description weakly coupled to gravity is possible, in such a way that a full UV QG description is needed.
While in general, specially if working in asymptotic regions of moduli space (such as the large radius regime $\rho\to 0$ to which the 9d potentials from Table \ref{tab:param subgroups} generically drives us) we have a parametric separation\footnote{In the bulk of moduli space, specially close to the so-called \emph{desert point} \cite{Long:2021jlv,vandeHeisteeg:2022btw,vandeHeisteeg:2023ubh,vandeHeisteeg:2023dlw}, we can have $m_{\rm KK}\sim M_{{\rm Pl,}d+1}\sim M_{{\rm Pl,}d}$, so that $\frac{m_{\rm KK}}{\Lambda_{{\rm QG},D}}=\mathcal{O}(1)$.}
\begin{equation}
    m_{\rm KK}\sim\frac{1}{R}\ll\Lambda_{{\rm QG},d}\lesssim\Lambda_{{\rm QG},d+1}\lesssim M_{{\rm Pl},d+1}=M_{{\rm Pl},d}(M_{{\rm Pl},d+1}R)^{-\frac{1}{d-2}}\ll M_{{\rm Pl},d}\,,
\end{equation}
if ${\rm cl}(G)=\infty$ then we will find monodromies along the compact $\mathbb{S}^1$ with arbitrarily large commutator length. Then the associated bordism $\mathcal{B}_2$ that bounds it has an arbitrarily large genus, so that for
\begin{equation}
    \mathsf{a}(\mathcal{B}_2)=\sum_{i=1}^g\log\big|{\rm Tr}([\mathcal{M}_{a_i},\mathcal{M}_{b_i}])\big|\gtrsim\left(\frac{\Lambda_{{\rm QG},d+1}}{m_{\rm KK}}\right)^{\frac{1}{2}}
\end{equation}
higher curvature corrections dominate the (Euclidean) action, and the nucleation rate  is arbitrarily suppressed,
\begin{equation}\label{eq. decay rate genus}
    \Gamma\, V_{d}^{-1}\sim e^{-S_{\rm E}}\sim\left\{\begin{array}{ll}
        \exp\Big[-C\,(\frac{m_{\rm KK}}{M_{{\rm Pl},d}})^{-(d-2)}g\Big] &  \text{for } g\lesssim\mathsf{a}(\mathcal{B}_2)\ll\left(\frac{\Lambda_{{\rm QG},d+1}}{m_{\rm KK}}\right)^{1/2}\\
        0 & \text{for } \mathsf{a}(\mathcal{B}_2)\gtrsim g\gtrsim\left(\frac{\Lambda_{{\rm QG},d+1}}{m_{\rm KK}}\right)^{1/2}
    \end{array}
    \right.\quad,
\end{equation}
with $C>0$, again provided that the condition \eqref{eq. mean trace} is met.

\vspace{0.5cm}

The picture es qualitatively different if we consider a stack of $[p,q]$ 7-branes as bordism defect for the $\mathbb{S}^1$ going to nothing over a $\mathbb{D}^1$ base. Generalizing the analysis from \cite{Hassfeld:2023kpu,Witten:1981gj}, we have that for a $d$-dimensional generalization of Witten's BoN with a defect generating an angular defect $\delta\theta$ the decay rate is given by
\begin{equation}
    \Gamma\, V_{d}^{-1}\sim \exp\left\{-2\pi^{d-2}\Gamma\left(\tfrac{d}{2}\right)^{-1}\left(1-\frac{\delta\theta}{2\pi}\right)^{-(d-2)}\left(\frac{m_{\rm KK}}{M_{{\rm Pl},d}}\right)^{-(d-2)}\right\}\;,
\end{equation}
with $d=9$ for our particular 9d gauge supergravities. Individual D7 branes induce an angular defect of $\frac{\pi}{6}$, see e.g., \cite{Blumenhagen:2013fgp}, so that at most 12 D7 branes can be stacked in a supersymmetric bordism of $\mathbb{S}^1$ to nothing, see Section \ref{ss. Weierstrass}. However, this is no longer the case when considering non-BPS arrangements of $[p,q]$ 7-branes, located along different points of the bordism $\mathcal{B}_2$, where the individual $p_i+\tau q_i$ vectors are no longer aligned and supersymmetry is no longer preserved, see \cite{Gaberdiel:1997ud,Blumenhagen:2013fgp}. Furthermore, since our bordism manifold is non-compact (we can see it as a Riemann surface with one or two punctures), the $\tau$ tadpole needs not be canceled, which allows for an arbitrary number of such $[p,q]$ 7-branes once SUSY is broken. 

Even in the cases where there is a well defined $\delta\theta$, from the lower-dimensional perspective, we can estimate the BoN decay rate, where again the bulk term of the Euclidean action cancels on-shell, with Israel conditions \cite{Israel:1966rt} implying that the GHY term is proportional to the tension of the defect, $K-K_0=\frac{d-1}{d-2}\mathcal{T}_{\rm stack}$, in such a way that \cite{Hassfeld:2023kpu}
\begin{equation}
     \Gamma\, V_{d}^{-1}\sim\exp\left\{-C\,\mathcal{T}_{\rm defect}R^{d-1}\right\}=\exp\left\{-C' \frac{\mathcal{T}_{\rm defect}}{M_{{\rm Pl},d+1}^{d-1}}\left(\frac{m_{\rm KK}}{M_{{\rm Pl,}d}}\right)^{-(d-2)}\right\}\,,
\end{equation}
with $C$ and $C'>0$ some $\mathcal{O}(1)$ numbers. Now, for our $\mathsf{SL}(2,\mathbb{Z})$ monodromies the defects are given, as already discussed, by stacks of $[p,q]$ 7-branes realizing said action. We can estimate the total tension of said stack, where each $[p_i,q_i]$ 7-brane is located at $\vec{x}_i\in\mathbb{D}^2$ by \cite{Schwarz6}
\begin{align}
    \mathcal{T}_{\rm defect}\approx& \sum_{\rm stack}\mathcal{T}_{\rm stack}-\tfrac{1}{\kappa_{10}^2}\sum_{i\neq j}\tfrac{{\rm Re}\big[(p_i+\tau(\vec x_i) q_i)(p_j+\bar\tau(\vec x_j) q_j)\big]}{{\rm Im}\tau(\vec x_i)}\log\tfrac{r_{ij}}{R}+\tfrac{1}{\kappa_{10}^2}\int_{\mathbb{D}^2}\tfrac{\partial_M\tau\partial^M\bar\tau}{2({\rm Im}\,\tau)^2}\,,
\end{align}
where the different terms come from the individual BPS tension of the separate 7-branes, the interaction between each pair (where the logarithmic dependence comes from the codimension-2 of the branes and the KK scale $R$ acts as an IR regulator, and we consider interactions between non-mutually local 7-branes of different stacks), and finally the contribution from the kinetic energy of the $\tau$ due to the monodromies induced by the different branes. We can approximate the value of $\tau$ at each brane stack by taking the value of $\tau$ left invariant by the appropriate monodromy, depicted in Table \ref{tab:stacks}, where for stacks containing more than one brane the total $[p,q]$ can be approximated by the sum of those from the individual components. The middle term, which has the more complicated structure due to the explicit dependence on the brane positions, can be effectively ignored if we only care about the scaling of $\mathcal{T}_{\rm defect}$ with the number of stacks, since it is expected that the brane arrangement is such that minimizes the overall tension, and thus at most will be of the order of the other two (positive) contributions. Additionally, similar to in the $g>0$ case previously studied in this section, we can approximate
\begin{equation}
    \frac{1}{\kappa_{10}^2}\int_{\mathbb{D}^2}\frac{\partial_M\tau\partial^M\bar\tau}{2({\rm Im}\,\tau)^2}=M_{\rm Pl,10}^8\mathcal{O}\bigg(\sum_{\rm stacks}\log|{\rm Tr}(\mathcal{M}_{\rm stack})|\bigg)\,,
\end{equation}
since the above integral is nothing but the area in $\mathbb{H}$ of the image of $\mathbb{D}^2$ through $\tau:\mathbb{D}^2\to \mathbb{H}$, and thus grows linearly with the length of its sides. 
\begin{table}[htb!]
    \centering
    \begin{tabular}{|c||c|c|c|c|c|c|c|c|c|c|}\hline
    Brane stack & $[p,q]$ &$I_n$ &$II$  &$III$  &$IV$&$I_0^*$  &$I_n^*$  &$IV^*$  &$III^*$  & $II^*$  \\\hline
    $\mathcal{T}_{\rm stack}M_{\rm Pl,10}^{-8}$ & $\frac{\pi}{6}$ &$\frac{n\pi }{6}$ &$\frac{\pi}{3}$  &$\frac{\pi}{2}$  &$\frac{2\pi}{3}$ &$\pi$ &$\frac{(n+6)\pi}{6}$  &$\frac{4\pi}{3}$  &$\frac{3\pi}{2}$  & $\frac{5\pi}{3}$  \\\hline
    Fixed $\tau$ & $-\frac{p}{q}$ if $q\neq 0$ &$i\infty$ &$e^{\frac{2\pi i}{3}}$  &$i$  &$e^{\frac{\pi i}{3}}$ & - &$i\infty$  &$e^{\pm\frac{2\pi i}{3}}$  &$i$  & $e^{\frac{\pi i}{3}}$  \\\hline
\end{tabular}
    \caption{Tensions and fixed axio-dilaton for different brane stacks associated to Kodaira singularities.}
    \label{tab:stacks}
\end{table}

We thus can conclude that, following the analysis of Section \ref{eq. gen results} around \eqref{eq. bound no branes}, we have that
\begin{equation}\label{eq. tension defect}
    \mathcal{T}_{\rm defect}\sim M_{\rm Pl,10}^8\mathcal{O}(N_{\mathcal{M}})\lesssim M_{\rm Pl,10}^8\mathcal{O}(\|\mathcal{M}\|_2)\,,
\end{equation}
where $N_{\mathcal{M}}$ is the number of branes needed to realize the monodromy $\mathcal{M}$ of the compact cycle, with $\|\mathcal{M}\|_2$ its (spectral) norm. This way, we can estimate the decay rate through this channel as
\begin{equation}\label{eq. decay branes}
     \Gamma\, V_{d}^{-1}\gtrsim\exp\left\{-\mathcal{O}(\|\mathcal{M}\|_2)\left(\frac{m_{\rm KK}}{M_{{\rm Pl,}d}}\right)^{-(d-2)}\right\}\,,
\end{equation}
where the decay rate can be actually higher than the above expression when the number of branes needed is lower than $\mathcal{O}(\|\mathcal{M}\|_2)$, see \eqref{eq. bound no branes}. Furthermore, since away from the 7-brane stacks we expect our $\mathcal{B}_2\simeq\mathbb{D}^2$ bordism to have a controlled curvature $|\mathcal{R}_{\mathcal{B}_2}|\sim R_{\rm KK}^{-2}$, in general the higher derivative terms that dominated the Euclidean action for $\mathcal{B}_2$ with large genus do not appear in this setting.

Comparing the decay rates \eqref{eq. decay branes} (where $\mathcal{B}_2$ has trivial topology and the needed defects consist in stacks of 7-branes) with \eqref{eq. decay rate genus} (where we try to minimize the number of 7-brane by allowing $\mathcal{B}_2$ to have genus $g>0$ realizing the needed monodromies), and using \eqref{eq. growth of genus}, we conclude that the former channel dominates over the later, both for intermediate and large values of the monodromy norm $\|\mathcal{M}\|_2$.

\subsection{\textit{A conjecture}\label{ss. conjecture}}

In the above section we have seen how decay channels involving stacks of $[p,q]$ 7-branes dominate over those in which the monodromy of the boundary is (mostly) realized through gravitational solitons. While our computations explicitly used the properties of the $\mathsf{SL}(2,\mathbb{Z})$ groups, we can see that they rely in generic properties which allow the generalization to other duality groups $G$.
\begin{itemize}
    \item The main problematic of our results derives from the fact that we need bordisms $\mathcal{B}_2$ of arbitrarily large genus in order to implement elements from $G'\equiv \mathsf{SL}(2,\mathbb{Z})'$, which resulted in decay channels to nothing to be arbitrarily suppressed for $g\gtrsim\sqrt{\Lambda_{{\rm QG},d+1}m_{\rm KK}^{-1}}$. As we discussed in \eqref{eq. com len} and \eqref{eq. com w}, this is set by the commutator width ${\rm cl}(G)$ of the duality group. This is a purely algebraic property of $G$, and as we have seen is infinite for $\mathsf{SL}(2,\mathbb{Z})$. We will study it for other duality groups in Section \ref{sec: other groups}.
    \item When computing the on-shell Euclidean action \eqref{eq. onshell action}, there are two main contributions. The first is purely gravitational, being a function of the genus $g$ of the bordism through Gauss-Bonnet theorem. The second one comes from the kinetic term of the scalars (in this case the axio-dilaton) experiencing the monodromy. As explained in more detail in Appendix \ref{ap. onshell}, $\int_{\mathcal{B}_2}\dd^2 x\,\sqrt{h_2}\frac{\partial_\mu\tau\partial^\mu\bar\tau}{2({\rm Im}\,\tau)^2}=\frac{i}{2}\int_{\tau(\mathcal{B}_2)\subseteq\mathbb{H}}\frac{\dd \tau\wedge\dd \bar\tau}{({\rm Im}\,\tau)^2}$ corresponds to the area of the image of $\mathcal{B}_2$ through $\tau:\mathcal{B}_2\to\mathbb{H}/\mathsf{SL}(2,\mathbb{Z})$. This generalizes, for an arbitrary duality group $G$ and scalars $\vec{\varphi}$ parameterizing the moduli space $\mathcal{M}$, to
    \begin{equation}
        \int_{\mathcal{B}_2}\dd^2x\sqrt{h_2}\,\mathsf{G}_{ij}\partial_\mu\varphi^i\partial^{\mu}\varphi^j=2{\rm Area}(\vec{\varphi}(\mathcal{B}_2))\quad\text{under}\quad \vec\varphi:\mathcal{B}_2\to\mathcal{M}/G
    \end{equation}
    on shell, where $\mathsf{G}_{ij}$ is the moduli space metric of $\mathcal{M}$. Now, under large monodromy transformations, the image of $\quad \vec\varphi:\mathcal{B}_2\to\mathcal{M}/G$ wraps the fundamental domain of the duality group several times. As argued in \cite{Delgado:2024skw}, the moduli spaces of consistent theories of quantum gravity are argued to be \emph{compactifiable} (this is, its volume grows no faster than Euclidean space with the geodesic distance). Even if $\mathcal{M}$ itself is not compact, for most studied examples this results in $\mathcal{M}/G$ having finite volume (see \cite[Sections 3.4 and 4.3.3]{Delgado:2024skw} for examples where this is not the case, though the compactifiability condition still holds), with the moduli space metric $\mathsf{G}_{ij}$ being (approximately) hyperbolic. This implies that, as in our case, ${\rm Area}(\vec{\varphi}(\mathcal{B}_2))$ scales with the length of its sides, respectively corresponding to the monodromies implemented, as in Figure \ref{fig:examples}. Considering a simplified case where $\mathcal{M}/G\simeq \mathbb{H}^n/G$, then under some monodromy $\mathcal{N}\in G$, we can approximate the length along $\mathbb{H}^n$ traveled by $\vec\varphi\to\mathcal{N\vec\varphi}$ by $\sim\log|{\rm Tr}(\mathcal{N})|$, \cite{bridson_haefliger_1999}. This way we can estimate the area of $\vec{\varphi}(\mathcal{B}_2)$ as in \eqref{eq. onshell action}, and a similar effective action as \eqref{eq. onshell action 2} is expected for gravitational solitons with large genus.\footnote{Note that in principle here we are not requiring any restriction on the topology of $\mathcal{M}/G$, only on its volume. However, as argued in \cite{Ooguri:2006in,Dierigl:2020lai,Raman:2024fcv,Rudelius:2024vmc}, it is expected that closed paths on $\mathcal{M}/G$ are contractible, and there are problematic implications when this is not the case, as the image $\vec\varphi(\mathbb{S}^1)\subset \mathcal{M}/G$ could be mapped to a non-trivial 1-cycle in $\mathcal{M}/G$, which would prevent $\vec\varphi(\mathcal{B}_2)\subset \mathcal{M}/G$ from being well-defined.}
    \item On the other hand, when computing the effective action of the BoN with $[p,q]$ 7-brane stacks as the defects realizing the monodromy, we used that the brane stacks individually we expected to have a tension of $\mathcal{O}(M_{\rm Pl,8}^8)$, the same order as the contributions from the scalar kinetic terms (precisely of the order of the ${\rm Area}(\vec{\varphi}(\mathcal{B}_2))$ discussed above), while corrections coming from the interactions between stacks expected to be also $\mathcal{O}(M_{\rm Pl,8}^8)$, though in general their precise expressions might be difficult to compute. It is generally expect that codimension-2 objects in $D$ dimensions have a tension directly related to the angle deficit they create at asymptotic infinity, $\delta\theta=M_{{\rm Pl},D}^{-(D-2)}\mathcal{T}$, see \cite[Section 5]{Nevoa:2025xiq} for more details. This way, we expect $\mathcal{T}=\mathcal{O}(M_{{\rm Pl},D}^{D-2})$. Additional contributions to the effective tension $\mathcal{T}_{\rm defect}$ of the defect are expected to be subleading or of the same order, following a similar argument to that above \eqref{eq. tension defect}.
    \item Finally, given some discrete duality group $G$, arguments similar to those from Section \eqref{eq. gen results} allow us to estimate the growth of the genus $g_{\mathcal{M}}$ needed to realize some $\mathcal{M}\in G$ (modulo representatives of $G^{\rm ab}$), together with the number $N_\mathcal{M}$ of individual codimension-2 brane defects needed to give the analogous realization of $\mathcal{M}$:
    \begin{equation}
        g_{\mathcal{M}}\gtrsim\mathcal{O}(\|\mathcal{M}\|_2)\qquad\text{and}\qquad \mathcal{O}(\log\|\mathcal{M}\|_2)\lesssim N_{\mathcal{M}}\lesssim\mathcal{O}(\|\mathcal{M}\|_2)\;,
    \end{equation}
    as in \eqref{eq. growth of genus} and \eqref{eq. bound no branes}. The above follows from the expression of $\|\mathcal{M}\|_2$ as Frobenius norm for large entries, i.e. $\|\mathcal{M}\|_2\approx\sqrt{\sum_{ij}|\mathcal{M}_{ij}|^2}$, together with the growth of the norm of $\mathcal{M}$ as a product of $N_\mathcal{M}$ parabolic or hyperbolic matrices in $G$.
\end{itemize}

As a result of the above discussion, we then expect that the results obtained in the previous section for $\mathbb{S}^1$ compactifications with $\mathsf{SL}(2,\mathbb{Z})$ monodromies extend for arbitrary (discrete) duality groups $G$ appearing in gravitational EFTs with a consistent QG completion, in such a way that if ${\rm cl}(G)=\infty$ then, while all elements of $G'=[G,G]$ can be realized by gravitational solitons, and in principle one would only need defects associated to $G^{\rm ab}$ to kill the $\tilde\Omega_1^{\mathsf{Spin}}({\rm B}G)$ bordism classes, in practice gravitational instantons with too large genus are arbitrarily suppressed.\\

An intuitive argument on why bordisms with non-trivial topology are more suppressed that those which have trivial topology (but contain duality defects) follows from the size of region where gravity becomes strongly coupled and QG UV effects become important. Bordism defects, such as $[p,q]$ 7-branes, are UV objects, and it is expected that in their core gravity becomes strongly coupled and the EFT description breaks down. However, it can be argued that up to a cut-off region of size $\Lambda_{\rm QG}^{-1}$, spacetimes remains weakly coupled and with small curvature, \cite{Lanza:2020qmt,Lanza:2021udy}. On the other hand, as argued before, the characteristic size of the bordism core is $\mathcal{O}(m_{\rm KK}^{-1})$, and for large genus the mean curvature is quite large, which can eventually also go above the QG cutoff $\Lambda_{\rm QG}$ and break the EFT. Again, since we expect $m_{\rm KK}\ll \Lambda_{\rm QG}$, the region of spacetime where gravity is strongly coupled is much larger for the later than for the former case. This is illustrated in Figure \ref{fig:two walls}.

 \begin{figure}[hbt!]
    \centering
    \resizebox{0.9\textwidth}{!}{%
    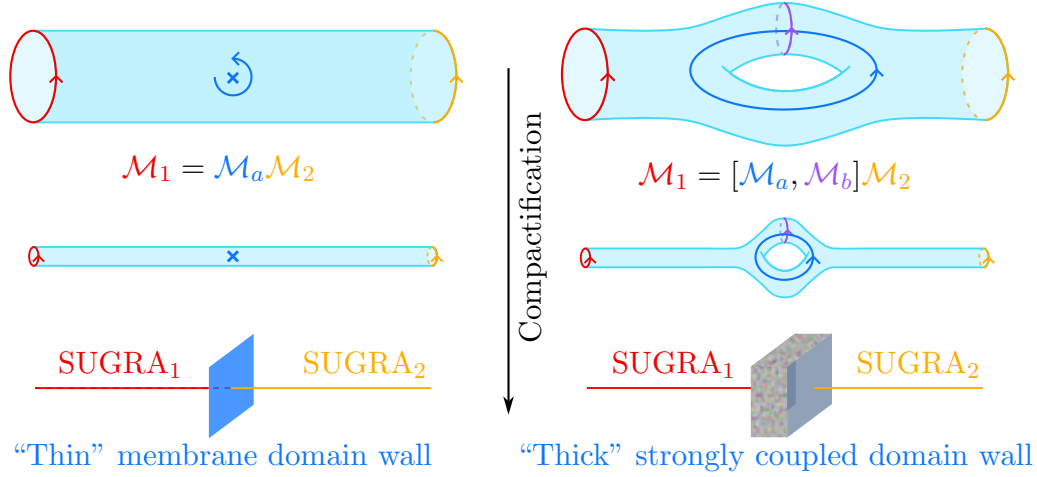
    }
    \caption{Sketch of the difference between interpolating bordisms where the topological defect consists in a codimension-2 singular object (defect killing $\tilde\Omega_1^{\mathsf{Spin}}({\rm B}G)$), which the lower-dimensional theory inherits as a domain wall consisting in the same object upon compactification, and a gravitational soliton with the appropriate monodromies, which upon dimensional reduction can be see as a ``thick'' domain wall where gravity is strongly coupled.}
    \label{fig:two walls}
\end{figure}

In general, it is expected that \emph{global symmetries} in EFTs consistently coupled to Quantum Gravity are broken by symmetry-breaking operators with suppression factor no smaller than $\exp\left(-C\,M_{{\rm Pl},d}^{d-2} \Lambda_{\rm QG}^{-(d-2)}\right)$ for some positive $C=\mathcal{O}(1)$, \cite{Fichet:2019ugl,Daus:2020vtf,Cordova:2022rer}. This bound in the decay rate is respected by decay channels to nothing where the monodromy defect is given by a codimension-two branes, see \eqref{eq. decay branes}, directly implied by the hierarchy $m_{\rm KK}\lesssim \Lambda_{\rm QG}$ in a given $d$-dimensional theory. Conversely, such bound is grossly violated by decay channels relying on gravitational solitons of large genus, see \eqref{eq. decay rate genus}. We know that, while the Cobordism Conjecture applied to $\tilde\Omega_1^{\mathsf{Spin}}({\rm B}\mathsf{SL}(2,\mathbb{Z}))\simeq \mathsf{SL}(2,\mathbb{Z})^{\rm ab}\simeq\mathbb{Z}_{12}$ only requires at most 11 D7-brane defects, the existence of a full spectrum of $[p,q]$ 7-branes implies that the former channel always exists. As discussed in Footnote \ref{fn.all same class}, all $[p,q]$ 7-branes are conjugate to the D7-brane, and thus map to the same element of $\mathsf{SL}(2,\mathbb{Z})^{\rm ab}\simeq\mathbb{Z}_{12}$ through abelianization. While 11 branes for a single choice of $[p_0,q_0]$ would seem to suffice, what we are saying here is that an unbounded number of the \emph{complete} $[p,q]$ spectrum is needed.\\

However, this might not be necessarily the case for a different theory with duality group $G$. In order not to have $\mathbb{S}^1$ compactifications with $G$ monodromy whose decay is arbitrarily suppressed (which would amount to an approximated global symmetry, arbitrarily weakly violated), we propose the following \emph{refinement} or \emph{conjecture} for the Cobordism Conjecture applied to $\tilde\Omega_1^{\mathsf{Spin}}({\rm B}G)$:

\newtcolorbox[use counter=myexample,number format=\Alph]{mybox}[2][]{%
colback=green!5!white,colframe=green!55!black,fonttitle=\bfseries, title=Conjecture \thetcbcounter #2,#1}

\begin{tcolorbox}[enhanced jigsaw,breakable,pad at break*=1mm,colback=YellowGreen!5!white,colframe=YellowGreen!90!black,title={\textbf{A conjecture for the commutator width of the duality group}}]

Given a gravitational EFT with a consistent completion to Quantum Gravity and duality group $G$, then if ${\rm cl}(G)=\infty$, consistency requires the existence of \textbf{\emph{infinitely many}} duality defects killing  elements in $G$, rather than only those associated to $\tilde{\Omega}_1^{\mathsf{Spin}}(BG)\simeq G^{\rm ab}$. Otherwise there will exist \emph{arbitrarily suppressed} bordisms between different $\mathbb{S}^1$ compactifications of the theory and/or nothing.
\end{tcolorbox}
As shown in Theorem \ref{th.inf rep} from Appendix \ref{app. results comm}, for finitely generated discrete infinite groups (we cannot have ${\rm cl}(G)=\infty$ for finite groups), for any upper bound on the number of allowed commutators, one still needs to include an infinite number of representatives (i.e., duality defects), even if $G^{\rm ab}$ is finite.

As already mentioned several times, type IIB string theory already satisfies this refinement through the existence of $[p,q]$ 7-branes. In the following section we will consider other theories with different duality groups, and see what our proposal has to say about them.

\section{Other duality groups}
\label{sec: other groups}

In the preceding sections of the paper we have discussed the fact that the group of monodromies of type IIB on $\mathbb{S}^1$ has infinite commutator width, what requires the whole spectrum of $[p,q]$ 7-branes to be used as bordism defects, rather than at most 11 7-branes. We can apply the same rationale to other duality groups in different effective theories, and see how the above conjecture is fulfilled (or equivalently, what would be required in order to do so). There are two main important properties of our duality group $G$ that must be taken into consideration:
\begin{itemize}
    \item If $G$ is \textbf{perfect}, this is, $G=[G,G]$ and $G^{\rm ab}$ is the trivial group, then all possible monodromies on $\mathbb{S}^1$ can be bounded by gravitational instantons with the appropriate monodromies. In principle, this would imply that at least from the cobordism conjecture perspective, no defects are required. Examples of perfect groups include $\mathsf{SL}(n,\mathbb{Z})$ for $n\geq 3$, $\mathsf{SL}(n,F)$ for $n\geq 2$ and $F$ any field except $\mathsf{SL}(2,\mathbb{F}_2)$ and $\mathsf{SL}(2,\mathbb{F}_3)$, the alternating groups $A_n$ for $n\geq5$, $\mathsf{SO}(n,n;\mathbb{Z})$ for $n=5,6$ and 7 \cite{Braeger:2025kra}, or $\mathsf{Sp}(2n,\mathbb{Z})$ for $n\geq 3$ \cite{bridson2014actions}. By definition non-trivial Abelian groups are not perfect.
    \item If $G$ has \textbf{finite commutator width} ${\rm cl}(G)<\infty$, then we have an upper bound on the genus of the bordism $\mathcal{B}_2$. Examples of groups with finite commutator width include \cite{Calegari2009scl} torsion groups, solvable and amenable groups, and $\mathsf{SL}(n,\mathbb{Z})$ for $n\geq 3$. On the other hand, examples of groups with infinite commutator width include free and hyperbolic groups, or groups of area-preserving diffeomorphisms on surfaces, such as $\mathsf{SL}(2,\mathbb{Z})$ (as we have seen firsthand).
    The computation of ${\rm cl}(G)$ is NP-hard \cite{heuer2020computingcommutatorlengthhard}, so sometimes it is easier to check ${\rm cl}(G)=\infty$ by computing the \emph{stable commutator length} of typical elements ${\rm scl}_G(g)=\lim_{n\to\infty}\frac{{\rm cl}_G(g^n)}{n}<\infty$ \cite{Calegari2009scl}. It is easy to see that if any $g\in G$ has ${\rm scl}_G(g)>0$ then ${\rm cl}(G)=\infty$. For example, as we have seen in Section \ref{sec: com len}, ${\rm scl}_{\mathsf{SL}(2,\mathbb{Z})}(T)=\frac{1}{12}$, see \eqref{eq.kplus1comm}. See \cite{Calegari2009scl} for extensive results on the stable commutator length for different groups.
\end{itemize}

In general the two concepts are independent, though as we will see for certain classes of groups they are related. In the following subsections we will study the perfectness and commutator width of different duality groups relevant in diverse theories.

\subsection{\texorpdfstring{$\mathsf{Mp}(2,\mathbb{Z})$}{Mp(2,Z)} and \texorpdfstring{$\mathsf{GL}^+(2,\mathbb{Z})$}{GL+(2,Z)}\label{ss. mp gl}}
Along the previous sections of the paper, we have focused on the bosonic sector of 10d type IIB string theory, featuring an $\mathsf{SL}(2,\mathbb{Z})$ duality group. As argued in Appendix A from \cite{Tachikawa:2018njr}, this duality group must be ``enlarged'' or ``refined'' when considering additional properties of the theory.

Before considering fermions, which will make us work with non-matrix groups, let us first go back to the F-theory interpretation of type IIB, seen as a limit of M-theory on an $\mathbb{T}^2$-fibered spacetime, where $\mathsf{SL}(2,\mathbb{Z})$ acts on the torus basis with transformations of determinant 1. Now, since M-theory is invariant under parity, we can consider the action of flipping individual $\mathbb{T}^2$ coordinates $y^1$ and $y^2$ (complexified in $z=y^1+iy^2$). In the F-theory limit, these only act on the internal torus, rather than the spacetime coordinates. From a type IIB worldsheet point of view, these are identified as the worldsheet orientation reversal, $\mathsf{R}_1=\Omega={\begin{psmallmatrix}
    -1&0\\0&1
\end{psmallmatrix}}$, and left-moving fermion number, $\mathsf{R}_2=(-1)^{\mathsf{F}_L}={\begin{psmallmatrix}
    1&0\\0&-1
\end{psmallmatrix}}$. Their action on the different type IIB bosonic fields is \cite{Tachikawa:2018njr}
\begin{equation}
    \begin{tabular}{l|cccccc}
         & $g_{\mu\nu}$&$\Phi$&$B_2$&$C_0$&$C_2$&$C_4$ \\\hline
         $\mathsf{R}_1=\Omega$&+&+&$-$&$-$&+&$-$\\
         $\mathsf{R}_2=(-1)^{\mathsf{F}_L}$&+&+&+&$-$&$-$&$-$     
    \end{tabular}
\end{equation}
It is clear that such transformations break SUSY, since for example D3-branes are mapped to $\overline{\text{D3}}$-branes under either transformation. See also that both send $\tau\to-\bar\tau$, which is clearly not an $\mathsf{SL}(2,\mathbb{Z})$ transformation. Now, since $\mathsf{R}_1^2=\mathsf{R}_2^2=1$ (but $\mathsf{R}_1,\,\mathsf{R}_2\not\in\mathsf{SL}(2,\mathbb{Z})$, while $\mathsf{R}_1\mathsf{R}_2\in\mathsf{SL}(2,\mathbb{Z})$), we have that $\mathsf{R}_1\mathsf{R}_2$ is identified with the non-trivial element of the center $Z(\mathsf{SL}(2,\mathbb{Z}))=\{\pm{\rm id}\}\simeq \mathbb{Z}_2$.\footnote{ We could have identified $\mathsf{R}_2\mathsf{R}_1$ as the non-trivial element in $Z(\mathsf{SL}(2,\mathbb{Z}))$. However, since both do not commute, with
\begin{equation}\label{eq.commR1R2}    [\mathsf{R}_1,\mathsf{R}_2]=\Omega(-1)^{\mathsf{F}_L}\Omega(-1)^{\mathsf{F}_L}=(-1)^{\mathsf{F}_R+\mathsf{F}_L}=(-1)^{\mathsf{F}}\;,
\end{equation} 
they would differ in a $(-1)^{\mathsf{F}}$ term that would be important once considering fermions.} Considering such transformations thus enlarges $\mathsf{SL}(2,\mathbb{Z})$ to $\mathsf{GL}(2,\mathbb{Z})$, which we can define via the following presentation
\begin{align}\label{eq. GL2Z}
    \mathsf{GL}(2,\mathbb{Z})&=\big\langle U,\,S,\,R\,|R^2=2,\,\,S^4=e,\, S^2=U^3,\,RSR^{-1}=S^{-1},\,RUR^{-1}=U^{-1}\big\rangle\notag\\&\simeq \mathsf{D}_{12}*_{\mathsf{D}_4}\mathsf{D}_8\;,
\end{align}
where we denote $\mathsf{D}_n$ as the dyhedral group of $n$ elements. Taking the reflection operator $R={\begin{psmallmatrix}
    0&1\\1&0
\end{psmallmatrix}}$, we then have $\mathsf{R}_1=SR$ and $\mathsf{R}_2=RS$. 

It is clear that $\mathsf{SL}(2,\mathbb
Z)$ is a subgroup of $\mathsf{GL}(2,\mathbb{Z})$, and thus considering the later as our enlarged (and SUSY breaking) transformation group, we could now use additional monodromies along the topologically non-trivial 1-cycles of the bordism $\mathcal{B}_2$ to nothing. We have that $\mathsf{GL}(2,\mathbb{Z})^{\rm ab}\simeq\mathbb{Z}_2\oplus\mathbb{Z}_2$, generated by $R$ and $S$, \cite{Debray:2023yrs}. Since all elements in $\mathsf{GL}(2,\mathbb{Z})'$ have determinant 1, but $S\in\mathsf{SL}(2,\mathbb{Z})$ while $S\not\in \mathsf{GL}(2,\mathbb{Z})'$, we have that $\mathsf{GL}(2,\mathbb{Z})'$ is a \emph{proper} subgroup of $\mathsf{SL}(2,\mathbb{Z})$, with index $2$.

Now, since $\mathsf{GL}(2,\mathbb{Z})$ (same for $\mathsf{SL}(2,\mathbb{Z})$) is \emph{virtually free group},\footnote{This is, they have a free subgroup of finite index. Considering the free group
\begin{equation}
    \Gamma(2)=\Big\langle\begin{pmatrix}
        1&2\\0&1
    \end{pmatrix},\,\begin{pmatrix}
        1&0\\2&1
    \end{pmatrix}\Big\rangle\simeq F_2\,,
\end{equation}
we have $\Gamma(2)=\ker({\rm mod}\; 2:\mathsf{SL}(2,\mathbb{Z})\to\mathsf{SL}(2,\mathbb{Z}_2))$. Since $\mathsf{SL}(2,\mathbb{Z}_2)\simeq S_3$ (one can simply check that there are only 6 matrices of such type), then $[\mathsf{SL}(2,\mathbb{Z}):\Gamma(2)]=6$, and thus $[\mathsf{GL}(2,\mathbb{Z}):\Gamma(2)]=12$. For the last step we have used that $[\mathsf{GL}(2,\mathbb{Z}):\mathsf{SL}(2,\mathbb{Z})]=2$, since $\mathsf{SL}(2,\mathbb{Z})=\ker(\det:\mathsf{GL}(2,\mathbb{Z})\to\{0,1\}\simeq\mathbb{Z}_2)$.

} by Corollary \ref{cor. virt free} from Appendix \ref{app. results comm} has \emph{infinite commutator width}. We could then think that we might find ourself in the same situation as before with $\mathsf{SL}(2,\mathbb{Z})$. However, one thing is crucial. For the constructions of the type IIB 9d gauged supergravities, we were not considering monodromies from the \emph{full} $\mathsf{GL}(2,\mathbb{Z})$ group (actually not even the full $\mathsf{SL}(2,\mathbb{Z})$), but rather \emph{only} the parabolic, elliptic and hyperbolic monodromies in $\exp\mathfrak{sl}(2,\mathbb{Z})\cap\mathsf{SL}(2,\mathbb{Z})$, with representatives described in Section \ref{ss:Fth}. After including products of said elements, we find ourselves with the type of elements depicted in Table \ref{tab.int}. Maybe surprisingly, as shown in Appendix \ref{App. real comm}, all such monodromies have ${\rm cl}_{\mathsf{GL}(2,\mathbb{Z})}\leq 2$, modulo a \emph{single} defect (such as one D7-brane) accounting for the $[\mathsf{SL}(2,\mathbb{Z}):\mathsf{GL}(2,\mathbb
Z)']=2$. From a more phenomenological perspective in line with the discussion from Section \ref{sec: com len}, this means that if we allow for $\mathsf{GL}(2,\mathbb{Z})$ monodromies along the non-trivial 1-cycles of our bordism $\mathcal{B}_2$, then these would be enhanced with respect to those with only $\mathsf{SL}(2,\mathbb{Z})$ monodromies, where a large number of brane defects or a topologically involved $\mathcal{B}_2$ were necessary.\footnote{While one could worry that including $\mathsf{GL}(2,\mathbb{Z})$ monodromies along our bordism would break supersymmetry, SUSY is generically spontaneously broken by the positive potential generated by the GSS compactification and depicted in Table \ref{tab:param subgroups}, as well as being the case when having gravitational solitons with non-zero genus.}

\subsubsection*{Including fermions}
So far we have only considered how spacetime bosonic degrees of freedom are affected by the elements of $\mathsf{SL}(2,\mathbb{Z})$ and $\mathsf{GL}(2,\mathbb{Z})$. As noted in \cite{Pantev:2016nze}, considering fermions things become a little bit more tricky, as under a $\begin{psmallmatrix}
    a&b\\c&d
\end{psmallmatrix}\in\mathsf{SL}(2,\mathbb{Z})$ element a chiral spinor transforms as $\psi\to\pm(c\tau+d)^{-1/2}\psi$. Since the sign of the transformation is not determined by the element $\begin{psmallmatrix}
    a&b\\c&d
\end{psmallmatrix}$, the action is not by $\mathsf{SL}(2,\mathbb{Z})$, but rather by a $\mathbb{Z}_2$ extension of it. Following again Appendix A from \cite{Tachikawa:2018njr}, we see that the element $\mathsf{R}_1\mathsf{R}_2\in Z(\mathsf{SL}(2,\mathbb{Z}))$, where both directions of the F-theory elliptic fiber are  reflected, amounts to a $\pi$ rotation of said $\mathbb{T}^2$ fiber. Thus $(\mathsf{R}_1\mathsf{R}_2)^2$ is a $2\pi$ turn, acting on bosons and fermions as $(-1)^{\mathsf{F}}$. Equivalently, since $\mathsf{R}_1^2=\mathsf{R}_2^2=1$, we have $(\mathsf{R}_1\mathsf{R}_2)^2=[\mathsf{R}_1,\mathsf{R}_2]=(-1)^{\mathsf{F}}$, see \eqref{eq.commR1R2}.

This means that, when fermions are being considered, we must extend $\mathsf{SL}(2,\mathbb{Z})$ by the \emph{metapletic} double cover
\begin{equation}
    1\to\mathbb{Z}_2\simeq\{1,(-1)^{\mathsf{F}}\}\to\mathsf{Mp}(2,\mathbb{Z})\to\mathsf{SL}(2,\mathbb{Z})\to 1\,,
\end{equation}
defined through the following presentation: 
\begin{equation}
    \mathsf{Mp}(2,\mathbb{Z})=\Big\langle\hat{U},\hat{S}\;|\; \hat{S}^8=e,\,\hat{S}^2=\hat{U}^3\Big\rangle\simeq \mathbb{Z}_{12}\ast_{\mathbb{Z}_4}\mathbb{Z}_{12}\;,
\end{equation}
which can be understood as the ``smallest'' extension of $\mathsf{SL}(2,\mathbb{Z})$ considering SUSY, with the Spin lift given by $\hat{S}^4=\hat{U}^6=(-1)^{\mathsf{F}}$. Note that, unlike $\mathsf{SL}(2,\mathbb{Z})$, $ \mathsf{Mp}(2,\mathbb{Z})$ is \emph{not} a matrix group.

Since in general it is not expected in F-theory constructions for spacetime to retain a Spin structure, but rather this can be correlated with duality transformations, the relevant structure to consider now is $({\rm Spin}(k)\times {\mathsf{Mp}(2,\mathbb{Z})})/\mathbb{Z}_2$ on $k$-manifolds. Thus, for us the relevant bordism group is now $\Omega_1^{\rm Spin-\mathsf{Mp}(2,\mathbb{Z})}(\rm pt)$. Due to the twisted Spin structure, we cannot perform a splitting as in \eqref{eq.reduced bord}. Nonetheless, one has \cite{Debray:2023yrs}
\begin{equation}
    \Omega_1^{\rm Spin-\mathsf{Mp}(2,\mathbb{Z})}(\rm pt)\simeq\mathsf{Mp}(2,\mathbb{Z})^{\rm ab} \simeq \mathbb{Z}_{24}\simeq \mathbb{Z}_{8}\oplus \mathbb{Z}_3\;,
\end{equation}
generated by $\mathbb{S}^1$ with monodromy $\hat{T}=\hat{S}^{-1}\hat{U}$. The discrepancy of a 2 factor in the order of the bordism group comes from the fact that in $\mathsf{SL}(2,\mathbb{Z})$ we have $S^4=U^6$, while here it is only the case for the bosonic sector. Since $\mathsf{Mp}(2,\mathbb{Z})$ is simply a double cover of $\mathsf{SL}(2,\mathbb{Z})$, by Corollary \ref{cor.double} from Appendix \ref{app. results comm}, we again have ${\rm cl}(\mathsf{Mp}(2,\mathbb{Z}))=\infty$, and the full spectrum of $[p,q]$-7-branes allows us to kill all non-trivial bordisms in a way that is not arbitrarily suppressed.

Finally, we take the reflections on the F-theory torus that resulted in $\mathsf{GL}(2,\mathbb{Z})$, see \eqref{eq. GL2Z}, together with spacetime fermions. The operations $\mathsf{R}_1$ and $\mathsf{R}_2$, which generated  $\mathbb{Z}_2\times\mathbb{Z}_2$, follow $[\mathsf{R}_1,\mathsf{R}_2]=(-1)^{\mathsf{F}}$, see \eqref{eq.commR1R2}. In order to take fermions into account, we have to consider the $\mathrm{Pin}^+$ double cover (see \cite{Debray:2023yrs} for more details)
\begin{equation}
    1\to\{1,(-1)^{\mathsf{F}}\}\simeq\mathbb{Z}_2\to \mathsf{GL}^+(2,\mathbb{Z})\to \mathsf{GL}(2,\mathbb{Z})\to 1\,,
\end{equation}
defined as
\begin{align}
\mathsf{GL}^+(2,\mathbb{Z})&=\Big\langle\hat{U},\hat{S},\hat{R}\,|\,\hat S^8=e,\, \hat S^2=\hat U^3,\, \hat R\hat S\hat R^{-1}=\hat S^{-1},\, \hat R\hat U\hat R^{-1}=\hat U^{-1},\, \hat R^2=e \Big\rangle\notag\\&\simeq \mathsf{D}_{16}*_{\mathsf{D}_8}\mathsf{D}_{24}\;,
\end{align}
with the lift choice $\hat{R}^2=1$ that takes the $\rm Pin^+$ lift, since M-theory is $\rm Pin^+$ (otherwise $\hat{R}^2=(-1)^{\mathsf{F}}$ would be $\rm Pin^-$). The same way as before, the above is not a matrix group.

As in the metaplectic case, again there is a non-trivial interplay between the Spin structure and the duality group, resulting in the relevant cobordism group to compute being $\Omega_1^{{\rm Spin}-\mathsf{GL}^+(2,\mathbb{Z})}(\rm pt)$. Again we are lucky and we have \cite{Debray:2023yrs}
\begin{equation}
    \Omega_1^{\rm Spin-\mathsf{GL}^+(2,\mathbb{Z})}(\rm pt)\simeq\mathsf{GL}^+(2,\mathbb{Z})^{\rm ab}\simeq \mathsf{GL}(2,\mathbb{Z})^{\rm ab}\simeq\mathbb{Z}_2\oplus\mathbb{Z}_2\;,
\end{equation}
generated by $\hat S$ and $\hat R$, respectively with bordism defects given by $[p,q]$-7-branes and the so-called \textbf{R7-brane} \cite{Dierigl:2022reg, Debray:2023yrs,Dierigl:2023jdp}, which consists in a non-supersymmetric defect, that nonetheless, might be stable, see \cite{Heckman:2025wqd}. 

Following the argument from Corollary \ref{cor.double}, we see that while $\mathsf{GL}^+(2,\mathbb Z)$ has infinite commutator width (since it is a double cover of $\mathsf{GL}(2,\mathbb Z)$), which would imply that defects associated to 9d compactifications on $\mathbb{S}^1$ with arbitrary $\mathsf{GL}^+(2,\mathbb Z)$ monodromy would require from the full spectrum of $[p,q]$-7- and R7-branes in order not to have compactifications with $\mathcal{B}_2$ arbitrarily suppressed. On the other hand, as previously argued in this section and explicitly computed in Appendix \ref{App. real comm}, the $\mathsf{SL}(2,\mathbb{Z})$ monodromies (or, for our purposes, their $\rm Pin^+$ lifts) used in the construction of type IIB 9d SUGRAs can all be obtained by at most gravitational instantons of genus-2 and a D7-brane.\footnote{While, as discussed in Theorem \ref{th. iso} and Corollary \ref{cor.double} from Appendix \ref{app. results comm}, the commutator width of the double cover can grow, this is not the case here, as the lift to $\mathsf{GL}(2,\mathbb{Z})$ of the monodromies resulting in gauged 9d supergravities must act trivially on the fermions (since SUSY is spontaneously, rather than explicitly, broken), and thus will only contain the trivial element of $\{1,(-1)^{\mathsf F}\}$.}

\subsection{U-duality and maximal supergravity}

In the previous sections we have considered the $\mathsf{SL}(2,\mathbb{Z})$ duality group of $D=10$ type IIB string theory (as well as its generalizations), and studied the bordisms $\mathcal{B}_2$ after compactifying the theory on $\mathbb{S}^1$ with the appropriate monodromy down to $d=9$. Similarly, we can study additional theories in diverse dimensions, with different duality groups $G$, and see what their abelianization $G^{\rm ab}$ and commutator widths ${\rm cl}(G)$ have to say about the bordisms upon compactification on $\mathbb{S}^1$ with $G$-monodromies.

Following the more general work of \cite{Braeger:2025kra}, a natural set of candidates for this are maximal supergravity theories with 32 supercharges in $D$ dimensions obtained by compactifying M-theory on $\mathbb{T}^{11-D}$. Their duality groups were first computed by Cremmer and Julia \cite{Julia:1980gr,Cremmer1981}, which after flux quantization are reduced to discrete groups $G_U^D$ represented in Table \ref{tab.Udual}.

\begin{table}[hbt!]
    \centering
    \begin{tabular}{|c|c|c|c|}\hline
        $D$ & $G_U^D$ & $\Omega^{\mathsf{Spin}}_1(BG_U^D)$&${\rm cl}(G_U^D)$ \\\hline
        10 & 1    & $\mathbb{Z}_2$  &  0 \\
        9 & $\mathsf{SL}(2,\mathbb{Z})$    & $\mathbb{Z}_2\oplus\mathbb{Z}_{12}$  &  $\infty$ \\
        8 & $\mathsf{SL}(2,\mathbb{Z})\times\mathsf{SL}(3,\mathbb{Z})$    & $\mathbb{Z}_2\oplus\mathbb{Z}_{12}$  &  $\infty$ \\
        7 & $\mathsf{SL}(5,\mathbb{Z})$    & $\mathbb{Z}_2$  &  $\leq 41$ \\
        6 & $\mathsf{SO}(5,5,\mathbb{Z})$    & $\mathbb{Z}_2$  &  $<\infty$ \\
        5 & $\mathsf{E}_{6(6)}(\mathbb
        Z)$    & $\mathbb{Z}_2$  &  $< \infty$ \\
        4 & $\mathsf{E}_{7(7)}(\mathbb
        Z)$    & $\mathbb{Z}_2$  &  $< \infty$ \\
        3 & $\mathsf{E}_{8(8)}(\mathbb
        Z)$    & $\mathbb{Z}_2$  &  $<\infty$ \\\hline
    \end{tabular}
    \caption{U-duality groups $G_U^D$ for maximal supergravity theories from $D=10$ to 3, together with their first Spin-bordism groups $\Omega_1^{\mathsf{Spin}}({\rm B}G_U^D)$ and their commutator widths ${\rm cl}(G_U^D)$}
    \label{tab.Udual}
\end{table}

Following \cite{Obers:1998fb}, the U-duality groups follow the general structure
\begin{equation}
    G_U^{11-l}=\mathsf{E}_{l(l)}(\mathbb{Z})=\mathsf{SL}(l,\mathbb{Z})\bowtie\mathsf{SO}(l-1,l-1,\mathbb{Z})\;,
\end{equation}
where $\bowtie$ is the co-called \emph{knit} or \emph{Zappa–Sz\'ep} product, a generalization of the direct and semidirect products denoting the group generated by the two non-commuting subgroups. In the above product, the first factor is associated to the modular group of $\mathbb{T}^l$. while the second cna be thought of as coming from the T-dualities of type IIA string theory on $\mathbb{T}^{l-1}$.

We first consider the first Spin-bordism groups associated to these groups, which from \eqref{eq.1spinBord} reduces to $\Omega_1^{\mathsf{Spin}}({\rm B}G_U^D)\simeq\mathbb{Z}_2\oplus(G_U^D)^{\rm ab}$. See that the $\mathbb{Z}_2$ factor from $\Omega_1^{\mathsf{Spin}}(\rm pt)\simeq\mathbb{Z}_2$ always remains. As explained in more detail in Section 7 from \cite{Braeger:2025kra}, $G_U^D$ is perfect for $D\leq 7$, being also known that $\mathsf{SL}(n,\mathbb{Z})$ is perfect for $n\geq 3$. Since abelianization respects the direct product, we have that the only problematic terms in Table \ref{tab.Udual} are the $\mathsf{SL}(2,\mathbb{Z})$ factors in $G_U^9$ and $G_U^8$, translating to a $\mathbb{Z}_{12}\simeq\mathbb{Z}_3\oplus\mathbb{Z}_4$ factor in $\Omega_1^{\mathsf{Spin}}({\rm B}G_U^D)$, in both cases generated by $\mathbb{S}^1$ with non-trivial $\mathsf{SL}(2,\mathbb{Z})$ (the same way as in 10d type IIB).

   As argued in \cite{Braeger:2025kra}, the cases with $(G_U^D)^{\rm ab}=0$ need no cobordism defects to kill $\Omega_1^{\mathsf{Spin}}({\rm B}G_U^D)$, which we have further corroborated by noticing that in such cases the commutator width of the U-duality groups remains finite. On the other hand, this is not the case when a $\mathsf{SL}(2,\mathbb{Z})$ factor is present in $G_U^9$ and $G_U^8$. While in the former case the defects are just 7-branes wrapped along the compact $\mathbb{S}^1$ compactifying from 10 to 9d, for the later things are not as straightforward. While as described above and in \cite{Braeger:2025kra}, these are given by singular geometries of the type $(\mathbb{T}^2\times \mathbb{C})/\mathbb{Z}_k$, with $k=3$ and 4, in a $\mathcal{B}_4$ bordism for $\mathbb{T}^3$.\footnote{After further compactifying the $D=8$ maximal supergravity on $\mathbb{S}^1$, locally this singularity will be $((\mathbb{T}^2\times \mathbb{C})/\mathbb{Z}_k)\times(-\epsilon,\epsilon)$, with $(-\epsilon,\epsilon)$ a neighborhood normal to the bordism direction.} The same way such singularities can be interpreted as III$^*$ and IV$^*$ in F-theory language (due to the presence of a $\mathbb{T}^2$ fiber), the full spectrum of defects to kill every element of $\mathsf{SL}(2,\mathbb{Z})$ is realized through the appropriate singular $\mathbb{T}^2$ fiber configuration (not necessarily SUSY preserving), the same way that in 10d type IIB string theory every $\mathsf{SL}(2,\mathbb{Z})$ element is realizable as a monodromy of some 7-brane configuration, with the appropriate F-theory realization, see \cite{Weigand:2018rez}. These singular elliptic fibers will not normally fall under Kodaira-N\'eron classification \cite{Kodaira+2015+1142+1156,Neron}, since they will generally break supersymmetry.\\

After taking care of the first Spin-bordism groups, we now center our attention on the commutator width of $G_U^D$, which also appear on Table \ref{tab.Udual}. In order to see this, we remind that $\mathsf{SL}(2,\mathbb{Z})$ has infinite commutator width, as we showed in Section \ref{sec: com len}. On the other hand, $\mathsf{SL}(n,\mathbb{Z})$ has finite commutator width for $n\geq 3$. While in general the computation of a given group's commutator width is quite complicated, there exist upper bounds in the literature for this family of groups \cite{Newman1987Unimodular,Avni_Meiri_2019},
\begin{equation}\label{eq.clSLnZ}
    {\rm cl}(\mathsf{SL}(n,\mathbb{Z}))\leq \min\left\{2\log\left(\tfrac{3}{2}\right)\log n+40,87\right\}\quad\text{for }\; n\geq3\;,
\end{equation}
    with the above bounds not expected to be optimal and actually quite larger than the actual commutator width. While it can be shown that for large enough $n$ (here $n\geq 42$, see Theorem 3 from \cite{DENNIS1988150}, though again this is expected to be a  quite conservative lower bound) $ {\rm cl}(\mathsf{SL}(n,\mathbb{Z}))\leq 6$, we are interested in smaller groups like $\mathsf{SL}(3,\mathbb{Z})$ and $\mathsf{SL}(5,\mathbb{Z})$, for which we turn to \eqref{eq.clSLnZ}. For $G_U^8\simeq \mathsf{SL}(2,\mathbb{Z})\times\mathsf{SL}(3,\mathbb{Z})$, we use Corollary \ref{cor. fin prod} from Appendix \ref{app. results comm} to see that, since ${\rm cl}(\mathsf{SL}(2,\mathbb{Z}))=\infty$,
    ${\rm cl}(G_U^8)=\infty$ even if the $\mathsf{SL}(3,\mathbb{Z})$ factor is perfect and has finite commutator width (at most ${\rm cl}(\mathsf{SL}(3,\mathbb{Z}))\leq 40$).

    For other families of groups of importance for us, such as $\mathsf{SO}(l,l,\mathbb{Z})$ or $\mathsf{E}_{l(l)}(\mathbb{Z})$, we obtain similar results in a slightly more involved way. From \cite{Hazrat2012CommutatorWI,Gvozdevsky2024VerbalWI}, since  $\mathsf{E}_{l(l)}(\mathbb{Z})$ are simply connected Chevalley groups (of type $E_l$) over $\mathbb{Z}$ and rank larger than 2, every word (including the commutator word) has finite width, with a bound depending only on the root system, thus resulting in $\mathsf{E}_{l(l)}(\mathbb{Z})$ having finite commutator width for $l=6,\,7$ and 8. Similarly, while $\mathsf{SO}(l,l,\mathbb{Z})$ is not a simply Chevalley group, $\mathsf{Spin}(l,l,\mathbb{Z})$ is (of type $D_l$), and as in the exceptional case, it has finite commutator width. Now, since commutator width differs at most in a multiplicative constant under finite index and double covering, see Theorem \ref{thm. finite index} and Corollary \ref{cor.double} from Appendix \ref{app. results comm}, we conclude that ${\rm cl}(\mathsf{SO}(l,l,\mathbb{Z}))<\infty$ for $l\geq 2$. Unlike in the $\mathsf{SL}(l,\mathbb{Z})$ case, no upper bounds for the commutator width are known in the literature.\\

    We thus conclude that barring $D=9 $ and 8, all relevant U-duality groups $G_U^D$ have finite commutator width, though in all such cases the actual value of ${\rm cl}(G_U^8)$ is still unknown. From a phenomenological point of view, it would be interesting to obtain sharper bounds on them, since they would translate in bounds on the decay rate of $\mathbb{S}^1$ compactifications of these theories.

    Finally, the same way that in Section \ref{ss. mp gl} we considered the metaplectic double cover $\mathsf{Mp}(2,\mathbb{Z})$ of $\mathsf{SL}(2,\mathbb{Z})$ to account for fermions in 10d Type IIB string theory, as well as possible reflections on the F-theory torus, arriving to the Pin$^+$-lift $\mathsf{GL}^+(2,\mathbb{Z})$, we can consider similar actions in our maximal supergravity theories, this is, studying the effect of $G_U^D$ on fermions, and reflections on the $\mathbb{T}^{11-D}$ torus we compactify 11d M-theory on, respectively. Following the extensive work of \cite{Pantev:2016nze,Chakrabhavi:2025bfi}, the Spin- and Pin$^+$-lifts of our duality groups are studied, defined through the double covers
    \begin{subequations}
        \begin{align}
            1\to\mathbb{Z}_2\to&\widetilde{G}_U^D\to G_U^D\to1\\
            1\to\mathbb{Z}_2\to&\widetilde{G}_U^{D+}\to (G_U^D\rtimes\mathbb{Z}_2^{\mathsf{R}})\to 1\;,
        \end{align}
    \end{subequations}
    where $\mathbb{Z}_2^{\mathsf{R}}$ is generated by reflections. As in the $\mathsf{SL}(2,\mathbb{Z})$ case, after lifting there is a non-trivial interplay between the Spin structure and the duality bundle, in such a way that the appropriate cobordism structures are now ${\rm Spin}-\widetilde{G}_U^{D(+)}$. While for arbitrary $k$ this greatly affects the structure of $\Omega_k^{{\rm Spin}-\widetilde{G}_U^{D(+)}}(\rm pt)$ with respect to $\Omega_k^{\mathsf{Spin}}({\rm B}G_U^D)$, see \cite{Chakrabhavi:2025bfi}, for the first bordism group one simply has
    \begin{equation}
        \Omega_1^{{\rm Spin}-\widetilde{G}_U^{D(+)}}({\rm pt})\simeq \big(\widetilde{G}_U^{D(+)}\big)^{\rm ab}\,,
    \end{equation}
    which greatly simplifies things for us, since bordisms $\mathcal{\mathcal{B}}_2$ will still be given by 2-manifolds with the appropriate $\widetilde{G}_U^{D(+)}$ monodromies on their non-trivial 1-cycles. We are thus interested in the commutator widths of the groups $\widetilde{G}_U^{D(+)}$. Using the results from Table \ref{tab.Udual}, and the fact that the commutator of the double cover $\widetilde{G}_U^{D}$ differs from that of $G_U^{D}$ by a multiplicative constant, and that ${\rm cl}(G_U^D\rtimes\mathbb{Z}_2^{\mathsf{R}})<\infty$ iff ${\rm cl}(G_U^D)<\infty$ (since $\mathbb{Z}_2$ is abelian and $[G_U^D,\mathbb{Z}_2^{\mathsf{R}}]=1$), see Theorem \ref{th. semidirect}, we conclude that if $G^D_U$ has finite commutator width, so will its Spin- and Pin$^+$-lifts. Following again Table \ref{tab.Udual} we see that again the only problems are for $D=9$ and 8, precisely due to the $\mathsf{SL}(2,\mathbb{Z})$ factor. The full spectrum of duality vortices is thus needed in order not to have $\mathcal{B}_2$ of arbitrary genus, which must include the new type IIA reflection 7-brane and reflection 6-brane discussed in Section 4 from \cite{Chakrabhavi:2025bfi}.

\subsection{Examples in 4d \texorpdfstring{$\mathcal{N}=2$}{N=2}}

In the previous subsections we have considered theories with a large number of supercharges, where the maximal supersymmetry sets the U-duality group for each dimension.\footnote{See also Appendix A.3 from \cite{Etheredge:2024tok}, where it is described how the Emergent String Conjecture \cite{Lee:2019xtm,Lee:2019wij} sets the the Weyl group $W(G_U^D)$ of the different U-duality groups.} As explained above, these are constructed by compactifying M-theory on $\mathbb{T}^{11-D}$ (or equivalently Type IIA string theory on $\mathbb{T}^{10-D}$). For our final set of examples,\footnote{We thank Damian van de Heisteeg for suggesting us these particular settings.} we consider the 4d $\mathcal{N}= 2$ SUGRA theories obtained from compactification of Type IIB string theory on a Calabi-Yau threefold $X$, which can be split into vector multiplet and hypermultiplet. As we will see, the different duality/monodromy groups will directly depend on the topological properties of $X$.

\subsubsection*{Vector multiplet sector}

Starting with the vector multiplet and the complex structure $\mathcal{M}_{\rm cs}$, parameterized by complex coordinates $\{z^ i\}_{i=1}^{h^{2,1}+1}$, we find a monodromy group $\Gamma\leq \mathsf{Sp}(2h^{2,1}+2,\mathbb{Z})$, induced by the non-trivial elements of $\pi_1(\mathcal{M}_{\rm cs})$, acting on the period vector
\begin{equation}
    \vec\Pi(z)\mapsto M\cdot \vec\Pi(z)\,,\quad\text{where}\quad\Pi^ I(z)=\int_{\Gamma_ I}\Omega(z)\,,
\end{equation}
with $\Sigma_I\in H^3(X;\mathbb{Z})\simeq\bigoplus_{p+q=3} H^{p,q}(X)\simeq \mathbb{Z}\oplus\mathbb{Z}^{h^{2,1}}\oplus\mathbb{Z}^{h^{2,1}}\oplus\mathbb{Z}$ and $\Omega$ the holomorphic $(3,0)$-form. Similarly, they act on the particle states
\begin{equation}\label{eq. wrapped N2}
    q=(q_{{\rm D0}},q_{{\rm D2},1},\dots,q_{{\rm D2},h^{2,1}},q_{{\rm D4},1},\dots,q_{{\rm D4},h^{2,1}},q_{{\rm D6}})^\intercal\,,\quad\text{as}\quad q\mapsto Mq\,,
\end{equation}
where the different states are understood as D3-branes wrapped on the different 3-cycles from the $H^3(X;\mathbb{Z})\simeq\bigoplus_{p+q=3}H^{p,q}(X)$.\footnote{More intuitively, we can understand the states encoded in $q=({\color{red}q_{{\rm D0}}},{\color{orange}\vec q_{{\rm D2}}},{\color{Green}\vec q_{{\rm D4}}},{\color{blue}q_{{\rm D6}}})^\intercal$ coming from D3-branes wrapped along 3-cycles from $H^3(X;\mathbb{Z})\simeq {\color{red}H^{3,0}(X)}\oplus{\color{orange}H^{2,1}(X)}\oplus{\color{Green}H^{1,2}(X)}\oplus{\color{blue}H^{03}(X)}$ as wrapped {\color{red}D0}-, {\color{orange}D2}-, {\color{Green}D4}- and {\color{blue}D6}-branes wrapped along the appropriate $p$-cycle in the mirror type IIA description, since through Mirror Symmetry the Hodge numbers are mapped as
$$
{\scriptscriptstyle
   \begin{array}{c}
         1\\
         0\qquad0\\
         0\qquad a\qquad 0\\
         {\color{red} 1}\qquad {\color{orange} b}\qquad {\color{Green} b}\qquad {\color{blue} 1}\\
         0\qquad a\qquad 0\\
                  0\qquad0\\
                  1
    \end{array}\quad\longleftrightarrow\quad
\begin{array}{c}
         {\color{red} 1}\\
         0\qquad0\\
         0\qquad {\color{orange} b}\qquad 0\\
         1\qquad a\qquad a\qquad 1\\
         0\qquad {\color{Green} b}\qquad 0\\
                  0\qquad0\\
                  {\color{blue} 1}
    \end{array}\quad\,},    
$$
which maybe makes clearer the notation used in \eqref{eq. wrapped N2}.
} These monodromy matrices are not just any element of $\mathsf{Sp}(2h^{2,1}+2,\mathbb{Z})$, but further be \emph{quasi-unipotent}, i.e.,
\begin{equation}
    (M^l-{\rm Id})^ k\neq 0\,,\quad\text{but}\quad (M^l-{\rm Id})^{k+1}=0\,,
\end{equation}
for integers $0\leq k\leq 3$ and $l\in\mathbb{Z}$, \cite{Schmid:1973cdo,Grimm:2018ohb}.\footnote{Furthermore, as explained in \cite{Delgado:2024skw}, compactifiability of the moduli space (which as discussed in Section \ref{ss. conjecture} was assumed for our purposes) implies that $\Gamma$ should act in semisimple way.}

The same way $\mathsf{SL}(2,\mathbb{Z})$ monodromies were realized in Type IIB string theory by stacks of $[p,q]$-7-branes, here we have codimension-two defects (which for $D=4$ correspond to strings) implementing the above duality transformations, corresponding to \emph{axionic strings}. These have been studied in the literature, both as spacetime solutions to de 4d equations of motion and from top-down constructions, and have been extensively featured in the context of the Swampland program, see \cite{Lanza:2021udy,Lanza:2022zyg,Grimm:2022sbl,Marchesano:2022avb,Martucci:2022krl,Wiesner:2022qys,Cota:2022yjw,Marchesano:2022axe,Martucci:2024trp,Grieco:2025bjy}, as well as under further reduction to 4d $\mathcal{N}=1$ theories. The string theory embedding of these 4d strings is not universal. For singularities in $\mathcal{M}_{\rm cs}$ corresponding to large complex structure limits, such strings correspond to KK-monopoles wrapped on 3-cycles, with the axion corresponding to the argument of the associated complex coordinates of $\mathcal{M}_{\rm cs}$.\footnote{In the mirror type IIA description, the axionic string corresponds to NS5-branes wrapped on divisors (4-cycles), with the axion resulting from the reduction of the $B_2$-field on the dual 2-cycle.} As for the strings realizing the monodromies associated to other singularities not at finite distance (and such not admitting a perturbative description in the 4d EFT), such as Landau-Ginzburg or conifold points, these do not have a clear top-down string theory embedding in either the type IIB or mirror type IIA descriptions.

For simplicity, and following on the results from Section 3.2 from \cite{Delgado:2024skw}, we will consider in more detail the case where $h^{2,1}=1$, with $\mathcal{M}_{\rm cs}\simeq\mathbb{P}^1\backslash\{0,1,\infty\}$. The point $z=\infty$ corresponds to a large complex structure limit, while $z=1$ results in a conifold point. The final $z=0$ can correspond to either a Landau-Ginzburg point (finite distance, finite monodromy), Conifold point (finite distance, infinite monodromy), a K point or a large complex structure point (both of infinite monodromy and at infinite distance, but the former corresponding to an emergent string limit while the later results in a decompactification). There are 14 Calabi-Yau threefolds of this type \cite{Doran:2005gu,almkvist2010tablescalabiyauequations}, with the singularities of their $\mathcal{M}_{\rm cs}$ moduli spaces and monodromy group $\Gamma$ identified, see e.g., Table 2 from \cite{Delgado:2024skw}. Of said fourteen manifolds, seven of them have infinite index $[\mathsf{Sp}(4,\mathbb{Z}):\Gamma]$, and are given by
\begin{equation}\label{eq.gamman2}
    \mathbb{Z}\ast\mathbb{Z}_5\,,\quad (\mathbb{Z}\oplus\mathbb{Z}_2)\ast_{\mathbb{Z}_2}\mathbb{Z}_8\,,\quad (\mathbb{Z}\oplus\mathbb{Z}_2)\ast_{\mathbb{Z}_2}\mathbb{Z}_{12}\,,\quad \mathbb{Z}\ast\mathbb{Z}\,,
\end{equation}
with the last shared by four different manifolds (with different singularity structure in the moduli space). The respective abelianizations, relevant for $\Omega_1^{\mathsf{Spin}}({\rm B}\Gamma)\simeq\mathbb{Z}_2\oplus\Gamma^{\rm ab}$ are given by\footnote{Here we have used that 
\begin{equation}
    (A\ast_CB)^{\rm ab}\simeq(A^{\rm ab}\oplus B^{\rm ab})/\langle\{i_G^{\rm ab}(c)i_H^{\rm ab}(c)^{-1}\}_{a\in C^{\rm ab}}\rangle,
\end{equation}
where $i_G:A\hookrightarrow G$ and $i_H:A\hookrightarrow G$ are the injective morphisms defining the amalgamated product, with $i_G^{\rm ab}$ and $i_H^{\rm ab}$ their images under the abelianization functor. Concretely, we have $(A\ast B)^{\rm ab}\simeq A^{\rm ab}\oplus B^{\rm ab} $.}
\begin{equation}
    \mathbb{Z}\oplus\mathbb{Z}_5\,,\quad \mathbb{Z}\oplus\mathbb{Z}_8\,,\quad \mathbb{Z}\oplus\mathbb{Z}_{12}\,,\quad \mathbb{Z}\oplus\mathbb{Z}\,,
\end{equation}
 The defects needed to kill the bordism classes are precisely the axionic strings whose monodromies generate the individual factors in the free (amalgamated) products from \eqref{eq.gamman2}. Unlike in the $\mathsf{SL}(2,\mathbb{Z})$ case, where in principle a single type of $[p,q]$ 7-branes (such as D7) where needed to kill the whole $\mathsf{SL}(2,\mathbb{Z})^{\rm ab}\simeq\mathbb{Z}_{12}\simeq\mathbb{Z}_3\oplus\mathbb{Z}_4$, here the axionic strings associated to two of the singularities will be needed. Related to this, the commutator width of all the above $\Gamma$ groups is \emph{infinite}, since they are all \emph{virtually free}. Following our reasoning, we thus need the full spectrum of codimension-2 defects able to generate the complete $\Gamma$. As we have just explained, due to the structure of the monodromy group, these were already needed to kill $\Gamma^{\rm ab}$, corresponding to different axionic strings. We again recall that, other than the large complex structure one, these to not have in general a clear top-down construction from string theory.

Regarding the additional seven CY$_3$ with $h^{2,1}=1$, they do not have an expression for $\Gamma$ in terms of a free or amalgamated product as in \eqref{eq.gamman2}, and further have finite index $[\mathsf{Sp}(4,\mathbb{Z}):\Gamma]$, \cite{Brav_2014,HOFMANN_2015}. The expression for their elements in terms of a set of generators is present in \cite{HOFMANN_2015}, but is in general quite more complicated that those of $\Gamma$ already discussed. In general, we expect to find again non-trivial abelianizations and commutator widths as in the previous cases, with the same conclusions applying.\\

In the above paragraphs we have studied the na\"ively simpler case where $h^{2,1}=1$, since the different monodromy groups have been classified in the literature. However, as it was the case in $\mathsf{SL}(2,\mathbb{Z})$ versus $\mathsf{SL}(n,\mathbb{Z})$ for $n\geq 3$, things are quite different for $\Gamma\leq \mathsf{Sp}(2h^{2,1}+2,\mathbb{Z})$ for $h^{2,1}\geq 2$. We have that $\mathsf{Sp}(2n,\mathbb{Z})$ for $n\geq 3$ is a simply connected Chevalley group (of type $C_n$) with rank higher than 2, and thus with finite commutator width, \cite{Hazrat2012CommutatorWI,Gvozdevsky2024VerbalWI}. Now, if $\Gamma\leq\mathsf{Sp}(2h^{2,1}+2,\mathbb{Z})$ has finite index, then we will have ${\rm cl}(\Gamma)<\infty$, see Theorem \ref{thm. finite index} from Appendix \ref{app. results comm}. While $\mathsf{Sp}(2n,\mathbb{Z})$ is perfect for $n\geq 3$ \cite{bridson2014actions}, this does not necessarily extend to its subgroups (be them of finite index or not), so one would need the appropriate axionic strings to generate $\Gamma^{\rm ab}$ in order to kill $\tilde\Omega_1^{\mathsf{Spin}}({\rm B}\Gamma)$, this time without needing to worry about arbitrarily involved bordisms $\mathcal{B}_2$.

Note that in general, if $[\mathsf{Sp}(2h^{2,1}+2,\mathbb{Z}):\Gamma]=\infty$, one might still have that ${\rm cl}(\Gamma)=\infty$, and thus duality an infinite set of vortices/codimension-two defects associated elements in $\Gamma$ (not just its generators) would be needed. This was already the case in the $h^{2,1}=1$ cases explained above, where the free group structure of $\Gamma\leq \mathsf{Sp}(4,\mathbb{Z})$ implied both infinite commutator width and the existence of a larger set of necessary duality defects. Understanding the physical distinction between finite and infinite index monodromy groups $\Gamma$ in $\mathsf{Sp}(2h^{2,1}+2,\mathbb{Z})$ might help understand the nature of the different axionic strings other than those associated to large complex structure limits.

\subsubsection*{Hypermultiplet sector}

We now consider the slightly more involved hypermultiplet sector. Compactifying again type IIB string theory on a CY threefold $X$, with $H^4(X;\mathbb{Z})\simeq H^2(X;\mathbb{Z})\simeq H^{1,1}(X)\simeq \mathbb{Z}^{h^{1,1}}$ spanned by the basis of 4-cycles $\{[\omega_a]\}_{a=1}^{h^{1,1}}$. The associated moduli space $\mathcal{M}_{\rm h}$ is a quaternionic K\"ahler manifold with $\dim_{\mathbb{R}}\mathcal{M}_{\rm h}=4(h^{1,1}+1)$, spanned by \cite{Alexandrov:2008gh,Alexandrov:2011va,Alexandrov:2013yva}
\begin{itemize}
    \item A universal part given by the type IIB axiodilaton $\tau=C_0+ie^{-\Phi}$.
    \item A set of $h^{1,1}$ complex scalars $\psi^a=b_2^a+it^a$ resulting from the expansion $B_2+i J=(b_2^a+it^a)\omega_a$, where $B_2$ is the NSNS 2-form and $J$ is the K\"ahler 2-form, in such a way that $t^a=\int_{[\omega_a]} J$ is the volume (in string units) of the divisor $[\omega_a]$. The total volume of $X$ is then $\mathcal{V}_X=\frac{1}{6}\kappa_{abc}t^ at^bt^ c$, with $\kappa_{abc}$ the intersection numbers.
    \item A set of $2h^{1,1}$ axions obtained from the reduction of the RR 2- and 4-forms, $C_2=c_2^a\omega_a$ and $C_4-\frac{1}{2}B_2\wedge C_2=-c_{4,a}\tilde{\omega}_a$, on the $\{\omega_ a\}_{a=1}^{h^{1,1}}$ basis of $H_2(X;\mathbb{Z})$ and the dual $\{\tilde \omega^ a\}_{a=1}^{h^{1,1}}$ basis of $H_4(X;\mathbb{Z})$.
    \item Two axions, $c_6$ and $b_6$, respectively obtained from the reduction of the RR 6-form over the whole $X$, $C_6-B_2\wedge C_4+\frac{1}{3}B_2\wedge B_2\wedge C_2=-c_6 \star 1_X$, and similarly by the NSNS 6-form $B_6=b_6 \star 1_X$ dual to $B_2$ (or from a 4d point of view, the axion dual to the 4d $B_2$ field).
\end{itemize}
Note that this results in $3(h^{1,1}+1)$ axions and $h^{1,1}+1$ non-compact coordinates. Associated to each of the axions we have a 4d string, electrically charged under the (4d) dual 2-form. As we have seen, from a top-down point of view, said axions are obtained by reducing $p$-form fields $A_p$ (either NSNS or RR) on $p$-cycles $\Sigma_p$. The associated 4d string defects will be given by wrapping the $7-p$-brane electrically charged under the dual $\tilde{A}_{8-p}$ along the Poincar\'e dual $\tilde{\Sigma}_{6-p}=\Sigma_p^\vee$, as depicted in Table \ref{tab.hyper strings}.

\begin{table}[hbt!]
    \centering
    \begin{tabular}{c|cccccc}
         Axion& $C_0$  & $b_2^a$ & $c_2^a$ & $c_{4,a}$ & $c_6$  & $b_6$ \\\hline
        String & D7 on $X$ & NS5 on $[\tilde\omega^a]$  & D5 on $[\tilde\omega^a]$  & D3 on $[\omega_a]$ & D1 & F1 \\
    \end{tabular}
    \caption{Top-down construction of the 4d strings associated with the different 4d axions and their monodromies. Adapted from Table 3 from \cite{Delgado:2024skw}.}
    \label{tab.hyper strings}
\end{table}
Regarding the dualities/monodromies acting on our fields and spectrum,\footnote{Note that, since type IIB string theory only has differential forms $A_p$ with $p$ odd, and for the hypermultiplet sector we are only acting on the even-dimensional cycles of $X$, in the 4d $\mathcal{N}=2$ theory we will only obtain 0- (axions, electrically coupled to instantons), 2- (coupled to our strings and dual to the aforementioned axions) and 4-form fields (acting as $\theta$ angles), the only charged objects of our theory under which the prospective duality/monodromy group $\Gamma_{\rm h}$ acts are precisely the extended 4d strings, transforming under $\Gamma_{\rm h}$ through the adjoint representation ${\rm ad}:\Gamma_{\rm h}\to{\rm Aut}(\Gamma_{\rm h})$.} we have that the $\mathsf{SL}(2,\mathbb{Z})$ duality from 10d type IIB is inherited in the 4d theory, acting on the hypermultiplet scalars in a non-trivial way, see Section 2.2 from \cite{Alexandrov:2013yva}. On the other hand, both from the $\mathcal{M}_{\rm h}$ moduli space metric \cite{Alexandrov:2013yva} and the large gauge transformations of the $B_2$, $C_2$ and $C_4$ fields, a discrete Heisenberg group action ${\rm H}^{2h^{1,1}+3}(\mathbb{Z})$ acting on the axions exists \cite{Alexandrov:2010np},
\begin{equation}
   T_{\vec\alpha,\vec\beta,\gamma}: (c_{2}^a,c_{4,a},b_6)\mapsto(c_{2}^a+\alpha^ a,c_{4,a}+\beta_a,b_6+2\gamma-\beta_ac_2^a+\alpha^ac_{4,a})\,,\quad \text{with }\;\alpha^a,\,\beta_b,\,\gamma\in\mathbb{Z}\,,
\end{equation}
satisfying the Heisenberg group law, $T_{\vec\alpha,\vec\beta,\gamma}T_{\vec\alpha',\vec\beta',\gamma'}=T_{\vec\alpha +\vec\alpha',\vec\beta+\vec \beta',\gamma+\gamma'+\frac{1}{2}(\vec\beta\cdot \vec\alpha'-\vec\beta'\cdot\vec\alpha)}$.\\
Finally, the hypermultiplet moduli space $\mathcal{M}_{\rm h}$ is fibered over the complex structure moduli space $\mathcal{M}_{\rm cs,mirror}$ of the mirror Calabi-Yau $X^\vee$, with the associated monodromy group $\Gamma_{\rm mirror}\leq\mathsf{Sp}(2h^{1,1}+2,\mathbb{Z})$ acting on the RR axions $q=(c_0,c_2^1,\dots,c_2^{h^{1,1}},c_{4,1},\dots,c_{4,h^{1,1}},c_6)^\intercal$ as in \eqref{eq. wrapped N2}. This way, the full monodromy group acting on the 4d $\mathcal{N}=2$ spectrum is given by
\begin{equation}\label{eq. semiprod hyper}
    \Gamma_{\rm h}=\mathsf{SL}(2,\mathbb{Z})\ltimes\left(\Gamma_{\rm mirror}\ltimes {\rm H}^{2h^{1,1}+3}(\mathbb{Z})\right)\,,
\end{equation}
a considerably more involved expression than the $\Gamma\leq\mathsf{Sp}(2h^{2,1}+2,\mathbb{Z})$ for the vector multiplet. Now, using Theorem \ref{th. semidirect},  we conclude that ${\rm cl}(\Gamma_{\rm h})=\infty$ due to the $\mathsf{SL}(2,\mathbb{Z})$ factor.\footnote{As for the abelianization, notice that
\begin{equation}
    G^{\rm ab}=(N\rtimes H)^{\rm ab}\simeq\frac{NH}{[NH,NH]}\simeq\frac{N}{[N,N]\cdot[N,H]}\cdot\frac{H}{[H,H]}\simeq\frac{N^{\rm ab}}{[N,H]}\cdot H^{\rm ab}\,,
\end{equation}
 with $(N\rtimes H)^{\rm ab}$ being ``smaller'' than $(N\times H)^{\rm ab}\simeq N^{\rm ab}\times H^{\rm ab}$.
}
 Again, the full spectrum of $[p,q]$-7-branes will be needed in order not to have arbitrarily complicated $\mathcal{B}_2$ bordisms associated to such factor. Regarding the $\Gamma_{\rm mirror}\ltimes {\rm H}^{2h^{1,1}+3}(\mathbb{Z})$ subgroup, it could be the case that itself had infinite commutator width, what would again motivate the need of a complete spectrum of codimension-two defects. As argued for the vector multiplet case, $\Gamma_{\rm mirror}\leq\mathsf{Sp}(2h^{1,1}+2,\mathbb{Z})$ will generically have infinite commutator width for $h^{1,1}=1$, while this will be finite if $h^{1,1}\geq 2$ and $[\mathsf{Sp}(2h^{1,1}+2,\mathbb{Z}):\Gamma_{\rm mirror}]<\infty$. Regarding factor associated to the Heisenberg group ${\rm H}^{2h^{1,1}+3}(\mathbb{Z})$, we have that it has infinite commutator width.\footnote{
Consider the general presentation Heisenberg group
\begin{equation}
    {\rm H}^{2n+1}(\mathbb{Z})=\langle x_1,\dots,x_n,y_1,\dots,y_n,z\,|[x_i,y_j]=z^{\delta_{ij}},\,[x_i,x_j]=[y_i,y_j]=[x_i,z]=[y_i,z]=1\rangle\,,
\end{equation}
being immediate that $[{\rm H}^{2n+1}(\mathbb{Z}),{\rm H}^{2n+1}(\mathbb{Z})]\simeq\mathbb{Z}\simeq\langle z\rangle$ and ${\rm H}^{2n+1}(\mathbb{Z})^{\rm ab}\simeq \langle x_1,\dots,x_n,y_1,\dots,y_n\rangle$. Since the commutator group consists in powers of $z$, which cannot be reduced further, it is clear that ${\rm cl}_{{\rm H}^{2n+1}(\mathbb{Z})}(z^k)=k$, and thus ${\rm cl}\big({\rm H}^{2n+1}(\mathbb{Z})\big)=\infty$.
} Thus the duality vortex whose stacking generates such factor is precisely the one associated to the shift in the $\{c_2^a,c_{4,a}\}_{a=1}^{h^{1,1}}$ axions, i.e., the D3- and D5-branes, see Table \ref{tab.hyper strings}.\\

We thus conclude that the hypermultiplet sector of 4d $\mathcal{N}=2$ theories contains the \emph{worst of both worlds}, since the issues that appeared for the 10d type IIB $\mathsf{SL}(2,\mathbb{Z})$ and $\Gamma\leq \mathsf{Sp}(2n+2,\mathbb{Z})$ resurface here through the semiproduct structure of \eqref{eq. semiprod hyper}, together with the Heisenberg group ${\rm H}^{2n+1}(\mathbb{Z})$ which by default has infinite commutator width and requires of stacks of D3- and D5-branes as duality defects.

\section{Conclusions\label{s. conc}}
In this paper, we have studied the topological properties of bordisms connecting theories given by circle compactifications with different monodromy actions on the compact dimension. To do this, we have used $\mathbb{S}^1$ compactifications of type IIB string theory, characterizing the bordisms between type IIB maximal and gauge supergravities in 9d (also considering bordisms to nothing). We focus on those gauge supergravities with a symmetry group given by $\mathbb{R}$, $\mathsf{SO}(2)$, and $\mathsf{SO}(1,1)$, one-parameter subgroups of $\mathsf{SL}(2,\mathbb{R})$, broken to discrete subgroups upon charge quantization. These supergravities can be obtained by doing a Generalized Scherk--Schwarz reduction from the type IIB parent theory in 10d, or equivalently, from F-theory on three-dimensional twisted tori $\mathbb{E}\hookrightarrow T^3_\mathcal{M}\to\mathbb{S}^1$ with $\mathsf{SL}(2,\mathbb{Z})$ monodromy along the base. A particular monodromy and non-trivial torsion homology group is associated to each twisted torus, as can be seen from table \ref{tab. hom T3}, being different to three-dimensional torus quotients from \cite{Montero,Pfaffle2000,Acharya:2019mcu}.

The first Spin bordism group of the 10d type IIB theory is given by $\Omega_1^{\mathsf{Spin}}({\rm B}\mathsf{SL}(2,\mathbb{Z}))\simeq \mathbb{Z}_2\oplus \mathbb{Z}_{12}$, where the $\mathsf{SL}(2,\mathbb{Z})^{\rm ab}=\mathbb{Z}_{12}$ factor results from the quotienting of $\mathsf{SL}(2,\mathbb{Z})$ by the commutator subgroup, since such monodromies can be achieved from gravitational instantons. From an F-theory perspective, such bordisms are given by elliptic fibrations $\mathbb{E}\hookrightarrow\mathcal{B}_4\to\tilde{\mathcal{B}}_2$, where the base manifold $\tilde{\mathcal{B}}_2$ corresponds to a Riemann surface of genus $g$ and $n=n_b+n_s$ punctures, with $n_b$ the number of boundaries and $n_s$ the number of singularities. This induces a non-trivial $\rho:\pi_1\big(\tilde{\mathcal{B}}_2^*\big)\rightarrow \mathsf{SL}(2,\mathbb{Z})$ action around the possible punctures and non-trivial 1-cycles of the base.
 By using the Leray--Serre spectral sequence, we were able to compute the interger homology of the $\mathcal{B}_4$ manifold.

We compute the monodromies that should be implemented by the bordisms $\mathbb{E}\hookrightarrow\mathcal{B}_4\to\tilde{\mathcal{B}}_2$ in order to connect a given 9d type IIB compactification with a different one and/or nothing, together with the topological properties of the interpolating manifold. We do this for trivial base $\tilde{\mathcal{B}}_2$, where we show the arrangement of Kodaira singularities that would be needed, as well as for non-trivial $\mathcal{B}_2$, where we specify the genus $g$ and the number (at most 11) of 7-branes. We find that almost all such bordisms explicitly break supersymmetry.

We have shown how the genus $g$ of such gravitational instanton grows linearly with the spectral norm of the monodromy.  Furthermore, since the commutator width of $\mathsf{SL}(2,\mathbb{Z})$ is unbounded, i.e., $\text{cl}(\mathsf{SL}(2\mathbb{Z}))=\infty$, there are monodromies that need to be expressed as a product of an arbitrarily large number of commutators, what would require gravitational solitons of arbitrarily large genus if we only allow for defects associated to $\mathsf{SL}(2,\mathbb{Z})^{\rm ab}\simeq \mathbb{Z}_{12}$. For large genus (equivalently, commutator length of the $\mathbb{S}^1$ monodromy) $g\gtrsim\sqrt{\Lambda_{{\rm QG},d+1}m_{\rm KK}^{-1}}$, such decay channels become arbitrarily suppressed, something that goes against QG expectations \cite{Fichet:2019ugl,Daus:2020vtf,Cordova:2022rer}. On the opposite, if the monodromy defects are implemented by stacks of $[p,q]$ 7-branes, such dramatic suppression does not occur. With this motivation, we conjecture that when the commutator width of the duality group $G$ is infinite, duality defects able to kill an infinite set of $G$ monodromies should exist, even if they could in principle be trivialized through gravitational solitons.

We argue that our computations rely on generic properties, which allow generalizations to other duality groups. We indeed test our proposal in duality groups for different spacetime dimensions and amounts of supersymmetry, and find that always such defects exists or conversely, they are not present but the commutator width is finite, thus confirming our expectations. One example of interest is the study of the $\mathsf{GL}^+(2,\mathbb{Z})$ duality group once fermions and SUSY-breaking transformations are considered, since it seems that it should be possible to trivialize the $\mathsf{SL}(2,\mathbb{Z})$ monodromies through at most a single defect and a genus zero to two gravitational soliton.

There are several possible directions along which we could continue the research from this paper. First of all, many of the results used and obtained rely on numerical factors that would benefit from more accuracy in their values. For example, when computing the decay rate of gravitational solitons in Section \ref{ss. decay rate}, we obtain some inequality \eqref{eq. mean trace} relating the $\mathcal{O}(1)$ factors of the gravitational and kinetic scalar terms of the bounce action with the average logarithmic trace of the different gravitational solitons. Sharpening said numerical values would allow us to more accurately compute the decay rate of said channels. Conversely, while it is known that important groups, such as those appearing in U-duality for $D\leq 7$ or some 4d $\mathcal{N}=2$ sectors, have finite commutator widths \cite{Newman1987Unimodular,Avni_Meiri_2019,DENNIS1988150,Hazrat2012CommutatorWI,Gvozdevsky2024VerbalWI}, in general the bounds upper bounds for ${\rm cl}(G)$ are not known in the literature or are very far from saturation. Better understanding of the commutator properties of duality groups relevant in string/M-theory compactifications would be of great help for our purposes, since it would allow us to better study bordisms between $\mathbb{S}^1$ compactification of such theories and/or nothing. The stability of theories where a single compact direction is privileged with respect to the others, such as \cite{Montero:2022prj}, might be of phenomenological importance, specially for the type of $\mathbb{S}^1$ compactifications with $\mathcal{M}\in G$ monodromy, which generically spontaneously break supersymmetry. Additionally, along this project we have not attempted to explicitly solve the equations of motion to obtain an expression for the metric and the backreacted fields, restricting ourselves to topological and algebraic properties that do not require the resolution of complicated systems of coupled equations. It would be interesting to employ techniques similar to \cite{Montero} to obtain analytic approximations to the geometry of the bordism, or directly numerical approach, since this would allow us to obtain sharper expressions for quantities such as the Euclidean action or the decay rates.

Along this paper, we have focused on bordisms with one or two boundaries, respectively corresponding to nullbordisms and interpolating bordisms. Our results apply and naturally extend to cases with $n_b\geq 3$, which would correspond to \emph{junctions} between theories, similar to \cite{Hellerman:2010dv,Altavista:2026edv,Tachikawa:2026jaj}. Through our purely geometric set-up, it would be interesting to study the transition/transmission rates of fields crossing the bordism and ending in one theory or another.\\

From a more formal point of view, we could consider if similar refinements to the Swampland Cobordism Conjecture could be made for higher $\Omega^{\mathsf{g}}_k({\rm B}G)$ groups (with $k>1$). While in this paper we have only needed to consider $\mathbb{S}^1$ as the possible compact topology, for higher-dimensional compact manifolds the possibilities rapidly grow with $k$, and generically one would need to consider different $G$-bundles supported over the non-trivial 1-cycles of the compactification manifold. Already for $k=2$ there seem to be interesting implications for the arrangement and relations between monodromy defects \cite{Dierigl:2026psj}. Furthermore, we could think what is the appropriate measure of the ``complexity'' of a bordism, see e.g., \cite{korkmaz2004stable}.

Finally, back to the gauged 9d supergravities, we have restricted our study to those with a clear construction from 10d type IIB string theory. As can be seen from Figure \ref{fig1}, there are additional theories in 9d that have no type IIB embedding. If such theories admit a consistent QG completion, then by the Cobordism Conjecture they should be connected to the ones studied in this paper. It has been conjectured that other 9d gauges supergravities, such as those obtained from torus-bundle compactifications of M-theory and M2-branes, are T-dual to the ones studied in this paper \cite{Hull:1998vy,Alonso-Alberca:2002pbb,GarciadelMoral:2016ffr}. Understanding the precise transformation could offer new insights to the behavior of bordism defects under T-duality, in line with \cite{Delgado:2023uqk}.\\

Much has been learned through the Cobordism Conjecture about qualitative properties of bordisms between theories and the properties of the needed defects. Motivated by such accumulated progress, we hope that this paper makes a small step in the better understanding of some of their more quantitative aspects, similar to those made by, see e.g. \cite{Hamada:2025duq,Heckman:2025wqd,Cavusoglu:2026xiv,Torres:2026vxx,Altavista:2026edv,Anastasi:2026cus,Tachikawa:2026jaj} for an incomplete list of recent progress in directions not completely orthogonal to ours. In doing so, we might learn new insights about the properties and stability of theories of quantum gravity, regardless of whether the universe they describe is close to ours or not.

\vspace{1cm}

\textbf{Acknowledgments}: 
We are grateful to Markus Dierigl, Maria Pilar García del Moral, Elizabeth Gasparim, Damian van de Heisteeg, Miguel Montero, Michelangelo Tartaglia, Ethan Torres and Irene Valenzuela for very illuminating discussions and comments. Particularly, we are specially thankful to Markus Dierigl and Irene Valenzuela for valuable comments on the draft. The authors want to thank IFT UAM-CSIC for hospitality during the initial development of this paper, as well the Spanish Agencia Estatal de Investigaci\'on through the grant ``IFT Centro de Excelencia Severo Ochoa'' CEX2020-001007-S and the grants PID2021-123017NB-I00 and PID2024-156043NB-I00, funded by MCIN/AEI/10.13039/ 501100011033 and by ERDF ``A way of making Europe''. The work of I.R.  is supported by the European Union through ERC Starting Grant SymQuaG-101163591 StG-2024, and he acknowledges the additional funding of the Spanish FPI grant No. PRE2020-094163 and the ERC Starting Grant QGuide101042568 - StG 2021. The work of C.L.H.  was supported by FONDECYT postdoctorado No. 3250615 2025.

\appendix
\section{A primer on solvmanifolds in three dimensions}
\label{Ape1}

 A \textit{solvmanifold} is a compact homogeneous space $G/\Lambda$, where $G$ is simply connected solvable Lie group\footnote{This is, given $\mathfrak{g}$ the Lie algebra associated to $G$, the sequence
 \begin{equation}
     \mathfrak{g}\supseteq[\mathfrak{g},\mathfrak{g}]\supseteq \big[[\mathfrak{g},\mathfrak{g}],[\mathfrak{g},\mathfrak{g}]\big]\supseteq \big[\big[[\mathfrak{g},\mathfrak{g}],[\mathfrak{g},\mathfrak{g}]\big],\big[[\mathfrak{g},\mathfrak{g}],[\mathfrak{g},\mathfrak{g}]\big]\big]\supseteq\dots\supset0\,,
 \end{equation}
 eventually terminates in the zero subalgebra.
 } and $\Lambda$ a lattice in $G$, i.e., a discrete co-compact subgroup. When $G$ is nilpotent, the existence of the lattice $\Lambda$ is guaranteed, and $G/\Lambda$ is a \textit{nilmanifold}. Every nilmanifold is a solvmanifold, but some solvmanifolds are not diffeomorphic to nilmanifolds.

 Almost abelian solvmanifolds (this is, those for which $G$ has a finite-index abelian subgroup) are characterized by Lie group $G$ being a semi-direct product, $G=\mathbb{R}^{d-1}\rtimes_\varphi \mathbb{R}$, with $\varphi$ a continuous one-parameter right group action $\varphi:\mathbb{R}\rightarrow \mbox{Aut}(\mathbb{R}^{d-1})$. We can regard $\varphi$ as a matrix $\varphi_x\in \mathsf{GL}(d-1,\mathbb{R})$ for each $x\in\mathbb{R}$.  Then, $G$ admits a lattice $\Lambda_G$ if and only if there exists $x_0\in\mathbb{R}$ such that 
 \begin{equation}
     \Sigma^{-1}\varphi_{x_0}\Sigma = M\,,
 \end{equation}
with $M\in \mathsf{SL}(d-1,\mathbb{Z})$ and $\Sigma\in \mathsf{GL}(d-1,\mathbb{R})$. Indeed, this condition also ensures the compactness of the solvmanifold, \cite{Bock,Aschieri}. 

The corresponding Lie algebra is as usual defined by $\left[ E_b,E_c\right]={f^a}_{bc}E_a$, where $\left\lbrace E_a \right\rbrace_a$ is a vector basis and ${f^a}_{bc}$ are the structure constants. Additionally, there is a set of dual one-forms $\left\lbrace \eta^a\right\rbrace$ satisfying a Maurer--Cartan equation
    \begin{eqnarray}
      \dd\eta^a=-\frac{1}{2}{f^a}_{bc}\eta^b\wedge \eta^c. 
    \end{eqnarray}

A \textit{twisted torus} is given by the quotient space $G/\Gamma$, where $G$ is a locally compact group and $\Gamma$ a lattice acting on $G$ by right multiplication \cite{Aschieri} . Consequently, when $G$ is a simply connected and nilpotent (or almost abelian solvable) Lie Group, then the corresponding solvmanifold is a twisted torus. A large class of twisted tori of interest in string compactifications comes in the form of fibrations over tori. The fibrations underlying the twisted tori, are called Mostow bundles \cite{Mostow}.
    \begin{thm}[Mostow bundles]
Let $\Lambda_G$ be a lattice in a connected and simply-connected solvable Lie group $G$ and $\mathbb{T}_{\Lambda_G}=G/\Lambda_G$ the associated solvmanifold. Let $N$ be the nilradical of $G$. Then $\Lambda_G N$ is a closed subgroup of $G$, $\Lambda_N:=\Lambda_G\cap N$ is a lattice in $N$, and $G/\Lambda_G N$ is a torus. It follows that the twisted torus $\mathbb{T}_{\Lambda_G}$ is a fibration over this torus with a nilmanifold fiber given by:
\begin{equation}
     N/\Lambda_N=\Lambda_G N/\Lambda_G \longrightarrow \mathbb{T}_{\Lambda_G}\longrightarrow G/\Lambda_G N\,.
\end{equation}

\end{thm}

Mostow bundles correspond to nilmanifolds fibered over a torus. The structure group is given by $\Lambda_G N/\Lambda_{G_0}$, where $\Lambda_{G_0}$ is the largest subgroup of $\Lambda_G$ which is normal in $\Lambda_G N$.

The Mostow bundle associated with an almost Abelian solvmanifold is
\begin{equation}
    N/\Lambda_N \simeq\mathbb{T}^{d-1}\longrightarrow \mathbb{T}_{\Lambda_G}\longrightarrow G/\Lambda_G N  \simeq \mathbb{T}\,,
\end{equation}
which corresponds with a torus bundle over a circle with monodromy $M$ in the mapping class group $\mathsf{SL}(d-1,\mathbb{Z})$ of orientation-preserving automorphisms up to the homotopy of the torus fibers.

Let us note that if $x\in\mathbb{S}^1$ is a coordinate over the circle and $\{y^1,\dots,y^d\}$ coordinates over the torus $\mathbb{T}^d$, we have that
\begin{equation}
    \left[E_a,E_x\right]={M_a}^bE_b \qquad \text{and} \qquad \left[E_a,E_b\right] = 0
\end{equation}
with again $\{E_1,\dots,E_d,E_x\}$ the generators of the Lie algebra and $M$ is the mass matrix such that $\varphi(x)=\exp{(Mx)}$. The corresponding basis of left-invariant one-forms are
\begin{equation}
    \eta^x=\dd x\,, \quad
    \eta^a = {\varphi(x)^a}_b \dd y^b,
    \end{equation}
    such that $\dd\eta^a = -{M^a}_b\eta^x\wedge\eta^b$

    \subsection{Orientable three-dimensional solvmanifolds}
    Three-dimensional solvmanifolds $G/\Gamma$ are almost abelian solvmanifolds. Consequently, we have that $$G=\mathbb{R}^2\rtimes_\varphi\mathbb{R}.$$
    For $\varphi_x={\rm Id}_2 \, \forall \, x\in\mathbb{R}$, then $G$ is abelian and the solvmanifold is a 3-torus (case (a) in Table \ref{tab1era}). The remaining cases correspond to $\varphi:\mathbb{R}\longrightarrow \mathsf{SL}(2,\mathbb{R})$,  which are classified into conjugacy classes given by their trace: \emph{parabolic} $\abs {\rm Tr(\rho)}=2$,  \emph{hyperbolic} $\abs {\rm Tr(\rho)}>2$, and \emph{elliptic} $\abs {\rm Tr(\rho)}<2$. Indeed, three-dimensional solvmanifolds (see Table \ref{tab1era}) are in correspondence with the isomorphisms classes of Lie algebras of connected solvable Lie groups that possess lattices (see Table \ref{tab2da}) \cite{Bock}. 
\begin{table}[htb]
    \centering
    \begin{tabular}{|c||c|c|c|c|}
 \hline
     & $b_1(G/\Gamma)$ & $G$ & Nilmanifold & Completely solvable  \\
    \hline\hline
    a) & 3 & $\mathbb{R}^3$ &yes & yes \\
       \hline
    b) & 2 & $\rm H^3$ &yes & yes\\
       \hline
    c) & 1 & ${\mathsf{ISO}}(1,1)$ &no & yes \\
       \hline
    d) & 1 & ${\mathsf{ISO}}(2)$ &no & no \\
    \hline
    \end{tabular}
    \caption{Solvmanifolds in three dimensions from \cite{Bock}, where $\rm H^3$, ${\mathsf{ISO}}(1,1)$ and ${\mathsf{ISO}}(2)$ corresponds to the Heisenberg, Poincare, and Euclidean groups in three dimensions.}
    \label{tab1era}
\end{table}
\begin{table}[htb]
    \centering
    \begin{tabular}{|c||c|c|c|}
 \hline
     & Algebra & Commutation relations & Completely solvable  \\
    \hline\hline
    a) & $3g_1$ &  -& abelian \\
       \hline
    b) & $g_{3,1}$ & $\left[E_2,E_3\right]=E_1 $ & nilpotent\\
       \hline
    c) & $g_{3,4}^{-1}$ & $\left[E_1,E_3\right]=E_1 $  , $\left[E_2,E_3\right]=-E_2 $ & yes \\
       \hline
    d) & $g_{3,5}^0$ & $\left[E_1,E_3\right]=-E_2 $  , $\left[E_2,E_3\right]=E_1 $  & no \\
    \hline
    \end{tabular}
    \caption{Non-isomorphic algebras in three dimensions from \cite{Bock}}
    \label{tab2da}
\end{table}
The Mostow fibrations of these solvmanifolds are given by torus bundles over a circle with monodromy in $\mbox{SL}(2,\mathbb{Z})$,
\begin{equation}
    \mathbb T^2\longrightarrow G/\Gamma \longrightarrow \mathbb S^1\,,
\end{equation}
equipped with some action $\rho:\pi_1(\mathbb{S}^1)\longrightarrow \mathsf{SL}(2,\mathbb{Z})$.

Therefore, the straight $\mathbb T^3$ corresponds to trivial monodromy (case (a)) in Tables \ref{tab1era}, \ref{tab2da} and first entry in \ref{tab. hom T3}), the Heisenberg nilmanifold (case (b)) corresponds to parabolic monodromies, and the twisted tori with hyperbolic, and elliptic monodromies are associated with cases (c) and (d), respectively. The last two cases are also denoted in the literature by $\varepsilon_{1,1}$ and $\varepsilon_2$, respectively. 

 The first homology group of these solvmanifolds is given by
 \begin{equation}
     H_1(\mathbb{T}_{\Lambda_G};\mathbb{Z})\simeq \mathbb{Z}\oplus\mathbb{Z}^{2-r}\oplus \bigoplus_{i=1}^r\mathbb{Z}_{m_i}
 \end{equation}
where $r={\rm rank}(M-{\rm Id}_{2})$ and $m_i=\frac{d_i(A)}{d_{i-1}(A)}$, with $d_i(A)$ the greatest common divisor of all $i\times i$ minors of $A$, with $d_0(A)=1$.

Solvmanifolds are widely used in the context of string compactifications; see, e.g., \cite{Silverstein,Andriot1,Andriot2,Grana,Grana2,Caviezel:2008ik,Cribiori:2021djm,Arboleya:2025jko,VanHemelryck:2025qok,Andriot:2026lac}. All the three-dimensional solvmanifolds seen in this appendix are technically known as twisted tori. However, most authors use the name of twisted torus for the Heisenberg nilmanifold (see \cite{Aschieri}).

\section{Computing the homology of \texorpdfstring{$\mathbb{E}\hookrightarrow\mathcal{B}_4\to\tilde{\mathcal{B}}_2$}{E->B4->B2} through the Leray-Serre spectral sequence\label{app.leraySerre}}      
In this appendix we will extract the topology of an elliptically fibered 4-bordism $\mathbb{T}^2\hookrightarrow\mathcal{B}_4\to\tilde{\mathcal{B}}_2$ manifold with boundary $\mathbb{T}^2\hookrightarrow\mathcal{C}_3\to\sqcup_i\mathbb{S}^1$. Our tool for this will be the \textbf{Leray-Serre (homological) spectral sequence} (see e.g. Section 5.1 from \cite{HatcherTOPOLOGY}), which for a fibration $F\hookrightarrow X\to B$ implies
\begin{equation}
E^2_{p,q}=H_p(B,H_q(F;\mathbb{Z}))\Longrightarrow H_{p+q}(X;\mathbb{Z})\,,\quad\text{with differential }\; d_2:E_{p,q}^2\to E_{p-2,q+1}^2\,.
\end{equation}
For our purposes, $F\simeq \mathbb{T}^2$, and we can see our two-dimensional base $\tilde{\mathcal{B}}^*_2$ as a genus-$g$ Riemann surface with $n=n_b+n_s$ punctures, accounting for both the $n_b$ boundary $\partial\tilde{\mathcal{B}}_2$ components and the $n_s$ points over $\tilde{B}_2$ where the fiber degenerates, i.e., $\tilde{B}_2\simeq\Sigma_{g,n}$. Associated to the non-trivial 1-cycles of the base we have a $G$-action (for the type IIB 9d gauged supergravities of this paper $G=\mathsf{SL}(2,\mathbb{Z})$) $\rho$
\begin{equation}   \rho:\pi_1(\Sigma_{g,n})\simeq\mathbb{Z}^{\ast(2g+n-1)}\to G\trianglelefteq\mathsf{GL}(2;\mathbb{Z})\,,
\end{equation}
generated by the matrices $\{\mathcal{M}_{a_i},\mathcal{M}_{b_i}\}_{i=1}^g$ and $\{\mathcal{M}_{c_j}\}_{j=1}^n$, with the following closing relation implied by the well definition of the above action
\begin{equation}\label{eq. closing rel}
    \prod_{i=1}^g[\mathcal{M}_{a_i},\mathcal{M}_{b_i}]\prod_{j=1}^n\mathcal{M}_{c_j}=\mathrm{Id}_2\,,
\end{equation}
associated to the generators of $\pi_1(\Sigma_{g,n})$, and acting on the elements of $H_\ast(\mathbb{T}^2;\mathbb{Z})$ as
\begin{equation}
    H_q(\mathbb{T}^2;\mathbb{Z})_\rho=\left\{\begin{array}{lll}
        \mathbb{Z} &q=0&\to\text{trivial }G\text{ action}  \\
         \mathbb{Z}^{\oplus 2} &q=1&\to\text{non trivial }G\text{ action}  \\
         \mathbb{Z} &q=2&\to\text{non trivial }G\text{ action if orientation is not preserved}  \\
         0&q\geq 3
    \end{array}\right.\quad.
\end{equation}
This way, the $E^2$-page is given by
\begin{equation}
    \begin{array}{c|ccc}
        q=2& \mathbb{Z}_{\det\circ\rho}&(\mathbb{Z}_{\det\circ\rho})^{\oplus2g+n-1}&0 \\
         q=1& (\mathbb{Z}^{\oplus 2})_{\rho}& H_1(\Sigma_{g,n};\mathbb{Z}^{\oplus 2}_\rho)&0\\
         q=0&\mathbb{Z} &\mathbb{Z}^{\oplus 2g+n-1}&0\\\hline
         E^2_{p,q}&p=0&p=1&p=2
    \end{array}\quad ,
\end{equation}
where we have used that $H_\ast(\Sigma_{g,n};\mathbb{Z})=(\mathbb{Z},\mathbb{Z}^{2g+n-1},0,\dots)$. Happily for us the $E^2$-page is narrow enough that the $d_2:E_{p,q}^2\to E_{p-2,q+1}^2$ differentials all vanish, and such the page collapses, $E^2_{p,q}=E^\infty_{p,q}$!\\

Of the non-trivial terms above, the first one is the twisted second homology of the fiber under the $G$-action,
\begin{equation}
    H_2(\mathbb{T}^2;\mathbb{Z})_\rho\simeq \mathbb{Z}_{\rm det\circ\rho}=\left\{\begin{array}{ll}
       \mathbb{Z} &\text{if all }\gamma\in\pi_1(\Sigma_{g,n}) \text{ respects orientation} \\
        \mathbb{Z}_2 &\text{otherwise} 
    \end{array}\right.\quad.
\end{equation}
For the purposes of this paper, since $\mathsf{SL}(2;\mathbb{Z})$ preserves orientation, the action is trivial and $H_2(\mathbb{T}^2;\mathbb{Z})_\rho\simeq\mathbb{Z}$. Had we considered $\mathsf{GL}(2,\mathbb{Z})$ instead, as is the case in Section \ref{ss. mp gl}, the action of the group would act non-trivially on $H_2(\mathbb{T}^2,\mathbb{Z})$, since orientation needs not be preserved in that case, and the $H_2(\mathbb{T}^2;\mathbb{Z})_\rho$ would depend on the precise $\mathsf{GL}(2,\mathbb{Z})$ monodromies along the base, with $ H_2(\mathbb{T}^2;\mathbb{Z})_\rho\simeq\mathbb{Z}_2$ is there is any path not preserving orientation.\\
All that is then left to compute the coinvariant $(\mathbb{Z}^2)_{\rho}$ and the twisted $H_1(\Sigma_{g,n};\mathbb{Z}^{\oplus 2}_\rho)$. The former is straightforward, since
\begin{equation}    (\mathbb{Z}^{\oplus 2})_\rho=\faktor{\mathbb{Z}^{\oplus 2}}{\langle(\rho(\gamma)-{\rm Id}_2)v\,:\,\gamma\in\pi_1(\Sigma_{g,n}),\,v\in\mathbb{Z}^{\oplus 2}\rangle}\,.
\end{equation}
Depending on the action $\rho$, we can have $(\mathbb{Z}^{\oplus 2})_\rho$ (if all monodromies are trivial), some torsion and/or free contribution or, if $\rho$ is ``dense'' enough, $(\mathbb{Z}^{\oplus 2})_\rho\simeq 0$. The above definition of the coinvariant $(\mathbb{Z}^{\oplus 2})_\rho$ is not very clear, since the free product structure of the $\pi_1(\Sigma_{g,n})\simeq\mathbb{Z}^{\ast 2g+n-1}$ translates in generic elements $\rho(\gamma)$ being words generated by the monodromy matrices associated by the generators of $\pi_1(\Sigma_{g,n})$. Note however that, given $A,\,B\in G\trianglelefteq \mathsf{GL}(2,\mathbb{Z})$, we have $AB-{\rm Id}_2=A(B-{\rm Id}_2)+(A-{\rm Id}_2)$. Furthermore, since $A$ is invertible, ${\rm im}(A(B-{\rm Id}_2))={\rm im}(B-{\rm Id}_2)$. This, way we have
\begin{align}
    {\langle(\rho(\gamma)-{\rm Id}_2)v\,:\,\gamma\in\pi_1(\Sigma_{g,n}),\,v\in\mathbb{Z}^{\oplus 2}\rangle}=\sum_{i=1}^g&\big[{\rm im}(\mathcal{M}_{a_i}-{\rm Id}_2)+ {\rm im}(\mathcal{M}_{b_i}-{\rm Id}_2)\big]\notag\\&+\sum_{j=1}^n{\rm im}(\mathcal{M}_{c_j}-{\rm Id}_2)\,.
\end{align}
Note that due to the closing relation \eqref{eq. closing rel}, one of the monodromy matrices is not independent from the others, and thus is redundant. Now, in order to compute the above subgroup of $\mathbb{Z}^{\oplus2}$, we employ the Smith normal form, so that given the $2\times2(2g+n)$-matrix
\begin{equation}
    \hat{\mathcal{M}}=\begin{pmatrix}
        \mathcal{M}_{a_1}-{\rm Id}_2|\cdots|\mathcal{M}_{a_g}-{\rm Id}_2|\mathcal{M}_{b_1}-{\rm Id}_2|\cdots|\mathcal{M}_{b_g}-{\rm Id}_2|\mathcal{M}_{c_1}-{\rm Id}_2|\cdots|\mathcal{M}_{c_n}-{\rm Id}_2
    \end{pmatrix}\,,
\end{equation}
we obtain, for an appropriate $\hat{\mathcal{S}}$ and $\hat{\mathcal{T}}$ $2\times 2 $ and $(2g+n)\times (2g+n)$ matrices with
\begin{equation}
    \hat{\mathcal{S}}\hat{\mathcal{M}}\hat{\mathcal{T}}=\begin{pmatrix}
        d_1&0&\cdots &0\\
        0&d_2&\dots &0
    \end{pmatrix}\,,
\end{equation}
then
\begin{equation}
    {\langle(\rho(\gamma)-{\rm Id}_2)v\,:\,\gamma\in\pi_1(\Sigma_{g,n}),\,v\in\mathbb{Z}^{\oplus 2}\rangle}=d_1\mathbb{Z}\oplus d_2\mathbb{Z}\,,\quad\text{so that}\quad (\mathbb{Z}^{\oplus 2})_\rho\simeq\mathbb{Z}_{d_1}\oplus\mathbb{Z}_{d_2}\,,
\end{equation}
where $\mathbb{Z}_0\equiv\mathbb{Z}$ and $\mathbb{Z}_1\equiv\{0\}$.

As for the twisted $H_1(\Sigma_{g,n};\mathbb{Z}^{\oplus 2}_\rho)$, we have
\begin{equation}
   H_1(\Sigma_{g,n};\mathbb{Z}^{\oplus 2}_\rho)\simeq\frac{\ker(\partial_1:C_1\otimes\mathbb{Z}^{\oplus 2}\to C_0\otimes\mathbb{Z}^{\oplus 2})}{{\rm im}(\partial_2:C_2\otimes\mathbb{Z}^{\oplus 2}\to C_1\otimes\mathbb{Z}^{\oplus 2})}\,,
\end{equation}
with
\begin{equation}
    0\to C_2\otimes\mathbb{Z}^{\oplus 2}\xrightarrow{\partial_2}C_1\otimes\mathbb{Z}^{\oplus 2}\xrightarrow{\partial_1}C_0\otimes\mathbb{Z}^{\oplus 2}\to 0\,.
\end{equation}
Now, notice that since $\Sigma_{g,n}$ has a boundary, there are no 2-cycles, $C_2=0$. This results in
\begin{equation}
    0\to C_1\otimes\mathbb{Z}^{\oplus2}\xrightarrow{\partial_1}C_0\otimes\mathbb{Z}^{\oplus 2}\to 0\,,
\end{equation}
where the twisted boundary operator $\partial_1$ acts as 
\begin{align}
    \partial_1\underbrace{(v_1,\dots,v_{2g+n})}_{\sigma\in C_1\otimes\mathbb{Z}^{\oplus2}}&=\sum_{i=1}^g\Big[(\mathcal{M}_{a_i}-{\rm Id}_2)v_{a_i}+(\mathcal{M}_{b_i}-{\rm Id}_2)v_{b_i}\Big]+\sum_{j=1}^n(\mathcal{M}_{c_j}-{\rm Id}_2)v_{c_j}\notag\\&=\sum_{\gamma\in\pi_1(\Sigma_{g,n})}(\mathcal{M}_{\gamma}-{\rm Id}_2)v_{\gamma}\,,
\end{align}
where we decompose the 1-cycles $\sigma\in C_1\otimes\mathbb{Z}^{\oplus2}$ in terms of the $2g+n$ generators of $\pi_1(\Sigma_{g,n})$ (this is, before modding out by the closing relation). This way, $H_1(\Sigma_{g,n};\mathbb{Z}^{\oplus 2}_\rho)=\ker(\partial_1)$. In general such twisted homology group can have both free and torsion parts. To compute the former we can use the Universal Coefficient Theorem, tensoring by $\mathbb{Q}$, in such a way that $\rank H_1(\Sigma_{g,n};\mathbb{Z}^{\oplus2}_\rho)=\dim H_1(\Sigma_{g,n};\mathbb{Q}^{2}_\rho)$. Now, we have the Euler characteristic is independent of the twisted structure
\begin{align}
\chi(\Sigma_{g,n};\mathbb{Q}^2_\rho)&=\underbrace{\chi(\Sigma_{g,n})\dim\mathbb{Q}^2}_{2(2-2g-n)}\notag\\&=\underbrace{\dim H_0(\Sigma_{g,n};\mathbb{Q}^{2}_\rho)}_{\dim(\mathbb{Q}^2)_\rho}-\dim H_1(\Sigma_{g,n};\mathbb{Q}^{2}_\rho)+\underbrace{\dim H_2(\Sigma_{g,n};\mathbb{Q}^{2}_\rho)}_{0}\,,
\end{align}
since at previously discussed there are no 2-cycles in $\Sigma_{g,n}$.
We recall that coinvariants is defined by
\begin{equation}
    (\mathbb{Q}^2)_\rho=\frac{\mathbb{Q}^2}{\bigoplus_{\gamma\in\pi_1(\Sigma_{g,n})}{\rm im}_{\mathbb{Q}}(\mathcal{M}_\gamma-{\rm Id}_2)}\,.
\end{equation}
Since by tensoring $(\mathbb{Z}^{\oplus 2})_\rho\otimes\mathbb{Q}\simeq (\mathbb{Q}^2)_\rho$ we get rid of the possible torsion, we have $\rank (\mathbb{Z}^2)_\rho=\dim (\mathbb{Q}^2)_\rho$. This way,
\begin{equation}
    \rank H_1(\Sigma_{g,n};\mathbb{Z}^{\oplus2}_\rho)=\dim H_1(\Sigma_{g,n};\mathbb{Q}^{2}_\rho)=2(2g+n-2)+\rank (\mathbb{Z}^2)_\rho\;.
\end{equation}
As for the torsion part of $H_1(\Sigma_{g,n};\mathbb{Z}^{\oplus2}_\rho)=E_{1,1}^2$ (i.e., it will contribute to $H_2(\mathcal{B}_4;\mathbb{Z})$), it arises in the same way as in $E_{0,1}^2$, this is, as torsion elements of $(\mathbb{Z}^{\oplus 2})_\rho$, since all the 1-cycles of the base and fiber are torsion free, so that any contribution will come from non trivial actions of the $\rho:\pi_1(\Sigma_{g,n})\to G$ action on 1-cycles of the fiber $\mathbb{T}^2$ upon embedding in $\mathcal{B}_4$. This way, ${\rm Tor}\,H_1(\Sigma_{g,n};\mathbb{Z}^{\oplus2}_\rho)\simeq{\rm Tor}\,(\mathbb{Z}^{\oplus 2})_\rho$.\\

We thus conclude that, using $H_k(\mathcal{B}_4;\mathbb{Z})\simeq\bigoplus_{p+q=k}E_{p,q}^{\infty}$
\begin{subequations}\label{eq. hom B4}
    \begin{align}     H_0(\mathcal{B}_4;\mathbb{Z})&\simeq\mathbb{Z}\\
    H_1(\mathcal{B}_4;\mathbb{Z})&\simeq\mathbb{Z}^{\oplus2g+n-1}\oplus(\mathbb{Z}^{\oplus 2})_{\rho}\\
    H_2(\mathcal{B}_4;\mathbb{Z})&\simeq\mathbb{Z}^{\oplus2(2g+n-2)+\rank(\mathbb{Z}^{\oplus 2})_\rho}\oplus \mathbb{Z}_{\det\circ\rho}\oplus{\rm Tor}\,(\mathbb{Z}^{\oplus 2})_{\rho}\notag\\&\simeq\mathbb{Z}^{\oplus2(2g+n-2)}\oplus \mathbb{Z}_{\det\circ\rho}\oplus(\mathbb{Z}^{\oplus 2})_{\rho}\label{eq.H2B4}\\
    H_3(\mathcal{B}_4;\mathbb{Z})&\simeq(\mathbb{Z}_{\det\circ\rho})^{\oplus2g+n-1}\\
    H_4(\mathcal{B}_4;\mathbb{Z})&\simeq0\,,
    \end{align}
\end{subequations}
where in \eqref{eq.H2B4} we have used the Fundamental Theorem of finitely generated Abelian groups. We recall that $\mathbb{Z}_{\det\circ\rho}\simeq \mathbb{Z}$ for $G=\mathsf{SL}(2;\mathbb{Z})$.\\

In a similar way we can compute the homology of each of the boundary components, $\mathbb{T}^2\hookrightarrow T^3_{\mathcal{M}}\to\mathbb{S}^1$, where now there is no boundary, and the torus fiber simply experiences a monodromy $\rho:\pi_1(\mathbb{S}^1)\simeq\mathbb{Z}\to G$ along the base, generated by a single $\mathcal{M}\in G$. As in the $\mathcal{B}_4$ case, $\rho$ can act non-trivially on the fiber homology, with now the $E^2$-page is given by
\begin{equation}
    \begin{array}{c|ccc}
        q=2& \mathbb{Z}_{\det\circ\rho}&\mathbb{Z}_{\det\circ\rho}&0 \\
         q=1& (\mathbb{Z}^{\oplus 2})_{\rho}& H_1(\mathbb{S}^1;\mathbb{Z}^{\oplus 2}_\rho)&0\\
         q=0&\mathbb{Z} &\mathbb{Z}&0\\\hline
         E^2_{p,q}&p=0&p=1&p=2
    \end{array}\quad ,
\end{equation}
where now
\begin{equation}
    (\mathbb{Z}^{\oplus})_\rho=\faktor{\mathbb{Z}^{\oplus2}}{{\rm im}\,(\mathcal{M}-{\rm Id}_2)}\,,\quad H_1(\mathbb{S}^1;\mathbb{Z}^{\oplus 2}_\rho)\simeq \ker(\partial_1)\simeq\ker(\mathcal{M}-{\rm Id}_2)\,,
\end{equation}
with the later now being free of torsion (which goes in line with the fact that oriented close 3-manifolds do not have torsion on their second homology group). This way,
\begin{equation}\label{eq. hom T3}
\begin{array}{cc}
     H_0(T^3_\mathcal{M};\mathbb{Z})\simeq\mathbb{Z}\,,& H_1(T^3_\mathcal{M};\mathbb{Z})\simeq\mathbb{Z}\oplus\faktor{\mathbb{Z}^{\oplus2}}{{\rm im}\,(\mathcal{M}-{\rm Id}_2)} \\
H_2(T^3_\mathcal{M};\mathbb{Z})\simeq\mathbb{Z}_{\det\circ\rho}\oplus \ker(\mathcal{M}-{\rm Id}_2)& H_3(T^3_\mathcal{M};\mathbb{Z})\simeq\mathbb{Z}_{\det\circ\rho}
\end{array}\;.   
\end{equation}
Again $\mathbb{Z}_{\det\circ\rho}\simeq\mathbb{Z}$ if $G$ respects orientation, as in the $\mathsf{SL}(2;\mathbb{Z})$ case.
Note that since $\ker(\mathcal{M}-{\rm Id}_2)$ and ${\mathbb{Z}^{\oplus2}}/{{\rm im}\,(\mathcal{M}-{\rm Id}_2)}$ have the same free part, Poincar\'e duality holds.

\section{Estimating the on-shell action of a gravitational soliton\label{ap. onshell}}
In this appendix we will estimate the Euclidean action of the gravitational soliton associated to the Bubble of Nothing with topology $\mathcal{B}_2$, a Riemann surface of genus $g$ and a single boundary, with a non-trivial $\mathsf{SL}(2,\mathbb{Z})$ bundle resulting in monodromies $\{\mathcal{M}_{a_i},\mathcal{M}_{b_i}\}_{i=1}^g$ along the non-trivial $2g$ 1-cycles associated to the handles of $\mathcal{B}_2$, with the monodromy along the boundary being given by $\mathcal{M}_{\partial\mathcal{B}_2}=\prod_{i=1}^g[\mathcal{M}_{a_i},\mathcal{M}_{b_i}]$. We start by considering the Euclidean action associated to $D=d+1$-dimensional gravity weakly coupled to an axio-dilaton $\tau$, together with the boundary Gibbons-Hawking-York term:
\begin{equation}
    S_{\rm E}=\underbrace{-\frac{1}{2\kappa_{D}^2}\int_\mathcal{X}\dd^Dx\sqrt{G}\left(\mathcal{R}_G-\frac{\partial_M\tau\partial^M\bar\tau}{2({\rm Im}\,\tau)^2}\right)}_{S_{\rm E}^{\rm (bulk)}}-\underbrace{\frac{1}{\kappa_D^2}\int_{\partial \mathcal{X}}\dd ^{D-1}x\sqrt{H}(K-K_0)}_{_{S_{\rm E}^{\rm (GHY)}}}\,,
\end{equation}
with $\mathcal{X}$ our spacetime and $K$ and $K_0$ the extrinsic curvature of the asymptotic boundary of the gravitational instanton and its background, respectively. The equations of motion $\mathcal{R}_G=\frac{\partial_M\tau\partial^M\bar\tau}{2({\rm Im}\,\tau)^2}$ make the bulk term of the action vanish on-shell, so that the only contribution comes from the GHY term at infinity. Following \cite{Montero,Ruiz:2024jiz}, we parameterize our geometry $\mathcal{X}$ with a $\mathsf{SO}(d-1)$-symmetric ansatz
\begin{equation}
    \dd s^2_\mathcal{X}=G_{MN}\dd x^M\dd x^N=W(y)^2R^2\dd\Omega^2_{d-1}+h_{\alpha\beta}(y)\dd y^\alpha\dd y^\beta\,,
\end{equation}
where we parameterize $\mathcal{B}_2$ by $y^1$ and $y^2$, with metric $h_{\alpha\beta}$. Furthermore, $W(y)$ is a warping function depending on $y$ and $R$ is the nucleation scale. Now, denoting the radial direction $r=RW(y)\in[0,\infty)$, far from the bubble core (i.e., close to the boundary of $\partial\mathcal{X}\simeq\mathbb{S}^{d-1}\times \partial\mathcal{B}_2\simeq \mathbb{S}^{d-1}\times \mathbb{S}^1$), $r\to\infty$ and, assuming a finite bubble mass $M$, we can take a Schwarzschild-like ansatz
\begin{equation}
    \dd s^2_\mathcal{X}\approx r^2\dd\Omega_{d-1}^2+\underbrace{f(r)^{-1}\dd r^2+f(r)\dd \theta^2}_{h_{\alpha\beta}(y)\dd y^\alpha\dd y^\beta}\;,\quad\text{with} \quad f(r)=1-\frac{M}{r^{d-2}}\,,
\end{equation}
and $\theta\in[0,2\pi R)$ an angular coordinate along the boundary. Considering the outward normal vector $n=\sqrt{f(r)}\partial_r$, the extrinsic curvature is given by
\begin{equation}
    K=\nabla_\mu n^\mu=\sqrt{f(r)}\left(\frac{d-1}{r}+\frac{f'(r)}{2f(r)}\right)=\frac{d-1}{r}-\frac{M}{2r^{d-1}}+\mathcal{O}(M^2r^{3-2d})\,.
\end{equation}
Subtracting the background extrinsic curvature ($K_0=\frac{d-1}{r}$ for $M=0$), we obtain that the Euclidean action is then given by
\begin{align}
    S_{\rm E}=\frac{1}{\kappa_D^2}\lim_{r\to\infty}\int_0^{2\pi R}\dd y\int_{\mathbb{S}^{d-1}}\dd\Omega_{d-1}f(r)r^{d-1}\bigg(-\frac{M}{2r^{d-1}}\bigg)=-\frac{\pi\Omega_{d-1}RM}{\kappa_D^2}\,,
\end{align}
with $\Omega_{d-1}=2\pi^{d/2}\Gamma(d/2)^{-1}$ the area of the unit $\mathbb{S}^{d-1}$. All that is left to compute is precisely the value of the ``mass'' $M$, for which we will need to approach the bubble wall and study the geometry of the $\mathcal{B}_2$ bulk. In general we do not have the explicit expression of $\mathcal{B}_2$ metric $h_{\alpha\beta}$ or the profile $\tau:\mathcal{B}_2\to \mathbb{H}$, and some assumptions and approximations will be in place. Splitting $M=M_{\rm grav}+M_{\tau}$, we will consider contributions both from the geometry of the bubble and from the profile of the axio-dilaton. Take first Witten's bubble of nothing (where $g=0$ and thus no monodromies are supported), or equivalently consider the very tip of $\mathcal{B}_2$, where our spacetime locally looks like $\mathbb{D}^2\times \mathbb{S}^{d-1}$. Approximating there 
\begin{equation}
    \dd s^2_{\mathcal{B}_2}= h_{\alpha\beta}(y)\dd y^\alpha\dd y^\beta\approx \left(1-\frac{M_0}{r^{d-2}}\right)^{-1}\dd r^2+ \left(1-\frac{M_0}{r^{d-2}}\right)\dd \theta^2\;,
\end{equation}
with $r\geq M_{0}^{\frac{1}{d-2}}$ and $\theta\in[0,2\pi R)$ periodic with respect to the KK circle. Now, similar to the original argument by Witten in \cite{Witten:1981gj}, in order to avoid a conical singularity at $r^{d-2}=M_0$, we need to identify $M_0=\left(\frac{d-2}{2}R\right)^{d-2}$. On the other hand, upon compactification the integral of the internal Ricci curvature acts like a source for the lower-dimensional potential. Through Gauss-Bonnet theorem (here the metric is dimensionfull) we find
    \begin{equation}\label{eq. Gauss Bonnet}
        \int_{\mathcal{B}_2}\sqrt{h_2} \dd^2x\mathcal{R}_{\mathcal{B}_2}+2\int_{\partial\mathcal{B}_2}\kappa_{\rm g}\dd s=4\pi(2-n-2g)\kappa^2\;,
    \end{equation}
    so that with $n=1$ we expect the contribution from the gravitational part of the mass to grow proportional to the Euler characteristic $\chi(\mathcal{B}_2)=1-2g$,
    \begin{equation}
        M_{\rm grav}=\mathcal{O}(\chi(\mathcal{B}_2)M_0)=\mathcal{O}(1-2g)R^{d-2}\,.
    \end{equation}
    Note that for $g>1$ the contribution is negative, lowering the effective mass of the bubble. Additional to this, we must consider the contribution from the kinetic term of the axio-dilaton $\tau$. Due to the Riemannian metric defined on $\mathbb{H}$, such contribution will be positive. Weighed by the normalized volume,
    \begin{equation}
        M_\tau\sim R^{d-2}\int_{\mathcal{B}_2}\dd^2 x\,\sqrt{h_2}\frac{\partial_\mu\tau\partial^\mu\bar\tau}{2({\rm Im}\,\tau)^2}=R^{d-2}\frac{i}{2}\int_{\tau(\mathcal{B}_2)\subseteq\mathbb{H}}\frac{\dd \tau\wedge\dd \bar\tau}{({\rm Im}\,\tau)^2}\,.
    \end{equation}
Seeing the axio-dilaton field as a map $\tau:\mathcal{B}_2\to\mathbb{H}$ to the upper-half plane, precisely the above integral amounts to the area of the image $\tau(\mathcal{B}_2)\subseteq\mathbb{H}$ with respect to the Poincar\'e metric. Now, in general such integral is not possible to compute (save for $\tau\equiv$constant) without explicitly having the expressions of $h_{ij}$ and $\tau$. We can, however, try to estimate its value. First of all, conformal invariance on $\mathcal{B}_2$ implies that the value of such integral does not depend on the size of the different cycles, what allows us to approximate it simply by properties of the different monodromies. We notice first that the area of the fundamental domain $\mathcal{F}=\mathbb{H}/\mathsf{SL}(2,\mathbb{Z})\simeq\mathbb{S}^2$ is finite, given by ${\rm vol}(\mathcal{F})=\frac{\pi}{3}\approx 1.0472$, but since $\mathcal{B}_2$ has a boundary, $\partial\tau(\mathcal{B}_2)\neq \emptyset$ too, and thus $\tau(\mathcal{B}_2)$ will not simply consist in a series of wrappings of $\mathcal{F}$. However, in a hyperbolic space the area of a polygon, such as $\tau(\mathcal{B}_2)$ for our purposes grows linearly with the length of its sides. We then take $\mathcal{B}_2$ to be a polygon of $4g+n$ sides as in Figure \ref{fig:examples}. Now, given some point in $x\in\mathcal{F}$, the distance to its image under a $\mathsf{SL}(2,\mathbb{Z})$ transformation $\mathcal{M}$ grows like \cite{KatokSvetlana1992Fg}
\begin{equation}
    {\rm dist}(x,\mathcal{M}(x))\sim\left\{\begin{array}{ll}
       \mathcal{O}(1)  & \text{for $\mathcal{M}$ elliptic} \\
       0&\text{for $\mathcal{M}$ parabolic}\\
       2\log({\rm |Tr(\mathcal{M})|})  & \text{for $\mathcal{M}$ hyperbolic}
    \end{array}\right.\;.
\end{equation}
Parabolic transformations, of the type $T^n$, which implements $\tau\to\tau+n$, return the initial point to itself within the fundamental domain $\mathcal{F}$. However, in doing so they are wrapping the image of $\mathcal{B}_2$ $n$ times along the $i\infty\in\mathcal{F}$ point, and thus actually multiplies the total area by $n$. Adding the contribution of the different handles, we can estimate that the kinetic contribution of the axio-dilaton scales like 
\begin{equation}
    \frac{i}{2}\int_{\tau(\mathcal{B}_2)\subseteq\mathbb{H}}\frac{\dd \tau\wedge\dd \bar\tau}{({\rm Im}\,\tau)^2}\approx\mathcal{O}\bigg(\sum_{i=1}^g\log\big|{\rm Tr}([\mathcal{M}_{a_i},\mathcal{M}_{b_i}])\big|\bigg)>0\,.
\end{equation}
This finally allows us to conclude 
\begin{equation}
    S_{\rm E}=(M_{{\rm Pl},d+1}R)^{d-1}\mathcal{O}\bigg(1-2g+\sum_{i=1}^g\log\big|{\rm Tr}([\mathcal{M}_{a_i},\mathcal{M}_{b_i}])\big|\bigg)
\end{equation}

From the equations of motion, $\mathcal{R}_{\mathcal{B}_2}=\frac{\partial_\mu\tau\partial^\mu\bar\tau}{2({\rm Im}\,\tau)^2}$, assuming the gravitational soliton $\mathcal{B}_2$ size is of the order of the KK scale, we conclude that the mean Ricci scalar curvature of the geometry is given by
\begin{equation}
    |\overline{\mathcal{R}}_{\mathcal{B}_2}|=\left|\frac{\int_{\mathcal{B}_2}\sqrt{h_2} \dd^2x\mathcal{R}_{\mathcal{B}_2}}{\int_{\mathcal{B}_2}\sqrt{h_2} \dd^2x}\right|=\mathcal{O}\bigg(m_{\rm KK}^2\times \sum_{i=1}^g\log\big|{\rm Tr}([\mathcal{M}_{a_i},\mathcal{M}_{b_i}])\big|\bigg)
\end{equation}

\section{Some results on commutators widths\label{app. results comm}}
In this appendix we present a series of results regarding commutator widths that are used along the main text. We give references for those results present in the literature, while those which have a relatively straightforward derivation are shown in detail.

\begin{thm}[{\cite[Example 2.6]{CULLER1981133}}] Let $G$ be some group and $u,\,v\in G$. Then $[u,v]^k$ can be written as the product of $\lfloor\frac{k}{2}\rfloor+1$ commutators
\end{thm}
\begin{cor}[{\cite[\textit{Lemme} 1.3]{Bavard1991}}]\label{cor.cor1}
 Let $G$ be a group and $\{u_i,v_i\}_{i=1}^r$. Then $\big([u_1,v_1]\dots[u_r,v_r]\big)^k$ can be written as the product of $k(r-1)+\lfloor\frac{k}{2}\rfloor+1$ commutators.
\end{cor}

\begin{thm}\label{thm. large n}
    Let $G$ be a discrete group with abelianization $G^{\rm ab}$. Given a set $T$ or representatives of $G^{\rm ab}$ and an element $g\in G-[G,G]$, then ${\rm cl}_G(g^n)$ is independent of the choice of $T$ in the $n\to\infty$ limit if the commutator subgroup is of finite index, $[G:[G,G]]<\infty$.
\end{thm}
\begin{proof}
    Take $G$ discrete, so that, by the fundamental theorem of finitely generated abelian groups,
    \begin{equation}
        G^{\rm ab}=G/[G,G]\simeq\mathbb{Z}^n\oplus\bigoplus_{i=1}^t\mathbb{Z}_{q_i}\;,
    \end{equation}
    where $\{q_i\}_{i=1}^t$ is a set of distinct powers of (not necessarily different) primes. Take now $g\in G$ and two different set of representatives of $G^{\rm ab}$, $T$ and $\tilde T\leq G$, in such a way that
    \begin{equation}
        g=t\alpha=\tilde t \beta\,,\quad\text{with }\,\alpha,\,\beta\in [G,G]\quad\text{and \,}t\in T,\,\tilde t\in\tilde T\,.
    \end{equation}
It is clear then that $[g]=[t]=[\tilde t]\in G^{\rm ab}$, and thus generate the same subgroup $H\trianglelefteq G^{\rm ab}$. Then $H$ can be either isomorphic to $\mathbb{Z}$ or $\mathbb{Z}_k$ for some $k\in\mathbb{N}$. By assumption, since $[G:[G,G]]$ is finite, we must have $n=0$ and $G^{\rm ab}$ is pure torsion.

Take then $H\simeq\mathbb{Z}_k$, generated by $\{l\in\mathbb{Z}_k\,|\,(l,k)=1\}$. Now, since for any $n\in \mathbb{Z}$ and $a$ generator of $\mathbb{Z}_k$ we have that
\begin{equation}
    n=a(a^{-1}_{\mathbb{Z}_k}n\mod k)+k\left\lfloor \frac{n-a(a^{-1}_{\mathbb{Z}_k}n\mod k)}{k}  \right\rfloor\,,
\end{equation}
where $a^{-1}_{\mathbb{Z}_k}$ is the inverse of $a$ in $\mathbb{Z}_k$, we conclude that for any generator $t$ of $H\simeq \mathbb{Z}_k$ (so that $t^k$ is a product of commutators), then
\begin{equation}
    t^n=(t^ a)^{a^{-1}_{\mathbb{Z}_k}n\mod k}(t^k)^{\left\lfloor \frac{n-a(a^{-1}_{\mathbb{Z}_k}n\mod k)}{k}  \right\rfloor}\,,
\end{equation}
where $t^k\in[G,G]$. Let us write
\begin{equation}
    t^k=[p_1,q_1]\dots[p_{{\rm cl}_G(t^k)},q_{{\rm cl}_G(t^k)}]\,.
\end{equation}
From Corollary \ref{cor.cor1}, we can then write $(t^k)^{\left\lfloor \frac{n-a(a^{-1}_{\mathbb{Z}_k}n\mod k)}{k}  \right\rfloor}$ as a product of
\begin{equation}
    \left\lfloor \frac{n-a(a^{-1}_{\mathbb{Z}_k}n\mod k)}{k}  \right\rfloor\big({\rm cl}_G(t^k)-1\big)+\left\lfloor\frac{1}{2}\left\lfloor \frac{n-a(a^{-1}_{\mathbb{Z}_k}n\mod k)}{k}  \right\rfloor\right\rfloor+1\,.
\end{equation}
Irrespective of the choice of $a$ (this is, $T$ or $\tilde T$), we have that the number of commutators grows like $\mathcal{O}\left(\lfloor\frac{n}{k}\rfloor[{\rm cl}_G(t^k)-1]+\lfloor\frac{n}{2k}\rfloor\right)$, where ${\rm cl}_G(t^k)$ does not depend on the choice of representatives since $t^k$ is itself a product of commutators.
    
\end{proof}
Note that if in the above proof $t$ and $\tilde t$ generate $\mathbb{Z}$, then it is not possible to give a general relation of how the commutator width grows. Given $t=\tilde tc$, with $c$ some product of commutators, then we will find a mismatch of $n[{\rm cl}_G(c)-1]+\lfloor\frac{n}{2}\rfloor+1$ commutators in $t^n$ with respect to $\tilde t^n$, with $c$ not being universal. On the other hand, since there is a good order in ${\rm cl}_G:[G,G]\to\mathbb{Z}_{>0}$, for every $g\in G$ we can define an ``absolute'' commutator width
\begin{equation}
    {\rm cl}_G(g)=\min_T {\rm cl}_{G,T}(g)
\end{equation}
for all possible representatives $T$ of $G$. We will not need to use this notion along the main text.

\begin{thm}\label{th.inf rep}
    Let $G$ be a discrete and finitely generated infinite group with infinite commutator width, and let $T$ be a set generating the representatives of $G^{\rm ab}$. If we allow a maximum number of commutators by which the elements of $G$ can be expressed, then an infinite number of additional representatives are needed.
\end{thm}
\begin{proof}
    Let 
    \begin{equation}
        W_K=\{x\in [G,G]\;:\;x=[g_1,h_1]\dots[g_K,h_K]\,,\text{with }g_i,\,h_i\in G\}\,,
    \end{equation}
    so that $\bigcup_{K=1}^\infty W_K=[G,G]$, while for every individual $W_K\neq [G,G]$ for finite $K$. Fix some $K\in \mathbb{Z}_{\geq 0}$, and for a given $x\in G$, write $x=r\cdot c$, with $c\in W_K$ and $r\in T_K$, some new set of representatives. We want to show that $T_K$ has to be infinite.\\
    By contradiction, assume then that $T_K=\{t_1,\dots,t_n\}$, so that one can write
    \begin{equation}
        G=\bigcup_{i=1}^nt_iW_K\,,\quad \text{and in particular}, \quad [G,G]\subseteq \bigcup_{i=1}^nt_iW_K\,.
    \end{equation}
    This then means that there is a subset $\{\tau_j\}_{j=1}^m\subset T_K$ such that $\tau_j\in[G,G]$, with commutator length ${\rm cl}(\tau_j)$. However, by assumption, for every $g\in [G,G]$, this means that
    \begin{equation}
        g=\mathsf{t}\cdot r\,,\quad\text{with}\; \mathsf{t}\in \{\tau_j\}_{j=1}^m\;\text{ and }\;r\in W_K\,,
    \end{equation}
    so that
    \begin{equation}
        {\rm cl}(g)\leq {\rm cl}(\mathsf{t})+{\rm cl}(r)\leq \max_j{\rm cl}(\tau_j)+K\,,
    \end{equation}
    where the r.h.s. of the inequality is clearly finite. This way there is an upper bound to the commutator length of elements in $[G,G]$, which clearly goes against assumptions. Them $T_K$ must be an infinite set of representatives.
\end{proof}

\begin{thm}\label{th.cl prod} Given two groups $A$ and $B$ with respective commutator widths ${\rm cl}(A)$ and ${\rm cl}(B)$, the commutator width of the direct product $A\times B$ fulfills
\begin{equation}
        \max\{{\rm cl}(A),{\rm cl}(B)\}\leq {\rm cl}(A\times B)\leq {\rm cl}(A)+{\rm cl}(B)\;,
    \end{equation}    
\end{thm}
\begin{proof}
    Notice that, due to factorization, $[A\times B,A\times B]=[A,A]\times[B,B]$. Being clear that $[(a,1_B),(a',1_B)]=([a,a'],1_B)$ and $[(1_A,b),(1_A,b')]=(1_A,[b,b'])$ one can always split
    \begin{equation}
        (x,y)=\left(\prod_{i=1}^{n_x}([a_i,a_i'],1_B)\right)\left(\prod_{j=1}^{n_y}(1_A,[b_,b_j'])\right)\,,\;\text{ with }\, n_x\leq{\rm cl}(A)\;\text{ and }\; n_y\leq{\rm cl}(B)\,,
    \end{equation}
thus obtaining the upper bound, while the lower one results from considering the elements $\{(a,1_B)\}_{a\in A}\cup\{(1_A,b)\}_{b\in B}\subset A\times B$.
\end{proof}
\begin{cor}\label{cor. fin prod}
    Given two groups $A$ and $B$, ${\rm cl}(A\times B)<\infty$ if and only if both ${\rm cl}(A)<\infty$ and ${\rm cl}(B)<\infty$.
\end{cor}

\begin{thm}\label{thm. finite index}
    Consider a group $G$ with finite commutator width, ${\rm cl}(G)<\infty$. Given a normal subgroup $H\trianglelefteq G$ with finite index, $[G:H]<\infty$, we have ${\rm cl}(H)<\infty$.
\end{thm}
\begin{proof}
    Take a group $G$ and a normal subgroup $H\trianglelefteq G$ of finite index, $[G:H]<\infty$. It is straightforward that we have the following relations, 
\begin{equation}
    H\leq H \cap [G,G]\leq [G,G]\;,\quad [H,H]\trianglelefteq H\trianglelefteq G\quad \text{and} \quad [H,H]\trianglelefteq [G,G]\;.
\end{equation}
Take now $T\subseteq G$ to be a set of representatives of $G/H$, finite by assumption. We show now that we can take $T$ such that for all $t,\,t'\in T$, $[t,t']\in H\cap [G,G]$. While it is clear that we always have $[t,t']\in [G,G]$, in order to find $[t,t']\in H$, we consider the projection and section
\begin{equation}
    \pi:G\longrightarrow G/H\quad\text{and}\quad s:G/H\longrightarrow G\,,
\end{equation}
and define $H_0=\pi^{-1}([G/H,G/H])$, with $H\trianglelefteq H_0\trianglelefteq G$ and $H_0/H=[G/H,G/H]$, in such a way that $G/H_0\simeq (G/H)/[G/H,G/H]=(G/H)^{\rm ab}$, abelian.

Taking $T=s(G/H_0)\subseteq G$, we find thus for all $t,\;t'\in T$ we have $q,\,q'\in G/H_0$ such that $\pi([t,t'])=[\pi(t),\pi(t')]=[q,q']=1$, and thus $[t,t']\in H$. With this choice, $[t,t']\in H\cap [G,G]$ for all $t,\,t'\in T$.

On the other hand, since $[G:H]$ is finite, so are $[[G,G]:[H,H]]$ and $[H\cap[G,G]:[H,H]]$. Denote now
\begin{equation}
    F=\{[t,t']\,|\,t,\,t'\in T\}\subseteq H\cap [G,G]\,,
\end{equation}
which is crucially finite, since $T$ is a finite set. We then have that for all $[t,t']\in F$, $[t,t']\in\tilde{g}[H,H]$, with $\tilde{g}\in H\cap [G,G]$. Such $\tilde{g}$ will then be able to be written as a finite product of commutators in $H$, and thus the same can be said about $[t,t']$. Since $F$ is finite group, we can then take $N_F<\infty$ as the maximum of all the commutator lengths in $H$ of the elements in $F$. Note that $N_F$ depends only on $T=s(G/H_0)$, and thus is independent of the individual elements of $F$.

Now, for any $h\in [H,H]$ we have $h=\prod_i^ N[g_i,g_i']$, with $N\leq {\rm cl}(G)$ and $g_i=h_it_i,\,g_i'=h_i't_i'\in G$, where $h_i,\,h_i'\in H$ and $t_i,\, t_i'\in T$ representatives of $G/H$. Simple computation shows that 
    \begin{equation}
        [g_i,g_i']=[h_it_i,h_i't_i']={}^{h_i}[t_i,h_i']\,{}^{h_ih_i'}[t_i,t_i']\,{}^{}[h_i,h_i']\,{}^{h_i'}[h_i,t_i']\;,
    \end{equation}
    where we denote ${}^ax:=axa^{-1}$. From our previous results, it is clear that the four commutators above belong to $[H,H]$, since this is not affected by conjugation. Since conjugation also respects normal subgroups and the commutator length, we conclude that $[g_i,g_i']$ can be written as the product of \emph{at most} $3+N_F$ commutators in $H$. We thus conclude that ${\rm cl}(H)\leq (3+N_F) {\rm cl}(G)$, finite if ${\rm cl}(G)<\infty$.

\end{proof}
\begin{thm}\label{th. iso}
    Consider a group $G$ with finite commutator width, ${\rm cl}(G)<\infty$, related to another group $\tilde G$ through the isogeny
    \begin{equation}
        1\to Z\to\tilde G\xrightarrow{\pi}G\to1\,,
    \end{equation}
    with finite kernel $Z$. Then $\tilde G$ also has finite commutator width, ${\rm cl}(\tilde G)<\infty$.
\end{thm}
\begin{proof}
    Take the derived isogenous group, $[\tilde G,\tilde G]$, in such a way that for all $\tilde x,\,\tilde y\in [\tilde G,\tilde G]$, $\pi([\tilde x,\tilde y])=[\pi(\tilde x),\pi(\tilde y)]\in [G,G]$, so that $\pi([\tilde G,\tilde G])=[G,G]$ and $[\tilde G,\tilde G]\leq\pi^{-1}([G,G])$.
    
    On the other hand, given $\tilde g\in\tilde G$ with $\pi(\tilde g)=g$, we have that $g$ admits a finite number $|Z|$ of lifts, $\{c\tilde g\}_{c\in Z}$. Take then a specific lift $c \tilde G$ of $G$, and an element $\tilde a\in[\tilde G,\tilde G]$ with $\pi(\tilde a)=a\in[G,G]$. Only those lifts $c\tilde a$ with $c\in Z\cap [\tilde G,\tilde G]$ (this is, the elements of $Z$ that can be written as products of commutators of $\tilde G$), will belong to $[\tilde G,\tilde G]$. If this is not the case, then our element $c\tilde a$ will be $c$ times the lift of the commutator expansion of $a$ in $[G,G]$, which is at most ${\rm cl}_G(a)$.

    If $c\in Z\cap [\tilde G,\tilde G]$, then both $\tilde a$ and $c\tilde a$ are in $[\tilde G,\tilde G]$. This means that such element is either the lift of a set of commutators from $G$, or that multiplied by $c$, which has commutator length ${\rm cl}_{\tilde G}(c)<\infty$ in $\tilde G$.

    We thus conclude that
    \begin{equation}
        {\rm cl}(G)\leq{\rm cl}(\tilde G)\leq {\rm cl}(G)+\max_{c\in Z\cap[\tilde G,\tilde G]} {\rm cl}_{\tilde G}(c)<\infty\,,
    \end{equation}
    where again the fact that $Z$ is finite plays a crucial role. The lower bound is saturated when the lift is to $c=1_Z$ or $c\not\in Z\cap[\tilde G,\tilde G]$.
    \end{proof}
    \begin{cor}\label{cor.double}
        Given any group $G$ with finite commutator width, ${\rm cl}(G)<\infty$, any double cover
        \begin{equation}
            1\to\mathbb{Z}_2\to\tilde G\to G\to 1\,,
        \end{equation}
        will have finite commutator width, ${\rm cl}(\tilde G)$.
    \end{cor}
    \begin{thm}[{\cite[Corollary 3.3]{CULLER1981133}}]
        Free groups have infinite commutator width.
    \end{thm}

    \begin{cor}\label{cor. virt free}
        Virtually free groups (this is, with a free subgroup of finite index) have infinite commutator width.
    \end{cor}
    \begin{proof}
       By contradiction, if $G$ had finite commutator width and a free group $F_r\trianglelefteq G$ of finite index, by Theorem \ref{thm. finite index}, $F_r$ would have finite commutator width too. Since this is not the case, we conclude that ${\rm cl}(G)=\infty$.  
    \end{proof}

    \begin{thm}\label{th. semidirect}
        Given a semidirect product $G=N\rtimes H$, we have
\begin{equation}
    {\rm cl}(A\rtimes B)<\infty\,\Longleftrightarrow\,{\rm cl}(A)<\infty\;\text{ and }\; {\rm cl}(B)<\infty\;\text{ and }\; {\rm cl}([A,B])<\infty\,.
\end{equation}
    \end{thm}
\begin{proof}
    Consider a semidirect product $G=N\rtimes H$, with $N\trianglelefteq G$. By definition, there is a group homomorphism $\varphi:H\to {\rm Aut}(N)$ given by $\varphi_h(n)=hnh^{-1}\in N$ for $h\in H$ and $n\in N$, what allows us to write $[h,n]=\varphi_h(n)h$. Since for all $g\in G$ we have some $n\in N$ and $h\in H$ such that $g=hn$, then we have that through some manipulation
\begin{equation}
                [g_1,g_2]=[n_1h_1,n_2h_2]=\underbrace{[\varphi_{h_1}(n_2),n_1]}_{\in [N,N]}\underbrace{[n_2,[n_1,h_1h_2h_1^{-1}]]}_{\in [N,N]}\underbrace{[n_1,h_1h_2h_1^{-1}]}_{\in[N,H]}\underbrace{[n_2,[h_1,h_2]}_{\in [N,H]}\underbrace{[h_1,h_2]}_{\in [H,H]}\,,
\end{equation}
so that we will have $[G,G]=[N,N]\cdot[N,H]\cdot[H,H]$, with further $[N,H]\leq N$ due to $N$ being a normal subgroup of $G$. This way all commutators in $G$ can be written as a product of commutators of the three above groups, with at least one commutator from each group for every commutator of $G$ before simplifying. It is clear that if the commutator lengths of $N$, $H$ or $[N,H]$ are infinite, there is no way that every element in $[G,G]$ can be written as a finite number of commutators, and vice-versa. 
\end{proof}

\subsection{Realization of monodromies in terms of commutators\label{App. real comm}}

We finally give an explicit explicit realization of the different $\mathsf{SL}(2,\mathbb{Z})$ elements depicted in Table \ref{tab.int},\footnote{We will only focus on the conjugacy class representatives from Table \ref{tab.int}. Since $g\left(\prod_i\mathcal{M}_i\right)g^{-1}=\prod_ig\mathcal{M}_ig^{-1}$ and $g[\mathcal{M}_1,\mathcal{M}_2]g^{-1}=[g\mathcal{M}_1g^{-1},g\mathcal{M}_2g^{-1}]$, the needed defect and gravitational solitons will be the same modulo conjugacy.} in terms of commutators (i.e. gravitational solitons) and defects associated to different monodromy groups.
 
 \subsection*{$\boxed{\mathsf{SL}(2,\mathbb{Z})}$}
 The monodromies from Table \ref{tab.int} are realized in terms of $\mathsf{SL}(2,\mathbb{Z})$ commutators and (anti-)D7 brane defects as
\begin{subequations}\label{eq.bordsl2z}
    \begin{align}
       \label{eq.bordsl2z-1} 
       \colorbox{yellow!15}{$\begin{pmatrix}
            1&n\\
            0&1
        \end{pmatrix}$}&=(\mathbf{T^{12}})^{\left\lfloor\frac{n}{12}\right\rfloor}T^{n\ {\rm mod}\, 12}\\
      \colorbox{red!15}{$\begin{pmatrix}
        n&nm-1\\1&m
    \end{pmatrix}$}&=\left\{\begin{array}{ll}
       T^{n+m-3\ {\rm mod}\, 12}(\mathbf{T^{12}})^{\left\lfloor\frac{n+m-3}{12}\right\rfloor}  \text{\scriptsize$\left[{\begin{pmatrix}
           1&1\\0&1
       \end{pmatrix},\begin{pmatrix}
           1-m&-m^2\\1&1+m
       \end{pmatrix}}\right]$}  &  n+m\geq 3\\
        \overline{T}^{3-n-m\ {\rm mod}\, 12}(\mathbf{\overline{T}^{12}})^{\left\lfloor\frac{3-n-m}{12}\right\rfloor}\text{\scriptsize$\left[{\begin{pmatrix}
           1&1\\0&1
       \end{pmatrix},\begin{pmatrix}
           1-m&-m^2\\1&1+m
       \end{pmatrix}}\right]$} &  n+m\leq 3
    \end{array}\right.\\
      \colorbox{blue!15}{$\begin{pmatrix}
        n&1-nm\\-1&m
    \end{pmatrix}$}&=\left\{\begin{array}{ll}
       T^{3-(n+m)\ {\rm mod}\, 12}(\mathbf{T^{12}})^{\left\lfloor\frac{3-(n+m)}{12}\right\rfloor}\text{\scriptsize$\left[{\begin{pmatrix}
           m-2&-(m-3)^2\\1&4-m
       \end{pmatrix},\begin{pmatrix}
           1&1\\0&1
       \end{pmatrix}}\right]$}  &  n+m\leq 3\\
        \overline{T}^{m+n-3\ {\rm mod}\, 12}(\mathbf{\overline{T}^{12}})^{\left\lfloor\frac{m+n-3}{12}\right\rfloor}\text{\scriptsize$\left[{\begin{pmatrix}
           m-2&-(m-3)^2\\1&4-m
       \end{pmatrix},\begin{pmatrix}
           1&1\\0&1
       \end{pmatrix}}\right]$}  &  n+m\geq 3
    \end{array}\right.\\
     \colorbox{green!15}{$ \begin{pmatrix}
        -1&n\\0&-1
    \end{pmatrix}$}&=\left\{\begin{array}{ll}
        \text{\scriptsize$\left[{ \begin{pmatrix}
            1&0\\1&1
        \end{pmatrix},\begin{pmatrix}
            1&1\\0&1
        \end{pmatrix}}\right]\left[{ \begin{pmatrix}
            4&-9\\1&-2
        \end{pmatrix},\begin{pmatrix}
            1&1\\0&1
        \end{pmatrix}}\right]$}(\mathbf{\overline{T}^{12}})^{\left\lfloor\frac{n-6}{12}\right\rfloor} \overline{T}^{n-6\ {\rm mod}\,12}& n\geq 6  \\
         \text{\scriptsize$\left[{ \begin{pmatrix}
            1&0\\1&1
        \end{pmatrix},\begin{pmatrix}
            1&1\\0&1
        \end{pmatrix}}\right]\left[{ \begin{pmatrix}
            4&-9\\1&-2
        \end{pmatrix},\begin{pmatrix}
            1&1\\0&1
        \end{pmatrix}}\right]$}(\mathbf{{T}^{12}})^{\left\lfloor\frac{6-n}{12}\right\rfloor} {T}^{6-n\ {\rm mod}\,12}& n\leq 6 
    \end{array}\right.
    \end{align}
\end{subequations}

where $ T={\begin{psmallmatrix}
    1&1\\0&1
\end{psmallmatrix}}$, $\overline{T}=T^{-1}={\begin{psmallmatrix}
    1&-1\\0&1
\end{psmallmatrix}}$,  $(\mathbf{T^{12}})^k$ is the set of $k+1$ commutators known to exist from \eqref{eq.kplus1comm} and $\mathbf{\overline{T}^{12}}$ is the inverse through $[A,B]^{-1}=[B,A]$, and we have used the following identities:
\begin{equation}
    \text{\small$\begin{array}{c}
        \begin{pmatrix}
            3&-1\\1&0
        \end{pmatrix}=\left[{ \begin{pmatrix}
            1&1\\0&1
        \end{pmatrix},\begin{pmatrix}
            1&0\\1&1
        \end{pmatrix}}\right]\;,\qquad \begin{pmatrix}
            -1&0\\0&-1
        \end{pmatrix}=\left[{ \begin{pmatrix}
            1&0\\1&1
        \end{pmatrix},\begin{pmatrix}
            1&1\\0&1
        \end{pmatrix}}\right] T^3\left[{ \begin{pmatrix}
            1&0\\1&1
        \end{pmatrix},\begin{pmatrix}
            1&1\\0&1
        \end{pmatrix}}\right]T^3\\\begin{pmatrix}
            n+m&-1\\1&0
        \end{pmatrix}=T^n\begin{pmatrix}
            m&-1\\1&0
        \end{pmatrix}
        \;,\quad \qquad T^n\begin{pmatrix}
            1&0\\1&1
        \end{pmatrix}\overline{T}^{n}=\begin{pmatrix}
            n+1&-n^2\\1&1-n
        \end{pmatrix}\,.
    \end{array}$}
\end{equation}
It is clear that for all one- (two-)parameter families labeled by $n$ ($n,m$), the number of commutators used grows like $\mathcal{O}\left(\frac{|n|}{6}\right)$ ($\mathcal{O}\left(\frac{|n-m|}{6}\right)$) in the $|n|\to\infty$ ($|n-m|\to \infty$) limit.

\subsection*{$\boxed{\mathsf{GL}(2,\mathbb{Z})}$}
Finally, we consider allow for $\mathsf{GL}(2,\mathbb{Z})$ along the 1-cycles of our gravitational solitons, such that
\begin{subequations}
    \begin{align}
          \colorbox{yellow!15}{$\begin{pmatrix}
            1&n\\
            0&1
        \end{pmatrix}$}&=\left\{
        \begin{array}{ll}
             T\text{\scriptsize$\left[{\begin{pmatrix}
                 1&0\\0&-1
             \end{pmatrix},\begin{pmatrix}
                 1&\tfrac{n-1}{2}\\0&-1
             \end{pmatrix}}\right]$}& n\text{ odd} \\
            \text{\scriptsize$\left[{\begin{pmatrix}
                 1&0\\0&-1
             \end{pmatrix},\begin{pmatrix}
                 1&\tfrac{n}{2}\\0&-1
             \end{pmatrix}}\right]$} & n\text{ even}
        \end{array}
        \right.\\
          \colorbox{green!15}{$\begin{pmatrix}
            -1&n\\
            0&-1
        \end{pmatrix}$}&=\left\{
        \begin{array}{ll}
             T\text{\scriptsize$\left[{\begin{pmatrix}
                 1&1\\0&-1
             \end{pmatrix},\begin{pmatrix}
                 n&\tfrac{n+1}{2}\\2&1
             \end{pmatrix}}\right]$}& n\text{ odd} \\
            \text{\scriptsize$\left[{\begin{pmatrix}
                 1&1\\0&-1
             \end{pmatrix},\begin{pmatrix}
                 n-1&\tfrac{n}{2}\\2&1
             \end{pmatrix}}\right]$} & n\text{ even}
        \end{array}
        \right.\\
          \colorbox{blue!15}{$ \begin{pmatrix}
            n&1-nm\\
            -1&m
        \end{pmatrix}$}&=\left\{ \begin{array}{ll}
             T\text{\scriptsize$\left[{ \begin{pmatrix}
                 1&0\\0&-1
             \end{pmatrix},\begin{pmatrix}
                1&-\tfrac{n+1}{2}\\0&-1
             \end{pmatrix}}\right]\left[{ \begin{pmatrix}
                 0&1\\1&0
             \end{pmatrix},\begin{pmatrix}
                \tfrac{m+1}{2}&\tfrac{m-1}{2}\\1&1
             \end{pmatrix}}\right]$}& n\text{ odd},\,m\text{ odd} \\
           \text{\scriptsize$\left[{ \begin{pmatrix}
                 1&0\\0&-1
             \end{pmatrix},\begin{pmatrix}
                1&-\tfrac{n}{2}\\0&-1
             \end{pmatrix}}\right]\left[{ \begin{pmatrix}
                 0&1\\1&0
             \end{pmatrix},\begin{pmatrix}
                \tfrac{m}{2}&\tfrac{m}{2}-1\\1&1
             \end{pmatrix}}\right]$}& n\text{ even},\,m\text{ odd} \\
             \text{\scriptsize$\left[{ \begin{pmatrix}
                 1&2\\0&-1
             \end{pmatrix},\begin{pmatrix}
                1&-\tfrac{n+5}{2}\\0&-1
             \end{pmatrix}}\right]\left[{ \begin{pmatrix}
                 1&0\\1&-1
             \end{pmatrix},\begin{pmatrix}
                \tfrac{m+3}{2}&-1\\1&0
             \end{pmatrix}}\right]$}& n\text{ odd},\,m\text{ even} \\
             \text{\scriptsize$\left[{ \begin{pmatrix}
                 1&0\\0&-1
             \end{pmatrix},\begin{pmatrix}
                1&-\tfrac{n}{2}\\0&-1
             \end{pmatrix}}\right]$} T\text{\scriptsize$\left[{ \begin{pmatrix}
                 0&1\\1&0
             \end{pmatrix},\begin{pmatrix}
                \tfrac{m}{2}&\tfrac{m}{2}-1\\1&1
             \end{pmatrix}}\right]$}& n\text{ even},\,m\text{ even} \\
        \end{array}\right.   \\
          \colorbox{red!15}{$\begin{pmatrix}
            n&mn-1\\
            1&m
        \end{pmatrix}$}&=\left\{ \begin{array}{ll}
            \text{\scriptsize$ \left[{ \begin{pmatrix}
                \tfrac{n+1}{2}&\tfrac{n-1}{2}\\1&1
             \end{pmatrix},\begin{pmatrix}
                 0&1\\1&0
             \end{pmatrix}}\right]\left[{ \begin{pmatrix}
                1&-\tfrac{m+1}{2}\\0&-1
             \end{pmatrix},\begin{pmatrix}
                 1&0\\0&-1
             \end{pmatrix}}\right]$}\overline{T}& n\text{ odd},\,m\text{ odd} \\
           \text{\scriptsize$\left[{ \begin{pmatrix}
                \tfrac{n}{2}&\tfrac{n}{2}-1\\1&1
             \end{pmatrix},\begin{pmatrix}
                 0&1\\1&0
             \end{pmatrix}}\right]\left[{ \begin{pmatrix}
                1&-\tfrac{m}{2}\\0&-1
             \end{pmatrix} \begin{pmatrix}
                 1&0\\0&-1
             \end{pmatrix}}\right]$}& n\text{ even},\,m\text{ odd} \\
           \text{\scriptsize$ \left[{ \begin{pmatrix}
                \tfrac{n+3}{2}&-1\\1&0
             \end{pmatrix},\begin{pmatrix}
                 1&0\\1&-1
             \end{pmatrix}}\right] \left[{ \begin{pmatrix}
                1&-\tfrac{m+5}{2}\\0&-1
             \end{pmatrix},\begin{pmatrix}
                 1&2\\0&-1
             \end{pmatrix}}\right]$}& n\text{ odd},\,m\text{ even} \\
             \text{\scriptsize$\left[{ \begin{pmatrix}
                \tfrac{n}{2}&\tfrac{n}{2}-1\\1&1
             \end{pmatrix},\begin{pmatrix}
                 0&1\\1&0
             \end{pmatrix}}\right]\overline{T}\left[{ \begin{pmatrix}
                1&-\tfrac{m}{2}\\0&-1
             \end{pmatrix},\begin{pmatrix}
                 1&0\\0&-1
             \end{pmatrix}}\right]$}& n\text{ even},\,m\text{ even} \\
        \end{array}\right.        
    \end{align}
\end{subequations}
We thus conclude that all the elements from Table \ref{tab.int} have at most commutator length 2 in $\mathsf{GL}(2,\mathbb{Z})$. While it could be the case that the two last two-parameter families can be expressed in terms of a single commutator (and possibly a $T$ D7-brane), and indeed such is the case for small values of $n$ and $m$, we have not been able to find a general expression with commutator length 1.

\vspace*{.5cm}
\bibliographystyle{JHEP2015}
\bibliography{ref}

@article{Makridou:2026jzy,
    author = "Makridou, Andriana and G{\'o}mez, Alejandro Javier Puga",
    title = "{Sharpened Dynamical Cobordism}",
    eprint = "2605.06793",
    archivePrefix = "arXiv",
    primaryClass = "hep-th",
    reportNumber = "IFT-UAM/CSIC-26-45",
    month = "5",
    year = "2026"
}

@article{Gvozdevsky2024VerbalWI,
    title={Verbal width in arithmetic Chevalley groups}, 
    author={Pavel Gvozdevsky},      
    journal={Journal of Algebra},  
    year={2024},  
    url={https://api.semanticscholar.org/CorpusID:266818514}
}

@inbook{Kodaira+2015+1142+1156,
url = {https://doi.org/10.1515/9781400869879-002},
title = {On Compact Analytic Surfaces},
booktitle = {Kunihiko Kodaira, Volume III},
author = {Kunihiko Kodaira},
publisher = {Princeton University Press},
address = {Princeton},
pages = {1142--1156},
doi = {doi:10.1515/9781400869879-002},
isbn = {9781400869879},
year = {2015},
lastchecked = {2026-01-08}
}

@article{Basile:2018irz,
    author = "Basile, I. and Mourad, J. and Sagnotti, A.",
    title = "{On Classical Stability with Broken Supersymmetry}",
    eprint = "1811.11448",
    archivePrefix = "arXiv",
    primaryClass = "hep-th",
    doi = "10.1007/JHEP01(2019)174",
    journal = "JHEP",
    volume = "01",
    pages = "174",
    year = "2019"
}

@article{Dudas:2004nd,
    author = "Dudas, E. and Pradisi, G. and Nicolosi, M. and Sagnotti, A.",
    title = "{On tadpoles and vacuum redefinitions in string theory}",
    eprint = "hep-th/0410101",
    archivePrefix = "arXiv",
    reportNumber = "CPTH-RR-057-0904, LPT-ORSAY-04-82, ROM2F-04-28",
    doi = "10.1016/j.nuclphysb.2004.11.028",
    journal = "Nucl. Phys. B",
    volume = "708",
    pages = "3--44",
    year = "2005"
}

@article{Dudas:2002dg,
    author = "Dudas, E. and Mourad, J. and Timirgaziu, Cristina",
    title = "{Time and space dependent backgrounds from nonsupersymmetric strings}",
    eprint = "hep-th/0209176",
    archivePrefix = "arXiv",
    reportNumber = "CPHT-RR-068-0902, LPT-ORSAY-02-136",
    doi = "10.1016/S0550-3213(03)00248-7",
    journal = "Nucl. Phys. B",
    volume = "660",
    pages = "3--24",
    year = "2003"
}

@article{Dudas:2000ff,
    author = "Dudas, E. and Mourad, J.",
    title = "{Brane solutions in strings with broken supersymmetry and dilaton tadpoles}",
    eprint = "hep-th/0004165",
    archivePrefix = "arXiv",
    reportNumber = "LPT-ORSAY-00-43, LPTM-00-25, LPT-00-43",
    doi = "10.1016/S0370-2693(00)00734-6",
    journal = "Phys. Lett. B",
    volume = "486",
    pages = "172--178",
    year = "2000"
}

@article{Mourad:2021roa,
    author = "Mourad, J. and Sagnotti, A.",
    title = "{On warped string vacuum profiles and cosmologies. Part II. Non-supersymmetric strings}",
    eprint = "2109.12328",
    archivePrefix = "arXiv",
    primaryClass = "hep-th",
    doi = "10.1007/JHEP12(2021)138",
    journal = "JHEP",
    volume = "12",
    pages = "138",
    year = "2021"
}

@article{Huertas:2023syg,
    author = "Huertas, Jes\'us and Uranga, Angel M.",
    title = "{Aspects of dynamical cobordism in AdS/CFT}",
    eprint = "2306.07335",
    archivePrefix = "arXiv",
    primaryClass = "hep-th",
    doi = "10.1007/JHEP08(2023)140",
    journal = "JHEP",
    volume = "08",
    pages = "140",
    year = "2023"
}

@article{Blumenhagen:2023abk,
    author = "Blumenhagen, Ralph and Kneissl, Christian and Wang, Chuying",
    title = "{Dynamical Cobordism Conjecture: solutions for end-of-the-world branes}",
    eprint = "2303.03423",
    archivePrefix = "arXiv",
    primaryClass = "hep-th",
    reportNumber = "MPP-2023-33",
    doi = "10.1007/JHEP05(2023)123",
    journal = "JHEP",
    volume = "05",
    pages = "123",
    year = "2023"
}

@article{Angius:2024pqk,
    author = "Angius, Roberta and Uranga, Angel Maria and Wang, Chuying",
    title = "{End of the world boundaries for chiral quantum gravity theories}",
    eprint = "2410.07322",
    archivePrefix = "arXiv",
    primaryClass = "hep-th",
    doi = "10.1007/JHEP03(2025)064",
    journal = "JHEP",
    volume = "03",
    pages = "064",
    year = "2025"
}

@article{Gaberdiel:1997ud,
    author = "Gaberdiel, Matthias R. and Zwiebach, Barton",
    title = "{Exceptional groups from open strings}",
    eprint = "hep-th/9709013",
    archivePrefix = "arXiv",
    reportNumber = "HUTP-97-A046, MIT-CTP-2670",
    doi = "10.1016/S0550-3213(97)00841-9",
    journal = "Nucl. Phys. B",
    volume = "518",
    pages = "151--172",
    year = "1998"
}

@article{Sen:1994fa,
    author = "Sen, Ashoke",
    title = "{Strong - weak coupling duality in four-dimensional string theory}",
    eprint = "hep-th/9402002",
    archivePrefix = "arXiv",
    reportNumber = "TIFR-TH-94-03",
    doi = "10.1142/S0217751X94001497",
    journal = "Int. J. Mod. Phys. A",
    volume = "9",
    pages = "3707--3750",
    year = "1994"
}

@article{Sen:1994wr,
    author = "Sen, Ashoke",
    title = "{Strong - weak coupling duality in three-dimensional string theory}",
    eprint = "hep-th/9408083",
    archivePrefix = "arXiv",
    reportNumber = "TIFR-TH-94-19",
    doi = "10.1016/0550-3213(94)00461-M",
    journal = "Nucl. Phys. B",
    volume = "434",
    pages = "179--209",
    year = "1995"
}

@article{Fukuda:2024pvu,
    author = "Fukuda, Masaki and Kobayashi, Shun K. and Watanabe, Kento and Yonekura, Kazuya",
    title = "{Black p-branes in heterotic string theory}",
    eprint = "2412.02277",
    archivePrefix = "arXiv",
    primaryClass = "hep-th",
    reportNumber = "TU-1247",
    doi = "10.1007/JHEP05(2025)043",
    journal = "JHEP",
    volume = "05",
    pages = "043",
    year = "2025"
}

@article{Kaidi:2024cbx,
    author = "Kaidi, Justin and Tachikawa, Yuji and Yonekura, Kazuya",
    title = "{On non-supersymmetric heterotic branes}",
    eprint = "2411.04344",
    archivePrefix = "arXiv",
    primaryClass = "hep-th",
    reportNumber = "TU-1248, KYUSHU-HET-294",
    doi = "10.1007/JHEP03(2025)211",
    journal = "JHEP",
    volume = "03",
    pages = "211",
    year = "2025"
}

@article{Kaidi:2023tqo,
    author = "Kaidi, Justin and Ohmori, Kantaro and Tachikawa, Yuji and Yonekura, Kazuya",
    title = "{Nonsupersymmetric Heterotic Branes}",
    eprint = "2303.17623",
    archivePrefix = "arXiv",
    primaryClass = "hep-th",
    doi = "10.1103/PhysRevLett.131.121601",
    journal = "Phys. Rev. Lett.",
    volume = "131",
    number = "12",
    pages = "121601",
    year = "2023"
}

@article{Cribiori:2021djm,
    author = "Cribiori, Niccol{\`o} and Junghans, Daniel and Van Hemelryck, Vincent and Van Riet, Thomas and Wrase, Timm",
    title = "{Scale-separated AdS4 vacua of IIA orientifolds and M-theory}",
    eprint = "2107.00019",
    archivePrefix = "arXiv",
    primaryClass = "hep-th",
    reportNumber = "UUITP-29/21",
    doi = "10.1103/PhysRevD.104.126014",
    journal = "Phys. Rev. D",
    volume = "104",
    number = "12",
    pages = "126014",
    year = "2021"
}

@article{VanHemelryck:2025qok,
    author = "Van Hemelryck, Vincent",
    title = "{Supersymmetric scale-separated AdS$_{3}$ orientifold vacua of type IIB}",
    eprint = "2502.04791",
    archivePrefix = "arXiv",
    primaryClass = "hep-th",
    reportNumber = "UUITP-04/25",
    doi = "10.1007/JHEP10(2025)109",
    journal = "JHEP",
    volume = "10",
    pages = "109",
    year = "2025"
}

@article{Caviezel:2008ik,
    author = "Caviezel, Claudio and Koerber, Paul and Kors, Simon and Lust, Dieter and Tsimpis, Dimitrios and Zagermann, Marco",
    title = "{The Effective theory of type IIA AdS(4) compactifications on nilmanifolds and cosets}",
    eprint = "0806.3458",
    archivePrefix = "arXiv",
    primaryClass = "hep-th",
    reportNumber = "MPP-2008-62, LMU-ASC-37-08",
    doi = "10.1088/0264-9381/26/2/025014",
    journal = "Class. Quant. Grav.",
    volume = "26",
    pages = "025014",
    year = "2009"
}

@article{Arboleya:2025jko,
    author = "Arboleya, {\'A}lvaro and Casagrande, Gabriele and Guarino, Adolfo and Morittu, Matteo",
    title = "{Taxonomy of type II orientifold flux vacua in 3D}",
    eprint = "2512.13433",
    archivePrefix = "arXiv",
    primaryClass = "hep-th",
    month = "12",
    year = "2025"
}

@article{Bomans:2021ara,
    author = "Bomans, Pieter and Cassani, Davide and Dibitetto, Giuseppe and Petri, Nicolo",
    title = "{Bubble instability of mIIA on $\mathrm{AdS}_4\times S^6$}",
    eprint = "2110.08276",
    archivePrefix = "arXiv",
    primaryClass = "hep-th",
    doi = "10.21468/SciPostPhys.12.3.099",
    journal = "SciPost Phys.",
    volume = "12",
    number = "3",
    pages = "099",
    year = "2022"
}

@article{Heckman:2024zdo,
    author = {Heckman, Jonathan J. and H{\"u}bner, Max},
    title = "{Celestial Topology, Symmetry Theories, and Evidence for a NonSUSY D3-Brane CFT}",
    eprint = "2406.08485",
    archivePrefix = "arXiv",
    primaryClass = "hep-th",
    doi = "10.1002/prop.202400270",
    journal = "Fortsch. Phys.",
    volume = "73",
    number = "4",
    pages = "2400270",
    year = "2025"
}

@article{Bandos:2023yyo,
    author = "Bandos, Igor and Blanco-Pillado, Jose J. and Sousa, Kepa and Urkiola, Mikel A.",
    title = "{Brane nucleation in supersymmetric models}",
    eprint = "2306.09412",
    archivePrefix = "arXiv",
    primaryClass = "hep-th",
    doi = "10.1007/JHEP10(2023)061",
    journal = "JHEP",
    volume = "10",
    pages = "061",
    year = "2023"
}

@article{Petri:2022yhy,
    author = "Petri, Nicol{\`o}",
    title = "{Bubbles of nothing and AdS instabilities}",
    eprint = "2205.00884",
    archivePrefix = "arXiv",
    primaryClass = "hep-th",
    doi = "10.22323/1.406.0170",
    journal = "PoS",
    volume = "CORFU2021",
    pages = "170",
    year = "2022"
}

@article{Draper:2021ujg,
    author = "Draper, Patrick and Garcia, Isabel Garcia and Lillard, Benjamin",
    title = "{Bubble of nothing decays of unstable theories}",
    eprint = "2105.08068",
    archivePrefix = "arXiv",
    primaryClass = "hep-th",
    doi = "10.1103/PhysRevD.104.L121701",
    journal = "Phys. Rev. D",
    volume = "104",
    number = "12",
    pages = "L121701",
    year = "2021"
}

@article{Dibitetto:2020csn,
    author = "Dibitetto, Giuseppe and Petri, Nicol{\`o} and Schillo, Marjorie",
    title = "{Nothing really matters}",
    eprint = "2002.01764",
    archivePrefix = "arXiv",
    primaryClass = "hep-th",
    doi = "10.1007/JHEP08(2020)040",
    journal = "JHEP",
    volume = "08",
    pages = "040",
    year = "2020"
}

@article{GarciaEtxebarria:2020xsr,
    author = "Garc{\'\i}a Etxebarria, I{\~n}aki and Montero, Miguel and Sousa, Kepa and Valenzuela, Irene",
    title = "{Nothing is certain in string compactifications}",
    eprint = "2005.06494",
    archivePrefix = "arXiv",
    primaryClass = "hep-th",
    doi = "10.1007/JHEP12(2020)032",
    journal = "JHEP",
    volume = "12",
    pages = "032",
    year = "2020"
}

@article{Blanco-Pillado:2011fcm,
    author = "Blanco-Pillado, Jose J. and Ramadhan, Handhika S. and Shlaer, Benjamin",
    title = "{Bubbles from Nothing}",
    eprint = "1104.5229",
    archivePrefix = "arXiv",
    primaryClass = "gr-qc",
    doi = "10.1088/1475-7516/2012/01/045",
    journal = "JCAP",
    volume = "01",
    pages = "045",
    year = "2012"
}

@article{Ooguri:2017njy,
    author = "Ooguri, Hirosi and Spodyneiko, Lev",
    title = "{New Kaluza-Klein instantons and the decay of AdS vacua}",
    eprint = "1703.03105",
    archivePrefix = "arXiv",
    primaryClass = "hep-th",
    reportNumber = "CALT-TH-2017-014, IPMU-17-0041",
    doi = "10.1103/PhysRevD.96.026016",
    journal = "Phys. Rev. D",
    volume = "96",
    number = "2",
    pages = "026016",
    year = "2017"
}

@article{Brown:2014rka,
    author = "Brown, Adam R.",
    title = "{Decay of hot Kaluza-Klein space}",
    eprint = "1408.5903",
    archivePrefix = "arXiv",
    primaryClass = "hep-th",
    doi = "10.1103/PhysRevD.90.104017",
    journal = "Phys. Rev. D",
    volume = "90",
    number = "10",
    pages = "104017",
    year = "2014"
}

@article{deAlwis:2013gka,
    author = "de Alwis, S. and Gupta, Rajesh and Hatefi, E. and Quevedo, F.",
    title = "{Stability, Tunneling and Flux Changing de Sitter Transitions in the Large Volume String Scenario}",
    eprint = "1308.1222",
    archivePrefix = "arXiv",
    primaryClass = "hep-th",
    reportNumber = "DAMTP-2013-42",
    doi = "10.1007/JHEP11(2013)179",
    journal = "JHEP",
    volume = "11",
    pages = "179",
    year = "2013"
}

@article{Brown:2011gt,
    author = "Brown, Adam R. and Dahlen, Alex",
    title = "{On 'nothing' as an infinitely negatively curved spacetime}",
    eprint = "1111.0301",
    archivePrefix = "arXiv",
    primaryClass = "hep-th",
    doi = "10.1103/PhysRevD.85.104026",
    journal = "Phys. Rev. D",
    volume = "85",
    pages = "104026",
    year = "2012"
}

@article{Dine:2004uw,
    author = "Dine, Michael and Fox, Patrick J. and Gorbatov, Elie",
    title = "{Catastrophic decays of compactified space-times}",
    eprint = "hep-th/0405190",
    archivePrefix = "arXiv",
    reportNumber = "SCIPP-2004-16",
    doi = "10.1088/1126-6708/2004/09/037",
    journal = "JHEP",
    volume = "09",
    pages = "037",
    year = "2004"
}

@article{Brown:2010mf,
    author = "Brown, Adam R. and Dahlen, Alex",
    title = "{Bubbles of Nothing and the Fastest Decay in the Landscape}",
    eprint = "1010.5240",
    archivePrefix = "arXiv",
    primaryClass = "hep-th",
    reportNumber = "PUPT-2357",
    doi = "10.1103/PhysRevD.84.043518",
    journal = "Phys. Rev. D",
    volume = "84",
    pages = "043518",
    year = "2011"
}

@article{Blanco-Pillado:2010vdp,
    author = "Blanco-Pillado, Jose J. and Ramadhan, Handhika S. and Shlaer, Benjamin",
    title = "{Decay of flux vacua to nothing}",
    eprint = "1009.0753",
    archivePrefix = "arXiv",
    primaryClass = "hep-th",
    doi = "10.1088/1475-7516/2010/10/029",
    journal = "JCAP",
    volume = "10",
    pages = "029",
    year = "2010"
}

@article{Blanco-Pillado:2010xww,
    author = "Blanco-Pillado, Jose J. and Shlaer, Benjamin",
    title = "{Bubbles of Nothing in Flux Compactifications}",
    eprint = "1002.4408",
    archivePrefix = "arXiv",
    primaryClass = "hep-th",
    doi = "10.1103/PhysRevD.82.086015",
    journal = "Phys. Rev. D",
    volume = "82",
    pages = "086015",
    year = "2010"
}

@article{Yang:2009wz,
    author = "Yang, I-Sheng",
    title = "{Stretched extra dimensions and bubbles of nothing in a toy model landscape}",
    eprint = "0910.1397",
    archivePrefix = "arXiv",
    primaryClass = "hep-th",
    doi = "10.1103/PhysRevD.81.125020",
    journal = "Phys. Rev. D",
    volume = "81",
    pages = "125020",
    year = "2010"
}

@article{Horowitz:2007pr,
    author = "Horowitz, Gary T. and Orgera, Jacopo and Polchinski, Joe",
    title = "{Nonperturbative Instability of AdS(5) x S**5/Z(k)}",
    eprint = "0709.4262",
    archivePrefix = "arXiv",
    primaryClass = "hep-th",
    doi = "10.1103/PhysRevD.77.024004",
    journal = "Phys. Rev. D",
    volume = "77",
    pages = "024004",
    year = "2008"
}

@article{Montero:2020icj,
    author = "Montero, Miguel and Vafa, Cumrun",
    title = "{Cobordism Conjecture, Anomalies, and the String Lamppost Principle}",
    eprint = "2008.11729",
    archivePrefix = "arXiv",
    primaryClass = "hep-th",
    doi = "10.1007/JHEP01(2021)063",
    journal = "JHEP",
    volume = "01",
    pages = "063",
    year = "2021"
}

@article{Demulder:2023vlo,
    author = "Demulder, Saskia and Lust, Dieter and Raml, Thomas",
    title = "{Topology change and non-geometry at infinite distance}",
    eprint = "2312.07674",
    archivePrefix = "arXiv",
    primaryClass = "hep-th",
    doi = "10.1007/JHEP06(2024)079",
    journal = "JHEP",
    volume = "06",
    pages = "079",
    year = "2024"
}

@article{Acharya:2004qe,
    author = "Acharya, Bobby Samir and Gukov, Sergei",
    title = "{M theory and singularities of exceptional holonomy manifolds}",
    eprint = "hep-th/0409191",
    archivePrefix = "arXiv",
    reportNumber = "HUTP-03-A053, RUNHETC-2003-26",
    doi = "10.1016/j.physrep.2003.10.017",
    journal = "Phys. Rept.",
    volume = "392",
    pages = "121--189",
    year = "2004"
}

@article{Greene:2000yb,
    author = "Greene, Brian R. and Schalm, Koenraad and Shiu, Gary",
    title = "{Dynamical topology change in M theory}",
    eprint = "hep-th/0010207",
    archivePrefix = "arXiv",
    reportNumber = "CU-TP-987, NIKHEF-00-024, YITP-SB-00-43, UPR-900-T, YITP-00-43",
    doi = "10.1063/1.1377038",
    journal = "J. Math. Phys.",
    volume = "42",
    pages = "3171--3187",
    year = "2001"
}

@article{Witten:1998zw,
	title        = {{Anti-de Sitter space, thermal phase transition, and confinement in gauge theories}},
	author       = {Witten, Edward},
	year         = 1998,
	journal      = {Adv. Theor. Math. Phys.},
	volume       = 2,
	pages        = {505--532},
	doi          = {10.4310/ATMP.1998.v2.n3.a3},
	editor       = {Bergstrom, L. and Lindstrom, U.},
	eprint       = {hep-th/9803131},
	archiveprefix = {arXiv},
	reportnumber = {IASSNS-HEP-98-21}
}

@article{Acharya:2019mcu,
	title        = {{Supersymmetry, Ricci Flat Manifolds and the String Landscape}},
	author       = {Acharya, Bobby Samir},
	year         = 2020,
	journal      = {JHEP},
	volume       = {08},
	pages        = 128,
	doi          = {10.1007/JHEP08(2020)128},
	eprint       = {1906.06886},
	archiveprefix = {arXiv},
	primaryclass = {hep-th},
	reportnumber = {KCL-2019-056}
}

@article{Bergman:1999km,
    author = "Bergman, Oren and Gaberdiel, Matthias R.",
    title = "{Dualities of type 0 strings}",
    eprint = "hep-th/9906055",
    archivePrefix = "arXiv",
    reportNumber = "CALT-68-2228, DAMTP-1999-74",
    doi = "10.1088/1126-6708/1999/07/022",
    journal = "JHEP",
    volume = "07",
    pages = "022",
    year = "1999"
}

@article{Harlow:2018fse,
    author = "Harlow, Daniel",
    title = "{TASI Lectures on the Emergence of Bulk Physics in AdS/CFT}",
    eprint = "1802.01040",
    archivePrefix = "arXiv",
    primaryClass = "hep-th",
    doi = "10.22323/1.305.0002",
    journal = "PoS",
    volume = "TASI2017",
    pages = "002",
    year = "2018"
}

@article{Hawking:1975vcx,
    author = "Hawking, S. W.",
    editor = "Gibbons, G. W. and Hawking, S. W.",
    title = "{Particle Creation by Black Holes}",
    doi = "10.1007/BF02345020",
    journal = "Commun. Math. Phys.",
    volume = "43",
    pages = "199--220",
    year = "1975",
    note = "[Erratum: Commun.Math.Phys. 46, 206 (1976)]"
}

@article{Bossard:2024mls,
    author = "Bossard, Guillaume and Casagrande, Gabriele and Dudas, Emilian",
    title = "{Twisted orientifold planes and S-duality without supersymmetry}",
    eprint = "2411.00955",
    archivePrefix = "arXiv",
    primaryClass = "hep-th",
    reportNumber = "CPHT-RR012.072024",
    doi = "10.1007/JHEP02(2025)062",
    journal = "JHEP",
    volume = "02",
    pages = "062",
    year = "2025"
}

@article{Fraiman:2025yrx,
    author = "Fraiman, Bernardo and Parra de Freitas, H{\'e}ctor",
    title = "{Symmetries and dualities in non-supersymmetric CHL strings}",
    eprint = "2511.01674",
    archivePrefix = "arXiv",
    primaryClass = "hep-th",
    reportNumber = "MPP-2025-79; IFT-UAM/CSIC-25-111",
    month = "11",
    year = "2025"
}

@article{Dudas:2001wd,
    author = "Dudas, E. and Mourad, J. and Sagnotti, A.",
    title = "{Charged and uncharged D-branes in various string theories}",
    eprint = "hep-th/0107081",
    archivePrefix = "arXiv",
    reportNumber = "LPT-ORSAY-01-56, ROM2F-01-18",
    doi = "10.1016/S0550-3213(01)00552-1",
    journal = "Nucl. Phys. B",
    volume = "620",
    pages = "109--151",
    year = "2002"
}

@article{Fabinger:2000jd,
    author = "Fabinger, Michal and Horava, Petr",
    title = "{Casimir effect between world branes in heterotic M theory}",
    eprint = "hep-th/0002073",
    archivePrefix = "arXiv",
    reportNumber = "CALT-68-2255, CITUSC-00-004, PRA-HEP-00-02",
    doi = "10.1016/S0550-3213(00)00255-8",
    journal = "Nucl. Phys. B",
    volume = "580",
    pages = "243--263",
    year = "2000"
}

@article{Blumenhagen:1999ad,
    author = "Blumenhagen, Ralph and Kumar, Alok",
    title = "{A Note on orientifolds and dualities of type 0B string theory}",
    eprint = "hep-th/9906234",
    archivePrefix = "arXiv",
    reportNumber = "HUB-EP-99-29, CERN-TH-99-190",
    doi = "10.1016/S0370-2693(99)01002-3",
    journal = "Phys. Lett. B",
    volume = "464",
    pages = "46--52",
    year = "1999"
}

@article{Kachru:1998yy,
    author = "Kachru, Shamit and Silverstein, Eva",
    title = "{Selfdual nonsupersymmetric type II string compactifications}",
    eprint = "hep-th/9808056",
    archivePrefix = "arXiv",
    reportNumber = "SLAC-PUB-7907, LBL-42139, LBNL-42139, UCB-PTH-98-40",
    doi = "10.1088/1126-6708/1998/11/001",
    journal = "JHEP",
    volume = "11",
    pages = "001",
    year = "1998"
}

@article{Blum:1997gw,
    author = "Blum, Julie D. and Dienes, Keith R.",
    title = "{Strong / weak coupling duality relations for nonsupersymmetric string theories}",
    eprint = "hep-th/9707160",
    archivePrefix = "arXiv",
    reportNumber = "IASSNS-HEP-97-80",
    doi = "10.1016/S0550-3213(97)00803-1",
    journal = "Nucl. Phys. B",
    volume = "516",
    pages = "83--159",
    year = "1998"
}

@article{Baykara:2025nnc,
    author = "Baykara, Zihni Kaan and Tomasiello, Alessandro and Vafa, Cumrun",
    title = "{Counting AdS Vacua}",
    eprint = "2512.04151",
    archivePrefix = "arXiv",
    primaryClass = "hep-th",
    month = "12",
    year = "2025"
}

@article{Douglas:2003um,
    author = "Douglas, Michael R.",
    title = "{The Statistics of string / M theory vacua}",
    eprint = "hep-th/0303194",
    archivePrefix = "arXiv",
    reportNumber = "RUNHETC-2003-09",
    doi = "10.1088/1126-6708/2003/05/046",
    journal = "JHEP",
    volume = "05",
    pages = "046",
    year = "2003"
}

@article{Font:1990gx,
    author = "Font, A. and Ibanez, Luis E. and Lust, D. and Quevedo, F.",
    title = "{Strong - weak coupling duality and nonperturbative effects in string theory}",
    reportNumber = "CERN-TH-5790-90",
    doi = "10.1016/0370-2693(90)90523-9",
    journal = "Phys. Lett. B",
    volume = "249",
    pages = "35--43",
    year = "1990"
}

@article{Taylor:2015xtz,
    author = "Taylor, Washington and Wang, Yi-Nan",
    title = "{The F-theory geometry with most flux vacua}",
    eprint = "1511.03209",
    archivePrefix = "arXiv",
    primaryClass = "hep-th",
    reportNumber = "MIT-CTP-4732",
    doi = "10.1007/JHEP12(2015)164",
    journal = "JHEP",
    volume = "12",
    pages = "164",
    year = "2015"
}

@article{Susskind:2003kw,
    author = "Susskind, Leonard",
    editor = "Carr, Bernard J.",
    title = "{The Anthropic landscape of string theory}",
    eprint = "hep-th/0302219",
    archivePrefix = "arXiv",
    pages = "247--266",
    month = "2",
    year = "2003"
}

@article{Bergman:1997rf,
    author = "Bergman, Oren and Gaberdiel, Matthias R.",
    title = "{A Nonsupersymmetric open string theory and S duality}",
    eprint = "hep-th/9701137",
    archivePrefix = "arXiv",
    reportNumber = "HUTP-97-A003, BRX-TH-402",
    doi = "10.1016/S0550-3213(97)00309-X",
    journal = "Nucl. Phys. B",
    volume = "499",
    pages = "183--204",
    year = "1997"
}

@article{Baykara:2026vdc,
    author = "Baykara, Zihni Kaan and Delgado, Matilda and Dudas, Emilian and De Freitas, Hector Parra and Vafa, Cumrun",
    title = "{A Duality Web for Non-Supersymmetric Strings}",
    eprint = "2604.07433",
    archivePrefix = "arXiv",
    primaryClass = "hep-th",
    month = "4",
    year = "2026"
}

@article{Baykara:2026gem,
    author = "Baykara, Zihni Kaan and Dudas, Emilian and Vafa, Cumrun",
    title = "{M-theory on $S^1\vee S^1$ as Type 0A}",
    eprint = "2603.13468",
    archivePrefix = "arXiv",
    primaryClass = "hep-th",
    month = "3",
    year = "2026"
}

@article{Altavista:2026evd,
    author = "Altavista, Chiara and Anastasi, Edoardo and Raucci, Salvatore and Uranga, Angel M. and Wang, Chuying",
    title = "{Ho{\v{r}}ava-Witten theory on ${\mathbf{S}}^1\vee{\mathbf{S}}^1$ as type 0 orientifold}",
    eprint = "2603.25786",
    archivePrefix = "arXiv",
    primaryClass = "hep-th",
    reportNumber = "IFT-UAM/CSIC-26-39",
    month = "3",
    year = "2026"
}

@article{Townsend:1995kk,
	title        = {{The eleven-dimensional supermembrane revisited}},
	author       = {Townsend, P. K.},
	year         = 1995,
	journal      = {Phys. Lett. B},
	volume       = 350,
	pages        = {184--187},
	doi          = {10.1016/0370-2693(95)00397-4},
	eprint       = {hep-th/9501068},
	archiveprefix = {arXiv},
	reportnumber = {DAMTP-R-95-2}
}

@inproceedings{Sen:1998kd,
    author = "Sen, Ashoke",
    title = "{Developments in superstring theory}",
    booktitle = "{29th International Conference on High-Energy Physics}",
    eprint = "hep-ph/9810356",
    archivePrefix = "arXiv",
    reportNumber = "MRI-PHY-981066",
    pages = "371--381",
    month = "7",
    year = "1998"
}

@article{Bavard1991,
  author  = {Bavard, Christophe},
  title   = {Longueur stable des commutateurs},
  journal = {L'Enseignement Math{\'e}matique. Revue Internationale. IIe S{\'e}rie},
  volume  = {37},
  year    = {1991},
  number  = {1-2},
  pages   = {109--150},
  issn    = {0013-8584},
  mrnumber = {1114911},
  zbl     = {0739.20019},
  url     = {http://doi.org/10.5169/seals-58742}
}

@article{Hellerman:2010dv,
    author = "Hellerman, Simeon and Kleban, Matthew",
    title = "{Dynamical Cobordisms in General Relativity and String Theory}",
    eprint = "1009.3277",
    archivePrefix = "arXiv",
    primaryClass = "hep-th",
    reportNumber = "IPMU-10-0156",
    doi = "10.1007/JHEP02(2011)022",
    journal = "JHEP",
    volume = "02",
    pages = "022",
    year = "2011"
}

@article{Etheredge:2024tok,
    author = "Etheredge, Muldrow and Heidenreich, Ben and Rudelius, Tom and Ruiz, Ignacio and Valenzuela, Irene",
    title = "{Taxonomy of infinite distance limits}",
    eprint = "2405.20332",
    archivePrefix = "arXiv",
    primaryClass = "hep-th",
    reportNumber = "ACFI-T24-04, CERN-TH-2024-067, IFT-UAM/CSIC-23-64",
    doi = "10.1007/JHEP03(2025)213",
    journal = "JHEP",
    volume = "03",
    pages = "213",
    year = "2025"
}

@inproceedings{Doran:2005gu,
    author = "Doran, Charles F. and Morgan, John W.",
    title = "{Mirror symmetry and integral variations of Hodge structure underlying one parameter families of Calabi-Yau threefolds}",
    booktitle = "{Workshop on Calabi-Yau Varieties and Mirror Symmetry}",
    eprint = "math/0505272",
    archivePrefix = "arXiv",
    pages = "517--537",
    month = "5",
    year = "2005"
}

@article{McNamara:2022lrw,
    author = "McNamara, Jacob and Reece, Matthew",
    title = "{Reflections on Parity Breaking}",
    eprint = "2212.00039",
    archivePrefix = "arXiv",
    primaryClass = "hep-th",
    month = "11",
    year = "2022"
}

@book{Blumenhagen:2013fgp,
    author = {Blumenhagen, Ralph and L{\"u}st, Dieter and Theisen, Stefan},
    title = "{Basic concepts of string theory}",
    doi = "10.1007/978-3-642-29497-6",
    isbn = "978-3-642-29496-9",
    publisher = "Springer",
    address = "Heidelberg, Germany",
    series = "Theoretical and Mathematical Physics",
    year = "2013"
}

@book{HatcherTOPOLOGY,
  address = {Cambridge},
  author = {Hatcher, Allen},
  isbn = {0-521-79160-X; 0-521-79540-0},
  mrclass = {55-01 (55-00)},
  mrnumber = {1867354 (2002k:55001)},
  mrreviewer = {Donald W. Kahn},
  pages = {xii+544},
  publisher = {Cambridge University Press},
  title = {Algebraic topology},
  username = {mwpb479},
  year = 2002
}

@article{Brav_2014,
   title={Thin monodromy in $\mathsf{Sp}(4)$},
   volume={150},
   ISSN={1570-5846},
   url={http://dx.doi.org/10.1112/S0010437X13007550},
   DOI={10.1112/s0010437x13007550},
   number={3},
   journal={Compositio Mathematica},
   publisher={Wiley},
   author={Brav, Christopher and Thomas, Hugh},
   year={2014},
   month=mar, pages={333–343} }

@article{HOFMANN_2015,
   title={Some monodromy groups of finite index in $\mathsf{Sp}(4)$},
   volume={99},
   ISSN={1446-8107},
   url={http://dx.doi.org/10.1017/S1446788715000014},
   DOI={10.1017/s1446788715000014},
   number={1},
   journal={Journal of the Australian Mathematical Society},
   publisher={Cambridge University Press (CUP)},
   author={Hofmann, J\"org and van Straten, Duco},
   year={2015},
   month=mar, pages={48–62} }

@misc{almkvist2010tablescalabiyauequations,
      title={Tables of Calabi--Yau equations}, 
      author={Gert Almkvist and Christian van Enckevort and Duco van Straten and Wadim Zudilin},
      year={2010},
      eprint={math/0507430},
      archivePrefix={arXiv},
      primaryClass={math.AG},
      url={https://arxiv.org/abs/math/0507430}, 
}

@article{Grieco:2025bjy,
    author = "Grieco, Alessandra and Ruiz, Ignacio and Valenzuela, Irene",
    title = "{EFT strings and dualities in 4d $\mathcal{N}=1$}",
    eprint = "2504.16984",
    archivePrefix = "arXiv",
    primaryClass = "hep-th",
    reportNumber = "IFT-UAM/CSIC-25-39, CERN-TH-2025-086",
    month = "4",
    year = "2025"
}

@article{Hazrat2012CommutatorWI,
    title={Commutator width in Chevalley groups},
    author={Roozbeh Hazrat and A. Stepanov and Nikolai Vavilov and Zuhong Zhang},
    journal={arXiv: Rings and Algebras},
    year={2012},
    volume={33},
    pages={139-170},
    url={https://api.semanticscholar.org/CorpusID:117154801} 
}

@misc{ConradSL2Z,
  author        = {Keith Conrad},
  title         = {$\rm SL_2(\mathbb{Z})$},
  month         = {December},
  year          = {2018},
  publisher={University of Connecticut, Math Department},
  url={\url{https://kconrad.math.uconn.edu/blurbs/grouptheory/SL(2,Z).pdf}}
}

@article{McNamara:2021cuo,
    author = "McNamara, Jacob",
    title = "{Gravitational Solitons and Completeness}",
    eprint = "2108.02228",
    archivePrefix = "arXiv",
    primaryClass = "hep-th",
    month = "8",
    year = "2021"
}

@article{Hamada:2025duq,
    author = "Hamada, Yu and Hamada, Yuta and Kimura, Hayate",
    title = "{Black String in the Standard Model}",
    eprint = "2501.05678",
    archivePrefix = "arXiv",
    primaryClass = "hep-th",
    reportNumber = "KEK-TH-2679, DESY-24-222",
    month = "1",
    year = "2025"
}

@article{Debray:2023yrs,
    author = "Debray, Arun and Dierigl, Markus and Heckman, Jonathan J. and Montero, Miguel",
    title = "{The Chronicles of IIBordia: Dualities, Bordisms, and the Swampland}",
    eprint = "2302.00007",
    archivePrefix = "arXiv",
    primaryClass = "hep-th",
    reportNumber = "LMU-ASC 06/23, IFT-UAM/CSIC-23-7",
    month = "1",
    year = "2023"
}

@article{Braeger:2025kra,
    author = "Braeger, Noah and Debray, Arun and Dierigl, Markus and Heckman, Jonathan J. and Montero, Miguel",
    title = "{Cobordism Utopia: U-Dualities, Bordisms, and the Swampland}",
    eprint = "2505.15885",
    archivePrefix = "arXiv",
    primaryClass = "hep-th",
    reportNumber = "CERN-TH-2025-103, IFT-025-51",
    month = "5",
    year = "2025"
}

@article{Schmid:1973cdo,
    author = "Schmid, Wilfried",
    title = "{Variation of hodge structure: The singularities of the period mapping}",
    doi = "10.1007/BF01389674",
    journal = "Invent. Math.",
    volume = "22",
    number = "3",
    pages = "211--319",
    year = "1973"
}

@article{Heckman:2025wqd,
    author = "Heckman, Jonathan J. and McNamara, Jacob and Parra-Martinez, Julio and Torres, Ethan",
    title = "{GSO Defects: IIA/IIB Walls and the Surprisingly Stable $\mathrm{R}7$-Brane}",
    eprint = "2507.21210",
    archivePrefix = "arXiv",
    primaryClass = "hep-th",
    reportNumber = "CERN-TH-2025-136",
    month = "7",
    year = "2025"
}

@article{Julia:1980gr,
    author = "Julia, B.",
    title = "{GROUP DISINTEGRATIONS}",
    reportNumber = "LPTENS-80-16",
    journal = "Conf. Proc. C",
    volume = "8006162",
    pages = "331--350",
    year = "1980"
}

@article{Avni_Meiri_2019,
    title={Words have bounded width in $\operatorname{SL}(n,\mathbb{Z})$},
    volume={155},
    DOI={10.1112/S0010437X19007334},
    number={7},
    journal={Compositio Mathematica},
    author={Avni, Nir and Meiri, Chen},
    year={2019},
    pages={1245–1258}}

@article{DENNIS1988150,
title = {On a question of M. Newman on the number of commutators},
journal = {Journal of Algebra},
volume = {118},
number = {1},
pages = {150-161},
year = {1988},
issn = {0021-8693},
doi = {https://doi.org/10.1016/0021-8693(88)90055-5},
url = {https://www.sciencedirect.com/science/article/pii/0021869388900555},
author = {R.K Dennis and L.N Vaserstein}
}

@article{Newman1987Unimodular,
  author    = {Morris Newman},
  title     = {Unimodular commutators},
  journal   = {Proceedings of the American Mathematical Society},
  volume    = {101},
  number    = {4},
  pages     = {605--609},
  year      = {1987},
  doi       = {10.1090/S0002-9939-1987-0911017-6},
  mrnumber  = {911017},
  issn      = {0002-9939},
}

@article{Cremmer1981,
  author  = {Cremmer, Eugène},
  title   = {Supergravities in 5 dimensions},
  journal = {In *Superspace and Supergravity*, edited by S.W. Hawking and M. Rocek},
  year    = {1981},
  publisher = {Cambridge University Press},
  address   = {Cambridge},
  pages   = {267--282},
  note    = {Based on lectures given at the Nuffield Workshop, Cambridge, June 16 - July 12, 1980}
}

@article{Dierigl:2023jdp,
    author = "Dierigl, Markus and Heckman, Jonathan J. and Montero, Miguel and Torres, Ethan",
    title = "{R7-branes as charge conjugation operators}",
    eprint = "2305.05689",
    archivePrefix = "arXiv",
    primaryClass = "hep-th",
    reportNumber = "LMU-ASC 16/23, IFT-UAM/CSIC-23-50",
    doi = "10.1103/PhysRevD.109.046004",
    journal = "Phys. Rev. D",
    volume = "109",
    number = "4",
    pages = "046004",
    year = "2024"
}

@article{Dierigl:2022reg,
    author = "Dierigl, Markus and Heckman, Jonathan J. and Montero, Miguel and Torres, Ethan",
    title = "{IIB string theory explored: Reflection 7-branes}",
    eprint = "2212.05077",
    archivePrefix = "arXiv",
    primaryClass = "hep-th",
    reportNumber = "LMU-ASC 31/22",
    doi = "10.1103/PhysRevD.107.086015",
    journal = "Phys. Rev. D",
    volume = "107",
    number = "8",
    pages = "086015",
    year = "2023"
}

@article{Andriot:2026lac,
    author = "Andriot, David",
    title = "{Dark energy from string theory: an introductory review}",
    eprint = "2603.25797",
    archivePrefix = "arXiv",
    primaryClass = "hep-th",
    month = "3",
    year = "2026"
}

@article{Dierigl:2026psj,
    author = "Dierigl, Markus and Ruiz, Ignacio",
    title = "{A missing link: Brane networks and the Cobordism Conjecture}",
    eprint = "2605.18952",
    archivePrefix = "arXiv",
    primaryClass = "hep-th",
    reportNumber = "CERN-TH-2026-108",
    month = "5",
    year = "2026"
}

@article{CULLER1981133,
title = {Using surfaces to solve equations in free groups},
journal = {Topology},
volume = {20},
number = {2},
pages = {133-145},
year = {1981},
issn = {0040-9383},
doi = {https://doi.org/10.1016/0040-9383(81)90033-1},
url = {https://www.sciencedirect.com/science/article/pii/0040938381900331},
author = {Marc Culler}
}

@article{Raman:2024fcv,
    author = "Raman, Sanjay and Vafa, Cumrun",
    title = "{Swampland and the Geometry of Marked Moduli Spaces}",
    eprint = "2405.11611",
    archivePrefix = "arXiv",
    primaryClass = "hep-th",
    month = "5",
    year = "2024"
}

@article{Rudelius:2024vmc,
    author = "Rudelius, Tom",
    title = "{A symmetry-centric perspective on the geometry of the string landscape and the swampland}",
    eprint = "2405.12980",
    archivePrefix = "arXiv",
    primaryClass = "hep-th",
    doi = "10.1142/S0218271824410037",
    journal = "Int. J. Mod. Phys. D",
    volume = "33",
    number = "15",
    pages = "2441003",
    year = "2024"
}

@article{Dierigl:2020lai,
    author = "Dierigl, Markus and Heckman, Jonathan J.",
    title = "{Swampland cobordism conjecture and non-Abelian duality groups}",
    eprint = "2012.00013",
    archivePrefix = "arXiv",
    primaryClass = "hep-th",
    doi = "10.1103/PhysRevD.103.066006",
    journal = "Phys. Rev. D",
    volume = "103",
    number = "6",
    pages = "066006",
    year = "2021"
}

@article{Kleban:2007kk,
    author = "Kleban, Matthew and Redi, Michele",
    title = "{Expanding F-Theory}",
    eprint = "0705.2020",
    archivePrefix = "arXiv",
    primaryClass = "hep-th",
    doi = "10.1088/1126-6708/2007/09/038",
    journal = "JHEP",
    volume = "09",
    pages = "038",
    year = "2007"
}

@article{Alexandrov:2008gh,
    author = "Alexandrov, Sergei and Pioline, Boris and Saueressig, Frank and Vandoren, Stefan",
    title = "{D-instantons and twistors}",
    eprint = "0812.4219",
    archivePrefix = "arXiv",
    primaryClass = "hep-th",
    reportNumber = "LPTA-08-073, ITP-UU-08-75, SPIN-08-58, IPHT-T08-203",
    doi = "10.1088/1126-6708/2009/03/044",
    journal = "JHEP",
    volume = "03",
    pages = "044",
    year = "2009"
}

@article{Alexandrov:2011va,
    author = "Alexandrov, Sergei",
    title = "{Twistor Approach to String Compactifications: a Review}",
    eprint = "1111.2892",
    archivePrefix = "arXiv",
    primaryClass = "hep-th",
    reportNumber = "L2C:11-255",
    doi = "10.1016/j.physrep.2012.09.005",
    journal = "Phys. Rept.",
    volume = "522",
    pages = "1--57",
    year = "2013"
}

@article{Tachikawa:2018njr,
    author = "Tachikawa, Yuji and Yonekura, Kazuya",
    title = "{Why are fractional charges of orientifolds compatible with Dirac quantization?}",
    eprint = "1805.02772",
    archivePrefix = "arXiv",
    primaryClass = "hep-th",
    reportNumber = "IPMU-18-0067",
    doi = "10.21468/SciPostPhys.7.5.058",
    journal = "SciPost Phys.",
    volume = "7",
    number = "5",
    pages = "058",
    year = "2019"
}

@article{Delgado:2024skw,
    author = "Delgado, Matilda and van de Heisteeg, Damian and Raman, Sanjay and Torres, Ethan and Vafa, Cumrun and Xu, Kai",
    title = "{Finiteness and the Emergence of Dualities}",
    eprint = "2412.03640",
    archivePrefix = "arXiv",
    primaryClass = "hep-th",
    reportNumber = "MPP-2024-224, CERN-TH-2024-204",
    month = "12",
    year = "2024"
}

@article{Weigand:2018rez,
    author = "Weigand, Timo",
    title = "{F-theory}",
    eprint = "1806.01854",
    archivePrefix = "arXiv",
    primaryClass = "hep-th",
    reportNumber = "CERN-TH-2018-126",
    journal = "PoS",
    volume = "TASI2017",
    pages = "016",
    year = "2018"
}

@article{Witten:1995im,
    author = "Witten, Edward",
    title = "{Bound states of strings and p-branes}",
    eprint = "hep-th/9510135",
    archivePrefix = "arXiv",
    reportNumber = "IASSNS-HEP-95-83",
    doi = "10.1016/0550-3213(95)00610-9",
    journal = "Nucl. Phys. B",
    volume = "460",
    pages = "335--350",
    year = "1996"
}

@article{Lee:2019wij,
    author = "Lee, Seung-Joo and Lerche, Wolfgang and Weigand, Timo",
    title = "{Emergent strings from infinite distance limits}",
    eprint = "1910.01135",
    archivePrefix = "arXiv",
    primaryClass = "hep-th",
    reportNumber = "CERN-TH-2019-159",
    doi = "10.1007/JHEP02(2022)190",
    journal = "JHEP",
    volume = "02",
    pages = "190",
    year = "2022"
}

@article{Lee:2019xtm,
    author = "Lee, Seung-Joo and Lerche, Wolfgang and Weigand, Timo",
    title = "{Emergent strings, duality and weak coupling limits for two-form fields}",
    eprint = "1904.06344",
    archivePrefix = "arXiv",
    primaryClass = "hep-th",
    reportNumber = "CERN-TH-2019-044",
    doi = "10.1007/JHEP02(2022)096",
    journal = "JHEP",
    volume = "02",
    pages = "096",
    year = "2022"
}

@misc{heuer2020computingcommutatorlengthhard,
      title={Computing commutator length is hard}, 
      author={Nicolaus Heuer},
      year={2020},
      eprint={2001.10230},
      archivePrefix={arXiv},
      primaryClass={math.GR},
      url={https://arxiv.org/abs/2001.10230}, 
}

@book{Calegari2009scl,
  author    = {Calegari, Danny},
  title     = {scl (stable commutator length)},
  series    = {MSJ Memoirs},
  volume    = {20},
  year      = {2009},
  publisher = {Mathematical Society of Japan},
  address   = {Tokyo},
  isbn      = {978-4-931469-53-2},
  mrnumber  = {2527220}
}

@article{Cota:2022yjw,
    author = "Cota, Cesar Fierro and Mininno, Alessandro and Weigand, Timo and Wiesner, Max",
    title = "{The asymptotic Weak Gravity Conjecture for open strings}",
    eprint = "2208.00009",
    archivePrefix = "arXiv",
    primaryClass = "hep-th",
    reportNumber = "ZMP-HH/22-15",
    doi = "10.1007/JHEP11(2022)058",
    journal = "JHEP",
    volume = "11",
    pages = "058",
    year = "2022"
}

@article{Harlow:2022ich,
    author = "Harlow, Daniel and Heidenreich, Ben and Reece, Matthew and Rudelius, Tom",
    title = "{Weak gravity conjecture}",
    eprint = "2201.08380",
    archivePrefix = "arXiv",
    primaryClass = "hep-th",
    reportNumber = "ACFI-T22-01",
    doi = "10.1103/RevModPhys.95.035003",
    journal = "Rev. Mod. Phys.",
    volume = "95",
    number = "3",
    pages = "035003",
    year = "2023"
}

@article{VanRiet:2023pnx,
    author = "Van Riet, Thomas and Zoccarato, Gianluca",
    title = "{Beginners lectures on flux compactifications and related Swampland topics}",
    eprint = "2305.01722",
    archivePrefix = "arXiv",
    primaryClass = "hep-th",
    doi = "10.1016/j.physrep.2023.11.003",
    journal = "Phys. Rept.",
    volume = "1049",
    pages = "1--51",
    year = "2024"
}

@article{Agmon:2022thq,
    author = "Agmon, Nathan Benjamin and Bedroya, Alek and Kang, Monica Jinwoo and Vafa, Cumrun",
    title = "{Lectures on the string landscape and the Swampland}",
    eprint = "2212.06187",
    archivePrefix = "arXiv",
    primaryClass = "hep-th",
    month = "12",
    year = "2022"
}

@article{vandeHeisteeg:2023dlw,
    author = "van de Heisteeg, Damian and Vafa, Cumrun and Wiesner, Max and Wu, David H.",
    title = "{Species scale in diverse dimensions}",
    eprint = "2310.07213",
    archivePrefix = "arXiv",
    primaryClass = "hep-th",
    doi = "10.1007/JHEP05(2024)112",
    journal = "JHEP",
    volume = "05",
    pages = "112",
    year = "2024"
}

@article{vandeHeisteeg:2023ubh,
    author = "van de Heisteeg, Damian and Vafa, Cumrun and Wiesner, Max",
    title = "{Bounds on Species Scale and the Distance Conjecture}",
    eprint = "2303.13580",
    archivePrefix = "arXiv",
    primaryClass = "hep-th",
    doi = "10.1002/prop.202300143",
    journal = "Fortsch. Phys.",
    volume = "71",
    number = "10-11",
    pages = "2300143",
    year = "2023"
}

@article{Blanco-Pillado:2016xvf,
    author = "Blanco-Pillado, Jose J. and Shlaer, Benjamin and Sousa, Kepa and Urrestilla, Jon",
    title = "{Bubbles of Nothing and Supersymmetric Compactifications}",
    eprint = "1606.03095",
    archivePrefix = "arXiv",
    primaryClass = "hep-th",
    doi = "10.1088/1475-7516/2016/10/002",
    journal = "JCAP",
    volume = "10",
    pages = "002",
    year = "2016"
}

@article{Ookouchi:2024tfz,
    author = "Ookouchi, Yutaka and Sato, Ryota and Tsukahara, Sohei",
    title = "{Decay of Kaluza-Klein vacuum via singular instanton}",
    eprint = "2404.13917",
    archivePrefix = "arXiv",
    primaryClass = "hep-th",
    doi = "10.1007/JHEP10(2025)040",
    journal = "JHEP",
    volume = "10",
    pages = "040",
    year = "2025"
}

@article{Long:2021jlv,
    author = "Long, Cody and Montero, Miguel and Vafa, Cumrun and Valenzuela, Irene",
    title = "{The desert and the swampland}",
    eprint = "2112.11467",
    archivePrefix = "arXiv",
    primaryClass = "hep-th",
    reportNumber = "IFT-UAM/CSIC-21-156",
    doi = "10.1007/JHEP03(2023)109",
    journal = "JHEP",
    volume = "03",
    pages = "109",
    year = "2023"
}

@article{Calderon-Infante:2023ler,
    author = "Calder\'on-Infante, Jos\'e and Castellano, Alberto and Herr\'aez, Alvaro and Ib\'a\~nez, Luis E.",
    title = "{Entropy bounds and the species scale distance conjecture}",
    eprint = "2306.16450",
    archivePrefix = "arXiv",
    primaryClass = "hep-th",
    doi = "10.1007/JHEP01(2024)039",
    journal = "JHEP",
    volume = "01",
    pages = "039",
    year = "2024"
}

@article{Obers:1998fb,
	Archiveprefix = {arXiv},
	Author = {Obers, N. A. and Pioline, B.},
	Doi = {10.1016/S0370-1573(99)00004-6},
	Eprint = {hep-th/9809039},
	Journal = {Phys. Rept.},
	Pages = {113--225},
	Reportnumber = {CERN-TH-98-282, CPHT-S639-0898},
	Title = {{U duality and M theory}},
	Volume = {318},
	Year = {1999}}

@article{Hull:1994ys,
	Archiveprefix = {arXiv},
	Author = {Hull, C. M. and Townsend, P. K.},
	Doi = {10.1016/0550-3213(94)00559-W},
	Eprint = {hep-th/9410167},
	Journal = {Nucl. Phys. B},
	Pages = {109--137},
	Reportnumber = {QMW-94-30, DAMTP-R-94-33},
	Title = {{Unity of superstring dualities}},
	Volume = {438},
	Year = {1995}}

@article{Greene:1995hu,
	Archiveprefix = {arXiv},
	Author = {Greene, Brian R. and Morrison, David R. and Strominger, Andrew},
	Doi = {10.1016/0550-3213(95)00371-X},
	Eprint = {hep-th/9504145},
	Journal = {Nucl. Phys.},
	Pages = {109-120},
	Primaryclass = {hep-th},
	Reportnumber = {CLNS-95-1335},
	Slaccitation = {%%CITATION = HEP-TH/9504145;%%},
	Title = {{Black hole condensation and the unification of string vacua}},
	Volume = {B451},
	Year = {1995}}

@article{Cordova:2022rer,
	Archiveprefix = {arXiv},
	Author = {Cordova, Clay and Ohmori, Kantaro and Rudelius, Tom},
	Eprint = {2202.05866},
	Month = {2},
	Primaryclass = {hep-th},
	Title = {{Generalized Symmetry Breaking Scales and Weak Gravity Conjectures}},
	Year = {2022}}

@article{Daus:2020vtf,
	Archiveprefix = {arXiv},
	Author = {Daus, Tristan and Hebecker, Arthur and Leonhardt, Sascha and March-Russell, John},
	Doi = {10.1016/j.nuclphysb.2020.115167},
	Eprint = {2002.02456},
	Journal = {Nucl. Phys. B},
	Pages = {115167},
	Primaryclass = {hep-th},
	Title = {{Towards a Swampland Global Symmetry Conjecture using weak gravity}},
	Volume = {960},
	Year = {2020}}

@article{Fichet:2019ugl,
	Archiveprefix = {arXiv},
	Author = {Fichet, Sylvain and Saraswat, Prashant},
	Doi = {10.1007/JHEP01(2020)088},
	Eprint = {1909.02002},
	Journal = {JHEP},
	Pages = {088},
	Primaryclass = {hep-th},
	Reportnumber = {CALT-TH-2019-32},
	Title = {{Approximate Symmetries and Gravity}},
	Volume = {01},
	Year = {2020}}

@article{Lanza:2022zyg,
    author = "Lanza, Stefano and Marchesano, Fernando and Martucci, Luca and Valenzuela, Irene",
    title = "{Large Field Distances from EFT strings}",
    eprint = "2205.04532",
    archivePrefix = "arXiv",
    primaryClass = "hep-th",
    doi = "10.22323/1.406.0169",
    journal = "PoS",
    volume = "CORFU2021",
    pages = "169",
    year = "2022"
}

@article{Marchesano:2022avb,
    author = "Marchesano, Fernando and Wiesner, Max",
    title = "{4d strings at strong coupling}",
    eprint = "2202.10466",
    archivePrefix = "arXiv",
    primaryClass = "hep-th",
    reportNumber = "IFT-UAM/CSIC-22-13",
    doi = "10.1007/JHEP08(2022)004",
    journal = "JHEP",
    volume = "08",
    pages = "004",
    year = "2022"
}

@article{Pantev:2016nze,
    author = "Pantev, T. and Sharpe, E.",
    title = "{Duality group actions on fermions}",
    eprint = "1609.00011",
    archivePrefix = "arXiv",
    primaryClass = "hep-th",
    doi = "10.1007/JHEP11(2016)171",
    journal = "JHEP",
    volume = "11",
    pages = "171",
    year = "2016"
}

@article{Chakrabhavi:2025bfi,
    author = "Chakrabhavi, Vivek and Debray, Arun and Dierigl, Markus and Heckman, Jonathan J.",
    title = "{Exploring Pintopia: Reflection Branes, Bordisms, and U-Dualities}",
    eprint = "2509.03573",
    archivePrefix = "arXiv",
    primaryClass = "hep-th",
    reportNumber = "CERN-TH-2025-180",
    month = "9",
    year = "2025"
}

@article{Alexandrov:2013yva,
    author = "Alexandrov, Sergei and Manschot, Jan and Persson, Daniel and Pioline, Boris",
    editor = "Donagi, Ron and Katz, Sheldon and Klemm, Albrecht and Morrison, David R.",
    title = "{Quantum hypermultiplet moduli spaces in N=2 string vacua: a review}",
    eprint = "1304.0766",
    archivePrefix = "arXiv",
    primaryClass = "hep-th",
    reportNumber = "CERN-PH-TH-2013-048",
    journal = "Proc. Symp. Pure Math.",
    volume = "90",
    pages = "181--212",
    year = "2015"
}

@article{Alexandrov:2010np,
    author = "Alexandrov, Sergei and Persson, Daniel and Pioline, Boris",
    title = "{On the topology of the hypermultiplet moduli space in type II/CY string vacua}",
    eprint = "1009.3026",
    archivePrefix = "arXiv",
    primaryClass = "hep-th",
    doi = "10.1103/PhysRevD.83.026001",
    journal = "Phys. Rev. D",
    volume = "83",
    pages = "026001",
    year = "2011"
}

@article{bridson2014actions,
  title={Actions of arithmetic groups on homology spheres and acyclic homology manifolds},
  author={Bridson, Martin R and Grunewald, Fritz and Vogtmann, Karen},
  journal={Mathematische Zeitschrift},
  volume={276},
  number={1},
  pages={387--395},
  year={2014},
  publisher={Springer}
}

@article{Martucci:2022krl,
    author = "Martucci, Luca and Risso, Nicolo and Weigand, Timo",
    title = "{Quantum gravity bounds on $ \mathcal{N} $ = 1 effective theories in four dimensions}",
    eprint = "2210.10797",
    archivePrefix = "arXiv",
    primaryClass = "hep-th",
    doi = "10.1007/JHEP03(2023)197",
    journal = "JHEP",
    volume = "03",
    pages = "197",
    year = "2023"
}

@article{Wiesner:2022qys,
    author = "Wiesner, Max",
    title = "{Light strings and strong coupling in F-theory}",
    eprint = "2210.14238",
    archivePrefix = "arXiv",
    primaryClass = "hep-th",
    doi = "10.1007/JHEP04(2023)088",
    journal = "JHEP",
    volume = "04",
    pages = "088",
    year = "2023"
}

@article{Marchesano:2022axe,
    author = "Marchesano, Fernando and Melotti, Luca",
    title = "{EFT strings and emergence}",
    eprint = "2211.01409",
    archivePrefix = "arXiv",
    primaryClass = "hep-th",
    doi = "10.1007/JHEP02(2023)112",
    journal = "JHEP",
    volume = "02",
    pages = "112",
    year = "2023"
}

@article{Martucci:2024trp,
    author = "Martucci, Luca and Risso, Nicol{\`o} and Valenti, Alessandro and Vecchi, Luca",
    title = "{Wormholes in the axiverse, and the species scale}",
    eprint = "2404.14489",
    archivePrefix = "arXiv",
    primaryClass = "hep-th",
    doi = "10.1007/JHEP07(2024)240",
    journal = "JHEP",
    volume = "07",
    pages = "240",
    year = "2024"
}

@article{Zeldovich:1976vq,
	Author = {Zeldovich, Ya.~B.},
	Doi = {10.1016/0375-9601(76)90783-0},
	Journal = {Phys. Lett. A},
	Pages = {254},
	Title = {{A New Type of Radioactive Decay: Gravitational Annihilation of Baryons}},
	Volume = {59},
	Year = {1976}}

@article{Zeldovich:1977be,
	Author = {Zeldovich, Ya.~B.},
	Journal = {Zh. Eksp. Teor. Fiz.},
	Pages = {18--21},
	Title = {{A Novel Type of Radioactive Decay: Gravitational Baryon Annihilation}},
	Volume = {72},
	Year = {1977}}

@article{McNamara:2019rup,
	Archiveprefix = {arXiv},
	Author = {McNamara, Jacob and Vafa, Cumrun},
	Eprint = {1909.10355},
	Month = {9},
	Primaryclass = {hep-th},
	Title = {{Cobordism Classes and the Swampland}},
	Year = {2019}}

@article{Yonekura:2022reu,
    author = "Yonekura, Kazuya",
    title = "{Heterotic global anomalies and torsion Witten index}",
    eprint = "2207.13858",
    archivePrefix = "arXiv",
    primaryClass = "hep-th",
    reportNumber = "TU-1163",
    doi = "10.1007/JHEP10(2022)114",
    journal = "JHEP",
    volume = "10",
    pages = "114",
    year = "2022"
}

@book{Brown1982,
  author    = {Brown, Kenneth S.},
  title     = {Cohomology of Groups},
  series    = {Graduate Texts in Mathematics},
  volume    = {87},
  year      = {1982},
  publisher = {Springer-Verlag},
  address   = {New York-Berlin},
  isbn      = {0-387-90688-6},
  doi       = {10.1007/978-1-4684-9327-6}
}

@book{Polchinski:1998rr,
	Author = {Polchinski, J.},
	Doi = {10.1017/CBO9780511618123},
	Isbn = {978-0-511-25228-0, 978-0-521-63304-8, 978-0-521-67228-3},
	Month = {12},
	Publisher = {Cambridge University Press},
	Series = {Cambridge Monographs on Mathematical Physics},
	Title = {{String theory. Vol. 2: Superstring theory and beyond}},
	Year = {2007}}

@article{Greene:1989ya,
    author = "Greene, Brian R. and Shapere, Alfred D. and Vafa, Cumrun and Yau, Shing-Tung",
    title = "{Stringy Cosmic Strings and Noncompact Calabi-Yau Manifolds}",
    reportNumber = "HUTP-89-A047, IASSNS-HEP-89-47",
    doi = "10.1016/0550-3213(90)90248-C",
    journal = "Nucl. Phys. B",
    volume = "337",
    pages = "1--36",
    year = "1990"
}

@article{Lanza:2020qmt,
	Archiveprefix = {arXiv},
	Author = {Lanza, Stefano and Marchesano, Fernando and Martucci, Luca and Valenzuela, Irene},
	Doi = {10.1007/JHEP02(2021)006},
	Eprint = {2006.15154},
	Journal = {JHEP},
	Pages = {006},
	Primaryclass = {hep-th},
	Title = {{Swampland Conjectures for Strings and Membranes}},
	Volume = {02},
	Year = {2021}}

@article{Banks:2010zn,
	Archiveprefix = {arXiv},
	Author = {Banks, Tom and Seiberg, Nathan},
	Doi = {10.1103/PhysRevD.83.084019},
	Eprint = {1011.5120},
	Journal = {Phys. Rev.},
	Pages = {084019},
	Primaryclass = {hep-th},
	Slaccitation = {%%CITATION = ARXIV:1011.5120;%%},
	Title = {{Symmetries and Strings in Field Theory and Gravity}},
	Volume = {D83},
	Year = {2011}}

@article{Vafa:2005ui,
	Archiveprefix = {arXiv},
	Author = {Vafa, Cumrun},
	Eprint = {hep-th/0509212},
	Primaryclass = {hep-th},
	Reportnumber = {HUTP-05-A043},
	Slaccitation = {%%CITATION = HEP-TH/0509212;%%},
	Title = {{The String landscape and the swampland}},
	Year = {2005}}

@article{Ooguri:2006in,
	Archiveprefix = {arXiv},
	Author = {Ooguri, Hirosi and Vafa, Cumrun},
	Doi = {10.1016/j.nuclphysb.2006.10.033},
	Eprint = {hep-th/0605264},
	Journal = {Nucl.Phys.},
	Pages = {21-33},
	Primaryclass = {hep-th},
	Reportnumber = {CALT-68-2600, HUTP-06-A017},
	Slaccitation = {%%CITATION = HEP-TH/0605264;%%},
	Title = {{On the Geometry of the String Landscape and the Swampland}},
	Volume = {B766},
	Year = {2007}}

@article{Banks:1988yz,
	Author = {Banks, Tom and Dixon, Lance J.},
	Doi = {10.1016/0550-3213(88)90523-8},
	Journal = {Nucl. Phys.},
	Pages = {93-108},
	Reportnumber = {PUPT-1086, SCIPP-8805},
	Slaccitation = {%%CITATION = NUPHA,B307,93;%%},
	Title = {{Constraints on String Vacua with Space-Time Supersymmetry}},
	Volume = {B307},
	Year = {1988}}

@article{Kallosh:1995hi,
	Archiveprefix = {arXiv},
	Author = {Kallosh, Renata and Linde, Andrei D. and Linde, Dmitri A. and Susskind, Leonard},
	Doi = {10.1103/PhysRevD.52.912},
	Eprint = {hep-th/9502069},
	Journal = {Phys.Rev.},
	Pages = {912-935},
	Primaryclass = {hep-th},
	Reportnumber = {SU-ITP-95-2},
	Slaccitation = {%%CITATION = HEP-TH/9502069;%%},
	Title = {{Gravity and global symmetries}},
	Volume = {D52},
	Year = {1995}}

@book{Ibanez:2012zz,
	Author = {Ib{\'a}{\~n}ez, Luis E. and Uranga, Angel M.},
	Isbn = {9780521517522, 9781139227421},
	Publisher = {Cambridge University Press},
	Slaccitation = {%%CITATION = INSPIRE-1112474;%%},
	Title = {{String theory and particle physics: An introduction to string phenomenology}},
	Url = {http://www.cambridge.org/de/knowledge/isbn/item6563092/?site_locale=de_DE},
	Year = {2012}}

@article{Brennan:2017rbf,
	Archiveprefix = {arXiv},
	Author = {Brennan, T. Daniel and Carta, Federico and Vafa, Cumrun},
	Booktitle = {{Proceedings, Theoretical Advanced Study Institute in Elementary Particle Physics: Physics at the Fundamental Frontier (TASI 2017): Boulder, CO, USA, June 5-30, 2017}},
	Doi = {10.22323/1.305.0015},
	Eprint = {1711.00864},
	Journal = {PoS},
	Pages = {015},
	Primaryclass = {hep-th},
	Reportnumber = {IFT-UAM-CSIC-17-105},
	Slaccitation = {%%CITATION = ARXIV:1711.00864;%%},
	Title = {{The String Landscape, the Swampland, and the Missing Corner}},
	Volume = {TASI2017},
	Year = {2017}}

@article{Grimm:2018ohb,
	Archiveprefix = {arXiv},
	Author = {Grimm, Thomas W. and Palti, Eran and Valenzuela, Irene},
	Doi = {10.1007/JHEP08(2018)143},
	Eprint = {1802.08264},
	Journal = {JHEP},
	Pages = {143},
	Primaryclass = {hep-th},
	Slaccitation = {%%CITATION = ARXIV:1802.08264;%%},
	Title = {{Infinite Distances in Field Space and Massless Towers of States}},
	Volume = {08},
	Year = {2018}}

@article{Draper:2021qtc,
    author = "Draper, Patrick and Garcia Garcia, Isabel and Lillard, Benjamin",
    title = "{De Sitter decays to infinity}",
    eprint = "2105.10507",
    archivePrefix = "arXiv",
    primaryClass = "hep-th",
    doi = "10.1007/JHEP12(2021)154",
    journal = "JHEP",
    volume = "12",
    pages = "154",
    year = "2021"
}

@article{Dvali:2007hz,
	Archiveprefix = {arXiv},
	Author = {Dvali, Gia},
	Doi = {10.1002/prop.201000009},
	Eprint = {0706.2050},
	Journal = {Fortsch. Phys.},
	Pages = {528-536},
	Primaryclass = {hep-th},
	Slaccitation = {%%CITATION = ARXIV:0706.2050;%%},
	Title = {{Black Holes and Large N Species Solution to the Hierarchy Problem}},
	Volume = {58},
	Year = {2010}}

@article{Dvali:2007wp,
	Archiveprefix = {arXiv},
	Author = {Dvali, Gia and Redi, Michele},
	Doi = {10.1103/PhysRevD.77.045027},
	Eprint = {0710.4344},
	Journal = {Phys. Rev.},
	Pages = {045027},
	Primaryclass = {hep-th},
	Slaccitation = {%%CITATION = ARXIV:0710.4344;%%},
	Title = {{Black Hole Bound on the Number of Species and Quantum Gravity at LHC}},
	Volume = {D77},
	Year = {2008}}

@article{vandeHeisteeg:2022btw,
    author = "van de Heisteeg, Damian and Vafa, Cumrun and Wiesner, Max and Wu, David H.",
    title = "{Moduli-dependent Species Scale}",
    eprint = "2212.06841",
    archivePrefix = "arXiv",
    primaryClass = "hep-th",
    doi = "10.4310/BPAM.2024.v1.n1.a1",
    month = "12",
    year = "2022"
}

@article{Castellano:2023aum,
    author = "Castellano, Alberto and Herr\'aez, Alvaro and Ib\'a\~nez, Luis E.",
    title = "{On the Species Scale, Modular Invariance and the Gravitational EFT expansion}",
    eprint = "2310.07708",
    archivePrefix = "arXiv",
    primaryClass = "hep-th",
    month = "10",
    year = "2023"
}

@article{Buratti:2021yia,
    author = "Buratti, Ginevra and Delgado, Matilda and Uranga, Angel M.",
    title = "{Dynamical tadpoles, stringy cobordism, and the SM from spontaneous compactification}",
    eprint = "2104.02091",
    archivePrefix = "arXiv",
    primaryClass = "hep-th",
    doi = "10.1007/JHEP06(2021)170",
    journal = "JHEP",
    volume = "06",
    pages = "170",
    year = "2021"
}

@article{Angius:2022mgh,
    author = "Angius, Roberta and Delgado, Matilda and Uranga, Angel M.",
    title = "{Dynamical Cobordism and the beginning of time: supercritical strings and tachyon condensation}",
    eprint = "2207.13108",
    archivePrefix = "arXiv",
    primaryClass = "hep-th",
    doi = "10.1007/JHEP08(2022)285",
    journal = "JHEP",
    volume = "08",
    pages = "285",
    year = "2022"
}

@article{Angius:2023uqk,
    author = "Angius, Roberta and Makridou, Andriana and Uranga, Angel M.",
    title = "{Intersecting end of the world branes}",
    eprint = "2312.16286",
    archivePrefix = "arXiv",
    primaryClass = "hep-th",
    doi = "10.1007/JHEP03(2024)110",
    journal = "JHEP",
    volume = "03",
    pages = "110",
    year = "2024"
}

@article{Blumenhagen:2022mqw,
    author = "Blumenhagen, Ralph and Cribiori, Niccol\`o and Kneissl, Christian and Makridou, Andriana",
    title = "{Dynamical cobordism of a domain wall and its companion defect 7-brane}",
    eprint = "2205.09782",
    archivePrefix = "arXiv",
    primaryClass = "hep-th",
    reportNumber = "MPP-2022-57",
    doi = "10.1007/JHEP08(2022)204",
    journal = "JHEP",
    volume = "08",
    pages = "204",
    year = "2022"
}

@article{Buratti:2021fiv,
    author = "Buratti, Ginevra and Calder\'on-Infante, Jos\'e and Delgado, Matilda and Uranga, Angel M.",
    title = "{Dynamical Cobordism and Swampland Distance Conjectures}",
    eprint = "2107.09098",
    archivePrefix = "arXiv",
    primaryClass = "hep-th",
    doi = "10.1007/JHEP10(2021)037",
    journal = "JHEP",
    volume = "10",
    pages = "037",
    year = "2021"
}

@article{Angius:2022aeq,
    author = "Angius, Roberta and Calder\'on-Infante, Jos\'e and Delgado, Matilda and Huertas, Jes\'us and Uranga, Angel M.",
    title = "{At the end of the world: Local Dynamical Cobordism}",
    eprint = "2203.11240",
    archivePrefix = "arXiv",
    primaryClass = "hep-th",
    reportNumber = "IFT-UAM/CSIC-22-31",
    doi = "10.1007/JHEP06(2022)142",
    journal = "JHEP",
    volume = "06",
    pages = "142",
    year = "2022"
}

@article{vanBeest:2021lhn,
    author = "van Beest, Marieke and Calder\'on-Infante, Jos\'e and Mirfendereski, Delaram and Valenzuela, Irene",
    title = "{Lectures on the Swampland Program in String Compactifications}",
    eprint = "2102.01111",
    archivePrefix = "arXiv",
    primaryClass = "hep-th",
    doi = "10.1016/j.physrep.2022.09.002",
    journal = "Phys. Rept.",
    volume = "989",
    pages = "1--50",
    year = "2022"
}

@article{Palti:2019pca,
    author = "Palti, Eran",
    title = "{The Swampland: Introduction and Review}",
    eprint = "1903.06239",
    archivePrefix = "arXiv",
    primaryClass = "hep-th",
    reportNumber = "MPP-2019-53",
    doi = "10.1002/prop.201900037",
    journal = "Fortsch. Phys.",
    volume = "67",
    number = "6",
    pages = "1900037",
    year = "2019"
}

@article{Lanza:2021udy,
    author = "Lanza, Stefano and Marchesano, Fernando and Martucci, Luca and Valenzuela, Irene",
    title = "{The EFT stringy viewpoint on large distances}",
    eprint = "2104.05726",
    archivePrefix = "arXiv",
    primaryClass = "hep-th",
    doi = "10.1007/JHEP09(2021)197",
    journal = "JHEP",
    volume = "09",
    pages = "197",
    year = "2021"
}

@article{korkmaz2004stable,
  title={Stable commutator length of a Dehn twist},
  author={Korkmaz, Mustafa},
  journal={Michigan Mathematical Journal},
  volume={52},
  number={1},
  pages={23--31},
  year={2004},
  publisher={University of Michigan, Department of Mathematics}
}

@article{Blum:1997cs,
    author = "Blum, Julie D. and Dienes, Keith R.",
    title = "{Duality without supersymmetry: The Case of the SO(16) x SO(16) string}",
    eprint = "hep-th/9707148",
    archivePrefix = "arXiv",
    reportNumber = "IASSNS-HEP-97-67",
    doi = "10.1016/S0370-2693(97)01172-6",
    journal = "Phys. Lett. B",
    volume = "414",
    pages = "260--268",
    year = "1997"
}

@article{Horava:1996ma,
    author = "Horava, Petr and Witten, Edward",
    title = "{Eleven-dimensional supergravity on a manifold with boundary}",
    eprint = "hep-th/9603142",
    archivePrefix = "arXiv",
    reportNumber = "IASSNS-HEP-96-17, PUPT-1597",
    doi = "10.1016/0550-3213(96)00308-2",
    journal = "Nucl. Phys. B",
    volume = "475",
    pages = "94--114",
    year = "1996"
}

@article{Horava:1995qa,
    author = "Horava, Petr and Witten, Edward",
    title = "{Heterotic and type I string dynamics from eleven-dimensions}",
    eprint = "hep-th/9510209",
    archivePrefix = "arXiv",
    reportNumber = "IASSNS-HEP-95-86, PUPT-1571A",
    doi = "10.1016/0550-3213(95)00621-4",
    journal = "Nucl. Phys. B",
    volume = "460",
    pages = "506--524",
    year = "1996"
}

@article{Witten:1995ex,
    author = "Witten, Edward",
    title = "{String theory dynamics in various dimensions}",
    eprint = "hep-th/9503124",
    archivePrefix = "arXiv",
    reportNumber = "IASSNS-HEP-95-18",
    doi = "10.1016/0550-3213(95)00158-O",
    journal = "Nucl. Phys. B",
    volume = "443",
    pages = "85--126",
    year = "1995"
}

@article{Grimm:2022sbl,
    author = "Grimm, Thomas W. and Lanza, Stefano and Li, Chongchuo",
    title = "{Tameness, Strings, and the Distance Conjecture}",
    eprint = "2206.00697",
    archivePrefix = "arXiv",
    primaryClass = "hep-th",
    doi = "10.1007/JHEP09(2022)149",
    journal = "JHEP",
    volume = "09",
    pages = "149",
    year = "2022"
}

@article{Grana:2021zvf,
    author = "Gra\~na, Mariana and Herr\'aez, Alvaro",
    title = "{The Swampland Conjectures: A Bridge from Quantum Gravity to Particle Physics}",
    eprint = "2107.00087",
    archivePrefix = "arXiv",
    primaryClass = "hep-th",
    doi = "10.3390/universe7080273",
    journal = "Universe",
    volume = "7",
    number = "8",
    pages = "273",
    year = "2021"
}

@article{Romans:1985tz,
    author = "Romans, L. J.",
    editor = "Salam, A. and Sezgin, E.",
    title = "{Massive N=2a Supergravity in Ten-Dimensions}",
    reportNumber = "NSF-ITP-85-148",
    doi = "10.1016/0370-2693(86)90375-8",
    journal = "Phys. Lett. B",
    volume = "169",
    pages = "374",
    year = "1986"
}

@article{Howe:1997qt,
    author = "Howe, Paul S. and Lambert, N. D. and West, Peter C.",
    title = "{A New massive type IIA supergravity from compactification}",
    eprint = "hep-th/9707139",
    archivePrefix = "arXiv",
    reportNumber = "KCL-TH-97-46",
    doi = "10.1016/S0370-2693(97)01199-4",
    journal = "Phys. Lett. B",
    volume = "416",
    pages = "303--308",
    year = "1998"
}

@article{PhysRevD.30.325,
  title = {$N=2$ supergravity in ten dimensions},
  author = {Giani, F. and Pernici, M.},
  journal = {Phys. Rev. D},
  volume = {30},
  issue = {2},
  pages = {325--333},
  numpages = {0},
  year = {1984},
  month = {Jul},
  publisher = {American Physical Society},
  doi = {10.1103/PhysRevD.30.325},
  url = {https://link.aps.org/doi/10.1103/PhysRevD.30.325}
}

@article{CAMPBELL1984112,
title = {N = 2, D = 10 non-chiral supergravity and its spontaneous compactification},
journal = {Nuclear Physics B},
volume = {243},
number = {1},
pages = {112-124},
year = {1984},
issn = {0550-3213},
doi = {https://doi.org/10.1016/0550-3213(84)90388-2},
url = {https://www.sciencedirect.com/science/article/pii/0550321384903882},
author = {I.C.G. Campbell and P.C. West}
}

@article{Huq:1983im,
    author = "Huq, M. and Namazie, M. A.",
    editor = "Salam, A. and Sezgin, E.",
    title = "{{Kaluza-Klein} Supergravity in Ten-dimensions}",
    reportNumber = "IC/83/210",
    doi = "10.1088/0264-9381/2/3/007",
    journal = "Class. Quant. Grav.",
    volume = "2",
    pages = "293",
    year = "1985",
    note = "[Erratum: Class.Quant.Grav. 2, 597 (1985)]"
}

@article{Montero:2022prj,
    author = "Montero, Miguel and Vafa, Cumrun and Valenzuela, Irene",
    title = "{The dark dimension and the Swampland}",
    eprint = "2205.12293",
    archivePrefix = "arXiv",
    primaryClass = "hep-th",
    doi = "10.1007/JHEP02(2023)022",
    journal = "JHEP",
    volume = "02",
    pages = "022",
    year = "2023"
}

@article{Strominger:1995cz,
    author = "Strominger, Andrew",
    title = "{Massless black holes and conifolds in string theory}",
    eprint = "hep-th/9504090",
    archivePrefix = "arXiv",
    doi = "10.1016/0550-3213(95)00287-3",
    journal = "Nucl. Phys. B",
    volume = "451",
    pages = "96--108",
    year = "1995"
}

@article{anderson1967structure,
  title={The structure of the {Spin} cobordism ring},
  author={Anderson, Donald W and Brown Jr, Edgar H and Peterson, Franklin P},
  journal={Annals of Mathematics},
  volume={86},
  number={2},
  pages={271--298},
  year={1967},
  publisher={JSTOR},
  doi={10.2307/1970690}
}

@article{Alonso-Alberca:2002pbb,
    author = "Alonso-Alberca, Natxo and Ortin, Tomas",
    title = "{Gauged / massive supergravities in diverse dimensions}",
    eprint = "hep-th/0210011",
    archivePrefix = "arXiv",
    reportNumber = "IFT-UAM-CSIC-02-18",
    doi = "10.1016/S0550-3213(02)01125-2",
    journal = "Nucl. Phys. B",
    volume = "651",
    pages = "263--290",
    year = "2003"
}

@article{GarciadelMoral:2016ffr,
    author = "Garcia del Moral, Maria Pilar and Pena, Joselen M. and Restuccia, Alvaro",
    title = "{Classification of M2-brane 2-torus bundles, U-duality invariance and type II gauged supergravities}",
    eprint = "1604.02579",
    archivePrefix = "arXiv",
    primaryClass = "hep-th",
    doi = "10.1103/PhysRevD.100.026005",
    journal = "Phys. Rev. D",
    volume = "100",
    number = "2",
    pages = "026005",
    year = "2019"
}

@article{Hull:1998vy,
    author = "Hull, C. M.",
    title = "{Massive string theories from M theory and F theory}",
    eprint = "hep-th/9811021",
    archivePrefix = "arXiv",
    reportNumber = "QMW-PH-98-36, LPTENS-98-32",
    doi = "10.1088/1126-6708/1998/11/027",
    journal = "JHEP",
    volume = "11",
    pages = "027",
    year = "1998"
}

@article{Tachikawa:2026jaj,
    author = "Tachikawa, Yuji",
    title = "{On the trivalent junction of three non-tachyonic heterotic string theories}",
    eprint = "2603.28133",
    archivePrefix = "arXiv",
    primaryClass = "hep-th",
    month = "3",
    year = "2026"
}

@article{Anastasi:2026cus,
    author = "Anastasi, Edoardo and Montero, Miguel and Uranga, Angel M. and Wang, Chuying",
    title = "{What IIB looks IIA string: String Cobordisms via Non-Compact CFTs}",
    eprint = "2603.00225",
    archivePrefix = "arXiv",
    primaryClass = "hep-th",
    reportNumber = "IFT-26-016",
    month = "2",
    year = "2026"
}

@article{Altavista:2026edv,
    author = "Altavista, Chiara and Anastasi, Edoardo and Angius, Roberta and Uranga, Angel M.",
    title = "{The Art of Branching: Cobordism Junctions of 10d String Theories}",
    eprint = "2603.24667",
    archivePrefix = "arXiv",
    primaryClass = "hep-th",
    month = "3",
    year = "2026"
}

@article{Cavusoglu:2026xiv,
    author = "{\c{C}}avu{\c{s}}o{\u{g}}lu, Atakan and Cveti{\v{c}}, Mirjam and Heckman, Jonathan J. and Kuntz, Jeffrey and Murdia, Chitraang",
    title = "{Gravitational Background of Alice-Vortices and R7-Branes}",
    eprint = "2602.13196",
    archivePrefix = "arXiv",
    primaryClass = "hep-th",
    month = "2",
    year = "2026"
}

@article{Torres:2026vxx,
    author = "Torres, Ethan",
    title = "{A Matrix Theory Construction of the IIA/IIB Wall}",
    eprint = "2603.02199",
    archivePrefix = "arXiv",
    primaryClass = "hep-th",
    reportNumber = "CERN-TH-2026-030",
    month = "3",
    year = "2026"
}

@article{Delgado:2023uqk,
    author = "Delgado, Matilda",
    title = "{The bubble of nothing under T-duality}",
    eprint = "2312.09291",
    archivePrefix = "arXiv",
    primaryClass = "hep-th",
    reportNumber = "IFT-UAM/CSIC-23-146",
    doi = "10.1007/JHEP05(2024)333",
    journal = "JHEP",
    volume = "05",
    pages = "333",
    year = "2024"
}

@article{Green,
author = {Smart, W.M and Green, R.M},
title = {Spherical astronomy.},
journal = {Astronomische Nachrichten},
volume = {309},
number = {4},
pages = {280-280},
doi = {10.1002/asna.2113090422},
url = {https://onlinelibrary.wiley.com/doi/abs/10.1002/asna.2113090422},
eprint = {https://onlinelibrary.wiley.com/doi/pdf/10.1002/asna.2113090422},
year = {1988}
}

@article{AbouZeid ,
      author         = "Abou-Zeid, Mohab and de Wit, Bernard and Lust, Dieter and
                        Nicolai, Hermann",
      title          = "{Space-time supersymmetry, IIA / B duality and M theory}",
      journal        = "Phys. Lett.",
      volume         = "B466",
      year           = "1999",
      pages          = "144-152",
      doi            = "10.1016/S0370-2693(99)01114-4",
      eprint         = "hep-th/9908169",
      archivePrefix  = "arXiv",
      primaryClass   = "hep-th",
      reportNumber   = "AEI-1999-13, HUB-EP-99-43, THU-99-22",
      SLACcitation   = "%%CITATION = HEP-TH/9908169;%%"
}

@article{Andriot1,
    author = "Andriot, David and Goi, Enrico and Minasian, Ruben and Petrini, Michela",
    title = "{Supersymmetry breaking branes on solvmanifolds and de Sitter vacua in string theory}",
    eprint = "1003.3774",
    archivePrefix = "arXiv",
    primaryClass = "hep-th",
    reportNumber = "IPHT-T10-022",
    doi = "10.1007/JHEP05(2011)028",
    journal = "JHEP",
    volume = "05",
    pages = "028",
    year = "2011"
}

@article{Andriot2,
    author = "Andriot, David",
    title = "{New supersymmetric vacua on solvmanifolds}",
    eprint = "1507.00014",
    archivePrefix = "arXiv",
    primaryClass = "hep-th",
    doi = "10.1007/JHEP02(2016)112",
    journal = "JHEP",
    volume = "02",
    pages = "112",
    year = "2016"
}

@article{Aschieri,
   title={Topological T-Duality for Twisted Tori},
   ISSN={1815-0659},
   url={http://dx.doi.org/10.3842/SIGMA.2021.012},
   DOI={10.3842/sigma.2021.012},
   journal={Symmetry, Integrability and Geometry: Methods and Applications},
   publisher={SIGMA (Symmetry, Integrability and Geometry: Methods and Application)},
   author={Aschieri, Paolo and Szabo, Richard J.},
   year={2021},
   month=feb }

@article{Bergshoeff5,
      author         = "Bergshoeff, E. and de Wit, T. and Gran, U. and Linares,
                        R. and Roest, D.",
      title          = "{(Non)Abelian gauged supergravities in nine-dimensions}",
      journal        = "JHEP",
      volume         = "10",
      year           = "2002",
      pages          = "061",
      doi            = "10.1088/1126-6708/2002/10/061",
      eprint         = "hep-th/0209205",
      archivePrefix  = "arXiv",
      primaryClass   = "hep-th",
      reportNumber   = "UG-02-41",
      SLACcitation   = "%%CITATION = HEP-TH/0209205;%%"
}

@article{Bergshoeff8,
    author = "Bergshoeff, E. and de Roo, M. and Green, Michael B. and Papadopoulos, G. and Townsend, P. K.",
    title = "{Duality of type II 7 branes and 8 branes}",
    eprint = "hep-th/9601150",
    archivePrefix = "arXiv",
    reportNumber = "DAMTP-R-95-55-REV, UG-15-95",
    doi = "10.1016/0550-3213(96)00171-X",
    journal = "Nucl. Phys. B",
    volume = "470",
    pages = "113--135",
    year = "1996"
}

@article{Bergshoeff9,
    author = "Bergshoeff, E. and Gran, U. and Roest, D.",
    title = "{Type IIB seven-brane solutions from nine-dimensional domain walls}",
    eprint = "hep-th/0203202",
    archivePrefix = "arXiv",
    reportNumber = "UG-02-38",
    doi = "10.1088/0264-9381/19/15/321",
    journal = "Class. Quant. Grav.",
    volume = "19",
    pages = "4207--4226",
    year = "2002"
}

@article{Bock,
   title={On low-dimensional solvmanifolds},
   volume={20},
   ISSN={1945-0036},
   url={http://dx.doi.org/10.4310/AJM.2016.v20.n2.a1},
   DOI={10.4310/ajm.2016.v20.n2.a1},
   number={2},
   journal={Asian Journal of Mathematics},
   publisher={International Press of Boston},
   author={Bock, Christoph},
   year={2016},
   pages={199–262} }

@article{Cowdall,
    author = "Cowdall, P. M.",
    title = "{Novel domain wall and Minkowski vacua of D = 9 maximal SO(2) gauged supergravity}",
    eprint = "hep-th/0009016",
    archivePrefix = "arXiv",
    reportNumber = "KCL-TH-00-51",
    doi = "10.1016/S0550-3213(01)00043-8",
    journal = "Nucl. Phys. B",
    volume = "600",
    pages = "81",
    year = "2001"
}

@article{Han:2004wt,
    author = "Han, Tao and Willenbrock, Scott",
    title = "{Scale of quantum gravity}",
    eprint = "hep-ph/0404182",
    archivePrefix = "arXiv",
    reportNumber = "MADPH-04-1367, ILL-TH-04-3, NSF-KITP-04-38",
    doi = "10.1016/j.physletb.2005.04.040",
    journal = "Phys. Lett. B",
    volume = "616",
    pages = "215--220",
    year = "2005"
}

@article{Nevoa:2025xiq,
    author = "Nevoa, Vinicius and Raman, Sanjay and Vafa, Cumrun",
    title = "{Elementary Constituents Conjecture}",
    eprint = "2511.13813",
    archivePrefix = "arXiv",
    primaryClass = "hep-th",
    month = "11",
    year = "2025"
}

@book{bridson_haefliger_1999,
  author    = {Bridson, Martin R. and Haefliger, Andr{\'e}},
  title     = {Metric Spaces of Non-Positive Curvature},
  series    = {Grundlehren der mathematischen Wissenschaften},
  volume    = {319},
  publisher = {Springer-Verlag},
  address   = {Berlin, Heidelberg},
  year      = {1999},
  doi       = {10.1007/978-3-662-12494-9},
  isbn      = {978-3-540-64324-1}
}

@article{Israel:1966rt,
    author = "Israel, W.",
    title = "{Singular hypersurfaces and thin shells in general relativity}",
    doi = "10.1007/BF02710419",
    journal = "Nuovo Cim. B",
    volume = "44S10",
    pages = "1",
    year = "1966",
    note = "[Erratum: Nuovo Cim.B 48, 463 (1967)]"
}

@article{Anber:2011ut,
    author = "Anber, Mohamed M. and Donoghue, John F.",
    title = "{On the running of the gravitational constant}",
    eprint = "1111.2875",
    archivePrefix = "arXiv",
    primaryClass = "hep-th",
    doi = "10.1103/PhysRevD.85.104016",
    journal = "Phys. Rev. D",
    volume = "85",
    pages = "104016",
    year = "2012"
}

@book{KatokSvetlana1992Fg,
series = {Chicago lectures in mathematics},
publisher = {University of Chicago Press},
isbn = {0226425827},
year = {1992},
title = {Fuchsian groups},
language = {eng},
address = {Chicago},
author = {Katok, Svetlana},
lccn = {92006535},
}

@article{Hassfeld:2023kpu,
    author = "Hassfeld, Bjoern and Hebecker, Arthur and Walcher, Johannes",
    title = "{Cobordism and bubbles of anything in the string landscape}",
    eprint = "2310.06021",
    archivePrefix = "arXiv",
    primaryClass = "hep-th",
    doi = "10.1007/JHEP02(2024)127",
    journal = "JHEP",
    volume = "02",
    pages = "127",
    year = "2024"
}

@article{Witten:1981gj,
    author = "Witten, Edward",
    title = "{Instability of the Kaluza-Klein Vacuum}",
    reportNumber = "PRINT-81-0441 (PRINCETON)",
    doi = "10.1016/0550-3213(82)90007-4",
    journal = "Nucl. Phys. B",
    volume = "195",
    pages = "481--492",
    year = "1982"
}

@phdthesis{Castellano:2024bna,
    author = "Castellano, Alberto",
    title = "{The Quantum Gravity Scale and the Swampland}",
    eprint = "2409.10003",
    archivePrefix = "arXiv",
    primaryClass = "hep-th",
    school = "U. Autonoma, Madrid (main)",
    year = "2024"
}

@article{Bedroya:2024ubj,
    author = "Bedroya, Alek and Mishra, Rashmish K. and Wiesner, Max",
    title = "{Density of states, black holes and the Emergent String Conjecture}",
    eprint = "2405.00083",
    archivePrefix = "arXiv",
    primaryClass = "hep-th",
    doi = "10.1007/JHEP01(2025)144",
    journal = "JHEP",
    volume = "01",
    pages = "144",
    year = "2025"
}

@article{Calderon-Infante:2025ldq,
    author = "Calder{\'o}n-Infante, Jos{\'e} and Castellano, Alberto and Herr{\'a}ez, Alvaro",
    title = "{The double EFT expansion in quantum gravity}",
    eprint = "2501.14880",
    archivePrefix = "arXiv",
    primaryClass = "hep-th",
    reportNumber = "CERN-TH-2025-009, EFI-25-1, MPP-2025-6",
    doi = "10.21468/SciPostPhys.19.4.096",
    journal = "SciPost Phys.",
    volume = "19",
    number = "4",
    pages = "096",
    year = "2025"
}

@article{Caron-Huot:2024lbf,
    author = "Caron-Huot, Simon and Li, Yue-Zhou",
    title = "{Gravity and a universal cutoff for field theory}",
    eprint = "2408.06440",
    archivePrefix = "arXiv",
    primaryClass = "hep-th",
    doi = "10.1007/JHEP02(2025)115",
    journal = "JHEP",
    volume = "02",
    pages = "115",
    year = "2025"
}

@article{ValeixoBento:2025bmv,
    author = "Valeixo Bento, Bruno and Melo, Jo{\~a}o F.",
    title = "{EFT {\&} species scale: friends or foes?}",
    eprint = "2501.08230",
    archivePrefix = "arXiv",
    primaryClass = "hep-th",
    doi = "10.1007/JHEP05(2025)212",
    journal = "JHEP",
    volume = "05",
    pages = "212",
    year = "2025"
}

@article{Grana,
      author         = "Grana, Mariana and Minasian, Ruben and Petrini, Michela
                        and Tomasiello, Alessandro",
      title          = "{A Scan for new N=1 vacua on twisted tori}",
      journal        = "JHEP",
      volume         = "05",
      year           = "2007",
      pages          = "031",
      doi            = "10.1088/1126-6708/2007/05/031",
      eprint         = "hep-th/0609124",
      archivePrefix  = "arXiv",
      primaryClass   = "hep-th",
      reportNumber   = "SPHT-T06-094, SU-ITP-06-23",
      SLACcitation   = "%%CITATION = HEP-TH/0609124;%%"
}

@article{Grana2,
    author = "Grana, Mariana",
    title = "{Flux compactifications in string theory: A Comprehensive review}",
    eprint = "hep-th/0509003",
    archivePrefix = "arXiv",
    reportNumber = "LPTENS-05-26, CPHT-RR-049-0805",
    doi = "10.1016/j.physrep.2005.10.008",
    journal = "Phys. Rept.",
    volume = "423",
    pages = "91--158",
    year = "2006"
}

@article{Hull,
      author         = "Hull, C. M. and Townsend, P. K.",
      title          = "{Unity of superstring dualities}",
      journal        = "Nucl. Phys.",
      volume         = "B438",
      year           = "1995",
      pages          = "109-137",
      doi            = "10.1016/0550-3213(94)00559-W",
      note           = "[,236(1994)]",
      eprint         = "hep-th/9410167",
      archivePrefix  = "arXiv",
      primaryClass   = "hep-th",
      reportNumber   = "QMW-94-30, DAMTP-R-94-33",
      SLACcitation   = "%%CITATION = HEP-TH/9410167;%%"
}

@article{Hull4,
      author         = "Hull, C. M.",
      title          = "{Gauged D=9 supergravities and Scherk-Schwarz reduction}",
      journal        = "Class. Quant. Grav.",
      volume         = "21",
      year           = "2004",
      number         = "2",
      pages          = "509-516",
      doi            = "10.1088/0264-9381/21/2/014",
      eprint         = "hep-th/0203146",
      archivePrefix  = "arXiv",
      primaryClass   = "hep-th",
      reportNumber   = "QMW-PH-02-05",
      SLACcitation   = "%%CITATION = HEP-TH/0203146;%%"
}

@article{Hull5,
      author         = "Dabholkar, Atish and Hull, Chris",
      title          = "{Duality twists, orbifolds, and fluxes}",
      journal        = "JHEP",
      volume         = "09",
      year           = "2003",
      pages          = "054",
      doi            = "10.1088/1126-6708/2003/09/054",
      eprint         = "hep-th/0210209",
      archivePrefix  = "arXiv",
      primaryClass   = "hep-th",
      reportNumber   = "TIFR-TH-01-24, QMUL-PH-02-16",
      SLACcitation   = "%%CITATION = HEP-TH/0210209;%%"
}

@article{Hull8,
      author         = "Hull, C. M.",
      title          = "{Massive string theories from M theory and F theory}",
      journal        = "JHEP",
      volume         = "11",
      year           = "1998",
      pages          = "027",
      doi            = "10.1088/1126-6708/1998/11/027",
      eprint         = "hep-th/9811021",
      archivePrefix  = "arXiv",
      primaryClass   = "hep-th",
      reportNumber   = "QMW-PH-98-36, LPTENS-98-32",
      SLACcitation   = "%%CITATION = HEP-TH/9811021;%%"
}

@article{Kaloper2,
    author = "Kaloper, Nemanja and Khuri, Ramzi R. and Myers, Robert C.",
    title = "{On generalized axion reductions}",
    eprint = "hep-th/9803066",
    archivePrefix = "arXiv",
    reportNumber = "NSF-ITP-98-023, SU-ITP-98-08, MCGILL-98-04, QMW-PH-98-08",
    doi = "10.1016/S0370-2693(98)00409-2",
    journal = "Phys. Lett. B",
    volume = "428",
    pages = "297--302",
    year = "1998"
}

@article{Melgarejo,
      author         = "Fernandez-Melgarejo, J. J. and Ortin, T. and
                        Torrente-Lujan, E.",
      title          = "{The general gaugings of maximal d=9 supergravity}",
      journal        = "JHEP",
      volume         = "10",
      year           = "2011",
      pages          = "068",
      doi            = "10.1007/JHEP10(2011)068",
      eprint         = "1106.1760",
      archivePrefix  = "arXiv",
      primaryClass   = "hep-th",
      reportNumber   = "IFT-UAM-CSIC-11-18, UM-TH-11-09",
      SLACcitation   = "%%CITATION = ARXIV:1106.1760;%%"
}

@article{Meessen,
    author = "Gheerardyn, J. and Meessen, P.",
    title = "{Supersymmetry of massive D = 9 supergravity}",
    eprint = "hep-th/0111130",
    archivePrefix = "arXiv",
    reportNumber = "KUL-TF-01-21, SISSA-83-01-EP",
    doi = "10.1016/S0370-2693(01)01429-0",
    journal = "Phys. Lett. B",
    volume = "525",
    pages = "322--330",
    year = "2002"
}

@article{Aspinwall:1993yb,
	title        = {{Multiple mirror manifolds and topology change in string theory}},
	author       = {Aspinwall, Paul S. and Greene, Brian R. and Morrison, David R.},
	year         = 1993,
	journal      = {Phys. Lett. B},
	volume       = 303,
	pages        = {249--259},
	doi          = {10.1016/0370-2693(93)91428-P},
	eprint       = {hep-th/9301043},
	archiveprefix = {arXiv},
	reportnumber = {IASSNS-HEP-93-4}
}

@article{Atiyah:2000zz,
	title        = {{An M theory flop as a large N duality}},
	author       = {Atiyah, Michael and Maldacena, Juan Martin and Vafa, Cumrun},
	year         = 2001,
	journal      = {J. Math. Phys.},
	volume       = 42,
	pages        = {3209--3220},
	doi          = {10.1063/1.1376159},
	eprint       = {hep-th/0011256},
	archiveprefix = {arXiv},
	reportnumber = {HUTP-00-A045}
}

@article{Giveon:1994fu,
	title        = {{Target space duality in string theory}},
	author       = {Giveon, Amit and Porrati, Massimo and Rabinovici, Eliezer},
	year         = 1994,
	journal      = {Phys. Rept.},
	volume       = 244,
	pages        = {77--202},
	doi          = {10.1016/0370-1573(94)90070-1},
	eprint       = {hep-th/9401139},
	archiveprefix = {arXiv},
	reportnumber = {RI-1-94, NYU-TH-94-01-01}
}

@article{Greene:1990ud,
	title        = {{Duality in {Calabi-Yau} Moduli Space}},
	author       = {Greene, Brian R. and Plesser, M. R.},
	year         = 1990,
	journal      = {Nucl. Phys. B},
	volume       = 338,
	pages        = {15--37},
	doi          = {10.1016/0550-3213(90)90622-K},
	reportnumber = {HUTP-89-A043A, HUTMP-89-B245}
}

@article{Garcia-Etxebarria:2018ajm,
	title        = {{Dai-Freed anomalies in particle physics}},
	author       = {Garc\'\i{}a-Etxebarria, I\~naki and Montero, Miguel},
	year         = 2019,
	journal      = {JHEP},
	volume       = {08},
	pages        = {003},
	doi          = {10.1007/JHEP08(2019)003},
	eprint       = {1808.00009},
	archiveprefix = {arXiv},
	primaryclass = {hep-th},
	reportnumber = {MPP-2018-188}
}

@article{Giddings:1987cg,
	title        = {{Axion Induced Topology Change in Quantum Gravity and String Theory}},
	author       = {Giddings, Steven B. and Strominger, Andrew},
	year         = 1988,
	journal      = {Nucl. Phys. B},
	volume       = 306,
	pages        = {890--907},
	doi          = {10.1016/0550-3213(88)90446-4},
	reportnumber = {HUTP-87-A067}
}

@article{Witten:1993yc,
	title        = {{Phases of N=2 theories in two-dimensions}},
	author       = {Witten, Edward},
	year         = 1993,
	journal      = {Nucl. Phys. B},
	volume       = 403,
	pages        = {159--222},
	doi          = {10.1016/0550-3213(93)90033-L},
	editor       = {Greene, B. and Yau, Shing-Tung},
	eprint       = {hep-th/9301042},
	archiveprefix = {arXiv},
	reportnumber = {IASSNS-HEP-93-3}
}

@article{Pfaffle2000,
  author  = {Pf\"affle, Frank},
  title   = {The {D}irac spectrum of {B}ieberbach manifolds},
  journal = {Journal of Geometry and Physics},
  volume  = {35},
  number  = {4},
  pages   = {367--385},
  year    = {2000},
  doi     = {10.1016/S0393-0440(00)00008-8},
  publisher = {Elsevier}
}

@article{Montero,
    author = "Garc\'\i{}a Etxebarria, Iñaki and Montero, Miguel and Sousa, Kepa and Valenzuela, Irene",
    title = "{Nothing is certain in string compactifications}",
    eprint = "2005.06494",
    archivePrefix = "arXiv",
    primaryClass = "hep-th",
    doi = "10.1007/JHEP12(2020)032",
    journal = "JHEP",
    volume = "12",
    pages = "032",
    year = "2020"
}

@article{Mostow,
      author         = "Mostow, G.D.",
      title          = "{Factor spaces of solvable groups}",
      journal        = "Annals of Mathematics.",
      volume         = "60",
      year           = "1954",
      pages          = "1-27",
      doi            = "https://doi.org/10.2307/1969700",
}

@article{mpgm2,
      author         = "Garcia del Moral, M. P. and Pena, J. M. and Restuccia,
                        A.",
      title          = "{Supermembrane origin of type II gauged supergravities in
                        9D}",
      journal        = "JHEP",
      volume         = "09",
      year           = "2012",
      pages          = "063",
      doi            = "10.1007/JHEP09(2012)063",
      eprint         = "1203.2767",
      archivePrefix  = "arXiv",
      primaryClass   = "hep-th",
      reportNumber   = "FPAUO-12-02",
      SLACcitation   = "%%CITATION = ARXIV:1203.2767;%%"
}

@article{mpgm7,
      author         = "Garcia del Moral, Maria Pilar and Pena, Joselen M. and
                        Restuccia, Alvaro",
      title          = "{Classification of M2-brane 2-torus bundles, U-duality
                        invariance and type II gauged supergravities}",
      journal        = "Phys. Rev.",
      volume         = "D100",
      year           = "2019",
      number         = "2",
      pages          = "026005",
      doi            = "10.1103/PhysRevD.100.026005",
      eprint         = "1604.02579",
      archivePrefix  = "arXiv",
      primaryClass   = "hep-th",
      SLACcitation   = "%%CITATION = ARXIV:1604.02579;%%"
}

@article{mpgm10,
    author = "Garcia del Moral, M. P. and Las Heras, C. and Leon, P. and Pena, J. M. and Restuccia, A.",
    title = "{Fluxes, twisted tori, monodromy and $U(1)$ supermembranes}",
    eprint = "2005.06397",
    archivePrefix = "arXiv",
    primaryClass = "hep-th",
    doi = "10.1007/JHEP09(2020)097",
    journal = "JHEP",
    volume = "09",
    pages = "097",
    year = "2020"
}

@article{mpgm23,
    author = "del Moral, Maria Pilar Garcia and las Heras, Camilo and Restuccia, Alvaro",
    title = "{Type IIB parabolic (p, q)-strings from M2-branes with fluxes}",
    eprint = "2201.04896",
    archivePrefix = "arXiv",
    primaryClass = "hep-th",
    doi = "10.1007/JHEP03(2023)143",
    journal = "JHEP",
    volume = "03",
    pages = "143",
    year = "2023"
}

@article{mpgm24,
    author = "del Moral, Maria Pilar Garcia and las Heras, Camilo and Restuccia, Alvaro",
    title = "{U-duality in quantum M2-branes and gauged supergravities}",
    eprint = "2405.02775",
    archivePrefix = "arXiv",
    primaryClass = "hep-th",
    reportNumber = "IFT-UAM/CSIC-24-68",
    doi = "10.1007/JHEP12(2024)163",
    journal = "JHEP",
    volume = "12",
    pages = "163",
    year = "2024"
}

@article{Neron,
     author = {N\'eron, Andr\'e},
     title = {Mod\`eles minimaux des vari\'et\'es ab\'eliennes sur les corps locaux et globaux},
     journal = {Publications Math\'ematiques de l'IH\'ES},
     pages = {5--128},
     publisher = {Institut des Hautes \'Etudes Scientifiques},
     volume = {21},
     year = {1964},
     zbl = {0132.41403},
     mrnumber = {31 #3423},
     language = {fr},
     url = {http://www.numdam.org/item/PMIHES_1964__21__5_0/}
}

@article{Ortin,
    author = "Meessen, Patrick and Ortin, Tomas",
    title = "{An Sl(2,Z) multiplet of nine-dimensional type II supergravity theories}",
    eprint = "hep-th/9806120",
    archivePrefix = "arXiv",
    reportNumber = "IFT-UAM-CSIC-98-3",
    doi = "10.1016/S0550-3213(98)00780-9",
    journal = "Nucl. Phys. B",
    volume = "541",
    pages = "195--245",
    year = "1999"
}

@article{Pope,
	doi = {10.1088/0264-9381/15/8/008},
	url = {https://doi.org/10.1088%2F0264-9381%2F15%2F8%2F008},
	year = 1998,
	month = {aug},
	publisher = {{IOP} Publishing},
	volume = {15},
	number = {8},
	pages = {2239--2256},
	author = {I V Lavrinenko and H Lü and C N Pope},
	title = {Fibre bundles and generalized dimensional reductions},
	journal = {Classical and Quantum Gravity},
}

@article{Ruiz:2024jiz,
    author = "Ruiz, Ignacio",
    title = "{Morse-Bott inequalities, topology change and cobordisms to nothing}",
    eprint = "2410.21372",
    archivePrefix = "arXiv",
    primaryClass = "hep-th",
    reportNumber = "IFT-UAM/CSIC-24-154",
    doi = "10.1007/JHEP06(2025)030",
    journal = "JHEP",
    volume = "06",
    pages = "030",
    year = "2025"
}

@article{Samtleben,
      author         = "Samtleben, Henning",
      title          = "{Lectures on Gauged Supergravity and Flux
                        Compactifications}",
      booktitle      = "{Strings, supergravity and gauge theories. Proceedings,
                        European RTN Winter School, CERN, Geneva, Switzerland,
                        January 21-25, 2008}",
      journal        = "Class. Quant. Grav.",
      volume         = "25",
      year           = "2008",
      pages          = "214002",
      doi            = "10.1088/0264-9381/25/21/214002",
      eprint         = "0808.4076",
      archivePrefix  = "arXiv",
      primaryClass   = "hep-th",
      reportNumber   = "ENSL-00315624",
      SLACcitation   = "%%CITATION = ARXIV:0808.4076;%%"
}

@article{Schwarz5,
      author         = "Scherk, Joel and Schwarz, John H.",
      title          = "{How to Get Masses from Extra Dimensions}",
      journal        = "Nucl. Phys.",
      volume         = "B153",
      year           = "1979",
      pages          = "61-88",
      doi            = "10.1016/0550-3213(79)90592-3",
      note           = "[,79(1979)]",
      reportNumber   = "LPTENS-79-2",
      SLACcitation   = "%%CITATION = NUPHA,B153,61;%%"
}

@article{Schwarz6,
    author = "Schwarz, John H.",
    title = "{An SL(2,Z) multiplet of type IIB superstrings}",
    eprint = "hep-th/9508143",
    archivePrefix = "arXiv",
    reportNumber = "CALT-68-2013",
    doi = "10.1016/0370-2693(95)01405-5",
    journal = "Phys. Lett. B",
    volume = "360",
    pages = "13--18",
    year = "1995",
    note = "[Erratum: Phys.Lett.B 364, 252 (1995)]"
}

@article{Silverstein,
   title={Simple de Sitter solutions},
   volume={77},
   ISSN={1550-2368},
   url={http://dx.doi.org/10.1103/PhysRevD.77.106006},
   DOI={10.1103/physrevd.77.106006},
   number={10},
   journal={Physical Review D},
   publisher={American Physical Society (APS)},
   author={Silverstein, Eva},
   year={2008},
   month=may }

@article{Zwiebach,
    author = {DeWolfe, Oliver and Hauer, Tamas and Iqbal, Amer and Zwiebach, Barton"},
    title = "{Uncovering the symmetries on [p,q] seven-branes: Beyond the Kodaira classification}",
    eprint = "hep-th/9812028",
    archivePrefix = "arXiv",
    reportNumber = "MIT-CTP-2804",
    doi = "10.4310/ATMP.1999.v3.n6.a5",
    journal = "Adv. Theor. Math. Phys.",
    volume = "3",
    pages = "1785--1833",
    year = "1999"
}

@article{Zwiebach2,
    author = "DeWolfe, Oliver and Hauer, Tamas and Iqbal, Amer and Zwiebach, Barton",
    title = "{Uncovering infinite symmetries on [p, q] 7-branes: Kac-Moody algebras and beyond}",
    eprint = "hep-th/9812209",
    archivePrefix = "arXiv",
    reportNumber = "MIT-CTP-2808",
    doi = "10.4310/ATMP.1999.v3.n6.a6",
    journal = "Adv. Theor. Math. Phys.",
    volume = "3",
    pages = "1835--1891",
    year = "1999"
}

\end{document}